\documentclass[acmsmall, screen, authorversion, nonacm]{acmart}

\setcitestyle{nosort}

\usepackage[bw]{agda}

\renewcommand{\AgdaKeywordFontStyle}[1]{%
  \ensuremath{\textsf{\textbf{#1}}}}
\renewcommand{\AgdaFontStyle}[1]{%
  \ensuremath{\mathit{#1}}}
\renewcommand{\AgdaBoundFontStyle}[1]{%
  \ensuremath{\mathit{#1}}}
\renewcommand{\AgdaSymbol}[1]{%
  \AgdaNoSpaceMath{\ensuremath{#1}}}
\renewcommand{\AgdaNumber}[1]{%
  \AgdaNoSpaceMath{\ensuremath{#1}}}
\renewcommand{\AgdaInductiveConstructor}[1]{%
  \AgdaNoSpaceMath{\ensuremath{\mathsf{\AgdaFormat{#1}{#1}}}}}
\renewcommand{\AgdaField}[1]{%
  \AgdaNoSpaceMath{\ensuremath{\mathsf{\AgdaFormat{#1}{#1}}}}}

\usepackage{newunicodechar}
\newunicodechar{▸}{\ensuremath{\mathrel{\blacktriangleright}}}
\newunicodechar{◯}{\ensuremath{\bigcirc}}
\newunicodechar{∞}{\ensuremath{\infty}}
\newunicodechar{‐}{\textit{-}}

\usepackage{mathpartir}

\usepackage{polytable}

\newcommand{\Erased}{\ensuremath{\AgdaDatatype{Erased}}}
\newcommand{\boxop}{%
  $\AgdaOperator{\AgdaInductiveConstructor{[\AgdaUnderscore{}]}}$}
\newcommand{\boxcong}{\ensuremath{\AgdaFunction{[]‐cong}}}

\newcommand{\MCbreak}{\vspace{-1.8ex}\\}

\specialcomment{acks}{%
  \begingroup
  \section*{Acknowledgements}
  \phantomsection\addcontentsline{toc}{section}{Acknowledgements}
}{%
  \endgroup
}
\makeatletter
\if@ACM@anonymous
  \excludecomment{acks}
\fi
\makeatother

\makeatletter
\newcommand{\ifAnonymous}[2]{\if@ACM@anonymous #1\else #2\fi}
\makeatother

\usepackage{xstring}
\newcommand{\ifAppendicesIncluded}[2]{%
  \IfBeginWith*{\jobname}{Paper-with-appendices}{#1}{#2}}
\newcommand{\appendixA}{%
  Appendix~\ifAppendicesIncluded{\ref{sec:full-typing-rules}}{A}}
\newcommand{\appendixB}{%
  Appendix~\ifAppendicesIncluded{\ref{sec:target-language}}{B}}

\setcopyright{cc}
\copyrightyear{2026}

\newcommand{\excludeAllowed}[1]{}

\newcommand{\wfCtxt}[1]{⊢ #1}
\newcommand{\wfTy}[2]{#1 ⊢ #2}
\newcommand{\wrTy}[4]{#1 ▸ #2 ⊢^{#3} #4}
\newcommand{\wfTm}[3]{#1 ⊢ #2 : #3}
\newcommand{\wrTm}[5]{#1 ▸ #2 ⊢ #3 :^{\!#4} #5}
\newcommand{\wrTmExt}[5]{#1 ▸ #2 ⊢_{\mathrm{E}} #3 :^{\!#4} #5}
\newcommand{\eqTy}[3]{#1 ⊢ #2 ≡ #3}
\newcommand{\eqTm}[4]{#1 ⊢ #2 ≡ #3 : #4}
\newcommand{\eqTmExt}[4]{#1 ⊢_{\mathrm{E}} #2 ≡ #3 : #4}
\newcommand{\eqTmR}[4]{#1 ⊢_{\mathrm{R}} #2 ≡ #3 : #4}
\newcommand{\wfSubst}[3]{#1 ⊢ #2 : #3}
\newcommand{\wrSubst}[5]{#1 ▸ #2 ⊢ #3 :^{\!#4} #5}
\newcommand{\wrSubstExt}[5]{#1 ▸ #2 ⊢_{\mathrm{E}} #3 :^{\!#4} #5}
\newcommand{\wrSubstR}[5]{#1 ▸ #2 ⊢_{\mathrm{R}} #3 :^{\!#4} #5}
\newcommand{\wfWk}[3]{#1 ⊢ #2 : #3}
\newcommand{\appSubstExt}[2]{#1\,#2}
\newcommand{\trName}{\mathit{tr}}
\newcommand{\tr}[1]{\trName\ #1}
\newcommand{\varTy}[3]{#1 : #2 ∈ #3}
\newcommand{\lookup}[2]{#1(#2)}
\newcommand{\leqG}[2]{#1 ≤ #2}
\newcommand{\redTy}[3]{#1 ⊢ #2 ⇒ #3}
\newcommand{\redsTy}[3]{#1 ⊢ #2 ⇒^{*} #3}
\newcommand{\redTm}[4]{#1 ⊢ #2 ⇒ #3 : #4}
\newcommand{\redsTm}[4]{#1 ⊢ #2 ⇒^{*} #3 : #4}
\newcommand{\redSucName}{⇒^{\text{suc}}}
\newcommand{\redsSucName}{⇒^{\text{suc}*}}
\newcommand{\redSuc}[3]{#1 ⊢ #2 \redSucName #3 : \Nat}
\newcommand{\redsSuc}[3]{#1 ⊢ #2 \redsSucName #3 : \Nat}

\newcommand{\var}[1]{#1}
\newcommand{\emptyC}{ε}
\newcommand{\consC}[2]{#1 , #2}
\newcommand{\wk}[2]{#1\ #2}
\newcommand{\wkO}[1]{{#1}{↑}}
\newcommand{\wkn}[2]{{#2}{↑^{#1}}}
\newcommand{\wkT}[1]{\wkn{2}{#1}}
\newcommand{\wklift}[3]{{#3}{↑^{#1}_{#2}}}
\newcommand{\head}[1]{\mathit{head}\ #1}
\newcommand{\tail}[1]{\mathit{tail}\ #1}
\newcommand{\subst}[2]{#1\ #2}
\newcommand{\substZ}[2]{#1 [#2/\var{0}]}
\newcommand{\substO}[3]{#1 [#2/\var{1}, #3/\var{0}]}
\newcommand{\substZUp}[2]{#1 [\wkO{\mathsf{\mathsf{id}}},#2/\var{0}]}
\newcommand{\substZUpn}[3]{#2 [\wkn{#1}{\mathsf{\mathsf{id}}},#3/\var{0}]}
\newcommand{\substZUpT}[2]{\substZUpn{2}{#1}{#2}}

\newcommand{\zeroG}{0}
\newcommand{\omegaG}{ω}
\newcommand{\mul}[2]{#1#2}

\newcommand{\emptyGC}{ε}
\newcommand{\consGC}[2]{#1 , #2}
\newcommand{\zeroGC}{\mathbf{0}}

\newcommand{\varZ}{0}
\newcommand{\varO}{1}
\newcommand{\varSuc}[1]{#1 + 1}
\newcommand{\UName}{\mathsf{U}}
\newcommand{\U}[1]{\UName_{#1}}
\newcommand{\zeroL}{0}
\newcommand{\sucL}[1]{#1 + 1}
\newcommand{\maxL}[2]{#1 ⊔ #2}
\newcommand{\PiWithLevel}[1]{Π_{#1}}
\newcommand{\PiT}[3]{\PiWithLevel{#1}\,#2\,#3}
\newcommand{\lam}[2]{λ^{#1}\,#2}
\newcommand{\app}[3]{#1{}^{#2}#3}
\newcommand{\SigmaWith}[2]{Σ_{#1,#2}}
\newcommand{\SigmaT}[4]{\SigmaWith{#1}{#2}\,#3\,#4}
\newcommand{\prodTm}[4]{(#3,#4)_{#1,#2}}
\newcommand{\fstWith}[1]{\mathsf{fst}_{#1}}
\newcommand{\fst}[2]{\fstWith{#1}\,#2}
\newcommand{\snd}[2]{\mathsf{snd}_{#1}\,#2}
\newcommand{\prodrecName}{\mathsf{prodrec}}
\newcommand{\prodrecWith}[2]{\prodrecName_{#1,#2}}
\newcommand{\prodrec}[5]{\prodrecWith{#1}{#2}\,#3\,#4\,#5}
\newcommand{\IdName}{\mathsf{Id}}
\newcommand{\IdT}[3]{\IdName\,#1\,#2\,#3}
\newcommand{\IdRules}[1]{\IdName\,#1}
\newcommand{\rfl}{\mathsf{refl}}
\newcommand{\JName}{\mathsf{J}}
\newcommand{\JWithGrades}[2]{\JName_{#1,#2}}
\newcommand{\J}[8]{\JWithGrades{#1}{#2}\,#3\,#4\,#5\,#6\,#7\,#8}
\newcommand{\JRules}[1]{\JName\,#1}
\newcommand{\KName}{\mathsf{K}}
\newcommand{\K}[6]{\KName_{#1}\,#2\,#3\,#4\,#5\,#6}
\newcommand{\KRules}[1]{\KName\,#1}
\newcommand{\EName}{\mathsf{Erased}}
\newcommand{\E}[2]{\EName_{#1}\,#2}
\newcommand{\ES}[3]{\EName_{#1,#2}\,#3}
\newcommand{\bx}[1]{[#1]}
\newcommand{\bxS}[2]{[#2]_{#1}}

\newcommand{\erasedName}{\mathsf{erased}}
\newcommand{\erased}[2]{\erasedName\,#1\,#2}
\newcommand{\bcName}{\mathsf{[]‐cong}}
\newcommand{\bcWithLevel}[1]{\bcName_{#1}}
\newcommand{\bcSWith}[1]{\bcName_{#1}}
\newcommand{\bc}[5]{\bcWithLevel{#1}\,#2\,#3\,#4\,#5}
\newcommand{\bcS}[6]{\bcName_{#1,#2}\,#3\,#4\,#5\,#6}
\newcommand{\QuotName}{{\text{\textunderscore{}}/\text{\textunderscore{}}}}
\newcommand{\Quot}[2]{#1 / #2}
\newcommand{\className}{\mathsf{class}}
\newcommand{\class}[1]{\className{}\,#1}
\newcommand{\respName}{\mathsf{resp}}
\newcommand{\resp}[5]{\respName{}\,#1\,#2\,#3\,#4\,#5}
\newcommand{\setName}{\mathsf{set}}
\newcommand{\set}[6]{\setName{}\,#1\,#2\,#3\,#4\,#5\,#6}
\newcommand{\qrecName}{\mathsf{qrec}}
\newcommand{\qrec}[5]{\qrecName{}\,#1\,#2\,#3\,#4\,#5}

\newcommand{\respType}[4]{\mathsf{Resp‐type}\,#1\,#2\,#3\,#4}
\newcommand{\LiftWith}[1]{\mathsf{Lift}_{#1}}
\newcommand{\LiftT}[2]{\LiftWith{#1}\,#2}
\newcommand{\lift}[1]{\mathsf{lift}\,#1}
\newcommand{\lowerTm}[1]{\mathsf{lower}\,#1}
\newcommand{\Empty}{⊥}
\newcommand{\emptyrecName}{\mathsf{emptyrec}}
\newcommand{\emptyrecWithLevel}[1]{\emptyrecName_{#1}}
\newcommand{\emptyrec}[3]{\emptyrecWithLevel{#1}\,#2\,#3}
\newcommand{\weak}{\mathbf{w}}
\newcommand{\strongS}{\mathbf{s}}
\newcommand{\Unit}[1]{⊤_{#1}}
\newcommand{\starTm}[1]{\star_{#1}}
\newcommand{\unitrecName}{\mathsf{unitrec}}
\newcommand{\unitrecWithLevel}[1]{\unitrecName_{#1}}
\newcommand{\unitrec}[4]{\unitrecWithLevel{#1}\,#2\,#3\,#4}
\newcommand{\Nat}{\mathsf{Nat}}
\newcommand{\zero}{\mathsf{zero}}
\newcommand{\suc}[1]{\mathsf{suc}\,#1}
\newcommand{\natrec}[6]{\mathsf{natrec}_{#1,#2}\,#3\,#4\,#5\,#6}
\newcommand{\bcVar}{\mathsf{bc}}
\newcommand{\bcVarZero}{\mathsf{bc}_0}
\newcommand{\JEName}{\mathsf{J^E}}
\newcommand{\JE}[6]{\JEName\,#1\,#2\,#3\,#4\,#5\,#6}

\newcommand{\hastype}[2]{{#1:#2}}
\newcommand{\identifier}[1]{\mathit{#1}}

\newcommand{\varT}[1]{\textsf{var}\,#1}
\newcommand{\starT}{\textsf{star}}
\newcommand{\unitrecT}[2]{\textsf{unitrec}\,#1\,#2}
\newcommand{\lamT}[1]{\textsf{lam}\,#1}
\newcommand{\appT}[2]{\textsf{app}\,#1\,#2}
\newcommand{\prodT}[2]{\textsf{pair}\,#1\,#2}
\newcommand{\fstT}[1]{\textsf{fst}\,#1}
\newcommand{\sndT}[1]{\textsf{snd}\,#1}
\newcommand{\prodrecT}[2]{\textsf{prodrec}\,#1\,#2}
\newcommand{\zeroName}{\mathsf{zero}}
\newcommand{\sucName}{\mathsf{suc}}
\newcommand{\sucT}[1]{\sucName\,#1}
\newcommand{\natrecT}[3]{\textsf{natrec}\,#1\,#2\,#3}
\newcommand{\dummy}{{↯}}

\newcommand{\stepName}{\Rightarrow}
\newcommand{\stepsName}{\Rightarrow^{*}}
\newcommand{\step}[2]{#1 \stepName #2}
\newcommand{\steps}[2]{#1 \stepsName #2}

\newcommand{\numeralAsTerm}[1]{\underline{n}}
\newcommand{\numeralAsTermExplanation}[1]{\sucName^{#1}\,\zeroName}

\newcommand{\extract}[1]{\mathit{erase}\ #1}
\newcommand{\loopT}{\mathit{loop}}

\newcommand{\extractExtName}{\mathit{erase}_{\mathrm{E}}}
\newcommand{\extractExt}[1]{\extractExtName\ #1}

\newcommand{\theTitle}{%
  Erased Postulates, Identity Types and Quotients}

\begin{code}[hide]%
\>[0]\AgdaSymbol{\{-\#}\AgdaSpace{}%
\AgdaKeyword{OPTIONS}\AgdaSpace{}%
\AgdaPragma{--cubical}\AgdaSpace{}%
\AgdaSymbol{\#-\}}\<%
\\
\\[\AgdaEmptyExtraSkip]%
\>[0]\AgdaKeyword{module}\AgdaSpace{}%
\AgdaModule{Cubical}\AgdaSpace{}%
\AgdaKeyword{where}\<%
\\
\\[\AgdaEmptyExtraSkip]%
\>[0]\AgdaKeyword{import}\AgdaSpace{}%
\AgdaModule{Equality.Path}\<%
\\
\>[0]\AgdaKeyword{open}\AgdaSpace{}%
\AgdaKeyword{import}\AgdaSpace{}%
\AgdaModule{Prelude}\<%
\\
\\[\AgdaEmptyExtraSkip]%
\>[0]\AgdaKeyword{private}\AgdaSpace{}%
\AgdaKeyword{variable}\<%
\\
\>[0][@{}l@{\AgdaIndent{0}}]%
\>[2]\AgdaGeneralizable{a}%
\>[6]\AgdaSymbol{:}\AgdaSpace{}%
\AgdaPostulate{Level}\<%
\\
\>[2]\AgdaGeneralizable{A}%
\>[6]\AgdaSymbol{:}\AgdaSpace{}%
\AgdaPrimitive{Type}\AgdaSpace{}%
\AgdaGeneralizable{a}\<%
\\
\>[2]\AgdaGeneralizable{x}\AgdaSpace{}%
\AgdaGeneralizable{y}\AgdaSpace{}%
\AgdaSymbol{:}\AgdaSpace{}%
\AgdaGeneralizable{A}\<%
\\
\\[\AgdaEmptyExtraSkip]%
\>[0]\AgdaFunction{Path}\AgdaSpace{}%
\AgdaSymbol{:}\AgdaSpace{}%
\AgdaSymbol{(}\AgdaBound{A}\AgdaSpace{}%
\AgdaSymbol{:}\AgdaSpace{}%
\AgdaPrimitive{Type}\AgdaSpace{}%
\AgdaGeneralizable{a}\AgdaSymbol{)}\AgdaSpace{}%
\AgdaSymbol{→}\AgdaSpace{}%
\AgdaBound{A}\AgdaSpace{}%
\AgdaSymbol{→}\AgdaSpace{}%
\AgdaBound{A}\AgdaSpace{}%
\AgdaSymbol{→}\AgdaSpace{}%
\AgdaPrimitive{Type}\AgdaSpace{}%
\AgdaGeneralizable{a}\<%
\\
\>[0]\AgdaFunction{Path}\AgdaSpace{}%
\AgdaSymbol{\AgdaUnderscore{}}\AgdaSpace{}%
\AgdaSymbol{=}\AgdaSpace{}%
\AgdaOperator{\AgdaFunction{Equality.Path.\AgdaUnderscore{}≡\AgdaUnderscore{}}}\<%
\\
\\[\AgdaEmptyExtraSkip]%
\>[0]\AgdaKeyword{data}\AgdaSpace{}%
\AgdaDatatype{Erased}\AgdaSpace{}%
\AgdaSymbol{(}\AgdaSymbol{@0}\AgdaSpace{}%
\AgdaBound{A}\AgdaSpace{}%
\AgdaSymbol{:}\AgdaSpace{}%
\AgdaPrimitive{Type}\AgdaSpace{}%
\AgdaGeneralizable{a}\AgdaSymbol{)}\AgdaSpace{}%
\AgdaSymbol{:}\AgdaSpace{}%
\AgdaPrimitive{Type}\AgdaSpace{}%
\AgdaBound{a}\AgdaSpace{}%
\AgdaKeyword{where}\<%
\\
\>[0][@{}l@{\AgdaIndent{0}}]%
\>[2]\AgdaOperator{\AgdaInductiveConstructor{[\AgdaUnderscore{}]}}\AgdaSpace{}%
\AgdaSymbol{:}\AgdaSpace{}%
\AgdaSymbol{@0}\AgdaSpace{}%
\AgdaBound{A}\AgdaSpace{}%
\AgdaSymbol{→}\AgdaSpace{}%
\AgdaDatatype{Erased}\AgdaSpace{}%
\AgdaBound{A}\<%
\end{code}

\begin{code}[hide]%
\>[0]\AgdaKeyword{postulate}\<%
\\
\>[0][@{}l@{\AgdaIndent{0}}]%
\>[2]\AgdaPostulate{\AgdaUnderscore{}}%
\>[46I]\AgdaSymbol{:}\<%
\end{code}
\newcommand{\PathAxy}{%
\begin{code}[inline]%
\>[46I][@{}l@{\AgdaIndent{1}}]%
\>[6]\AgdaFunction{Path}\AgdaSpace{}%
\AgdaGeneralizable{A}\AgdaSpace{}%
\AgdaGeneralizable{x}\AgdaSpace{}%
\AgdaGeneralizable{y}\<%
\end{code}}

\newcommand{\boxcongPath}{%
\begin{code}%
\>[0]\AgdaFunction{[\ensuremath{\mkern1.5mu}]‐cong‐Path}\AgdaSpace{}%
\AgdaSymbol{:}\AgdaSpace{}%
\AgdaSymbol{\{}\AgdaSymbol{@0}\AgdaSpace{}%
\AgdaBound{A}\AgdaSpace{}%
\AgdaSymbol{:}\AgdaSpace{}%
\AgdaPrimitive{Type}\AgdaSpace{}%
\AgdaGeneralizable{a}\AgdaSymbol{\}}\AgdaSpace{}%
\AgdaSymbol{\{}\AgdaSymbol{@0}\AgdaSpace{}%
\AgdaBound{x}\AgdaSpace{}%
\AgdaBound{y}\AgdaSpace{}%
\AgdaSymbol{:}\AgdaSpace{}%
\AgdaBound{A}\AgdaSymbol{\}}\AgdaSpace{}%
\AgdaSymbol{→}\AgdaSpace{}%
\AgdaSymbol{@0}\AgdaSpace{}%
\AgdaFunction{Path}\AgdaSpace{}%
\AgdaBound{A}\AgdaSpace{}%
\AgdaBound{x}\AgdaSpace{}%
\AgdaBound{y}\AgdaSpace{}%
\AgdaSymbol{→}\AgdaSpace{}%
\AgdaFunction{Path}\AgdaSpace{}%
\AgdaSymbol{(}\AgdaDatatype{Erased}\AgdaSpace{}%
\AgdaBound{A}\AgdaSymbol{)}\AgdaSpace{}%
\AgdaOperator{\AgdaInductiveConstructor{[}}\AgdaSpace{}%
\AgdaBound{x}\AgdaSpace{}%
\AgdaOperator{\AgdaInductiveConstructor{]}}\AgdaSpace{}%
\AgdaOperator{\AgdaInductiveConstructor{[}}\AgdaSpace{}%
\AgdaBound{y}\AgdaSpace{}%
\AgdaOperator{\AgdaInductiveConstructor{]}}\<%
\\
\>[0]\AgdaFunction{[\ensuremath{\mkern1.5mu}]‐cong‐Path}\AgdaSpace{}%
\AgdaBound{eq}\AgdaSpace{}%
\AgdaSymbol{=}\AgdaSpace{}%
\AgdaSymbol{λ}\AgdaSpace{}%
\AgdaBound{i}\AgdaSpace{}%
\AgdaSymbol{→}\AgdaSpace{}%
\AgdaOperator{\AgdaInductiveConstructor{[}}\AgdaSpace{}%
\AgdaBound{eq}\AgdaSpace{}%
\AgdaBound{i}\AgdaSpace{}%
\AgdaOperator{\AgdaInductiveConstructor{]}}\<%
\end{code}}

\begin{code}[hide]%
\>[0]\AgdaKeyword{import}\AgdaSpace{}%
\AgdaModule{Equality}\AgdaSpace{}%
\AgdaSymbol{as}\AgdaSpace{}%
\AgdaModule{E}\<%
\\
\\[\AgdaEmptyExtraSkip]%
\>[0]\AgdaKeyword{module}\AgdaSpace{}%
\AgdaModule{Paper}\<%
\\
\>[0][@{}l@{\AgdaIndent{0}}]%
\>[2]\AgdaSymbol{\{}\AgdaBound{e⁺}\AgdaSymbol{\}}\AgdaSpace{}%
\AgdaSymbol{(}\AgdaBound{eq‐J}\AgdaSpace{}%
\AgdaSymbol{:}\AgdaSpace{}%
\AgdaSymbol{∀}\AgdaSpace{}%
\AgdaSymbol{\{}\AgdaBound{a}\AgdaSpace{}%
\AgdaBound{p}\AgdaSymbol{\}}\AgdaSpace{}%
\AgdaSymbol{→}\AgdaSpace{}%
\AgdaRecord{E.Equality‐with‐J}\AgdaSpace{}%
\AgdaBound{a}\AgdaSpace{}%
\AgdaBound{p}\AgdaSpace{}%
\AgdaBound{e⁺}\AgdaSymbol{)}\AgdaSpace{}%
\AgdaKeyword{where}\<%
\\
\\[\AgdaEmptyExtraSkip]%
\>[0]\AgdaKeyword{open}\AgdaSpace{}%
\AgdaModule{E.Derived-definitions-and-properties}\AgdaSpace{}%
\AgdaBound{eq‐J}\<%
\\
\>[0][@{}l@{\AgdaIndent{0}}]%
\>[2]\AgdaKeyword{using}\AgdaSpace{}%
\AgdaSymbol{(}\AgdaFunction{Decidable‐equality}\AgdaSymbol{)}\<%
\\
\\[\AgdaEmptyExtraSkip]%
\>[0]\AgdaKeyword{open}\AgdaSpace{}%
\AgdaKeyword{import}\AgdaSpace{}%
\AgdaModule{Logical-equivalence}\AgdaSpace{}%
\AgdaKeyword{using}\AgdaSpace{}%
\AgdaSymbol{(}\AgdaOperator{\AgdaRecord{\AgdaUnderscore{}⇔\AgdaUnderscore{}}}\AgdaSymbol{)}\<%
\\
\>[0]\AgdaKeyword{open}\AgdaSpace{}%
\AgdaKeyword{import}\AgdaSpace{}%
\AgdaModule{Prelude}\<%
\\
\>[0][@{}l@{\AgdaIndent{0}}]%
\>[2]\AgdaKeyword{using}\<%
\\
\>[2][@{}l@{\AgdaIndent{0}}]%
\>[4]\AgdaSymbol{(}\AgdaPostulate{Level}\AgdaSymbol{;}\AgdaSpace{}%
\AgdaPrimitive{lsuc}\AgdaSymbol{;}\AgdaSpace{}%
\AgdaOperator{\AgdaPrimitive{\AgdaUnderscore{}⊔\AgdaUnderscore{}}}\AgdaSymbol{;}\AgdaSpace{}%
\AgdaPrimitive{Type}\AgdaSymbol{;}\AgdaSpace{}%
\AgdaOperator{\AgdaFunction{\AgdaUnderscore{}∘\AgdaUnderscore{}}}\AgdaSymbol{;}\AgdaSpace{}%
\AgdaRecord{⊤}\AgdaSymbol{;}\AgdaSpace{}%
\AgdaInductiveConstructor{tt}\AgdaSymbol{;}\AgdaSpace{}%
\AgdaOperator{\AgdaFunction{¬\AgdaUnderscore{}}}\AgdaSymbol{;}\AgdaSpace{}%
\AgdaDatatype{ℕ}\AgdaSymbol{;}\AgdaSpace{}%
\AgdaInductiveConstructor{zero}\AgdaSymbol{;}\AgdaSpace{}%
\AgdaInductiveConstructor{suc}\AgdaSymbol{;}\AgdaSpace{}%
\AgdaOperator{\AgdaPrimitive{\AgdaUnderscore{}+\AgdaUnderscore{}}}\AgdaSymbol{;}\AgdaSpace{}%
\AgdaOperator{\AgdaDatatype{\AgdaUnderscore{}⊎\AgdaUnderscore{}}}\AgdaSymbol{;}\<%
\\
\>[4][@{}l@{\AgdaIndent{0}}]%
\>[5]\AgdaFunction{Bool}\AgdaSymbol{;}\AgdaSpace{}%
\AgdaInductiveConstructor{true}\AgdaSymbol{;}\AgdaSpace{}%
\AgdaInductiveConstructor{false}\AgdaSymbol{;}\AgdaSpace{}%
\AgdaRecord{Σ}\AgdaSymbol{;}\AgdaSpace{}%
\AgdaFunction{∃}\AgdaSymbol{;}\AgdaSpace{}%
\AgdaOperator{\AgdaFunction{\AgdaUnderscore{}×\AgdaUnderscore{}}}\AgdaSymbol{;}\AgdaSpace{}%
\AgdaOperator{\AgdaInductiveConstructor{\AgdaUnderscore{},\AgdaUnderscore{}}}\AgdaSymbol{;}\AgdaSpace{}%
\AgdaDatatype{List}\AgdaSymbol{)}\<%
\\
\>[2]\AgdaKeyword{renaming}\AgdaSpace{}%
\AgdaSymbol{(}\AgdaField{proj\ensuremath{{}_{\mathrm{2}}}}\AgdaSpace{}%
\AgdaSymbol{to}\AgdaSpace{}%
\AgdaField{snd}\AgdaSymbol{;}\AgdaSpace{}%
\AgdaFunction{⊥\ensuremath{{}_{\mathrm{0}}}}\AgdaSpace{}%
\AgdaSymbol{to}\AgdaSpace{}%
\AgdaFunction{⊥}\AgdaSymbol{)}\<%
\\
\>[0]\AgdaKeyword{open}\AgdaSpace{}%
\AgdaKeyword{import}\AgdaSpace{}%
\AgdaModule{Embedding}\AgdaSpace{}%
\AgdaBound{eq‐J}\AgdaSpace{}%
\AgdaKeyword{using}\AgdaSpace{}%
\AgdaSymbol{(}\AgdaFunction{Is‐embedding}\AgdaSymbol{)}\<%
\\
\>[0]\AgdaKeyword{open}\AgdaSpace{}%
\AgdaKeyword{import}\AgdaSpace{}%
\AgdaModule{Equivalence}\AgdaSpace{}%
\AgdaBound{eq‐J}\AgdaSpace{}%
\AgdaSymbol{as}\AgdaSpace{}%
\AgdaModule{Eq}\AgdaSpace{}%
\AgdaKeyword{using}\AgdaSpace{}%
\AgdaSymbol{(}\AgdaFunction{Is‐equivalence}\AgdaSymbol{)}\<%
\\
\>[0]\AgdaKeyword{open}\AgdaSpace{}%
\AgdaKeyword{import}\AgdaSpace{}%
\AgdaModule{Function-universe}\AgdaSpace{}%
\AgdaBound{eq‐J}\AgdaSpace{}%
\AgdaKeyword{using}\AgdaSpace{}%
\AgdaSymbol{(}\AgdaDatatype{Kind}\AgdaSymbol{;}\AgdaSpace{}%
\AgdaOperator{\AgdaFunction{\AgdaUnderscore{}↝[\AgdaUnderscore{}]\AgdaUnderscore{}}}\AgdaSymbol{)}\<%
\\
\>[0]\AgdaKeyword{open}\AgdaSpace{}%
\AgdaKeyword{import}\AgdaSpace{}%
\AgdaModule{H-level}\AgdaSpace{}%
\AgdaBound{eq‐J}\AgdaSpace{}%
\AgdaKeyword{using}\AgdaSpace{}%
\AgdaSymbol{(}\AgdaFunction{H‐level}\AgdaSymbol{)}\<%
\\
\>[0]\AgdaKeyword{open}\AgdaSpace{}%
\AgdaKeyword{import}\AgdaSpace{}%
\AgdaModule{List}\AgdaSpace{}%
\AgdaBound{eq‐J}\AgdaSpace{}%
\AgdaKeyword{using}\AgdaSpace{}%
\AgdaSymbol{(}\AgdaFunction{length}\AgdaSymbol{)}\<%
\\
\>[0]\AgdaKeyword{open}\AgdaSpace{}%
\AgdaKeyword{import}\AgdaSpace{}%
\AgdaModule{Surjection}\AgdaSpace{}%
\AgdaBound{eq‐J}\AgdaSpace{}%
\AgdaKeyword{using}\AgdaSpace{}%
\AgdaSymbol{(}\AgdaFunction{Split‐surjective}\AgdaSymbol{)}\<%
\\
\\[\AgdaEmptyExtraSkip]%
\>[0]\AgdaKeyword{variable}\<%
\\
\>[0][@{}l@{\AgdaIndent{0}}]%
\>[2]\AgdaGeneralizable{a}\AgdaSpace{}%
\AgdaGeneralizable{b}\AgdaSpace{}%
\AgdaGeneralizable{c}\AgdaSpace{}%
\AgdaGeneralizable{r}%
\>[34]\AgdaSymbol{:}\AgdaSpace{}%
\AgdaPostulate{Level}\<%
\\
\>[2]\AgdaGeneralizable{A}\AgdaSpace{}%
\AgdaGeneralizable{B}\AgdaSpace{}%
\AgdaGeneralizable{C}%
\>[34]\AgdaSymbol{:}\AgdaSpace{}%
\AgdaPrimitive{Type}\AgdaSpace{}%
\AgdaSymbol{\AgdaUnderscore{}}\<%
\\
\>[2]\AgdaSymbol{@0}\AgdaSpace{}%
\AgdaGeneralizable{R}%
\>[34]\AgdaSymbol{:}\AgdaSpace{}%
\AgdaGeneralizable{A}\AgdaSpace{}%
\AgdaSymbol{→}\AgdaSpace{}%
\AgdaGeneralizable{A}\AgdaSpace{}%
\AgdaSymbol{→}\AgdaSpace{}%
\AgdaPrimitive{Type}\AgdaSpace{}%
\AgdaSymbol{\AgdaUnderscore{}}\<%
\\
\>[2]\AgdaGeneralizable{P}\AgdaSpace{}%
\AgdaOperator{\AgdaGeneralizable{Has‐type\AgdaUnderscore{}}}%
\>[34]\AgdaSymbol{:}\AgdaSpace{}%
\AgdaGeneralizable{A}\AgdaSpace{}%
\AgdaSymbol{→}\AgdaSpace{}%
\AgdaPrimitive{Type}\AgdaSpace{}%
\AgdaSymbol{\AgdaUnderscore{}}\<%
\\
\>[2]\AgdaOperator{\AgdaGeneralizable{Has‐type[\AgdaUnderscore{}]\AgdaUnderscore{}}}%
\>[34]\AgdaSymbol{:}\AgdaSpace{}%
\AgdaSymbol{(}\AgdaBound{A}\AgdaSpace{}%
\AgdaSymbol{:}\AgdaSpace{}%
\AgdaPrimitive{Type}\AgdaSpace{}%
\AgdaGeneralizable{a}\AgdaSymbol{)}\AgdaSpace{}%
\AgdaSymbol{→}\AgdaSpace{}%
\AgdaBound{A}\AgdaSpace{}%
\AgdaSymbol{→}\AgdaSpace{}%
\AgdaPrimitive{Type}\AgdaSpace{}%
\AgdaSymbol{\AgdaUnderscore{}}\<%
\\
\>[2]\AgdaGeneralizable{k}%
\>[34]\AgdaSymbol{:}\AgdaSpace{}%
\AgdaDatatype{Kind}\<%
\\
\>[2]\AgdaGeneralizable{eq}\AgdaSpace{}%
\AgdaGeneralizable{f\,}\AgdaSpace{}%
\AgdaGeneralizable{g}\AgdaSpace{}%
\AgdaGeneralizable{n}\AgdaSpace{}%
\AgdaGeneralizable{t\ensuremath{{}_{\mathrm{1}}}}\AgdaSpace{}%
\AgdaGeneralizable{t\ensuremath{{}_{\mathrm{2}}}}\AgdaSpace{}%
\AgdaGeneralizable{x}\AgdaSpace{}%
\AgdaGeneralizable{x′}\AgdaSpace{}%
\AgdaGeneralizable{xs}\AgdaSpace{}%
\AgdaGeneralizable{y}\AgdaSpace{}%
\AgdaGeneralizable{y′}\AgdaSpace{}%
\AgdaGeneralizable{p}\AgdaSpace{}%
\AgdaGeneralizable{q}\AgdaSpace{}%
\AgdaSymbol{:}\AgdaSpace{}%
\AgdaGeneralizable{A}\<%
\\
\\[\AgdaEmptyExtraSkip]%
\>[0]\AgdaFunction{Σ‐syntax}\AgdaSpace{}%
\AgdaSymbol{:}\AgdaSpace{}%
\AgdaSymbol{(}\AgdaBound{A}\AgdaSpace{}%
\AgdaSymbol{:}\AgdaSpace{}%
\AgdaPrimitive{Type}\AgdaSpace{}%
\AgdaGeneralizable{a}\AgdaSymbol{)}\AgdaSpace{}%
\AgdaSymbol{→}\AgdaSpace{}%
\AgdaSymbol{(}\AgdaBound{A}\AgdaSpace{}%
\AgdaSymbol{→}\AgdaSpace{}%
\AgdaPrimitive{Type}\AgdaSpace{}%
\AgdaGeneralizable{p}\AgdaSymbol{)}\AgdaSpace{}%
\AgdaSymbol{→}\AgdaSpace{}%
\AgdaPrimitive{Type}\AgdaSpace{}%
\AgdaSymbol{(}\AgdaGeneralizable{a}\AgdaSpace{}%
\AgdaOperator{\AgdaPrimitive{⊔}}\AgdaSpace{}%
\AgdaGeneralizable{p}\AgdaSymbol{)}\<%
\\
\>[0]\AgdaFunction{Σ‐syntax}\AgdaSpace{}%
\AgdaSymbol{=}\AgdaSpace{}%
\AgdaRecord{Σ}\<%
\\
\\[\AgdaEmptyExtraSkip]%
\>[0]\AgdaKeyword{syntax}\AgdaSpace{}%
\AgdaFunction{Σ‐syntax}\AgdaSpace{}%
\AgdaBound{A}\AgdaSpace{}%
\AgdaSymbol{(λ}\AgdaSpace{}%
\AgdaBound{x}\AgdaSpace{}%
\AgdaSymbol{→}\AgdaSpace{}%
\AgdaBound{P}\AgdaSymbol{)}\AgdaSpace{}%
\AgdaSymbol{=}\AgdaSpace{}%
\AgdaSymbol{(}%
\AgdaBound{x}\AgdaSpace{}%
\AgdaSymbol{:}\AgdaSpace{}%
\AgdaBound{A}\AgdaSymbol{)}\AgdaSpace{}%
\AgdaFunction{×}\AgdaSpace{}%
\AgdaBound{P}\<%
\\
\\[\AgdaEmptyExtraSkip]%
\>[0]\AgdaFunction{Lift}\AgdaSpace{}%
\AgdaSymbol{:}\AgdaSpace{}%
\AgdaPrimitive{Type}\AgdaSpace{}%
\AgdaGeneralizable{a}\AgdaSpace{}%
\AgdaSymbol{→}\AgdaSpace{}%
\AgdaPrimitive{Type}\AgdaSpace{}%
\AgdaSymbol{(}\AgdaGeneralizable{a}\AgdaSpace{}%
\AgdaOperator{\AgdaPrimitive{⊔}}\AgdaSpace{}%
\AgdaGeneralizable{b}\AgdaSymbol{)}\<%
\\
\>[0]\AgdaFunction{Lift}\AgdaSpace{}%
\AgdaSymbol{\{}\AgdaBound{b}\AgdaSymbol{\}}\AgdaSpace{}%
\AgdaSymbol{=}\AgdaSpace{}%
\AgdaRecord{Prelude.↑}\AgdaSpace{}%
\AgdaBound{b}\<%
\\
\\[\AgdaEmptyExtraSkip]%
\>[0]\AgdaKeyword{infix}\AgdaSpace{}%
\AgdaNumber{-100}\AgdaSpace{}%
\AgdaOperator{\AgdaFunction{\AgdaUnderscore{}invisible}}\<%
\\
\\[\AgdaEmptyExtraSkip]%
\>[0]\AgdaOperator{\AgdaFunction{\AgdaUnderscore{}invisible}}\AgdaSpace{}%
\AgdaSymbol{:}\AgdaSpace{}%
\AgdaSymbol{\{}\AgdaBound{A}\AgdaSpace{}%
\AgdaSymbol{:}\AgdaSpace{}%
\AgdaPrimitive{Type}\AgdaSpace{}%
\AgdaGeneralizable{a}\AgdaSymbol{\}}\AgdaSpace{}%
\AgdaSymbol{→}\AgdaSpace{}%
\AgdaBound{A}\AgdaSpace{}%
\AgdaSymbol{→}\AgdaSpace{}%
\AgdaBound{A}\<%
\\
\>[0]\AgdaBound{x}\AgdaSpace{}%
\AgdaOperator{}\AgdaSpace{}%
\AgdaSymbol{=}\AgdaSpace{}%
\AgdaBound{x}\<%
\end{code}

\begin{code}[hide]%
\>[0]\AgdaKeyword{module}\AgdaSpace{}%
\AgdaModule{Extra-Id}\AgdaSpace{}%
\AgdaKeyword{where}\<%
\\
\>[0][@{}l@{\AgdaIndent{0}}]%
\>[2]\AgdaKeyword{data}\AgdaSpace{}%
\AgdaDatatype{Id}\AgdaSpace{}%
\AgdaSymbol{(}\AgdaBound{A}\AgdaSpace{}%
\AgdaSymbol{:}\AgdaSpace{}%
\AgdaPrimitive{Type}\AgdaSpace{}%
\AgdaGeneralizable{a}\AgdaSymbol{)}\AgdaSpace{}%
\AgdaSymbol{(}\AgdaBound{x}\AgdaSpace{}%
\AgdaSymbol{:}\AgdaSpace{}%
\AgdaBound{A}\AgdaSymbol{)}\AgdaSpace{}%
\AgdaSymbol{:}\AgdaSpace{}%
\AgdaBound{A}\AgdaSpace{}%
\AgdaSymbol{→}\AgdaSpace{}%
\AgdaPrimitive{Type}\AgdaSpace{}%
\AgdaBound{a}\AgdaSpace{}%
\AgdaKeyword{where}\<%
\end{code}
\newcommand{\refl}{%
\begin{code}[inline]%
\>[2][@{}l@{\AgdaIndent{1}}]%
\>[4]\AgdaInductiveConstructor{refl}\AgdaSpace{}%
\AgdaSymbol{:}\AgdaSpace{}%
\AgdaDatatype{Id}\AgdaSpace{}%
\AgdaBound{A}\AgdaSpace{}%
\AgdaBound{x}\AgdaSpace{}%
\AgdaBound{x}\<%
\end{code}}

\begin{code}[hide]%
\>[0]\AgdaKeyword{postulate}\<%
\\
\>[0][@{}l@{\AgdaIndent{0}}]%
\>[2]\AgdaPostulate{\AgdaUnderscore{}}%
\>[223I]\AgdaSymbol{:}\<%
\end{code}
\newcommand{\IdAxy}{%
\begin{code}[inline]%
\>[.][@{}l@{}]\<[223I]%
\>[4]\AgdaDatatype{Id}\AgdaSpace{}%
\AgdaGeneralizable{A}\AgdaSpace{}%
\AgdaGeneralizable{x}\AgdaSpace{}%
\AgdaGeneralizable{y}\<%
\end{code}}

\newcommand{\congType}{%
\begin{code}[inline]%
\>[0]\AgdaFunction{cong}\AgdaSpace{}%
\AgdaSymbol{:}\AgdaSpace{}%
\AgdaSymbol{(}\AgdaBound{f\,}\AgdaSpace{}%
\AgdaSymbol{:}\AgdaSpace{}%
\AgdaGeneralizable{A}\AgdaSpace{}%
\AgdaSymbol{→}\AgdaSpace{}%
\AgdaGeneralizable{B}\AgdaSymbol{)}\AgdaSpace{}%
\AgdaSymbol{→}\AgdaSpace{}%
\AgdaDatatype{Id}\AgdaSpace{}%
\AgdaGeneralizable{A}\AgdaSpace{}%
\AgdaGeneralizable{x}\AgdaSpace{}%
\AgdaGeneralizable{y}\AgdaSpace{}%
\AgdaSymbol{→}\AgdaSpace{}%
\AgdaDatatype{Id}\AgdaSpace{}%
\AgdaGeneralizable{B}\AgdaSpace{}%
\AgdaSymbol{(}\AgdaBound{f\,}\AgdaSpace{}%
\AgdaGeneralizable{x}\AgdaSymbol{)}\AgdaSpace{}%
\AgdaSymbol{(}\AgdaBound{f\,}\AgdaSpace{}%
\AgdaGeneralizable{y}\AgdaSymbol{)}\<%
\end{code}}
\begin{code}[hide]%
\>[0]\AgdaFunction{cong}\AgdaSpace{}%
\AgdaBound{f\,}\AgdaSpace{}%
\AgdaInductiveConstructor{refl}\AgdaSpace{}%
\AgdaSymbol{=}\AgdaSpace{}%
\AgdaInductiveConstructor{refl}\<%
\end{code}

\begin{code}[hide]%
\>[0]\AgdaKeyword{infix}%
\>[7]\AgdaNumber{-1}\AgdaSpace{}%
\AgdaOperator{\AgdaPostulate{\AgdaUnderscore{}□}}\<%
\\
\>[0]\AgdaKeyword{infixr}\AgdaSpace{}%
\AgdaNumber{-2}\AgdaSpace{}%
\AgdaOperator{\AgdaPostulate{\AgdaUnderscore{}→⟨⟩\AgdaUnderscore{}}}\AgdaSpace{}%
\AgdaOperator{\AgdaPostulate{\AgdaUnderscore{}⇔⟨⟩\AgdaUnderscore{}}}\AgdaSpace{}%
\AgdaOperator{\AgdaPostulate{\AgdaUnderscore{}≃⟨⟩\AgdaUnderscore{}}}\<%
\\
\\[\AgdaEmptyExtraSkip]%
\>[0]\AgdaKeyword{postulate}\<%
\\
\>[0][@{}l@{\AgdaIndent{0}}]%
\>[2]\AgdaOperator{\AgdaPostulate{\AgdaUnderscore{}→⟨⟩\AgdaUnderscore{}}}%
\>[9]\AgdaSymbol{:}\AgdaSpace{}%
\AgdaSymbol{(}\AgdaBound{A}\AgdaSpace{}%
\AgdaSymbol{:}\AgdaSpace{}%
\AgdaPrimitive{Type}\AgdaSpace{}%
\AgdaGeneralizable{a}\AgdaSymbol{)}\AgdaSpace{}%
\AgdaSymbol{→}\AgdaSpace{}%
\AgdaSymbol{(}\AgdaGeneralizable{B}\AgdaSpace{}%
\AgdaSymbol{→}\AgdaSpace{}%
\AgdaGeneralizable{C}\AgdaSymbol{)}\AgdaSpace{}%
\AgdaSymbol{→}\AgdaSpace{}%
\AgdaSymbol{(}\AgdaBound{A}\AgdaSpace{}%
\AgdaSymbol{→}\AgdaSpace{}%
\AgdaGeneralizable{C}\AgdaSymbol{)}\<%
\\
\>[2]\AgdaOperator{\AgdaPostulate{\AgdaUnderscore{}⇔⟨⟩\AgdaUnderscore{}}}%
\>[9]\AgdaSymbol{:}\AgdaSpace{}%
\AgdaSymbol{(}\AgdaBound{A}\AgdaSpace{}%
\AgdaSymbol{:}\AgdaSpace{}%
\AgdaPrimitive{Type}\AgdaSpace{}%
\AgdaGeneralizable{a}\AgdaSymbol{)}\AgdaSpace{}%
\AgdaSymbol{→}\AgdaSpace{}%
\AgdaGeneralizable{B}\AgdaSpace{}%
\AgdaOperator{\AgdaRecord{⇔}}\AgdaSpace{}%
\AgdaGeneralizable{C}\AgdaSpace{}%
\AgdaSymbol{→}\AgdaSpace{}%
\AgdaBound{A}\AgdaSpace{}%
\AgdaOperator{\AgdaRecord{⇔}}\AgdaSpace{}%
\AgdaGeneralizable{C}\<%
\\
\>[2]\AgdaOperator{\AgdaPostulate{\AgdaUnderscore{}≃⟨⟩\AgdaUnderscore{}}}%
\>[9]\AgdaSymbol{:}\AgdaSpace{}%
\AgdaSymbol{(}\AgdaBound{A}\AgdaSpace{}%
\AgdaSymbol{:}\AgdaSpace{}%
\AgdaPrimitive{Type}\AgdaSpace{}%
\AgdaGeneralizable{a}\AgdaSymbol{)}\AgdaSpace{}%
\AgdaSymbol{→}\AgdaSpace{}%
\AgdaGeneralizable{B}\AgdaSpace{}%
\AgdaOperator{\AgdaRecord{Eq.≃}}\AgdaSpace{}%
\AgdaGeneralizable{C}\AgdaSpace{}%
\AgdaSymbol{→}\AgdaSpace{}%
\AgdaBound{A}\AgdaSpace{}%
\AgdaOperator{\AgdaRecord{Eq.≃}}\AgdaSpace{}%
\AgdaGeneralizable{C}\<%
\\
\>[2]\AgdaOperator{\AgdaPostulate{\AgdaUnderscore{}□}}%
\>[9]\AgdaSymbol{:}\AgdaSpace{}%
\AgdaSymbol{(}\AgdaBound{A}\AgdaSpace{}%
\AgdaSymbol{:}\AgdaSpace{}%
\AgdaPrimitive{Type}\AgdaSpace{}%
\AgdaGeneralizable{a}\AgdaSymbol{)}\AgdaSpace{}%
\AgdaSymbol{→}\AgdaSpace{}%
\AgdaBound{A}\AgdaSpace{}%
\AgdaOperator{\AgdaFunction{↝[}}\AgdaSpace{}%
\AgdaGeneralizable{k}\AgdaSpace{}%
\AgdaOperator{\AgdaFunction{]}}\AgdaSpace{}%
\AgdaBound{A}\<%
\end{code}

\begin{code}[hide]%
\>[0]\AgdaKeyword{module}\AgdaSpace{}%
\AgdaModule{Subst}\AgdaSpace{}%
\AgdaKeyword{where}\<%
\end{code}
\newcommand{\substDef}{%
\begin{code}%
\>[0][@{}l@{\AgdaIndent{1}}]%
\>[2]\AgdaFunction{subst}\AgdaSpace{}%
\AgdaSymbol{:}\AgdaSpace{}%
\AgdaSymbol{(}\AgdaBound{P}\AgdaSpace{}%
\AgdaSymbol{:}\AgdaSpace{}%
\AgdaGeneralizable{A}\AgdaSpace{}%
\AgdaSymbol{→}\AgdaSpace{}%
\AgdaPrimitive{Type}\AgdaSpace{}%
\AgdaGeneralizable{p}\AgdaSymbol{)}\AgdaSpace{}%
\AgdaSymbol{→}\AgdaSpace{}%
\AgdaDatatype{Id}\AgdaSpace{}%
\AgdaGeneralizable{A}\AgdaSpace{}%
\AgdaGeneralizable{x}\AgdaSpace{}%
\AgdaGeneralizable{y}\AgdaSpace{}%
\AgdaSymbol{→}\AgdaSpace{}%
\AgdaBound{P}\AgdaSpace{}%
\AgdaGeneralizable{x}\AgdaSpace{}%
\AgdaSymbol{→}\AgdaSpace{}%
\AgdaBound{P}\AgdaSpace{}%
\AgdaGeneralizable{y}\<%
\\
\>[2]\AgdaFunction{subst}\AgdaSpace{}%
\AgdaBound{P}\AgdaSpace{}%
\AgdaInductiveConstructor{refl}\AgdaSpace{}%
\AgdaBound{p}\AgdaSpace{}%
\AgdaSymbol{=}\AgdaSpace{}%
\AgdaBound{p}\<%
\end{code}}

\begin{code}[hide]%
\>[0]\AgdaFunction{Is‐set}\AgdaSpace{}%
\AgdaSymbol{:}\AgdaSpace{}%
\AgdaPrimitive{Type}\AgdaSpace{}%
\AgdaGeneralizable{a}\AgdaSpace{}%
\AgdaSymbol{→}\AgdaSpace{}%
\AgdaPrimitive{Type}\AgdaSpace{}%
\AgdaGeneralizable{a}\<%
\end{code}
\newcommand{\Isset}{%
\begin{code}[inline]%
\>[0]\AgdaFunction{Is‐set}\AgdaSpace{}%
\AgdaBound{A}\AgdaSpace{}%
\AgdaSymbol{=}\AgdaSpace{}%
\AgdaSymbol{\{}\AgdaBound{x}\AgdaSpace{}%
\AgdaBound{y}\AgdaSpace{}%
\AgdaSymbol{:}\AgdaSpace{}%
\AgdaBound{A}\AgdaSymbol{\}}\AgdaSpace{}%
\AgdaSymbol{(}\AgdaBound{p}\AgdaSpace{}%
\AgdaBound{q}\AgdaSpace{}%
\AgdaSymbol{:}\AgdaSpace{}%
\AgdaDatatype{Id}\AgdaSpace{}%
\AgdaBound{A}\AgdaSpace{}%
\AgdaBound{x}\AgdaSpace{}%
\AgdaBound{y}\AgdaSymbol{)}\AgdaSpace{}%
\AgdaSymbol{→}\AgdaSpace{}%
\AgdaDatatype{Id}\AgdaSpace{}%
\AgdaSymbol{(}\AgdaDatatype{Id}\AgdaSpace{}%
\AgdaBound{A}\AgdaSpace{}%
\AgdaBound{x}\AgdaSpace{}%
\AgdaBound{y}\AgdaSymbol{)}\AgdaSpace{}%
\AgdaBound{p}\AgdaSpace{}%
\AgdaBound{q}\<%
\end{code}}

\begin{code}[hide]%
\>[0]\AgdaKeyword{module}\AgdaSpace{}%
\AgdaModule{Quotient}\AgdaSpace{}%
\AgdaKeyword{where}\<%
\\
\>[0][@{}l@{\AgdaIndent{0}}]%
\>[2]\AgdaKeyword{open}\AgdaSpace{}%
\AgdaModule{Subst}\<%
\\
\\[\AgdaEmptyExtraSkip]%
\>[2]\AgdaKeyword{module}\AgdaSpace{}%
\AgdaModule{Hidden}\AgdaSpace{}%
\AgdaKeyword{where}\<%
\\
\>[2][@{}l@{\AgdaIndent{0}}]%
\>[4]\AgdaKeyword{postulate}\<%
\end{code}
\newcommand{\QuotAPI}{%
\begin{code}%
\>[4][@{}l@{\AgdaIndent{1}}]%
\>[6]\AgdaOperator{\AgdaPostulate{\AgdaUnderscore{}/\AgdaUnderscore{}}}%
\>[15]\AgdaSymbol{:}\AgdaSpace{}%
\AgdaSymbol{(}\AgdaBound{A}\AgdaSpace{}%
\AgdaSymbol{:}\AgdaSpace{}%
\AgdaPrimitive{Type}\AgdaSpace{}%
\AgdaGeneralizable{a}\AgdaSymbol{)}\AgdaSpace{}%
\AgdaSymbol{→}\AgdaSpace{}%
\AgdaSymbol{@0}\AgdaSpace{}%
\AgdaSymbol{(}\AgdaBound{A}\AgdaSpace{}%
\AgdaSymbol{→}\AgdaSpace{}%
\AgdaBound{A}\AgdaSpace{}%
\AgdaSymbol{→}\AgdaSpace{}%
\AgdaPrimitive{Type}\AgdaSpace{}%
\AgdaGeneralizable{a}\AgdaSymbol{)}\AgdaSpace{}%
\AgdaSymbol{→}\AgdaSpace{}%
\AgdaPrimitive{Type}\AgdaSpace{}%
\AgdaGeneralizable{a}\<%
\\
\>[6]\AgdaOperator{\AgdaPostulate{[\AgdaUnderscore{}]}}%
\>[15]\AgdaSymbol{:}\AgdaSpace{}%
\AgdaSymbol{\{}\AgdaSymbol{@0}\AgdaSpace{}%
\AgdaBound{R}\AgdaSpace{}%
\AgdaSymbol{:}\AgdaSpace{}%
\AgdaGeneralizable{A}\AgdaSpace{}%
\AgdaSymbol{→}\AgdaSpace{}%
\AgdaGeneralizable{A}\AgdaSpace{}%
\AgdaSymbol{→}\AgdaSpace{}%
\AgdaPrimitive{Type}\AgdaSpace{}%
\AgdaGeneralizable{a}\AgdaSymbol{\}}\AgdaSpace{}%
\AgdaSymbol{→}\AgdaSpace{}%
\AgdaGeneralizable{A}\AgdaSpace{}%
\AgdaSymbol{→}\AgdaSpace{}%
\AgdaGeneralizable{A}\AgdaSpace{}%
\AgdaOperator{\AgdaPostulate{/}}\AgdaSpace{}%
\AgdaBound{R}\<%
\\
\>[6]\AgdaSymbol{@}\AgdaNumber{0}\AgdaSpace{}%
\AgdaInductiveConstructor{resp}%
\>[15]\AgdaSymbol{:}\AgdaSpace{}%
\AgdaGeneralizable{R}\AgdaSpace{}%
\AgdaGeneralizable{x}\AgdaSpace{}%
\AgdaGeneralizable{y}\AgdaSpace{}%
\AgdaSymbol{→}\AgdaSpace{}%
\AgdaDatatype{Id}\AgdaSpace{}%
\AgdaSymbol{(}\AgdaGeneralizable{A}\AgdaSpace{}%
\AgdaOperator{\AgdaPostulate{/}}\AgdaSpace{}%
\AgdaGeneralizable{R}\AgdaSymbol{)}\AgdaSpace{}%
\AgdaOperator{\AgdaPostulate{[}}\AgdaSpace{}%
\AgdaGeneralizable{x}\AgdaSpace{}%
\AgdaOperator{\AgdaPostulate{]}}\AgdaSpace{}%
\AgdaOperator{\AgdaPostulate{[}}\AgdaSpace{}%
\AgdaGeneralizable{y}\AgdaSpace{}%
\AgdaOperator{\AgdaPostulate{]}}\<%
\\
\>[6]\AgdaSymbol{@}\AgdaNumber{0}\AgdaSpace{}%
\AgdaInductiveConstructor{set}%
\>[15]\AgdaSymbol{:}\AgdaSpace{}%
\AgdaFunction{Is‐set}\AgdaSpace{}%
\AgdaSymbol{(}\AgdaGeneralizable{A}\AgdaSpace{}%
\AgdaOperator{\AgdaPostulate{/}}\AgdaSpace{}%
\AgdaGeneralizable{R}\AgdaSymbol{)}\<%
\\
\>[6]\AgdaPostulate{qrec}%
\>[15]\AgdaSymbol{:}%
\>[407I]\AgdaSymbol{\{}\AgdaSymbol{@0}\AgdaSpace{}%
\AgdaBound{R}\AgdaSpace{}%
\AgdaSymbol{:}\AgdaSpace{}%
\AgdaGeneralizable{A}\AgdaSpace{}%
\AgdaSymbol{→}\AgdaSpace{}%
\AgdaGeneralizable{A}\AgdaSpace{}%
\AgdaSymbol{→}\AgdaSpace{}%
\AgdaPrimitive{Type}\AgdaSpace{}%
\AgdaGeneralizable{a}\AgdaSymbol{\}}\AgdaSpace{}%
\AgdaSymbol{→}\AgdaSpace{}%
\AgdaSymbol{(}\AgdaBound{P}\AgdaSpace{}%
\AgdaSymbol{:}\AgdaSpace{}%
\AgdaGeneralizable{A}\AgdaSpace{}%
\AgdaOperator{\AgdaPostulate{/}}\AgdaSpace{}%
\AgdaBound{R}\AgdaSpace{}%
\AgdaSymbol{→}\AgdaSpace{}%
\AgdaPrimitive{Type}\AgdaSpace{}%
\AgdaGeneralizable{p}\AgdaSymbol{)}\AgdaSpace{}%
\AgdaSymbol{→}\AgdaSpace{}%
\AgdaSymbol{(}\AgdaBound{f\,}\AgdaSpace{}%
\AgdaSymbol{:}\AgdaSpace{}%
\AgdaSymbol{∀}\AgdaSpace{}%
\AgdaBound{x}\AgdaSpace{}%
\AgdaSymbol{→}\AgdaSpace{}%
\AgdaBound{P}\AgdaSpace{}%
\AgdaOperator{\AgdaPostulate{[}}\AgdaSpace{}%
\AgdaBound{x}\AgdaSpace{}%
\AgdaOperator{\AgdaPostulate{]}}\AgdaSymbol{)}\AgdaSpace{}%
\AgdaSymbol{→}\<%
\\
\>[.][@{}l@{}]\<[407I]%
\>[17]\AgdaSymbol{@0}\AgdaSpace{}%
\AgdaSymbol{(∀}\AgdaSpace{}%
\AgdaSymbol{\{}\AgdaBound{x}\AgdaSpace{}%
\AgdaBound{y}\AgdaSymbol{\}}\AgdaSpace{}%
\AgdaSymbol{(}\AgdaBound{r}\AgdaSpace{}%
\AgdaSymbol{:}\AgdaSpace{}%
\AgdaBound{R}\AgdaSpace{}%
\AgdaBound{x}\AgdaSpace{}%
\AgdaBound{y}\AgdaSymbol{)}\AgdaSpace{}%
\AgdaSymbol{→}\AgdaSpace{}%
\AgdaDatatype{Id}\AgdaSpace{}%
\AgdaSymbol{(}\AgdaBound{P}\AgdaSpace{}%
\AgdaOperator{\AgdaPostulate{[}}\AgdaSpace{}%
\AgdaBound{y}\AgdaSpace{}%
\AgdaOperator{\AgdaPostulate{]}}\AgdaSymbol{)}\AgdaSpace{}%
\AgdaSymbol{(}\AgdaFunction{subst}\AgdaSpace{}%
\AgdaBound{P}\AgdaSpace{}%
\AgdaSymbol{(}\AgdaInductiveConstructor{resp}\AgdaSpace{}%
\AgdaBound{r}\AgdaSymbol{)}\AgdaSpace{}%
\AgdaSymbol{(}\AgdaBound{f\,}\AgdaSpace{}%
\AgdaBound{x}\AgdaSymbol{))}\AgdaSpace{}%
\AgdaSymbol{(}\AgdaBound{f\,}\AgdaSpace{}%
\AgdaBound{y}\AgdaSymbol{))}\AgdaSpace{}%
\AgdaSymbol{→}\<%
\\
\>[17]\AgdaSymbol{@0}\AgdaSpace{}%
\AgdaSymbol{(∀}\AgdaSpace{}%
\AgdaBound{x}\AgdaSpace{}%
\AgdaSymbol{→}\AgdaSpace{}%
\AgdaFunction{Is‐set}\AgdaSpace{}%
\AgdaSymbol{(}\AgdaBound{P}\AgdaSpace{}%
\AgdaBound{x}\AgdaSymbol{))}\AgdaSpace{}%
\AgdaSymbol{→}\AgdaSpace{}%
\AgdaSymbol{(}\AgdaBound{x}\AgdaSpace{}%
\AgdaSymbol{:}\AgdaSpace{}%
\AgdaGeneralizable{A}\AgdaSpace{}%
\AgdaOperator{\AgdaPostulate{/}}\AgdaSpace{}%
\AgdaBound{R}\AgdaSymbol{)}\AgdaSpace{}%
\AgdaSymbol{→}\AgdaSpace{}%
\AgdaBound{P}\AgdaSpace{}%
\AgdaBound{x}\<%
\end{code}}

\begin{code}[hide]%
\>[2]\AgdaKeyword{data}\AgdaSpace{}%
\AgdaOperator{\AgdaDatatype{\AgdaUnderscore{}/\AgdaUnderscore{}}}\AgdaSpace{}%
\AgdaSymbol{(}\AgdaBound{A}\AgdaSpace{}%
\AgdaSymbol{:}\AgdaSpace{}%
\AgdaPrimitive{Type}\AgdaSpace{}%
\AgdaGeneralizable{a}\AgdaSymbol{)}\AgdaSpace{}%
\AgdaSymbol{(}\AgdaSymbol{@0}\AgdaSpace{}%
\AgdaBound{R}\AgdaSpace{}%
\AgdaSymbol{:}\AgdaSpace{}%
\AgdaBound{A}\AgdaSpace{}%
\AgdaSymbol{→}\AgdaSpace{}%
\AgdaBound{A}\AgdaSpace{}%
\AgdaSymbol{→}\AgdaSpace{}%
\AgdaPrimitive{Type}\AgdaSpace{}%
\AgdaGeneralizable{a}\AgdaSymbol{)}\AgdaSpace{}%
\AgdaSymbol{:}\AgdaSpace{}%
\AgdaPrimitive{Type}\AgdaSpace{}%
\AgdaBound{a}\AgdaSpace{}%
\AgdaKeyword{where}\<%
\\
\>[2][@{}l@{\AgdaIndent{0}}]%
\>[4]\AgdaOperator{\AgdaInductiveConstructor{[\AgdaUnderscore{}]}}\AgdaSpace{}%
\AgdaSymbol{:}\AgdaSpace{}%
\AgdaBound{A}\AgdaSpace{}%
\AgdaSymbol{→}\AgdaSpace{}%
\AgdaBound{A}\AgdaSpace{}%
\AgdaOperator{\AgdaDatatype{/}}\AgdaSpace{}%
\AgdaBound{R}\<%
\\
\\[\AgdaEmptyExtraSkip]%
\>[2]\AgdaKeyword{postulate}\<%
\\
\>[2][@{}l@{\AgdaIndent{0}}]%
\>[4]\AgdaSymbol{@}\AgdaNumber{0}\AgdaSpace{}%
\AgdaInductiveConstructor{resp}%
\>[13]\AgdaSymbol{:}\AgdaSpace{}%
\AgdaGeneralizable{R}\AgdaSpace{}%
\AgdaGeneralizable{x}\AgdaSpace{}%
\AgdaGeneralizable{y}\AgdaSpace{}%
\AgdaSymbol{→}\AgdaSpace{}%
\AgdaDatatype{Id}\AgdaSpace{}%
\AgdaSymbol{(}\AgdaGeneralizable{A}\AgdaSpace{}%
\AgdaOperator{\AgdaDatatype{/}}\AgdaSpace{}%
\AgdaGeneralizable{R}\AgdaSymbol{)}\AgdaSpace{}%
\AgdaOperator{\AgdaInductiveConstructor{[}}\AgdaSpace{}%
\AgdaGeneralizable{x}\AgdaSpace{}%
\AgdaOperator{\AgdaInductiveConstructor{]}}\AgdaSpace{}%
\AgdaOperator{\AgdaInductiveConstructor{[}}\AgdaSpace{}%
\AgdaGeneralizable{y}\AgdaSpace{}%
\AgdaOperator{\AgdaInductiveConstructor{]}}\<%
\\
\>[4]\AgdaSymbol{@}\AgdaNumber{0}\AgdaSpace{}%
\AgdaInductiveConstructor{set}%
\>[13]\AgdaSymbol{:}\AgdaSpace{}%
\AgdaFunction{Is‐set}\AgdaSpace{}%
\AgdaSymbol{(}\AgdaGeneralizable{A}\AgdaSpace{}%
\AgdaOperator{\AgdaDatatype{/}}\AgdaSpace{}%
\AgdaGeneralizable{R}\AgdaSymbol{)}\<%
\\
\\[\AgdaEmptyExtraSkip]%
\>[2]\AgdaFunction{qrec}\AgdaSpace{}%
\AgdaSymbol{:}%
\>[519I]\AgdaSymbol{(}\AgdaBound{P}\AgdaSpace{}%
\AgdaSymbol{:}\AgdaSpace{}%
\AgdaGeneralizable{A}\AgdaSpace{}%
\AgdaOperator{\AgdaDatatype{/}}\AgdaSpace{}%
\AgdaGeneralizable{R}\AgdaSpace{}%
\AgdaSymbol{→}\AgdaSpace{}%
\AgdaPrimitive{Type}\AgdaSpace{}%
\AgdaGeneralizable{p}\AgdaSymbol{)}\AgdaSpace{}%
\AgdaSymbol{→}\AgdaSpace{}%
\AgdaSymbol{(}\AgdaBound{f\,}\AgdaSpace{}%
\AgdaSymbol{:}\AgdaSpace{}%
\AgdaSymbol{∀}\AgdaSpace{}%
\AgdaBound{x}\AgdaSpace{}%
\AgdaSymbol{→}\AgdaSpace{}%
\AgdaBound{P}\AgdaSpace{}%
\AgdaOperator{\AgdaInductiveConstructor{[}}\AgdaSpace{}%
\AgdaBound{x}\AgdaSpace{}%
\AgdaOperator{\AgdaInductiveConstructor{]}}\AgdaSymbol{)}\AgdaSpace{}%
\AgdaSymbol{→}\<%
\\
\>[.][@{}l@{}]\<[519I]%
\>[9]\AgdaSymbol{@}\AgdaNumber{0}\AgdaSpace{}%
\AgdaSymbol{(∀}\AgdaSpace{}%
\AgdaSymbol{\{}\AgdaBound{x}\AgdaSpace{}%
\AgdaBound{y}\AgdaSymbol{\}}\AgdaSpace{}%
\AgdaSymbol{(}\AgdaBound{r}\AgdaSpace{}%
\AgdaSymbol{:}\AgdaSpace{}%
\AgdaGeneralizable{R}\AgdaSpace{}%
\AgdaBound{x}\AgdaSpace{}%
\AgdaBound{y}\AgdaSymbol{)}\AgdaSpace{}%
\AgdaSymbol{→}\AgdaSpace{}%
\AgdaDatatype{Id}\AgdaSpace{}%
\AgdaSymbol{(}\AgdaBound{P}\AgdaSpace{}%
\AgdaOperator{\AgdaInductiveConstructor{[}}\AgdaSpace{}%
\AgdaBound{y}\AgdaSpace{}%
\AgdaOperator{\AgdaInductiveConstructor{]}}\AgdaSymbol{)}\AgdaSpace{}%
\AgdaSymbol{(}\AgdaFunction{subst}\AgdaSpace{}%
\AgdaBound{P}\AgdaSpace{}%
\AgdaSymbol{(}\AgdaInductiveConstructor{resp}\AgdaSpace{}%
\AgdaBound{r}\AgdaSymbol{)}\AgdaSpace{}%
\AgdaSymbol{(}\AgdaBound{f\,}\AgdaSpace{}%
\AgdaBound{x}\AgdaSymbol{))}\AgdaSpace{}%
\AgdaSymbol{(}\AgdaBound{f\,}\AgdaSpace{}%
\AgdaBound{y}\AgdaSymbol{))}\AgdaSpace{}%
\AgdaSymbol{→}\<%
\\
\>[9]\AgdaSymbol{@}\AgdaNumber{0}\AgdaSpace{}%
\AgdaSymbol{(∀}\AgdaSpace{}%
\AgdaBound{x}\AgdaSpace{}%
\AgdaSymbol{→}\AgdaSpace{}%
\AgdaFunction{Is‐set}\AgdaSpace{}%
\AgdaSymbol{(}\AgdaBound{P}\AgdaSpace{}%
\AgdaBound{x}\AgdaSymbol{))}\AgdaSpace{}%
\AgdaSymbol{→}\AgdaSpace{}%
\AgdaSymbol{(}\AgdaBound{x}\AgdaSpace{}%
\AgdaSymbol{:}\AgdaSpace{}%
\AgdaGeneralizable{A}\AgdaSpace{}%
\AgdaOperator{\AgdaDatatype{/}}\AgdaSpace{}%
\AgdaGeneralizable{R}\AgdaSymbol{)}\AgdaSpace{}%
\AgdaSymbol{→}\AgdaSpace{}%
\AgdaBound{P}\AgdaSpace{}%
\AgdaBound{x}\<%
\end{code}
\newcommand{\qrecdef}{%
\begin{code}[inline]%
\>[2]\AgdaFunction{qrec}\AgdaSpace{}%
\AgdaSymbol{\AgdaUnderscore{}}\AgdaSpace{}%
\AgdaBound{f\,}\AgdaSpace{}%
\AgdaSymbol{\AgdaUnderscore{}}\AgdaSpace{}%
\AgdaSymbol{\AgdaUnderscore{}}\AgdaSpace{}%
\AgdaOperator{\AgdaInductiveConstructor{[}}\AgdaSpace{}%
\AgdaBound{x}\AgdaSpace{}%
\AgdaOperator{\AgdaInductiveConstructor{]}}\AgdaSpace{}%
\AgdaSymbol{=}\AgdaSpace{}%
\AgdaBound{f\,}\AgdaSpace{}%
\AgdaBound{x}\<%
\end{code}}

\begin{code}[hide]%
\>[2]\AgdaKeyword{postulate}\<%
\\
\>[2][@{}l@{\AgdaIndent{0}}]%
\>[4]\AgdaPostulate{\AgdaUnderscore{}}%
\>[586I]\AgdaSymbol{:}\<%
\end{code}
\newcommand{\QuotAR}{%
\begin{code}[inline]%
\>[.][@{}l@{}]\<[586I]%
\>[6]\AgdaGeneralizable{A}\AgdaSpace{}%
\AgdaOperator{\AgdaDatatype{/}}\AgdaSpace{}%
\AgdaGeneralizable{R}\<%
\end{code}}

\newcommand{\quotbox}{\AgdaInductiveConstructor{[\AgdaUnderscore{}]}}%

\title{\theTitle}
\author{Nils Anders Danielsson}
\orcid{0000-0001-8688-0333}
\affiliation{
  \department{Department of Computer Science and Engineering}
  \institution{University of Gothenburg and Chalmers University of Technology}
  \city{Gothenburg}
  \country{Sweden}
}
\email{nad@cse.gu.se}

\begin{document}

\begin{abstract}
  This text is concerned with the question of whether, in type theory
  with erasure annotations, one can postulate that some type is
  inhabited and still have a guarantee that a program will not get
  stuck.
  Previous work has provided such guarantees for consistent erased
  postulates, i.e.\ postulates that are restricted to be used in
  erased contexts.
  Here those guarantees are extended to type theory with identity
  types.

  Similar ideas provide a simple way to support quotient types: it is
  shown that one can let things like ``the equivalence classes for two
  related values are equal'' be erased postulates and have an
  eliminator that only computes for the equivalence class constructor,
  and still get a guarantee that programs will compute correctly.

  Another question is whether programs compute correctly if one is
  allowed to transport (cast) using erased identity proofs.
  It is shown that this is safe in the absence of quotients and
  postulates, and in the presence of quotients and erased postulates
  that can be implemented using equality reflection.
  However, unrestricted transports of this kind are not compatible
  with erased, postulated univalence.
  For that reason the text includes a study of the function
  \boxcong{}, which encapsulates a limited form of transport for
  erased identity proofs.

  The text is accompanied by machine-checked Agda proofs.
\end{abstract}

\maketitle

\section{Introduction}
\label{sec:introduction}

\begin{code}[hide]%
\>[0]\AgdaKeyword{module}\AgdaSpace{}%
\AgdaModule{Introduction}\AgdaSpace{}%
\AgdaKeyword{where}\<%
\\
\\[\AgdaEmptyExtraSkip]%
\>[0][@{}l@{\AgdaIndent{0}}]%
\>[2]\AgdaKeyword{open}\AgdaSpace{}%
\AgdaModule{List}\AgdaSpace{}%
\AgdaKeyword{using}\AgdaSpace{}%
\AgdaSymbol{(}\AgdaInductiveConstructor{[\ensuremath{\mkern1.5mu}]}\AgdaSymbol{;}\AgdaSpace{}%
\AgdaOperator{\AgdaInductiveConstructor{\AgdaUnderscore{}∷\AgdaUnderscore{}}}\AgdaSymbol{)}\<%
\\
\>[2]\AgdaKeyword{open}\AgdaSpace{}%
\AgdaModule{Subst}\AgdaSpace{}%
\AgdaKeyword{public}\<%
\end{code}

Dependent type theories/programming languages like Agda
\citep{agda-2025}, Idris \citep{brady-2021}, Lean \citep{lean-2026}
and Rocq \citep{rocq-2026} support \emph{identity types}.
A key feature of identity types is the ability to substitute equals
for equals, as witnessed by the following Agda function:
\substDef{}%
In any context expressible as a function \AgdaBound{P} from some type
\AgdaBound{A} to the universe
\begin{code}[hide]%
\>[2]\AgdaKeyword{postulate}\<%
\\
\>[2][@{}l@{\AgdaIndent{0}}]%
\>[4]\AgdaPostulate{\AgdaUnderscore{}}%
\>[597I]\AgdaSymbol{:}\<%
\end{code}
\begin{code}[inline*]%
\>[.][@{}l@{}]\<[597I]%
\>[6]\AgdaPrimitive{Type}\AgdaSpace{}%
\AgdaGeneralizable{p}\<%
\end{code}
(the type of types with universe level \AgdaBound{p}) one can replace
\AgdaBound{x} with \AgdaBound{y}, given that the identity type
\IdAxy{} is inhabited.
Note that \AgdaFunction{subst} is implemented by matching on the
constructor \refl{}.

A potential drawback of identity types in plain, intensional
Martin-Löf type theory is that one cannot prove certain properties,
like function extensionality – extensionally equal functions are equal
\citep{altenkirch-1999} – or univalence – which implies that types
that are in bijective correspondence are equal \citep{hott-2013}.
A workaround is to \emph{postulate} the properties that one needs
(assuming that they are consistent).
However, if the type theory is used as a programming language, then a
postulate could lead to programs becoming stuck during execution.

This problem was addressed by \citet{abel-et-al-2023}, who study a
family of type theories with support for \emph{erasure} inspired by
the work of \citet{mcbride-2016} and \citet{atkey-2018}, and similar
to the support for erasure in Agda and Idris 2.
The idea is basically that erased data will be removed by compilers,
and that the type system ensures that no run-time decision is made
based on erased data.
\citet[Theorem 8.3]{abel-et-al-2023} show that, for a program $t$ of
type $ℕ$ that only uses postulates in erased contexts, if the
postulates are mutually consistent, then $t$ will reduce to a numeral,
and the corresponding compiled program will reduce to the same
numeral.

\citeauthor{abel-et-al-2023} did not support identity types.
One contribution of this work is to add support for identity types
(see Theorem \ref{thm:soundness-of-erasure}), thus establishing that a
postulate like the following one is ``safe'':
\begin{code}%
\>[2]\AgdaKeyword{postulate}\AgdaSpace{}%
\AgdaSymbol{@}\AgdaNumber{0}\AgdaSpace{}%
\AgdaPostulate{funext}\AgdaSpace{}%
\AgdaSymbol{:}\AgdaSpace{}%
\AgdaSymbol{(∀}\AgdaSpace{}%
\AgdaBound{x}\AgdaSpace{}%
\AgdaSymbol{→}\AgdaSpace{}%
\AgdaDatatype{Id}\AgdaSpace{}%
\AgdaSymbol{(}\AgdaGeneralizable{P}\AgdaSpace{}%
\AgdaBound{x}\AgdaSymbol{)}\AgdaSpace{}%
\AgdaSymbol{(}\AgdaGeneralizable{f\,}\AgdaSpace{}%
\AgdaBound{x}\AgdaSymbol{)}\AgdaSpace{}%
\AgdaSymbol{(}\AgdaGeneralizable{g}\AgdaSpace{}%
\AgdaBound{x}\AgdaSymbol{))}\AgdaSpace{}%
\AgdaSymbol{→}\AgdaSpace{}%
\AgdaDatatype{Id}\AgdaSpace{}%
\AgdaSymbol{((}\AgdaBound{x}\AgdaSpace{}%
\AgdaSymbol{:}\AgdaSpace{}%
\AgdaGeneralizable{A}\AgdaSymbol{)}\AgdaSpace{}%
\AgdaSymbol{→}\AgdaSpace{}%
\AgdaGeneralizable{P}\AgdaSpace{}%
\AgdaBound{x}\AgdaSymbol{)}\AgdaSpace{}%
\AgdaGeneralizable{f\,}\AgdaSpace{}%
\AgdaGeneralizable{g}\<%
\end{code}
This code uses Agda notation: the annotation \AgdaSymbol{@0} means
that the postulate is erased, and may only be used in erased contexts,
for instance in erased function arguments.

This work also goes further:
\begin{itemize}
\item It is shown how erased postulates make it easy to implement a
  form of \emph{quotient types} (§~\ref{sec:quotients}).
  As an aside, the associated proofs of normalisation, canonicity,
  consistency, etc.\ are possibly the first mechanised proofs of that
  kind for type theory with quotients
  (§~\ref{sec:meta-theory}).
\item It is shown that in some cases erased postulates can be
  supported even in the presence of \emph{erased matches}
  \citep{abel-et-al-2023}, in particular variants of
  \AgdaFunction{subst} that take \emph{erased} identity proofs
  (§~\ref{sec:erased-postulates-and-erased-matches}).
\item Some forms of erased matches are incompatible with univalence,
  and for that reason this text includes a study of the function
  \boxcong{}, which provides a restricted form of erased match for
  identity proofs (§~\ref{sec:box-cong}).
\item This work provides a first step towards \emph{Computational Book
    HoTT}: a variant of Book HoTT \citep{hott-2013} with erased
  univalence, higher inductive types with erased higher constructors,
  and a guarantee that programs will compute properly
  (§~\ref{sec:computational-book-hott}).
\item The technique of \emph{fording}
  \citep{chapman-et-al-2010,pujet-leray-tabareau-2025} – encoding
  inductive families using regular inductive types and identity proofs
  – is adapted for use with erased identity proofs.
  (The text only treats the case of vectors.)
  The technique uses erased matches for identity types
  (§~\ref{sec:fording}).
\end{itemize}

\paragraph{Formalisation}

All results have been formalised in Agda, and the code has been made
available \ifAnonymous{as supplementary material}{to download
  \citep{danielsson-2026}}.
The results in §§~\ref{sec:typing}-\ref{sec:no-box-cong} have been
proved in the context of a formalisation of graded type theory
\citep{abel-et-al-2023,abel-et-al-2026}, while those in
§§~\ref{sec:quotients}–\ref{sec:fording} and
§§~\ref{sec:modalities}–\ref{sec:box-cong-from-funext} are proved for
Agda.
Note that there are minor differences between the code and the
presentation in the text.

The graded type theory formalisation of \citet{abel-et-al-2023} builds
on earlier work for non-graded type theory by \citet{abel-et-al-2017},
extended by Oskar Eriksson, Gaëtan Gilbert and Wojciech Nawrocki.
The graded type theory formalisation has in turn been extended with a
weak unit type (Eriksson), top-level, possibly opaque definitions
\citep{danielsson-geng-2025}, and universe polymorphism
\citep{danielsson-et-al-2026}.

\section{Erasure}
\label{sec:erasure}

Let me start by presenting parts of Agda's support for erasure (the
design of which was inspired by the work of \citet{mcbride-2016} and
\citet{atkey-2018}) through some examples.

Erased arguments are marked with $\AgdaSymbol{@0}$ (or
$\AgdaSymbol{@erased}$), and the idea is that no run-time decision
should be made based on an erased argument.
The following definition of \AgdaFunction{ok} is OK, because the
result of the function can be computed without making use of the
(implicit) type argument $\AgdaBound{A}$:\footnote{I sometimes omit
  declarations of implicit arguments, but not erased ones.}\MCbreak{}
\begin{minipage}[t]{0.49\linewidth}
\begin{code}%
\>[2]\AgdaFunction{ok}\AgdaSpace{}%
\AgdaSymbol{:}\AgdaSpace{}%
\AgdaSymbol{\{}\AgdaSymbol{@}\AgdaNumber{0}\AgdaSpace{}%
\AgdaBound{A}\AgdaSpace{}%
\AgdaSymbol{:}\AgdaSpace{}%
\AgdaPrimitive{Type}\AgdaSymbol{\}}\AgdaSpace{}%
\AgdaSymbol{→}\AgdaSpace{}%
\AgdaBound{A}\AgdaSpace{}%
\AgdaSymbol{→}\AgdaSpace{}%
\AgdaBound{A}\<%
\\
\>[2]\AgdaFunction{ok}\AgdaSpace{}%
\AgdaBound{x}\AgdaSpace{}%
\AgdaSymbol{=}\AgdaSpace{}%
\AgdaBound{x}\<%
\end{code}
\end{minipage}
\begin{minipage}[t]{0.49\linewidth}
\begin{code}[hide]%
\>[2]\AgdaKeyword{module}\AgdaSpace{}%
\AgdaModule{Hide\ensuremath{{}_{\mathrm{1}}}}\AgdaSpace{}%
\AgdaKeyword{where}\<%
\\
\>[2][@{}l@{\AgdaIndent{0}}]%
\>[4]\AgdaSymbol{@}\AgdaNumber{0}\<%
\end{code}
\begin{code}%
\>[4][@{}l@{\AgdaIndent{1}}]%
\>[6]\AgdaFunction{not‐ok}\AgdaSpace{}%
\AgdaSymbol{:}\AgdaSpace{}%
\AgdaSymbol{\{}\AgdaSymbol{@}\AgdaNumber{0}\AgdaSpace{}%
\AgdaBound{A}\AgdaSpace{}%
\AgdaSymbol{:}\AgdaSpace{}%
\AgdaPrimitive{Type}\AgdaSymbol{\}}\AgdaSpace{}%
\AgdaSymbol{→}\AgdaSpace{}%
\AgdaSymbol{@}\AgdaNumber{0}\AgdaSpace{}%
\AgdaBound{A}\AgdaSpace{}%
\AgdaSymbol{→}\AgdaSpace{}%
\AgdaBound{A}\<%
\\
\>[4]\AgdaFunction{not‐ok}\AgdaSpace{}%
\AgdaBound{x}\AgdaSpace{}%
\AgdaSymbol{=}\AgdaSpace{}%
\AgdaBound{x}\<%
\end{code}
\end{minipage}\\
However, \AgdaFunction{not‐ok} is not OK (in a non-erased context),
because the argument of type $\AgdaBound{A}$ is marked as erased.
Erased arguments and definitions can always be used in erased
contexts, which include type signatures and erased function arguments
(for instance the subexpression $\AgdaBound{e}$ of the expression
$\mbox{\AgdaBound{f\,} \AgdaBound{e}}$, if the first explicit argument
of $\AgdaBound{f\,}$ is erased).

\newcommand{\ErasedDef}{%
\begin{code}%
\>[2]\AgdaKeyword{data}\AgdaSpace{}%
\AgdaDatatype{Erased}\AgdaSpace{}%
\AgdaSymbol{(}\AgdaSymbol{@0}\AgdaSpace{}%
\AgdaBound{A}\AgdaSpace{}%
\AgdaSymbol{:}\AgdaSpace{}%
\AgdaPrimitive{Type}\AgdaSpace{}%
\AgdaGeneralizable{a}\AgdaSymbol{)}\AgdaSpace{}%
\AgdaSymbol{:}\AgdaSpace{}%
\AgdaPrimitive{Type}\AgdaSpace{}%
\AgdaBound{a}\AgdaSpace{}%
\AgdaKeyword{where}\<%
\\
\>[2][@{}l@{\AgdaIndent{0}}]%
\>[4]\AgdaOperator{\AgdaInductiveConstructor{[\AgdaUnderscore{}]}}\AgdaSpace{}%
\AgdaSymbol{:}\AgdaSpace{}%
\AgdaSymbol{@0}\AgdaSpace{}%
\AgdaBound{A}\AgdaSpace{}%
\AgdaSymbol{→}\AgdaSpace{}%
\AgdaDatatype{Erased}\AgdaSpace{}%
\AgdaBound{A}\<%
\end{code}}

\begin{code}[hide]%
\>[2]\AgdaKeyword{postulate}\<%
\\
\>[2][@{}l@{\AgdaIndent{0}}]%
\>[4]\AgdaPostulate{\AgdaUnderscore{}}%
\>[665I]\AgdaSymbol{:}\<%
\end{code}
\newcommand{\ErasedA}{%
\mbox{\begin{code}[inline]%
\>[.][@{}l@{}]\<[665I]%
\>[6]\AgdaDatatype{Erased}\AgdaSpace{}%
\AgdaGeneralizable{A}\<%
\end{code}}}

One key type in this text is \Erased{}
\citep{mishra-linger-2008}.
The data type \ErasedA{} has a single constructor \boxop{} (``box'')
that takes a single, erased argument of type \AgdaBound{A} – it is
basically a type-level variant of \AgdaSymbol{@0}:\MCbreak{}
\begin{minipage}[t]{0.59\linewidth}
\ErasedDef{}%
\end{minipage}
\begin{minipage}[t]{0.39\linewidth}
\begin{code}%
\>[2]\AgdaSymbol{@}\AgdaNumber{0}\AgdaSpace{}%
\AgdaFunction{erased}\AgdaSpace{}%
\AgdaSymbol{:}\AgdaSpace{}%
\AgdaDatatype{Erased}\AgdaSpace{}%
\AgdaGeneralizable{A}\AgdaSpace{}%
\AgdaSymbol{→}\AgdaSpace{}%
\AgdaGeneralizable{A}\<%
\\
\>[2]\AgdaFunction{erased}\AgdaSpace{}%
\AgdaOperator{\AgdaInductiveConstructor{[}}\AgdaSpace{}%
\AgdaBound{x}\AgdaSpace{}%
\AgdaOperator{\AgdaInductiveConstructor{]}}\AgdaSpace{}%
\AgdaSymbol{=}\AgdaSpace{}%
\AgdaBound{x}\<%
\end{code}
\end{minipage}\\
\Erased{} comes with the projection function \AgdaFunction{erased},
which is erased.
Erased definitions like \AgdaFunction{erased} can only be used in
erased contexts, but their bodies can make use of anything that is
erased.

Note that the type argument \AgdaBound{A} of \Erased{} is erased.
This means that \Erased{}, seen as a function from types to types,
takes an erased argument.
This is not necessarily OK for all types and arguments, see the
discussion of \AgdaDatatype{Id\ensuremath{{}_{\mathrm{0}}}} in §~\ref{sec:box-cong};
\citet{abel-et-al-2021} provide more examples.

This text is mainly focused on a variant of \Erased{} without
η-equality, like the one above.
Initially I investigated a variant with η-equality.
However, in a setting with linear types such a type may not work very
well \citep{abel-et-al-2023}: the η-long normal form of the linear
identity function on \ErasedA{} is basically
\begin{code}[hide]%
\>[2]\AgdaFunction{\AgdaUnderscore{}}\AgdaSpace{}%
\AgdaSymbol{:}\AgdaSpace{}%
\AgdaDatatype{Erased}\AgdaSpace{}%
\AgdaGeneralizable{A}\AgdaSpace{}%
\AgdaSymbol{→}\AgdaSpace{}%
\AgdaDatatype{Erased}\AgdaSpace{}%
\AgdaGeneralizable{A}\<%
\\
\>[2]\AgdaSymbol{\AgdaUnderscore{}}%
\>[684I]\AgdaSymbol{=}\<%
\end{code}
\mbox{\begin{code}[inline]%
\>[.][@{}l@{}]\<[684I]%
\>[4]\AgdaSymbol{λ}\AgdaSpace{}%
\AgdaBound{x}\AgdaSpace{}%
\AgdaSymbol{→}\AgdaSpace{}%
\AgdaOperator{\AgdaInductiveConstructor{[}}\AgdaSpace{}%
\AgdaFunction{erased}\AgdaSpace{}%
\AgdaBound{x}\AgdaSpace{}%
\AgdaOperator{\AgdaInductiveConstructor{]}}\<%
\end{code}},
which is arguably not a linear function because the argument of the
box op\-er\-a\-tor is erased.
For that reason I have chosen to focus on a variant of \Erased{}
without η-equality.

Pattern matching on erased arguments is not allowed in non-erased
contexts for data types with two or more constructors.
However, such \emph{erased matches} \citep{abel-et-al-2023} are
allowed for the empty type, a data type with zero constructors.
Agda optionally allows erased matches for non-indexed,
single-constructor data types.
If such a match is allowed, then all the constructor's arguments are
treated as erased in the right-hand side.
For instance, \AgdaFunction{fst'} is disallowed, but
\AgdaFunction{fst''} is allowed:\MCbreak{}
\begin{minipage}[t]{0.49\linewidth}
\begin{code}[hide]%
\>[2]\AgdaSymbol{@}\AgdaNumber{0}\<%
\end{code}
\begin{code}%
\>[2][@{}l@{\AgdaIndent{1}}]%
\>[4]\AgdaFunction{fst'}\AgdaSpace{}%
\AgdaSymbol{:}\AgdaSpace{}%
\AgdaSymbol{@}\AgdaNumber{0}\AgdaSpace{}%
\AgdaGeneralizable{A}\AgdaSpace{}%
\AgdaOperator{\AgdaFunction{×}}\AgdaSpace{}%
\AgdaGeneralizable{B}\AgdaSpace{}%
\AgdaSymbol{→}\AgdaSpace{}%
\AgdaGeneralizable{A}\<%
\\
\>[2]\AgdaFunction{fst'}\AgdaSpace{}%
\AgdaSymbol{(}\AgdaBound{x}\AgdaSpace{}%
\AgdaOperator{\AgdaInductiveConstructor{,}}\AgdaSpace{}%
\AgdaSymbol{\AgdaUnderscore{})}\AgdaSpace{}%
\AgdaSymbol{=}\AgdaSpace{}%
\AgdaBound{x}\<%
\end{code}
\end{minipage}
\begin{minipage}[t]{0.49\linewidth}
\begin{code}%
\>[2]\AgdaFunction{fst''}\AgdaSpace{}%
\AgdaSymbol{:}\AgdaSpace{}%
\AgdaSymbol{@}\AgdaNumber{0}\AgdaSpace{}%
\AgdaGeneralizable{A}\AgdaSpace{}%
\AgdaOperator{\AgdaFunction{×}}\AgdaSpace{}%
\AgdaGeneralizable{B}\AgdaSpace{}%
\AgdaSymbol{→}\AgdaSpace{}%
\AgdaDatatype{Erased}\AgdaSpace{}%
\AgdaGeneralizable{A}\<%
\\
\>[2]\AgdaFunction{fst''}\AgdaSpace{}%
\AgdaSymbol{(}\AgdaBound{x}\AgdaSpace{}%
\AgdaOperator{\AgdaInductiveConstructor{,}}\AgdaSpace{}%
\AgdaSymbol{\AgdaUnderscore{})}\AgdaSpace{}%
\AgdaSymbol{=}\AgdaSpace{}%
\AgdaOperator{\AgdaInductiveConstructor{[}}\AgdaSpace{}%
\AgdaBound{x}\AgdaSpace{}%
\AgdaOperator{\AgdaInductiveConstructor{]}}\<%
\end{code}
\end{minipage}

\section{Quotient Types}
\label{sec:quotients}

Higher inductive types \citep{hott-2013} are inductive types that can
have \emph{higher constructors} imposing equalities on the values of
the type (or equalities between equalities of values, and so on).
One example is given by \emph{quotient types}
\citep{hofmann-1995,hott-2013}.
Quotient types are used in mathematics, but can also be used for
programming.
For instance, efficient data structures often allow a single logical
value to be represented in multiple ways.
Quotients allow one to make all representations that stand for the
same logical value propositionally equal (if ``standing for the same
logical value'' can be defined in a suitable way), and an advantage of
propositional equality over custom equivalence relations is that
propositional equality is necessarily preserved by all functions.

Erased postulates can be used to provide a simple implementation of
quotient types.
The technique presented below should work for more higher inductive
types, but the type theories presented in §~\ref{sec:typing} only have
one data type with higher constructors, that of quotients.
However, note that \emph{propositional truncation}, which is used
prominently in homotopy type theory \citep{hott-2013}, can be
implemented as a quotient with the trivial relation.

The following interface, along with a specific computation rule, is
shown to be safe (the interface is based on the set quotient HIT, due
to Qian and Mörtberg, in the cubical library):
\QuotAPI{}%
The type \QuotAR{} is \AgdaBound{A} quotiented by the relation
\AgdaBound{R}.
The \emph{point constructor} \quotbox{} takes a value to its
equivalence class.
The higher constructor \AgdaInductiveConstructor{resp} ensures that
the equivalence classes of related values are equal, and the higher
constructor \AgdaInductiveConstructor{set} ensures that the quotients
are \emph{set-truncated}, following \citet{hott-2013}: \Isset{} states
that \AgdaBound{A} is a \emph{set}, i.e.\ a type for which all
identity proofs are equal.
(If one chooses to work in a setting with uniqueness of identity
proofs, where every type is a set, then one can ignore this truncation
constructor and the corresponding argument of the eliminator.)

The higher constructors are erased.
In turn the corresponding arguments of the eliminator
\AgdaPostulate{qrec} are also erased.
The eliminator computes in the usual way when applied to the point
constructor (\qrecdef{}), but no computation rules are associated to
the higher constructors, unlike quotient types implemented in cubical
type theory \citep{vezzosi-et-al-2019}.
However, cubical type theory is arguably rather complicated.
This interface – a variant of Licata's trick \citep{hott-2013} –
should be comparatively easy to implement given a type theory with
suitable erasure annotations.

A drawback of this simple interface is that terms can get stuck.
However, the fact that the higher constructors are erased ensures that
programs of type $ℕ$ will compute to the correct numeral, and this is
true also in the presence of consistent erased postulates (see
Theorem \ref{thm:soundness-of-erasure}).

\begin{code}[hide]%
\>[2]\AgdaFunction{Is‐prop}\AgdaSpace{}%
\AgdaSymbol{:}\AgdaSpace{}%
\AgdaPrimitive{Type}\AgdaSpace{}%
\AgdaGeneralizable{a}\AgdaSpace{}%
\AgdaSymbol{→}\AgdaSpace{}%
\AgdaPrimitive{Type}\AgdaSpace{}%
\AgdaGeneralizable{a}\<%
\end{code}
\newcommand{\IsPropositionDef}{%
\begin{code}[inline]%
\>[2]\AgdaFunction{Is‐prop}\AgdaSpace{}%
\AgdaBound{A}\AgdaSpace{}%
\AgdaSymbol{=}\AgdaSpace{}%
\AgdaSymbol{(}\AgdaBound{x}\AgdaSpace{}%
\AgdaBound{y}\AgdaSpace{}%
\AgdaSymbol{:}\AgdaSpace{}%
\AgdaBound{A}\AgdaSymbol{)}\AgdaSpace{}%
\AgdaSymbol{→}\AgdaSpace{}%
\AgdaDatatype{Id}\AgdaSpace{}%
\AgdaBound{A}\AgdaSpace{}%
\AgdaBound{x}\AgdaSpace{}%
\AgdaBound{y}\<%
\end{code}}

\begin{code}[hide]%
\>[2]\AgdaKeyword{postulate}\<%
\\
\>[2][@{}l@{\AgdaIndent{0}}]%
\>[4]\AgdaPostulate{Is‐equivalence‐relation}\AgdaSpace{}%
\AgdaSymbol{:}\AgdaSpace{}%
\AgdaSymbol{(}\AgdaGeneralizable{A}\AgdaSpace{}%
\AgdaSymbol{→}\AgdaSpace{}%
\AgdaGeneralizable{A}\AgdaSpace{}%
\AgdaSymbol{→}\AgdaSpace{}%
\AgdaPrimitive{Type}\AgdaSpace{}%
\AgdaGeneralizable{a}\AgdaSymbol{)}\AgdaSpace{}%
\AgdaSymbol{→}\AgdaSpace{}%
\AgdaPrimitive{Type}\AgdaSpace{}%
\AgdaGeneralizable{a}\<%
\end{code}

One postulate that one may in particular want to use is that of
\emph{propositional extensionality}:
\begin{code}%
\>[2]\AgdaKeyword{postulate}\AgdaSpace{}%
\AgdaSymbol{@}\AgdaNumber{0}\AgdaSpace{}%
\AgdaPostulate{propext}\AgdaSpace{}%
\AgdaSymbol{:}\AgdaSpace{}%
\AgdaFunction{Is‐prop}\AgdaSpace{}%
\AgdaGeneralizable{A}\AgdaSpace{}%
\AgdaSymbol{→}\AgdaSpace{}%
\AgdaFunction{Is‐prop}\AgdaSpace{}%
\AgdaGeneralizable{B}\AgdaSpace{}%
\AgdaSymbol{→}\AgdaSpace{}%
\AgdaGeneralizable{A}\AgdaSpace{}%
\AgdaOperator{\AgdaRecord{⇔}}\AgdaSpace{}%
\AgdaGeneralizable{B}\AgdaSpace{}%
\AgdaSymbol{→}\AgdaSpace{}%
\AgdaDatatype{Id}\AgdaSpace{}%
\AgdaSymbol{(}\AgdaPrimitive{Type}\AgdaSpace{}%
\AgdaGeneralizable{a}\AgdaSymbol{)}\AgdaSpace{}%
\AgdaGeneralizable{A}\AgdaSpace{}%
\AgdaGeneralizable{B}\<%
\end{code}
A type is a proposition if all its elements are equal
\citep{hott-2013}: \mbox{\IsPropositionDef{}}.
Propositional extensionality, which follows from univalence, ensures
that any logically equivalent propositions in the same universe are
equal.
Using erased propositional extensionality and erased function
extensionality (sometimes abbreviated to propext and funext) one can
prove that certain quotient relations \AgdaBound{R} are
\emph{effective} \citep{hofmann-1995,hott-2013,chapman-et-al-2015},
formulated in the following way:
\begin{code}[hide]%
\>[2]\AgdaKeyword{postulate}\<%
\\
\>[2][@{}l@{\AgdaIndent{0}}]%
\>[4]\AgdaPostulate{\AgdaUnderscore{}}%
\>[763I]\AgdaSymbol{:}\AgdaSpace{}%
\AgdaKeyword{let}\AgdaSpace{}%
\AgdaKeyword{open}\AgdaSpace{}%
\AgdaModule{Quotient}\AgdaSpace{}%
\AgdaKeyword{in}\<%
\end{code}
\mbox{\begin{code}[inline]%
\>[.][@{}l@{}]\<[763I]%
\>[6]\AgdaDatatype{Id}\AgdaSpace{}%
\AgdaSymbol{(}\AgdaGeneralizable{A}\AgdaSpace{}%
\AgdaOperator{\AgdaDatatype{/}}\AgdaSpace{}%
\AgdaGeneralizable{R}\AgdaSymbol{)}\AgdaSpace{}%
\AgdaOperator{\AgdaInductiveConstructor{[}}\AgdaSpace{}%
\AgdaGeneralizable{x}\AgdaSpace{}%
\AgdaOperator{\AgdaInductiveConstructor{]}}\AgdaSpace{}%
\AgdaOperator{\AgdaInductiveConstructor{[}}\AgdaSpace{}%
\AgdaGeneralizable{y}\AgdaSpace{}%
\AgdaOperator{\AgdaInductiveConstructor{]}}\AgdaSpace{}%
\AgdaSymbol{→}\AgdaSpace{}%
\AgdaGeneralizable{R}\AgdaSpace{}%
\AgdaGeneralizable{x}\AgdaSpace{}%
\AgdaGeneralizable{y}\<%
\end{code}}.
In erased contexts all propositional equivalence relations are
effective, and one can even prove something that holds in non-erased
contexts:
\begin{code}[hide]%
\>[2]\AgdaKeyword{postulate}\<%
\\
\>[2][@{}l@{\AgdaIndent{0}}]%
\>[4]\AgdaPostulate{\AgdaUnderscore{}}%
\>[781I]\AgdaSymbol{:}\AgdaSpace{}%
\AgdaKeyword{let}\AgdaSpace{}%
\AgdaKeyword{open}\AgdaSpace{}%
\AgdaModule{Quotient}\AgdaSpace{}%
\AgdaKeyword{in}\<%
\end{code}
\begin{code}%
\>[.][@{}l@{}]\<[781I]%
\>[6]\AgdaSymbol{@0}\AgdaSpace{}%
\AgdaPostulate{Is‐equivalence‐relation}\AgdaSpace{}%
\AgdaGeneralizable{R}\AgdaSpace{}%
\AgdaSymbol{→}\AgdaSpace{}%
\AgdaSymbol{@0}\AgdaSpace{}%
\AgdaFunction{Is‐prop}\AgdaSpace{}%
\AgdaSymbol{(}\AgdaGeneralizable{R}\AgdaSpace{}%
\AgdaGeneralizable{x}\AgdaSpace{}%
\AgdaGeneralizable{y}\AgdaSymbol{)}\AgdaSpace{}%
\AgdaSymbol{→}\AgdaSpace{}%
\AgdaGeneralizable{R}\AgdaSpace{}%
\AgdaGeneralizable{x}\AgdaSpace{}%
\AgdaGeneralizable{x}\AgdaSpace{}%
\AgdaSymbol{→}\AgdaSpace{}%
\AgdaDatatype{Id}\AgdaSpace{}%
\AgdaSymbol{(}\AgdaGeneralizable{A}\AgdaSpace{}%
\AgdaOperator{\AgdaDatatype{/}}\AgdaSpace{}%
\AgdaGeneralizable{R}\AgdaSymbol{)}\AgdaSpace{}%
\AgdaOperator{\AgdaInductiveConstructor{[}}\AgdaSpace{}%
\AgdaGeneralizable{x}\AgdaSpace{}%
\AgdaOperator{\AgdaInductiveConstructor{]}}\AgdaSpace{}%
\AgdaOperator{\AgdaInductiveConstructor{[}}\AgdaSpace{}%
\AgdaGeneralizable{y}\AgdaSpace{}%
\AgdaOperator{\AgdaInductiveConstructor{]}}\AgdaSpace{}%
\AgdaSymbol{→}\AgdaSpace{}%
\AgdaGeneralizable{R}\AgdaSpace{}%
\AgdaGeneralizable{x}\AgdaSpace{}%
\AgdaGeneralizable{y}\<%
\end{code}

As an example of what effectiveness can be used for, let us try to
decide erased equality for \mbox{\QuotAR{}}.
What does ``decide erased equality'' mean?
This can be formulated using the type
\begin{code}[hide]%
\>[2]\AgdaFunction{Dec‐Erased}\AgdaSpace{}%
\AgdaSymbol{:}\AgdaSpace{}%
\AgdaPrimitive{Type}\AgdaSpace{}%
\AgdaGeneralizable{a}\AgdaSpace{}%
\AgdaSymbol{→}\AgdaSpace{}%
\AgdaPrimitive{Type}\AgdaSpace{}%
\AgdaGeneralizable{a}\<%
\end{code}
\begin{code}[inline]%
\>[2]\AgdaFunction{Dec‐Erased}\AgdaSpace{}%
\AgdaBound{A}\AgdaSpace{}%
\AgdaSymbol{=}\AgdaSpace{}%
\AgdaDatatype{Erased}\AgdaSpace{}%
\AgdaBound{A}\AgdaSpace{}%
\AgdaOperator{\AgdaDatatype{⊎}}\AgdaSpace{}%
\AgdaDatatype{Erased}\AgdaSpace{}%
\AgdaSymbol{(}\AgdaOperator{\AgdaFunction{¬}}\AgdaSpace{}%
\AgdaBound{A}\AgdaSymbol{)}\<%
\end{code}.
A value of this type gives you either an erased proof that
\AgdaBound{A} is inhabited, or an erased proof that \AgdaBound{A} is
not inhabited.
Now let us try to prove
\begin{code}[hide]%
\>[2]\AgdaKeyword{postulate}\<%
\\
\>[2][@{}l@{\AgdaIndent{0}}]%
\>[4]\AgdaPostulate{\AgdaUnderscore{}}%
\>[827I]\AgdaSymbol{:}\AgdaSpace{}%
\AgdaKeyword{let}\AgdaSpace{}%
\AgdaKeyword{open}\AgdaSpace{}%
\AgdaModule{Quotient}\AgdaSpace{}%
\AgdaKeyword{in}\<%
\end{code}
\mbox{\begin{code}[inline]%
\>[.][@{}l@{}]\<[827I]%
\>[6]\AgdaSymbol{(∀}\AgdaSpace{}%
\AgdaBound{x}\AgdaSpace{}%
\AgdaBound{y}\AgdaSpace{}%
\AgdaSymbol{→}\AgdaSpace{}%
\AgdaFunction{Dec‐Erased}\AgdaSpace{}%
\AgdaSymbol{(}\AgdaGeneralizable{R}\AgdaSpace{}%
\AgdaBound{x}\AgdaSpace{}%
\AgdaBound{y}\AgdaSymbol{))}\AgdaSpace{}%
\AgdaSymbol{→}\<%
\end{code}}
\mbox{\begin{code}[inline]%
\>[6]\AgdaSymbol{∀}\AgdaSpace{}%
\AgdaBound{x}\AgdaSpace{}%
\AgdaBound{y}\AgdaSpace{}%
\AgdaSymbol{→}\AgdaSpace{}%
\AgdaFunction{Dec‐Erased}\AgdaSpace{}%
\AgdaSymbol{(}\AgdaDatatype{Id}\AgdaSpace{}%
\AgdaSymbol{(}\AgdaGeneralizable{A}\AgdaSpace{}%
\AgdaOperator{\AgdaDatatype{/}}\AgdaSpace{}%
\AgdaGeneralizable{R}\AgdaSymbol{)}\AgdaSpace{}%
\AgdaBound{x}\AgdaSpace{}%
\AgdaBound{y}\AgdaSymbol{)}\<%
\end{code}}.
We use \AgdaFunction{qrec} to match on \AgdaBound{x} and
\AgdaBound{y}: the arguments corresponding to the higher constructors
can be filled in using the fact that, in erased contexts, erased
funext implies that
\begin{code}[hide]%
\>[2]\AgdaKeyword{postulate}\<%
\\
\>[2][@{}l@{\AgdaIndent{0}}]%
\>[4]\AgdaPostulate{\AgdaUnderscore{}}%
\>[850I]\AgdaSymbol{:}\AgdaSpace{}%
\AgdaKeyword{let}\AgdaSpace{}%
\AgdaKeyword{open}\AgdaSpace{}%
\AgdaModule{Quotient}\AgdaSpace{}%
\AgdaKeyword{in}\<%
\end{code}
\begin{code}[inline*]%
\>[.][@{}l@{}]\<[850I]%
\>[6]\AgdaFunction{Dec‐Erased}\AgdaSpace{}%
\AgdaSymbol{(}\AgdaDatatype{Id}\AgdaSpace{}%
\AgdaSymbol{(}\AgdaGeneralizable{A}\AgdaSpace{}%
\AgdaOperator{\AgdaDatatype{/}}\AgdaSpace{}%
\AgdaGeneralizable{R}\AgdaSymbol{)}\AgdaSpace{}%
\AgdaGeneralizable{x}\AgdaSpace{}%
\AgdaGeneralizable{y}\AgdaSymbol{)}\<%
\end{code}
is a proposition (note that every proposition is a set).
We get
\begin{code}[hide]%
\>[2]\AgdaKeyword{postulate}\<%
\\
\>[2][@{}l@{\AgdaIndent{0}}]%
\>[4]\AgdaPostulate{\AgdaUnderscore{}}%
\>[861I]\AgdaSymbol{:}\AgdaSpace{}%
\AgdaOperator{\AgdaGeneralizable{Has‐type}}\<%
\end{code}
\mbox{\begin{code}[inline]%
\>[.][@{}l@{}]\<[861I]%
\>[6]\AgdaOperator{\AgdaInductiveConstructor{[}}\AgdaSpace{}%
\AgdaGeneralizable{x′}\AgdaSpace{}%
\AgdaOperator{\AgdaInductiveConstructor{]}}\<%
\end{code}}
and
\begin{code}[hide]%
\>[2]\AgdaKeyword{postulate}\<%
\\
\>[2][@{}l@{\AgdaIndent{0}}]%
\>[4]\AgdaPostulate{\AgdaUnderscore{}}%
\>[865I]\AgdaSymbol{:}\AgdaSpace{}%
\AgdaOperator{\AgdaGeneralizable{Has‐type}}\<%
\end{code}
\mbox{\begin{code}[inline]%
\>[.][@{}l@{}]\<[865I]%
\>[6]\AgdaOperator{\AgdaInductiveConstructor{[}}\AgdaSpace{}%
\AgdaGeneralizable{y′}\AgdaSpace{}%
\AgdaOperator{\AgdaInductiveConstructor{]}}\<%
\end{code}},
and check whether
\begin{code}[hide]%
\>[2]\AgdaKeyword{postulate}\<%
\\
\>[2][@{}l@{\AgdaIndent{0}}]%
\>[4]\AgdaPostulate{\AgdaUnderscore{}}%
\>[869I]\AgdaSymbol{:}\<%
\end{code}
\begin{code}[inline*]%
\>[.][@{}l@{}]\<[869I]%
\>[6]\AgdaGeneralizable{R}\AgdaSpace{}%
\AgdaGeneralizable{x′}\AgdaSpace{}%
\AgdaGeneralizable{y′}\<%
\end{code}
holds.
If it does, then we use \AgdaInductiveConstructor{resp} to conclude that
\AgdaBound{x} is equal to \AgdaBound{y}.
If not, then we can conclude that \AgdaBound{x} is not equal to
\AgdaBound{y} if \AgdaBound{R} is effective in erased contexts.

In order to demonstrate that the quotient interface presented here
works in practice I have implemented finite sets as red-black trees
quotiented by set equivalence.
(I implemented the empty set, insertion, deletion, the membership
relation, and a membership test, along with proofs of
correctness.)

The red-black trees are implemented following \citet{okasaki-1999},
using a technique due to \citet{mcbride-2014} to handle ordering
constraints.
They are parametrised by a strict total order with a comparison
function of type
\begin{code}[hide]%
\>[2]\AgdaKeyword{postulate}\<%
\\
\>[2][@{}l@{\AgdaIndent{0}}]%
\>[4]\AgdaPostulate{\AgdaUnderscore{}}%
\>[872I]\AgdaSymbol{:}\AgdaSpace{}%
\AgdaSymbol{(}\AgdaOperator{\AgdaBound{\AgdaUnderscore{}<\AgdaUnderscore{}}}\AgdaSpace{}%
\AgdaSymbol{:}\AgdaSpace{}%
\AgdaGeneralizable{A}\AgdaSpace{}%
\AgdaSymbol{→}\AgdaSpace{}%
\AgdaGeneralizable{A}\AgdaSpace{}%
\AgdaSymbol{→}\AgdaSpace{}%
\AgdaPrimitive{Type}\AgdaSpace{}%
\AgdaGeneralizable{a}\AgdaSymbol{)}\AgdaSpace{}%
\AgdaSymbol{→}\<%
\end{code}
\begin{code}[inline]%
\>[.][@{}l@{}]\<[872I]%
\>[6]\AgdaSymbol{∀}\AgdaSpace{}%
\AgdaBound{x}\AgdaSpace{}%
\AgdaBound{y}\AgdaSpace{}%
\AgdaSymbol{→}\AgdaSpace{}%
\AgdaDatatype{Erased}\AgdaSpace{}%
\AgdaSymbol{(}\AgdaBound{x}\AgdaSpace{}%
\AgdaOperator{\AgdaBound{<}}\AgdaSpace{}%
\AgdaBound{y}\AgdaSymbol{)}\AgdaSpace{}%
\AgdaOperator{\AgdaDatatype{⊎}}\AgdaSpace{}%
\AgdaDatatype{Erased}\AgdaSpace{}%
\AgdaSymbol{(}\AgdaDatatype{Id}\AgdaSpace{}%
\AgdaGeneralizable{A}\AgdaSpace{}%
\AgdaBound{x}\AgdaSpace{}%
\AgdaBound{y}\AgdaSymbol{)}\AgdaSpace{}%
\AgdaOperator{\AgdaDatatype{⊎}}\AgdaSpace{}%
\AgdaDatatype{Erased}\AgdaSpace{}%
\AgdaSymbol{(}\AgdaBound{y}\AgdaSpace{}%
\AgdaOperator{\AgdaBound{<}}\AgdaSpace{}%
\AgdaBound{x}\AgdaSymbol{)}\<%
\end{code},
and the tree invariants are erased.

\begin{code}[hide]%
\>[2]\AgdaKeyword{module}\AgdaSpace{}%
\AgdaModule{\AgdaUnderscore{}}\AgdaSpace{}%
\AgdaSymbol{(}\AgdaBound{a}\AgdaSpace{}%
\AgdaSymbol{:}\AgdaSpace{}%
\AgdaPostulate{Level}\AgdaSymbol{)}\AgdaSpace{}%
\AgdaSymbol{(}\AgdaBound{A}\AgdaSpace{}%
\AgdaSymbol{:}\AgdaSpace{}%
\AgdaPrimitive{Type}\AgdaSpace{}%
\AgdaBound{a}\AgdaSymbol{)}\AgdaSpace{}%
\AgdaKeyword{where}\<%
\\
\>[2][@{}l@{\AgdaIndent{0}}]%
\>[4]\AgdaKeyword{open}\AgdaSpace{}%
\AgdaModule{Quotient}\<%
\\
\\[\AgdaEmptyExtraSkip]%
\>[4]\AgdaKeyword{postulate}\<%
\\
\>[4][@{}l@{\AgdaIndent{0}}]%
\>[6]\AgdaPostulate{Tree}\AgdaSpace{}%
\AgdaSymbol{:}\AgdaSpace{}%
\AgdaPrimitive{Type}\AgdaSpace{}%
\AgdaBound{a}\<%
\end{code}
\newcommand{\treeMembershipType}{%
  \begin{code}[inline]%
\>[6]\AgdaOperator{\AgdaPostulate{\AgdaUnderscore{}∈ᵀ\AgdaUnderscore{}}}\AgdaSpace{}%
\AgdaSymbol{:}\AgdaSpace{}%
\AgdaBound{A}\AgdaSpace{}%
\AgdaSymbol{→}\AgdaSpace{}%
\AgdaPostulate{Tree}\AgdaSpace{}%
\AgdaSymbol{→}\AgdaSpace{}%
\AgdaPrimitive{Type}\AgdaSpace{}%
\AgdaBound{a}\<%
\end{code}}

The finite sets are defined by
\begin{code}[hide]%
\>[4]\AgdaFunction{Set}\AgdaSpace{}%
\AgdaSymbol{:}\AgdaSpace{}%
\AgdaPrimitive{Type}\AgdaSpace{}%
\AgdaBound{a}\<%
\\
\>[4]\AgdaFunction{Set}\AgdaSpace{}%
\AgdaSymbol{=}\<%
\end{code}
\begin{code}[inline]%
\>[4][@{}l@{\AgdaIndent{1}}]%
\>[6]\AgdaPostulate{Tree}\AgdaSpace{}%
\AgdaOperator{\AgdaDatatype{/}}\AgdaSpace{}%
\AgdaSymbol{(λ}\AgdaSpace{}%
\AgdaBound{t\ensuremath{{}_{\mathrm{1}}}}\AgdaSpace{}%
\AgdaBound{t\ensuremath{{}_{\mathrm{2}}}}\AgdaSpace{}%
\AgdaSymbol{→}\AgdaSpace{}%
\AgdaSymbol{∀}\AgdaSpace{}%
\AgdaBound{x}\AgdaSpace{}%
\AgdaSymbol{→}\AgdaSpace{}%
\AgdaBound{x}\AgdaSpace{}%
\AgdaOperator{\AgdaPostulate{∈ᵀ}}\AgdaSpace{}%
\AgdaBound{t\ensuremath{{}_{\mathrm{1}}}}\AgdaSpace{}%
\AgdaOperator{\AgdaRecord{⇔}}\AgdaSpace{}%
\AgdaBound{x}\AgdaSpace{}%
\AgdaOperator{\AgdaPostulate{∈ᵀ}}\AgdaSpace{}%
\AgdaBound{t\ensuremath{{}_{\mathrm{2}}}}\AgdaSymbol{)}\<%
\end{code}.
Let me explain how tree membership (\treeMembershipType{}) can be
lifted to set membership.
One cannot simply define
\begin{code}[hide]%
\>[4]\AgdaOperator{\AgdaFunction{\AgdaUnderscore{}∈\AgdaUnderscore{}}}\AgdaSpace{}%
\AgdaSymbol{:}\AgdaSpace{}%
\AgdaBound{A}\AgdaSpace{}%
\AgdaSymbol{→}\AgdaSpace{}%
\AgdaFunction{Set}\AgdaSpace{}%
\AgdaSymbol{→}\AgdaSpace{}%
\AgdaPrimitive{Type}\AgdaSpace{}%
\AgdaBound{a}\<%
\end{code}
\mbox{\begin{code}[inline]%
\>[4]\AgdaBound{x}\AgdaSpace{}%
\AgdaOperator{\AgdaFunction{∈}}\AgdaSpace{}%
\AgdaOperator{\AgdaInductiveConstructor{[}}\AgdaSpace{}%
\AgdaBound{t}\AgdaSpace{}%
\AgdaOperator{\AgdaInductiveConstructor{]}}\AgdaSpace{}%
\AgdaSymbol{=}\AgdaSpace{}%
\AgdaBound{x}\AgdaSpace{}%
\AgdaOperator{\AgdaPostulate{∈ᵀ}}\AgdaSpace{}%
\AgdaBound{t}\<%
\end{code}},
because the eliminator \AgdaPostulate{qrec} requires that the motive
\AgdaBound{P} is a family of sets (in erased contexts), and in the
presence of univalence
\begin{code}[hide]%
\>[4]\AgdaKeyword{postulate}\<%
\\
\>[4][@{}l@{\AgdaIndent{0}}]%
\>[6]\AgdaPostulate{\AgdaUnderscore{}}%
\>[954I]\AgdaSymbol{:}\<%
\end{code}
\begin{code}[inline*]%
\>[.][@{}l@{}]\<[954I]%
\>[8]\AgdaPrimitive{Type}\AgdaSpace{}%
\AgdaBound{a}\<%
\end{code}
is not a set.
However, in erased contexts one can use propext and funext to prove
that the Σ-type
\begin{code}[hide]%
\>[4]\AgdaKeyword{postulate}\<%
\\
\>[4][@{}l@{\AgdaIndent{0}}]%
\>[6]\AgdaPostulate{\AgdaUnderscore{}}%
\>[956I]\AgdaSymbol{:}\<%
\end{code}
\begin{code}[inline*]%
\>[.][@{}l@{}]\<[956I]%
\>[8]\AgdaSymbol{(}%
\AgdaBound{A}\AgdaSpace{}%
\AgdaSymbol{:}\AgdaSpace{}%
\AgdaPrimitive{Type}\AgdaSpace{}%
\AgdaBound{a}\AgdaSymbol{)}\AgdaSpace{}%
\AgdaFunction{×}\AgdaSpace{}%
\AgdaDatatype{Erased}\AgdaSpace{}%
\AgdaSymbol{(}\AgdaFunction{Is‐prop}\AgdaSpace{}%
\AgdaBound{A}\AgdaSymbol{)}\<%
\end{code}
is a set, so one can use \AgdaFunction{qrec} to define types that are
propositions, and it suffices that one can prove that they are
propositions in erased contexts.
Such a proof is possible for tree membership.
If tree membership were not propositional, then one could have used
propositional truncation to turn the type into a proposition – and as
mentioned above propositional truncation can be implemented as a
quotient.

The use of \AgdaPostulate{qrec} in the implementation of
\AgdaFunction{\AgdaUnderscore{}∈\AgdaUnderscore{}} also requires that
the ``point function'' respects the relation (in erased contexts):
in the non-dependent case where \AgdaBound{P} is
\begin{code}[hide]%
\>[4]\AgdaKeyword{postulate}\<%
\\
\>[4][@{}l@{\AgdaIndent{0}}]%
\>[6]\AgdaPostulate{\AgdaUnderscore{}}%
\>[966I]\AgdaSymbol{:}\AgdaSpace{}%
\AgdaOperator{\AgdaGeneralizable{Has‐type[}}\AgdaSpace{}%
\AgdaSymbol{(}\AgdaBound{A}\AgdaSpace{}%
\AgdaOperator{\AgdaDatatype{/}}\AgdaSpace{}%
\AgdaGeneralizable{R}\AgdaSpace{}%
\AgdaSymbol{→}\AgdaSpace{}%
\AgdaPrimitive{Type}\AgdaSpace{}%
\AgdaGeneralizable{p}\AgdaSymbol{)}\AgdaSpace{}%
\AgdaOperator{\AgdaGeneralizable{]}}\<%
\end{code}
\begin{code}[inline*]%
\>[.][@{}l@{}]\<[966I]%
\>[8]\AgdaSymbol{λ}\AgdaSpace{}%
\AgdaBound{\AgdaUnderscore{}}\AgdaSpace{}%
\AgdaSymbol{→}\AgdaSpace{}%
\AgdaGeneralizable{B}\<%
\end{code}
the third explicit argument of \AgdaPostulate{qrec} can be simplified
to
\begin{code}[hide]%
\>[4]\AgdaKeyword{postulate}\<%
\\
\>[4][@{}l@{\AgdaIndent{0}}]%
\>[6]\AgdaPostulate{\AgdaUnderscore{}}%
\>[978I]\AgdaSymbol{:}\AgdaSpace{}%
\AgdaSymbol{(}\AgdaBound{f\,}\AgdaSpace{}%
\AgdaSymbol{:}\AgdaSpace{}%
\AgdaBound{A}\AgdaSpace{}%
\AgdaSymbol{→}\AgdaSpace{}%
\AgdaGeneralizable{B}\AgdaSymbol{)}\AgdaSpace{}%
\AgdaSymbol{→}\<%
\end{code}
\begin{code}[inline]%
\>[.][@{}l@{}]\<[978I]%
\>[8]\AgdaSymbol{∀}\AgdaSpace{}%
\AgdaSymbol{\{}\AgdaBound{x}\AgdaSpace{}%
\AgdaBound{y}\AgdaSymbol{\}}\AgdaSpace{}%
\AgdaSymbol{→}\AgdaSpace{}%
\AgdaGeneralizable{R}\AgdaSpace{}%
\AgdaBound{x}\AgdaSpace{}%
\AgdaBound{y}\AgdaSpace{}%
\AgdaSymbol{→}\AgdaSpace{}%
\AgdaDatatype{Id}\AgdaSpace{}%
\AgdaGeneralizable{B}\AgdaSpace{}%
\AgdaSymbol{(}\AgdaBound{f\,}\AgdaSpace{}%
\AgdaBound{x}\AgdaSymbol{)}\AgdaSpace{}%
\AgdaSymbol{(}\AgdaBound{f\,}\AgdaSpace{}%
\AgdaBound{y}\AgdaSymbol{)}\<%
\end{code}.
Again one can use propext, which ensures that
\begin{code}[hide]%
\>[4]\AgdaKeyword{postulate}\<%
\\
\>[4][@{}l@{\AgdaIndent{0}}]%
\>[6]\AgdaPostulate{\AgdaUnderscore{}}%
\>[998I]\AgdaSymbol{:}\<%
\end{code}
\mbox{\begin{code}[inline]%
\>[.][@{}l@{}]\<[998I]%
\>[8]\AgdaGeneralizable{x}\AgdaSpace{}%
\AgdaOperator{\AgdaPostulate{∈ᵀ}}\AgdaSpace{}%
\AgdaGeneralizable{t\ensuremath{{}_{\mathrm{1}}}}\AgdaSpace{}%
\AgdaOperator{\AgdaRecord{⇔}}\AgdaSpace{}%
\AgdaGeneralizable{x}\AgdaSpace{}%
\AgdaOperator{\AgdaPostulate{∈ᵀ}}\AgdaSpace{}%
\AgdaGeneralizable{t\ensuremath{{}_{\mathrm{2}}}}\<%
\end{code}}
implies
\begin{code}[hide]%
\>[4]\AgdaKeyword{postulate}\<%
\\
\>[4][@{}l@{\AgdaIndent{0}}]%
\>[6]\AgdaPostulate{\AgdaUnderscore{}}%
\>[1005I]\AgdaSymbol{:}\<%
\end{code}
\mbox{\begin{code}[inline]%
\>[.][@{}l@{}]\<[1005I]%
\>[8]\AgdaDatatype{Id}\AgdaSpace{}%
\AgdaSymbol{(}\AgdaPrimitive{Type}\AgdaSpace{}%
\AgdaBound{a}\AgdaSymbol{)}\AgdaSpace{}%
\AgdaSymbol{(}\AgdaGeneralizable{x}\AgdaSpace{}%
\AgdaOperator{\AgdaPostulate{∈ᵀ}}\AgdaSpace{}%
\AgdaGeneralizable{t\ensuremath{{}_{\mathrm{1}}}}\AgdaSymbol{)}\AgdaSpace{}%
\AgdaSymbol{(}\AgdaGeneralizable{x}\AgdaSpace{}%
\AgdaOperator{\AgdaPostulate{∈ᵀ}}\AgdaSpace{}%
\AgdaGeneralizable{t\ensuremath{{}_{\mathrm{2}}}}\AgdaSymbol{)}\<%
\end{code}}.
Using funext one can prove that
\begin{code}[hide]%
\>[4]\AgdaKeyword{postulate}\<%
\\
\>[4][@{}l@{\AgdaIndent{0}}]%
\>[6]\AgdaPostulate{\AgdaUnderscore{}}%
\>[1014I]\AgdaSymbol{:}\<%
\end{code}
\mbox{\begin{code}[inline]%
\>[.][@{}l@{}]\<[1014I]%
\>[8]\AgdaDatatype{Erased}\AgdaSpace{}%
\AgdaSymbol{(}\AgdaFunction{Is‐prop}\AgdaSpace{}%
\AgdaBound{A}\AgdaSymbol{)}\<%
\end{code}}
is a proposition, which allows us to conclude.

\section{Erased Postulates and Erased Matches}
\label{sec:erased-postulates-and-erased-matches}

The previous section illustrates one way of working with erasure:
``Proofs'', like the higher quotient constructors or the invariants of
the red-black tree data structure, are marked as erased.
``Programs'' are not erased, but can contain erased parts, like some
arguments of \AgdaFunction{qrec} or the results of the comparison
function used for the red-black trees.
This separation allows postulates like funext and propext to
be used in proofs, without affecting the computational behaviour of
the non-erased parts of a program.
However, sometimes one may wish to use an erased proof to show that
something non-erased is type-correct.

In order to illustrate this, let us discuss two definitions of vectors
(length-indexed lists):\MCbreak{}
\begin{minipage}[t]{0.58\linewidth}
\begin{code}%
\>[2]\AgdaKeyword{data}\AgdaSpace{}%
\AgdaDatatype{Vec\ensuremath{{}^{\mkern1mu\mathrm{D}}}}\AgdaSpace{}%
\AgdaSymbol{(}\AgdaBound{A}\AgdaSpace{}%
\AgdaSymbol{:}\AgdaSpace{}%
\AgdaPrimitive{Type}\AgdaSpace{}%
\AgdaGeneralizable{a}\AgdaSymbol{)}\AgdaSpace{}%
\AgdaSymbol{:}\AgdaSpace{}%
\AgdaSymbol{@0}\AgdaSpace{}%
\AgdaDatatype{ℕ}\AgdaSpace{}%
\AgdaSymbol{→}\AgdaSpace{}%
\AgdaPrimitive{Type}\AgdaSpace{}%
\AgdaBound{a}\AgdaSpace{}%
\AgdaKeyword{where}\<%
\\
\>[2][@{}l@{\AgdaIndent{0}}]%
\>[4]\AgdaInductiveConstructor{[\ensuremath{\mkern1.5mu}]}%
\>[9]\AgdaSymbol{:}\AgdaSpace{}%
\AgdaDatatype{Vec\ensuremath{{}^{\mkern1mu\mathrm{D}}}}\AgdaSpace{}%
\AgdaBound{A}\AgdaSpace{}%
\AgdaInductiveConstructor{zero}\<%
\\
\>[4]\AgdaOperator{\AgdaInductiveConstructor{\AgdaUnderscore{}∷\AgdaUnderscore{}}}%
\>[9]\AgdaSymbol{:}\AgdaSpace{}%
\AgdaSymbol{\{}\AgdaSymbol{@0}\AgdaSpace{}%
\AgdaBound{n}\AgdaSpace{}%
\AgdaSymbol{:}\AgdaSpace{}%
\AgdaDatatype{ℕ}\AgdaSymbol{\}}\AgdaSpace{}%
\AgdaSymbol{→}\AgdaSpace{}%
\AgdaBound{A}\AgdaSpace{}%
\AgdaSymbol{→}\AgdaSpace{}%
\AgdaDatatype{Vec\ensuremath{{}^{\mkern1mu\mathrm{D}}}}\AgdaSpace{}%
\AgdaBound{A}\AgdaSpace{}%
\AgdaBound{n}\AgdaSpace{}%
\AgdaSymbol{→}\AgdaSpace{}%
\AgdaDatatype{Vec\ensuremath{{}^{\mkern1mu\mathrm{D}}}}\AgdaSpace{}%
\AgdaBound{A}\AgdaSpace{}%
\AgdaSymbol{(}\AgdaInductiveConstructor{suc}\AgdaSpace{}%
\AgdaBound{n}\AgdaSymbol{)}\<%
\end{code}
\end{minipage}
\begin{minipage}[t]{0.41\linewidth}
\begin{code}%
\>[2]\AgdaFunction{Vec\ensuremath{{}^{\mkern2mu\mathrm{L}}}}\AgdaSpace{}%
\AgdaSymbol{:}\AgdaSpace{}%
\AgdaPrimitive{Type}\AgdaSpace{}%
\AgdaGeneralizable{a}\AgdaSpace{}%
\AgdaSymbol{→}\AgdaSpace{}%
\AgdaSymbol{@}\AgdaNumber{0}\AgdaSpace{}%
\AgdaDatatype{ℕ}\AgdaSpace{}%
\AgdaSymbol{→}\AgdaSpace{}%
\AgdaPrimitive{Type}\AgdaSpace{}%
\AgdaGeneralizable{a}\<%
\\
\>[2]\AgdaFunction{Vec\ensuremath{{}^{\mkern2mu\mathrm{L}}}}\AgdaSpace{}%
\AgdaBound{A}\AgdaSpace{}%
\AgdaBound{n}\AgdaSpace{}%
\AgdaSymbol{=}\AgdaSpace{}%
\AgdaSymbol{(}%
\AgdaBound{xs}\AgdaSpace{}%
\AgdaSymbol{:}\AgdaSpace{}%
\AgdaDatatype{List}\AgdaSpace{}%
\AgdaBound{A}\AgdaSymbol{)}\AgdaSpace{}%
\AgdaFunction{×}\<%
\\
\>[2][@{}l@{\AgdaIndent{0}}]%
\>[4]\AgdaDatatype{Erased}\AgdaSpace{}%
\AgdaSymbol{(}\AgdaDatatype{Id}\AgdaSpace{}%
\AgdaDatatype{ℕ}\AgdaSpace{}%
\AgdaSymbol{(}\AgdaFunction{length}\AgdaSpace{}%
\AgdaBound{xs}\AgdaSymbol{)}\AgdaSpace{}%
\AgdaBound{n}\AgdaSymbol{)}\<%
\end{code}
\end{minipage}\\
\AgdaDatatype{Vec\ensuremath{{}^{\mkern1mu\mathrm{D}}}} is an indexed data type.
Note that the (implicit) natural number argument of the cons
constructor is erased.
A goal here is to ensure that there is no need to keep the numbers
representing the vectors' lengths around at run-time, so all natural
numbers are erased in this example.
\AgdaFunction{Vec\ensuremath{{}^{\mkern2mu\mathrm{L}}}} uses plain lists, but pairs them up with erased
proofs showing that the lists have correct lengths.

It is easy to convert a list to a vector.
Can we convert from \AgdaFunction{Vec\ensuremath{{}^{\mkern2mu\mathrm{L}}}} to
\AgdaDatatype{Vec\ensuremath{{}^{\mkern1mu\mathrm{D}}}}?\MCbreak{}
\begin{minipage}[t]{0.47\linewidth}
\begin{code}%
\>[2]\AgdaFunction{conv\ensuremath{{}^{\mkern2mu\mathrm{L}}}}\AgdaSpace{}%
\AgdaSymbol{:}\AgdaSpace{}%
\AgdaSymbol{(}\AgdaBound{xs}\AgdaSpace{}%
\AgdaSymbol{:}\AgdaSpace{}%
\AgdaDatatype{List}\AgdaSpace{}%
\AgdaGeneralizable{A}\AgdaSymbol{)}\AgdaSpace{}%
\AgdaSymbol{→}\AgdaSpace{}%
\AgdaDatatype{Vec\ensuremath{{}^{\mkern1mu\mathrm{D}}}}\AgdaSpace{}%
\AgdaGeneralizable{A}\AgdaSpace{}%
\AgdaSymbol{(}\AgdaFunction{length}\AgdaSpace{}%
\AgdaBound{xs}\AgdaSymbol{)}\<%
\\
\>[2]\AgdaFunction{conv\ensuremath{{}^{\mkern2mu\mathrm{L}}}}%
\>[9]\AgdaInductiveConstructor{[\ensuremath{\mkern1.5mu}]}%
\>[19]\AgdaSymbol{=}\AgdaSpace{}%
\AgdaInductiveConstructor{[\ensuremath{\mkern1.5mu}]}\<%
\\
\>[2]\AgdaFunction{conv\ensuremath{{}^{\mkern2mu\mathrm{L}}}}%
\>[9]\AgdaSymbol{(}\AgdaBound{x}\AgdaSpace{}%
\AgdaOperator{\AgdaInductiveConstructor{∷}}\AgdaSpace{}%
\AgdaBound{xs}\AgdaSymbol{)}%
\>[19]\AgdaSymbol{=}\AgdaSpace{}%
\AgdaBound{x}\AgdaSpace{}%
\AgdaOperator{\AgdaInductiveConstructor{∷}}\AgdaSpace{}%
\AgdaFunction{conv\ensuremath{{}^{\mkern2mu\mathrm{L}}}}\AgdaSpace{}%
\AgdaBound{xs}\<%
\end{code}
\end{minipage}
\begin{code}[hide]%
\>[2]\AgdaKeyword{module}\AgdaSpace{}%
\AgdaModule{Hide\ensuremath{{}_{\mathrm{2}}}}\AgdaSpace{}%
\AgdaKeyword{where}\<%
\end{code}
\begin{minipage}[t]{0.51\linewidth}
\begin{code}%
\>[2][@{}l@{\AgdaIndent{1}}]%
\>[4]\AgdaFunction{conv\ensuremath{{}^{\mkern1mu\mathrm{V}}}}\AgdaSpace{}%
\AgdaSymbol{:}\AgdaSpace{}%
\AgdaSymbol{\{}\AgdaSymbol{@}\AgdaNumber{0}\AgdaSpace{}%
\AgdaBound{n}\AgdaSpace{}%
\AgdaSymbol{:}\AgdaSpace{}%
\AgdaDatatype{ℕ}\AgdaSymbol{\}}\AgdaSpace{}%
\AgdaSymbol{→}\AgdaSpace{}%
\AgdaFunction{Vec\ensuremath{{}^{\mkern2mu\mathrm{L}}}}\AgdaSpace{}%
\AgdaGeneralizable{A}\AgdaSpace{}%
\AgdaBound{n}\AgdaSpace{}%
\AgdaSymbol{→}\AgdaSpace{}%
\AgdaDatatype{Vec\ensuremath{{}^{\mkern1mu\mathrm{D}}}}\AgdaSpace{}%
\AgdaGeneralizable{A}\AgdaSpace{}%
\AgdaBound{n}\<%
\\
\>[4]\AgdaFunction{conv\ensuremath{{}^{\mkern1mu\mathrm{V}}}}\AgdaSpace{}%
\AgdaSymbol{(}\AgdaInductiveConstructor{[\ensuremath{\mkern1.5mu}]}%
\>[19]\AgdaOperator{\AgdaInductiveConstructor{,}}\AgdaSpace{}%
\AgdaOperator{\AgdaInductiveConstructor{[}}\AgdaSpace{}%
\AgdaBound{eq}\AgdaSpace{}%
\AgdaOperator{\AgdaInductiveConstructor{]}}\AgdaSymbol{)}\AgdaSpace{}%
\AgdaSymbol{=}\AgdaSpace{}%
\AgdaHole{?}\<%
\\
\>[4]\AgdaFunction{conv\ensuremath{{}^{\mkern1mu\mathrm{V}}}}\AgdaSpace{}%
\AgdaSymbol{(}\AgdaBound{x}\AgdaSpace{}%
\AgdaOperator{\AgdaInductiveConstructor{∷}}\AgdaSpace{}%
\AgdaBound{xs}%
\>[19]\AgdaOperator{\AgdaInductiveConstructor{,}}\AgdaSpace{}%
\AgdaOperator{\AgdaInductiveConstructor{[}}\AgdaSpace{}%
\AgdaBound{eq}\AgdaSpace{}%
\AgdaOperator{\AgdaInductiveConstructor{]}}\AgdaSymbol{)}\AgdaSpace{}%
\AgdaSymbol{=}\AgdaSpace{}%
\AgdaHole{?}\<%
\end{code}
\end{minipage}\\
In the nil case we should return something of type
\begin{code}[hide]%
\>[2]\AgdaKeyword{postulate}\<%
\\
\>[2][@{}l@{\AgdaIndent{0}}]%
\>[4]\AgdaPostulate{\AgdaUnderscore{}}%
\>[1117I]\AgdaSymbol{:}\<%
\end{code}
\begin{code}[inline]%
\>[.][@{}l@{}]\<[1117I]%
\>[6]\AgdaDatatype{Vec\ensuremath{{}^{\mkern1mu\mathrm{D}}}}\AgdaSpace{}%
\AgdaGeneralizable{A}\AgdaSpace{}%
\AgdaGeneralizable{n}\<%
\end{code},
given an \emph{erased} proof that shows that \AgdaBound{n} is equal to
zero.
Can we return \AgdaInductiveConstructor{[\ensuremath{\mkern1.5mu}]}?
That constructor has the index \AgdaInductiveConstructor{zero}, not
\AgdaBound{n}.
Can we use the function \AgdaFunction{subst} above to substitute
\AgdaBound{n} for \AgdaInductiveConstructor{zero}?
That function does not take an erased identity proof (the penultimate
argument is not annotated with \AgdaSymbol{@0}).\footnote{It is
  important for the discussion that the argument \AgdaBound{n} is
  erased: an erased identity proof between two non-erased natural
  numbers can be turned into a non-erased identity proof because
  equality of natural numbers is decidable.}
Is the following variant of \AgdaFunction{subst}, with an erased
identity proof, OK?
\begin{code}%
\>[2]\AgdaFunction{subst′}\AgdaSpace{}%
\AgdaSymbol{:}\AgdaSpace{}%
\AgdaSymbol{(}\AgdaBound{P}\AgdaSpace{}%
\AgdaSymbol{:}\AgdaSpace{}%
\AgdaGeneralizable{A}\AgdaSpace{}%
\AgdaSymbol{→}\AgdaSpace{}%
\AgdaPrimitive{Type}\AgdaSpace{}%
\AgdaGeneralizable{p}\AgdaSymbol{)}\AgdaSpace{}%
\AgdaSymbol{→}\AgdaSpace{}%
\AgdaSymbol{@}\AgdaNumber{0}\AgdaSpace{}%
\AgdaDatatype{Id}\AgdaSpace{}%
\AgdaGeneralizable{A}\AgdaSpace{}%
\AgdaGeneralizable{x}\AgdaSpace{}%
\AgdaGeneralizable{y}\AgdaSpace{}%
\AgdaSymbol{→}\AgdaSpace{}%
\AgdaBound{P}\AgdaSpace{}%
\AgdaGeneralizable{x}\AgdaSpace{}%
\AgdaSymbol{→}\AgdaSpace{}%
\AgdaBound{P}\AgdaSpace{}%
\AgdaGeneralizable{y}\<%
\\
\>[2]\AgdaFunction{subst′}\AgdaSpace{}%
\AgdaBound{P}\AgdaSpace{}%
\AgdaInductiveConstructor{refl}\AgdaSpace{}%
\AgdaBound{p}\AgdaSpace{}%
\AgdaSymbol{=}\AgdaSpace{}%
\AgdaBound{p}\<%
\end{code}
Maybe.
Note that this code involves a match on an erased argument: an erased
match.
It should not be possible to make run-time decisions based on erased
data, but in this case there is only one constructor –
\AgdaInductiveConstructor{refl} – so this might still make sense.
In the presence of (erased) univalence this code is not safe: it
allows one to construct two provably distinct booleans that only
differ in their erased parts (see \citet{abel-et-al-2021} or
§ \ref{sec:box-cong}).
However, there are settings in which it does make sense.

The soundness theorem presented by
\citet[Theorem 8.3]{abel-et-al-2023} supports erased matches for
single-constructor data types if there are no postulates: given a
program $t$ of type $ℕ$ with erased matches but no postulates, $t$
will reduce to a numeral, and the corresponding compiled program will
reduce to the same numeral.
However, the following statement does not hold in general: ``given a
program $t$ of type $ℕ$ with erased matches and erased postulates, if
the postulates are mutually consistent, then $t$ will reduce to a
numeral, and the corresponding compiled program will reduce to the
same numeral''.
In particular, the part ``$t$ will reduce to a numeral'' can fail if
an erased match gets stuck on a postulate.

One contribution of this work is a proof of soundness of erasure for a
class of programs with both erased matches and erased postulates
(Theorem \ref{thm:soundness-reflection}).
There are no restrictions on the erased match\-es, but the types of
the erased postulates must be inhabited in an extension of the type
theory with \emph{equality reflection}.
Under equality reflection, if there is some term of type
\mbox{\IdAxy{}}, then \AgdaBound{x} and \AgdaBound{y} are
\emph{judgementally} equal.
Type-checking is not decidable for Martin-Löf type theory with
equality reflection, and systems like Agda do not support this
feature.
However, it allows one to prove properties like function
extensionality and uniqueness of identity proofs (UIP: every type is a
set).
The statement of Theorem \ref{thm:soundness-reflection} is (roughly)
that, given a program $t$ of type $ℕ$ that uses erased matches and
erased postulates, if the postulates can be implemented in the type
theory extended with equality reflection, then there is a numeral $n$
such that the compiled program corresponding to $t$ reduces to $n$,
and furthermore $t$ with the implementations substituted for the
postulates is judgementally equal to $n$ in the extended type theory.

\section{Box-Cong}
\label{sec:box-cong}

As mentioned above, if unrestricted erased matches are allowed for the
identity type, then the theory becomes incompatible with (erased,
postulated) univalence \citep{abel-et-al-2021}. (A similar problem
affects Rocq with identity types in $\mathit{Prop}$
\citep{shulman-2012}.)
Using univalence one can turn bijections into equalities.
In particular, one can turn the not function into a proof that
$\AgdaDatatype{Bool}$ is equal to $\AgdaDatatype{Bool}$.
Let us use $\AgdaFunction{swap}$ to denote such a proof, constructed
using erased univalence, and therefore only available in erased
contexts.
Now consider the following code, which uses the function
\AgdaFunction{subst′} from
§ \ref{sec:erased-postulates-and-erased-matches}:\MCbreak{}
\begin{code}[hide]%
\>[2]\AgdaKeyword{postulate}\<%
\\
\>[2][@{}l@{\AgdaIndent{0}}]%
\>[4]\AgdaPostulate{swap}\AgdaSpace{}%
\AgdaSymbol{:}\AgdaSpace{}%
\AgdaDatatype{Id}\AgdaSpace{}%
\AgdaPrimitive{Type}\AgdaSpace{}%
\AgdaFunction{Bool}\AgdaSpace{}%
\AgdaFunction{Bool}\<%
\end{code}
\begin{minipage}[t]{0.47\linewidth}
\begin{code}%
\>[2]\AgdaFunction{should‐be‐true}\AgdaSpace{}%
\AgdaSymbol{:}\AgdaSpace{}%
\AgdaFunction{Bool}\<%
\\
\>[2]\AgdaFunction{should‐be‐true}\AgdaSpace{}%
\AgdaSymbol{=}\AgdaSpace{}%
\AgdaFunction{subst′}\AgdaSpace{}%
\AgdaSymbol{(λ}\AgdaSpace{}%
\AgdaBound{A}\AgdaSpace{}%
\AgdaSymbol{→}\AgdaSpace{}%
\AgdaBound{A}\AgdaSymbol{)}\AgdaSpace{}%
\AgdaInductiveConstructor{refl}\AgdaSpace{}%
\AgdaInductiveConstructor{true}\<%
\end{code}
\end{minipage}
\begin{minipage}[t]{0.49\linewidth}
\begin{code}%
\>[2]\AgdaFunction{should‐be‐false}\AgdaSpace{}%
\AgdaSymbol{:}\AgdaSpace{}%
\AgdaFunction{Bool}\<%
\\
\>[2]\AgdaFunction{should‐be‐false}\AgdaSpace{}%
\AgdaSymbol{=}\AgdaSpace{}%
\AgdaFunction{subst′}\AgdaSpace{}%
\AgdaSymbol{(λ}\AgdaSpace{}%
\AgdaBound{A}\AgdaSpace{}%
\AgdaSymbol{→}\AgdaSpace{}%
\AgdaBound{A}\AgdaSymbol{)}\AgdaSpace{}%
\AgdaPostulate{swap}\AgdaSpace{}%
\AgdaInductiveConstructor{true}\<%
\end{code}
\end{minipage}\\
The definition of \AgdaFunction{should‐be‐false} is ``OK'' because the
penultimate argument of \AgdaFunction{subst′} is erased.
The definition of \AgdaFunction{should‐be‐true} reduces to
\AgdaInductiveConstructor{true}.
Using properties of univalence one can also construct an erased proof
that shows that \AgdaFunction{should‐be‐false} is equal to
\AgdaInductiveConstructor{false} (and because equality is decidable
for $\AgdaDatatype{Bool}$ one can turn this into a non-erased proof).
However, the only difference between \AgdaFunction{should‐be‐true} and
\AgdaFunction{should‐be‐false} is the erased argument,\footnote{And
  perhaps details of how the implicit arguments are instantiated, but
  one could write them out in full. Similar comments apply below as
  well.} and Agda compiles these definitions to basically identical
code, even though they are provably distinct.

So, the function \AgdaFunction{subst′} is not safe in the presence of
erased univalence.
Univalence is also inconsistent in the presence of UIP, and thus also
in the presence of equality reflection, so
Theorem \ref{thm:soundness-reflection} which was discussed above does
not apply in the presence of erased, postulated univalence.
Is any kind of erased match compatible with that postulate?

Theorem \ref{thm:soundness-of-erasure}, which was mentioned above,
supports one kind of erased match (following \citet{abel-et-al-2023}).
As long as the erased postulates are mutually consistent it is the
case that erased matches for the empty type are allowed, as in
\begin{code}[inline]%
\>[2]\AgdaFunction{⊥‐elim}\AgdaSpace{}%
\AgdaSymbol{:}\AgdaSpace{}%
\AgdaSymbol{\{}\AgdaSymbol{@}\AgdaNumber{0}\AgdaSpace{}%
\AgdaBound{A}\AgdaSpace{}%
\AgdaSymbol{:}\AgdaSpace{}%
\AgdaPrimitive{Type}\AgdaSpace{}%
\AgdaGeneralizable{a}\AgdaSymbol{\}}\AgdaSpace{}%
\AgdaSymbol{→}\AgdaSpace{}%
\AgdaSymbol{@}\AgdaNumber{0}\AgdaSpace{}%
\AgdaFunction{⊥}\AgdaSpace{}%
\AgdaSymbol{→}\AgdaSpace{}%
\AgdaBound{A}\<%
\end{code}.
\begin{code}[hide]%
\>[2]\AgdaFunction{⊥‐elim}\AgdaSpace{}%
\AgdaSymbol{()}\<%
\end{code}
This allows one to, for instance, use erased information to discard
impossible branches.
However, the theorem does not support erased matches for
single-constructor data types (except if there are no erased
postulates).

Fortunately there may be a way forward.
One contribution of this work is an investigation of a function called
\boxcong{} (``box-cong''), which encapsulates limited forms of erased
matches for identity types.
The work of \citet{abel-et-al-2021} suggests to me that the following
definition, with an erased match, might be safe in the presence of
erased univalence (more about this in
§ \ref{sec:computational-book-hott}):
\begin{code}%
\>[2]\AgdaFunction{[\ensuremath{\mkern1.5mu}]‐cong}\AgdaSpace{}%
\AgdaSymbol{:}\AgdaSpace{}%
\AgdaSymbol{\{}\AgdaSymbol{@}\AgdaNumber{0}\AgdaSpace{}%
\AgdaBound{A}\AgdaSpace{}%
\AgdaSymbol{:}\AgdaSpace{}%
\AgdaPrimitive{Type}\AgdaSpace{}%
\AgdaGeneralizable{a}\AgdaSymbol{\}}\AgdaSpace{}%
\AgdaSymbol{\{}\AgdaSymbol{@}\AgdaNumber{0}\AgdaSpace{}%
\AgdaBound{x}\AgdaSpace{}%
\AgdaBound{y}\AgdaSpace{}%
\AgdaSymbol{:}\AgdaSpace{}%
\AgdaBound{A}\AgdaSymbol{\}}\AgdaSpace{}%
\AgdaSymbol{→}\AgdaSpace{}%
\AgdaSymbol{@}\AgdaNumber{0}\AgdaSpace{}%
\AgdaDatatype{Id}\AgdaSpace{}%
\AgdaBound{A}\AgdaSpace{}%
\AgdaBound{x}\AgdaSpace{}%
\AgdaBound{y}\AgdaSpace{}%
\AgdaSymbol{→}\AgdaSpace{}%
\AgdaDatatype{Id}\AgdaSpace{}%
\AgdaSymbol{(}\AgdaDatatype{Erased}\AgdaSpace{}%
\AgdaBound{A}\AgdaSymbol{)}\AgdaSpace{}%
\AgdaOperator{\AgdaInductiveConstructor{[}}\AgdaSpace{}%
\AgdaBound{x}\AgdaSpace{}%
\AgdaOperator{\AgdaInductiveConstructor{]}}\AgdaSpace{}%
\AgdaOperator{\AgdaInductiveConstructor{[}}\AgdaSpace{}%
\AgdaBound{y}\AgdaSpace{}%
\AgdaOperator{\AgdaInductiveConstructor{]}}\<%
\\
\>[2]\AgdaFunction{[\ensuremath{\mkern1.5mu}]‐cong}\AgdaSpace{}%
\AgdaInductiveConstructor{refl}\AgdaSpace{}%
\AgdaSymbol{=}\AgdaSpace{}%
\AgdaInductiveConstructor{refl}\<%
\end{code}
Using \boxcong{} – as well as the erased projection function
\AgdaFunction{erased} – one can define
\AgdaFunction{subst\ensuremath{{}^{\mkern2mu\mathrm{E}}}}:
\begin{code}%
\>[2]\AgdaFunction{subst\ensuremath{{}^{\mkern2mu\mathrm{E}}}}\AgdaSpace{}%
\AgdaSymbol{:}\AgdaSpace{}%
\AgdaSymbol{\{}\AgdaSymbol{@}\AgdaNumber{0}\AgdaSpace{}%
\AgdaBound{A}\AgdaSpace{}%
\AgdaSymbol{:}\AgdaSpace{}%
\AgdaPrimitive{Type}\AgdaSpace{}%
\AgdaGeneralizable{a}\AgdaSymbol{\}}\AgdaSpace{}%
\AgdaSymbol{\{}\AgdaSymbol{@}\AgdaNumber{0}\AgdaSpace{}%
\AgdaBound{x}\AgdaSpace{}%
\AgdaBound{y}\AgdaSpace{}%
\AgdaSymbol{:}\AgdaSpace{}%
\AgdaBound{A}\AgdaSymbol{\}}\AgdaSpace{}%
\AgdaSymbol{(}\AgdaBound{P}\AgdaSpace{}%
\AgdaSymbol{:}\AgdaSpace{}%
\AgdaSymbol{@}\AgdaNumber{0}\AgdaSpace{}%
\AgdaBound{A}\AgdaSpace{}%
\AgdaSymbol{→}\AgdaSpace{}%
\AgdaPrimitive{Type}\AgdaSpace{}%
\AgdaGeneralizable{p}\AgdaSymbol{)}\AgdaSpace{}%
\AgdaSymbol{→}\AgdaSpace{}%
\AgdaSymbol{@}\AgdaNumber{0}\AgdaSpace{}%
\AgdaDatatype{Id}\AgdaSpace{}%
\AgdaBound{A}\AgdaSpace{}%
\AgdaBound{x}\AgdaSpace{}%
\AgdaBound{y}\AgdaSpace{}%
\AgdaSymbol{→}\AgdaSpace{}%
\AgdaBound{P}\AgdaSpace{}%
\AgdaBound{x}\AgdaSpace{}%
\AgdaSymbol{→}\AgdaSpace{}%
\AgdaBound{P}\AgdaSpace{}%
\AgdaBound{y}\<%
\\
\>[2]\AgdaFunction{subst\ensuremath{{}^{\mkern2mu\mathrm{E}}}}\AgdaSpace{}%
\AgdaBound{P}\AgdaSpace{}%
\AgdaBound{eq}\AgdaSpace{}%
\AgdaBound{p}\AgdaSpace{}%
\AgdaSymbol{=}\AgdaSpace{}%
\AgdaFunction{subst}\AgdaSpace{}%
\AgdaSymbol{(λ}\AgdaSpace{}%
\AgdaBound{x}\AgdaSpace{}%
\AgdaSymbol{→}\AgdaSpace{}%
\AgdaBound{P}\AgdaSpace{}%
\AgdaSymbol{(}\AgdaFunction{erased}\AgdaSpace{}%
\AgdaBound{x}\AgdaSymbol{))}\AgdaSpace{}%
\AgdaSymbol{(}\AgdaFunction{[\ensuremath{\mkern1.5mu}]‐cong}\AgdaSpace{}%
\AgdaBound{eq}\AgdaSymbol{)}\AgdaSpace{}%
\AgdaBound{p}\<%
\end{code}
Note that the identity proof argument is erased, but in turn
\AgdaBound{P} takes an erased argument: this makes the use of
\AgdaFunction{erased} in the right-hand side OK.
The type function
\begin{code}[hide]%
\>[2]\AgdaKeyword{postulate}\<%
\\
\>[2][@{}l@{\AgdaIndent{0}}]%
\>[4]\AgdaPostulate{\AgdaUnderscore{}}%
\>[1255I]\AgdaSymbol{:}\AgdaSpace{}%
\AgdaOperator{\AgdaGeneralizable{Has‐type}}\<%
\end{code}
\begin{code}[inline*]%
\>[.][@{}l@{}]\<[1255I]%
\>[6]\AgdaDatatype{Vec\ensuremath{{}^{\mkern1mu\mathrm{D}}}}\AgdaSpace{}%
\AgdaGeneralizable{A}\<%
\end{code}
also takes an erased argument, so one can use \AgdaFunction{subst\ensuremath{{}^{\mkern2mu\mathrm{E}}}} to
implement \AgdaFunction{conv\ensuremath{{}^{\mkern1mu\mathrm{V}}}}:
\begin{code}[hide]%
\>[2]\AgdaFunction{conv\ensuremath{{}^{\mkern1mu\mathrm{V}}}}\AgdaSpace{}%
\AgdaSymbol{:}\AgdaSpace{}%
\AgdaSymbol{\{}\AgdaSymbol{@}\AgdaNumber{0}\AgdaSpace{}%
\AgdaBound{n}\AgdaSpace{}%
\AgdaSymbol{:}\AgdaSpace{}%
\AgdaDatatype{ℕ}\AgdaSymbol{\}}\AgdaSpace{}%
\AgdaSymbol{→}\AgdaSpace{}%
\AgdaFunction{Vec\ensuremath{{}^{\mkern2mu\mathrm{L}}}}\AgdaSpace{}%
\AgdaGeneralizable{A}\AgdaSpace{}%
\AgdaBound{n}\AgdaSpace{}%
\AgdaSymbol{→}\AgdaSpace{}%
\AgdaDatatype{Vec\ensuremath{{}^{\mkern1mu\mathrm{D}}}}\AgdaSpace{}%
\AgdaGeneralizable{A}\AgdaSpace{}%
\AgdaBound{n}\<%
\end{code}
\begin{code}[inline]%
\>[2]\AgdaFunction{conv\ensuremath{{}^{\mkern1mu\mathrm{V}}}}\AgdaSpace{}%
\AgdaSymbol{\{}\AgdaBound{A}\AgdaSymbol{\}}\AgdaSpace{}%
\AgdaSymbol{(}\AgdaBound{xs}\AgdaSpace{}%
\AgdaOperator{\AgdaInductiveConstructor{,}}\AgdaSpace{}%
\AgdaOperator{\AgdaInductiveConstructor{[}}\AgdaSpace{}%
\AgdaBound{eq}\AgdaSpace{}%
\AgdaOperator{\AgdaInductiveConstructor{]}}\AgdaSymbol{)}\AgdaSpace{}%
\AgdaSymbol{=}\AgdaSpace{}%
\AgdaFunction{subst\ensuremath{{}^{\mkern2mu\mathrm{E}}}}\AgdaSpace{}%
\AgdaSymbol{(}\AgdaDatatype{Vec\ensuremath{{}^{\mkern1mu\mathrm{D}}}}\AgdaSpace{}%
\AgdaBound{A}\AgdaSymbol{)}\AgdaSpace{}%
\AgdaBound{eq}\AgdaSpace{}%
\AgdaSymbol{(}\AgdaFunction{conv\ensuremath{{}^{\mkern2mu\mathrm{L}}}}\AgdaSpace{}%
\AgdaBound{xs}\AgdaSymbol{)}\<%
\end{code}.

The implementation of \AgdaFunction{conv\ensuremath{{}^{\mkern1mu\mathrm{V}}}} is just a simple example,
but it should illustrate that erased matches for identity types can
sometimes be useful.
The text also contains other examples of the utility of \boxcong{}:
\begin{itemize}
\item \emph{Fording} is a method for encoding indexed data types using
  regular data types and identity proofs
  \citep{chapman-et-al-2010,pujet-leray-tabareau-2025}.
  Does fording work if the identity proofs are erased?
  In §~\ref{sec:fording} it is shown that fording with erased identity
  proofs works for \AgdaDatatype{Vec\ensuremath{{}^{\mkern1mu\mathrm{D}}}} if \boxcong{} is available, and
  also that such fording does not work (in a certain sense) unless a
  limited variant of \boxcong{} can be defined.
\item In the absence of erased matches for single-constructor data
  types, and the absence of equality reflection, \boxcong{} cannot be
  defined, and the limited variant of \boxcong{} discussed above also
  cannot be defined (see Theorem~\ref{thm:no-box-cong}).
\item A number of statements are logically equivalent to ``\boxcong{}
  can be defined'', see §~\ref{sec:box-cong-iff}.
  One of these statements is that a specific function is an
  equivalence – basically a bijection – from
  \begin{code}[hide]%
\>[2]\AgdaKeyword{postulate}\<%
\\
\>[2][@{}l@{\AgdaIndent{0}}]%
\>[4]\AgdaPostulate{\AgdaUnderscore{}}%
\>[1284I]\AgdaSymbol{:}\<%
\end{code}
  \begin{code}[inline*]%
\>[.][@{}l@{}]\<[1284I]%
\>[6]\AgdaDatatype{Id}\AgdaSpace{}%
\AgdaSymbol{(}\AgdaDatatype{Erased}\AgdaSpace{}%
\AgdaGeneralizable{A}\AgdaSymbol{)}\AgdaSpace{}%
\AgdaGeneralizable{x}\AgdaSpace{}%
\AgdaGeneralizable{y}\<%
\end{code}
  to
  \begin{code}[hide]%
\>[2]\AgdaKeyword{postulate}\<%
\\
\>[2][@{}l@{\AgdaIndent{0}}]%
\>[4]\AgdaPostulate{\AgdaUnderscore{}}%
\>[1289I]\AgdaSymbol{:}\<%
\end{code}
  \mbox{\begin{code}[inline]%
\>[.][@{}l@{}]\<[1289I]%
\>[6]\AgdaDatatype{Erased}\AgdaSpace{}%
\AgdaSymbol{(}\AgdaDatatype{Id}\AgdaSpace{}%
\AgdaGeneralizable{A}\AgdaSpace{}%
\AgdaSymbol{(}\AgdaFunction{erased}\AgdaSpace{}%
\AgdaGeneralizable{x}\AgdaSymbol{)}\AgdaSpace{}%
\AgdaSymbol{(}\AgdaFunction{erased}\AgdaSpace{}%
\AgdaGeneralizable{y}\AgdaSymbol{))}\<%
\end{code}}.
  Note that this provides a characterisation of equality for
  \AgdaDatatype{Erased}.\footnote{In Martin-Löf type theory identity
  types can be underspecified.
  A classic example is provided by identity for Π-types.
  One formulation of function extensionality is that a certain
  function is an equivalence from
  \begin{code}[hide]%
\>[2]\AgdaKeyword{postulate}\<%
\\
\>[2][@{}l@{\AgdaIndent{0}}]%
\>[4]\AgdaPostulate{\AgdaUnderscore{}}%
\>[1296I]\AgdaSymbol{:}\<%
\end{code}
  \begin{code}[inline*]%
\>[.][@{}l@{}]\<[1296I]%
\>[6]\AgdaDatatype{Id}\AgdaSpace{}%
\AgdaSymbol{((}\AgdaBound{x}\AgdaSpace{}%
\AgdaSymbol{:}\AgdaSpace{}%
\AgdaGeneralizable{A}\AgdaSymbol{)}\AgdaSpace{}%
\AgdaSymbol{→}\AgdaSpace{}%
\AgdaGeneralizable{P}\AgdaSpace{}%
\AgdaBound{x}\AgdaSymbol{)}\AgdaSpace{}%
\AgdaGeneralizable{f\,}\AgdaSpace{}%
\AgdaGeneralizable{g}\<%
\end{code}
  to
  \begin{code}[hide]%
\>[2]\AgdaKeyword{postulate}\<%
\\
\>[2][@{}l@{\AgdaIndent{0}}]%
\>[4]\AgdaPostulate{\AgdaUnderscore{}}%
\>[1305I]\AgdaSymbol{:}\AgdaSpace{}%
\AgdaSymbol{\{}\AgdaBound{f\,}\AgdaSpace{}%
\AgdaBound{g}\AgdaSpace{}%
\AgdaSymbol{:}\AgdaSpace{}%
\AgdaSymbol{(}\AgdaBound{x}\AgdaSpace{}%
\AgdaSymbol{:}\AgdaSpace{}%
\AgdaGeneralizable{A}\AgdaSymbol{)}\AgdaSpace{}%
\AgdaSymbol{→}\AgdaSpace{}%
\AgdaGeneralizable{P}\AgdaSpace{}%
\AgdaBound{x}\AgdaSymbol{\}}\AgdaSpace{}%
\AgdaSymbol{→}\<%
\end{code}
  \mbox{\begin{code}[inline]%
\>[.][@{}l@{}]\<[1305I]%
\>[6]\AgdaSymbol{(}\AgdaBound{x}\AgdaSpace{}%
\AgdaSymbol{:}\AgdaSpace{}%
\AgdaGeneralizable{A}\AgdaSymbol{)}\AgdaSpace{}%
\AgdaSymbol{→}\AgdaSpace{}%
\AgdaDatatype{Id}\AgdaSpace{}%
\AgdaSymbol{(}\AgdaGeneralizable{P}\AgdaSpace{}%
\AgdaBound{x}\AgdaSymbol{)}\AgdaSpace{}%
\AgdaSymbol{(}\AgdaBound{f\,}\AgdaSpace{}%
\AgdaBound{x}\AgdaSymbol{)}\AgdaSpace{}%
\AgdaSymbol{(}\AgdaBound{g}\AgdaSpace{}%
\AgdaBound{x}\AgdaSymbol{)}\<%
\end{code}}
  \citep{hott-2013}.}
  Another statement is that \Erased{} and the constructor \boxop{}
  (η-expanded) form a \emph{modality} in the sense of
  \citet{rijke-et-al-2020} (see §~\ref{sec:modalities}).
  This modality is \emph{left exact}, which for instance implies that
  if \AgdaBound{A} is a proposition, then \ErasedA{} is a proposition
  (for a stronger statement, see §~\ref{sec:erasure-modality}).
\end{itemize}
It turns out that, in the presence of (non-erased) function
extensionality, \boxcong{} can be implemented.
This is discussed in §~\ref{sec:box-cong-from-funext}.

Above it was shown that \AgdaFunction{subst′} is incompatible with
erased univalence.
In the presence of \boxcong{} a variant of
\AgdaFunction{subst} with the type signature
\begin{code}[inline*]%
\>[2]\AgdaFunction{subst″}\AgdaSpace{}%
\AgdaSymbol{:}\AgdaSpace{}%
\AgdaSymbol{(}\AgdaSymbol{@}\AgdaNumber{0}\AgdaSpace{}%
\AgdaBound{P}\AgdaSpace{}%
\AgdaSymbol{:}\AgdaSpace{}%
\AgdaGeneralizable{A}\AgdaSpace{}%
\AgdaSymbol{→}\AgdaSpace{}%
\AgdaPrimitive{Type}\AgdaSpace{}%
\AgdaGeneralizable{p}\AgdaSymbol{)}\AgdaSpace{}%
\AgdaSymbol{→}\AgdaSpace{}%
\AgdaDatatype{Id}\AgdaSpace{}%
\AgdaGeneralizable{A}\AgdaSpace{}%
\AgdaGeneralizable{x}\AgdaSpace{}%
\AgdaGeneralizable{y}\AgdaSpace{}%
\AgdaSymbol{→}\AgdaSpace{}%
\AgdaBound{P}\AgdaSpace{}%
\AgdaGeneralizable{x}\AgdaSpace{}%
\AgdaSymbol{→}\AgdaSpace{}%
\AgdaBound{P}\AgdaSpace{}%
\AgdaGeneralizable{y}\<%
\end{code}
\begin{code}[hide]%
\>[2]\AgdaFunction{subst″}\AgdaSpace{}%
\AgdaBound{P}\AgdaSpace{}%
\AgdaInductiveConstructor{refl}\AgdaSpace{}%
\AgdaBound{p}\AgdaSpace{}%
\AgdaSymbol{=}\AgdaSpace{}%
\AgdaBound{p}\<%
\end{code}
is also problematic \citep{abel-et-al-2021}.
This function allows us to construct the following variants of
\AgdaFunction{should‐be‐true} and
\AgdaFunction{should‐be‐false}:\MCbreak{}
\begin{minipage}[t]{0.47\linewidth}
\begin{code}%
\>[2]\AgdaFunction{should‐be‐true′}\AgdaSpace{}%
\AgdaSymbol{:}\AgdaSpace{}%
\AgdaFunction{Bool}\<%
\\
\>[2]\AgdaFunction{should‐be‐true′}\AgdaSpace{}%
\AgdaSymbol{=}\<%
\\
\>[2][@{}l@{\AgdaIndent{0}}]%
\>[4]\AgdaFunction{subst″}\AgdaSpace{}%
\AgdaFunction{erased}\AgdaSpace{}%
\AgdaSymbol{(}\AgdaFunction{[\ensuremath{\mkern1.5mu}]‐cong}\AgdaSpace{}%
\AgdaInductiveConstructor{refl}\AgdaSymbol{)}\AgdaSpace{}%
\AgdaInductiveConstructor{true}\<%
\end{code}
\end{minipage}
\begin{minipage}[t]{0.49\linewidth}
\begin{code}%
\>[2]\AgdaFunction{should‐be‐false′}\AgdaSpace{}%
\AgdaSymbol{:}\AgdaSpace{}%
\AgdaFunction{Bool}\<%
\\
\>[2]\AgdaFunction{should‐be‐false′}\AgdaSpace{}%
\AgdaSymbol{=}\<%
\\
\>[2][@{}l@{\AgdaIndent{0}}]%
\>[4]\AgdaFunction{subst″}\AgdaSpace{}%
\AgdaFunction{erased}\AgdaSpace{}%
\AgdaSymbol{(}\AgdaFunction{[\ensuremath{\mkern1.5mu}]‐cong}\AgdaSpace{}%
\AgdaPostulate{swap}\AgdaSymbol{)}\AgdaSpace{}%
\AgdaInductiveConstructor{true}\<%
\end{code}
\end{minipage}\\
These booleans are provably equal to \AgdaInductiveConstructor{true}
and \AgdaInductiveConstructor{false}, respectively, but are equal
except for the erased proofs \AgdaInductiveConstructor{refl} and
\AgdaFunction{swap}.

\begin{code}[hide]%
\>[0]\AgdaKeyword{module}\AgdaSpace{}%
\AgdaModule{Id\ensuremath{{}_{\mathrm{0}}}}\AgdaSpace{}%
\AgdaKeyword{where}\<%
\\
\>[0][@{}l@{\AgdaIndent{0}}]%
\>[2]\AgdaKeyword{open}\AgdaSpace{}%
\AgdaModule{Introduction}\AgdaSpace{}%
\AgdaKeyword{using}\AgdaSpace{}%
\AgdaSymbol{(}\AgdaDatatype{Erased}\AgdaSymbol{;}\AgdaSpace{}%
\AgdaOperator{\AgdaInductiveConstructor{[\AgdaUnderscore{}]}}\AgdaSymbol{;}\AgdaSpace{}%
\AgdaFunction{erased}\AgdaSymbol{)}\<%
\end{code}
Let me end this section by discussing a variant of the identity type
with three erased arguments:
\begin{code}%
\>[2]\AgdaKeyword{data}\AgdaSpace{}%
\AgdaDatatype{Id\ensuremath{{}_{\mathrm{0}}}}\AgdaSpace{}%
\AgdaSymbol{(}\AgdaSymbol{@0}\AgdaSpace{}%
\AgdaBound{A}\AgdaSpace{}%
\AgdaSymbol{:}\AgdaSpace{}%
\AgdaPrimitive{Type}\AgdaSpace{}%
\AgdaGeneralizable{a}\AgdaSymbol{)}\AgdaSpace{}%
\AgdaSymbol{(}\AgdaSymbol{@0}\AgdaSpace{}%
\AgdaBound{x}\AgdaSpace{}%
\AgdaSymbol{:}\AgdaSpace{}%
\AgdaBound{A}\AgdaSymbol{)}\AgdaSpace{}%
\AgdaSymbol{:}\AgdaSpace{}%
\AgdaSymbol{@0}\AgdaSpace{}%
\AgdaBound{A}\AgdaSpace{}%
\AgdaSymbol{→}\AgdaSpace{}%
\AgdaPrimitive{Type}\AgdaSpace{}%
\AgdaBound{a}\AgdaSpace{}%
\AgdaKeyword{where}\<%
\\
\>[2][@{}l@{\AgdaIndent{0}}]%
\>[4]\AgdaInductiveConstructor{refl}\AgdaSpace{}%
\AgdaSymbol{:}\AgdaSpace{}%
\AgdaDatatype{Id\ensuremath{{}_{\mathrm{0}}}}\AgdaSpace{}%
\AgdaBound{A}\AgdaSpace{}%
\AgdaBound{x}\AgdaSpace{}%
\AgdaBound{x}\<%
\end{code}
In the presence of \boxcong{}\footnote{Either for this type, or for
  the identity type presented earlier. The code presented here assumes
  that \boxcong{}, univalence, \AgdaFunction{swap} and
  \AgdaFunction{subst} are stated using this variant of the identity
  type.} this type is incompatible with erased, postulated univalence.
Using the fact that the last argument of \AgdaDatatype{Id\ensuremath{{}_{\mathrm{0}}}} is erased
we can define the following function that takes an erased identity
proof to a non-erased identity proof:
\begin{code}[hide]%
\>[2]\AgdaKeyword{postulate}\<%
\\
\>[2][@{}l@{\AgdaIndent{0}}]%
\>[4]\AgdaPostulate{swap}%
\>[12]\AgdaSymbol{:}\AgdaSpace{}%
\AgdaDatatype{Id\ensuremath{{}_{\mathrm{0}}}}\AgdaSpace{}%
\AgdaPrimitive{Type}\AgdaSpace{}%
\AgdaFunction{Bool}\AgdaSpace{}%
\AgdaFunction{Bool}\<%
\\
\>[4]\AgdaPostulate{[\ensuremath{\mkern1.5mu}]‐cong}\AgdaSpace{}%
\AgdaSymbol{:}\AgdaSpace{}%
\AgdaSymbol{@0}\AgdaSpace{}%
\AgdaDatatype{Id\ensuremath{{}_{\mathrm{0}}}}\AgdaSpace{}%
\AgdaGeneralizable{A}\AgdaSpace{}%
\AgdaGeneralizable{x}\AgdaSpace{}%
\AgdaGeneralizable{y}\AgdaSpace{}%
\AgdaSymbol{→}\AgdaSpace{}%
\AgdaDatatype{Id\ensuremath{{}_{\mathrm{0}}}}\AgdaSpace{}%
\AgdaSymbol{(}\AgdaDatatype{Erased}\AgdaSpace{}%
\AgdaGeneralizable{A}\AgdaSymbol{)}\AgdaSpace{}%
\AgdaOperator{\AgdaInductiveConstructor{[}}\AgdaSpace{}%
\AgdaGeneralizable{x}\AgdaSpace{}%
\AgdaOperator{\AgdaInductiveConstructor{]}}\AgdaSpace{}%
\AgdaOperator{\AgdaInductiveConstructor{[}}\AgdaSpace{}%
\AgdaGeneralizable{y}\AgdaSpace{}%
\AgdaOperator{\AgdaInductiveConstructor{]}}\<%
\\
\>[4]\AgdaPostulate{subst}%
\>[12]\AgdaSymbol{:}\AgdaSpace{}%
\AgdaSymbol{(}\AgdaBound{P}\AgdaSpace{}%
\AgdaSymbol{:}\AgdaSpace{}%
\AgdaGeneralizable{A}\AgdaSpace{}%
\AgdaSymbol{→}\AgdaSpace{}%
\AgdaPrimitive{Type}\AgdaSpace{}%
\AgdaGeneralizable{p}\AgdaSymbol{)}\AgdaSpace{}%
\AgdaSymbol{→}\AgdaSpace{}%
\AgdaDatatype{Id\ensuremath{{}_{\mathrm{0}}}}\AgdaSpace{}%
\AgdaGeneralizable{A}\AgdaSpace{}%
\AgdaGeneralizable{x}\AgdaSpace{}%
\AgdaGeneralizable{y}\AgdaSpace{}%
\AgdaSymbol{→}\AgdaSpace{}%
\AgdaBound{P}\AgdaSpace{}%
\AgdaGeneralizable{x}\AgdaSpace{}%
\AgdaSymbol{→}\AgdaSpace{}%
\AgdaBound{P}\AgdaSpace{}%
\AgdaGeneralizable{y}\<%
\end{code}
\begin{code}%
\>[2]\AgdaFunction{resurrect}\AgdaSpace{}%
\AgdaSymbol{:}\AgdaSpace{}%
\AgdaSymbol{@}\AgdaNumber{0}\AgdaSpace{}%
\AgdaDatatype{Id\ensuremath{{}_{\mathrm{0}}}}\AgdaSpace{}%
\AgdaGeneralizable{A}\AgdaSpace{}%
\AgdaGeneralizable{x}\AgdaSpace{}%
\AgdaGeneralizable{y}\AgdaSpace{}%
\AgdaSymbol{→}\AgdaSpace{}%
\AgdaDatatype{Id\ensuremath{{}_{\mathrm{0}}}}\AgdaSpace{}%
\AgdaGeneralizable{A}\AgdaSpace{}%
\AgdaGeneralizable{x}\AgdaSpace{}%
\AgdaGeneralizable{y}\<%
\\
\>[2]\AgdaFunction{resurrect}\AgdaSpace{}%
\AgdaSymbol{\{}\AgdaBound{A}\AgdaSymbol{\}}\AgdaSpace{}%
\AgdaSymbol{\{}\AgdaBound{x}\AgdaSymbol{\}}\AgdaSpace{}%
\AgdaBound{eq}\AgdaSpace{}%
\AgdaSymbol{=}\AgdaSpace{}%
\AgdaPostulate{subst}\AgdaSpace{}%
\AgdaSymbol{(λ}\AgdaSpace{}%
\AgdaBound{y}\AgdaSpace{}%
\AgdaSymbol{→}\AgdaSpace{}%
\AgdaDatatype{Id\ensuremath{{}_{\mathrm{0}}}}\AgdaSpace{}%
\AgdaBound{A}\AgdaSpace{}%
\AgdaBound{x}\AgdaSpace{}%
\AgdaSymbol{(}\AgdaFunction{erased}\AgdaSpace{}%
\AgdaBound{y}\AgdaSymbol{))}\AgdaSpace{}%
\AgdaSymbol{(}\AgdaPostulate{[\ensuremath{\mkern1.5mu}]‐cong}\AgdaSpace{}%
\AgdaBound{eq}\AgdaSymbol{)}\AgdaSpace{}%
\AgdaInductiveConstructor{refl}\<%
\end{code}
Using \AgdaFunction{resurrect} it is easy to define variants of
\AgdaFunction{should‐be‐true} and
\AgdaFunction{should‐be‐false}:\MCbreak{}
\begin{minipage}[t]{0.47\linewidth}
\begin{code}%
\>[2]\AgdaFunction{should‐be‐true″}\AgdaSpace{}%
\AgdaSymbol{:}\AgdaSpace{}%
\AgdaFunction{Bool}\<%
\\
\>[2]\AgdaFunction{should‐be‐true″}\AgdaSpace{}%
\AgdaSymbol{=}\<%
\\
\>[2][@{}l@{\AgdaIndent{0}}]%
\>[4]\AgdaPostulate{subst}\AgdaSpace{}%
\AgdaSymbol{(λ}\AgdaSpace{}%
\AgdaBound{A}\AgdaSpace{}%
\AgdaSymbol{→}\AgdaSpace{}%
\AgdaBound{A}\AgdaSymbol{)}\AgdaSpace{}%
\AgdaSymbol{(}\AgdaFunction{resurrect}\AgdaSpace{}%
\AgdaInductiveConstructor{refl}\AgdaSymbol{)}\AgdaSpace{}%
\AgdaInductiveConstructor{true}\<%
\end{code}
\end{minipage}
\begin{minipage}[t]{0.49\linewidth}
\begin{code}%
\>[2]\AgdaFunction{should‐be‐false″}\AgdaSpace{}%
\AgdaSymbol{:}\AgdaSpace{}%
\AgdaFunction{Bool}\<%
\\
\>[2]\AgdaFunction{should‐be‐false″}\AgdaSpace{}%
\AgdaSymbol{=}\<%
\\
\>[2][@{}l@{\AgdaIndent{0}}]%
\>[4]\AgdaPostulate{subst}\AgdaSpace{}%
\AgdaSymbol{(λ}\AgdaSpace{}%
\AgdaBound{A}\AgdaSpace{}%
\AgdaSymbol{→}\AgdaSpace{}%
\AgdaBound{A}\AgdaSymbol{)}\AgdaSpace{}%
\AgdaSymbol{(}\AgdaFunction{resurrect}\AgdaSpace{}%
\AgdaPostulate{swap}\AgdaSymbol{)}\AgdaSpace{}%
\AgdaInductiveConstructor{true}\<%
\end{code}
\end{minipage}\\
Again these booleans are provably equal to
\AgdaInductiveConstructor{true} and \AgdaInductiveConstructor{false},
respectively, but are equal except for the erased proofs
\AgdaInductiveConstructor{refl} and \AgdaFunction{swap}.

\section{Towards Computational Book HoTT}
\label{sec:computational-book-hott}

The work presented in this text gives us a fairly simple way to
support univalence and computation in the same theory, without having
to implement cubical type theory \citep{cohen-et-al-2018}, which is
arguably rather complicated.
In the absence of erased matches for single-constructor data types one
is free to use univalence in the erased fragment of the theory.
Univalence does not compute, just like in ``Book HoTT''
\cite{hott-2013}, but that does not affect compiled programs (here
restricted to programs of type $ℕ$), which compute properly.
Book HoTT also supports higher inductive types (HITs).
In Book HoTT the eliminators for HITs only compute for the point
constructors, just like for the quotient types presented in
§ \ref{sec:quotients}.

Thus this work provides a first step towards \emph{Computational Book
  HoTT}.
One thing that is left for future work is to investigate if the
technique employed for quotient types works for all HITs supported by
Book HoTT.
If one wants to use \boxcong{}, then it would also make sense to prove
that the combination of \boxcong{} and postulated, erased univalence
is safe.
I have not done this, but I have tried to pave the way for such a
proof.
That is another contribution of this work.

Theorem \ref{thm:soundness-reflection}, discussed above, is proved
using a meta-theorem (Theorem \ref{thm:meta}) instantiated with type
theories with equality reflection.
If that meta-theorem can be instantiated with cubical type theories
\citep{cohen-et-al-2018} that support identity types with J and
\boxcong{} in addition to path types, then one may obtain the
following result: given a program $t$ of type $ℕ$ that uses erased
postulates and \boxcong{} (and perhaps certain other forms of erased
matches), if the postulates can be implemented in the type theory
extended with certain cubical features, then there is a numeral $n$
such that the compiled program corresponding to $t$ reduces to $n$,
and furthermore $t$ with the implementations substituted for the
postulates is judgementally equal to $n$ in the cubical type theory.
Note that univalence can be proved in cubical type theory.

\citet{abel-et-al-2021} sketch a proof, with some unproved
assumptions, of soundness of erasure for a variant of cubical type
theory with erasure and erased univalence (but without identity
types with the J rule).
In that system \boxcong{} can be implemented for the path type, like
in the following Cubical Agda code:
\boxcongPath{}%
A path in \PathAxy{} is basically a function from ``the interval'' to
\AgdaBound{A} that is equal to \AgdaBound{x} and \AgdaBound{y} at the
interval's two endpoints.
Note that the erased path \AgdaBound{eq} is used in the erased
argument of \boxop{}.
The work of \citeauthor{abel-et-al-2021} may be incomplete, but it
makes me hopeful that one could instantiate the meta-theorem in the
way described above.

\section{Fording with Erased Identity Proofs}
\label{sec:fording}

\begin{code}[hide]%
\>[0]\AgdaKeyword{module}\AgdaSpace{}%
\AgdaModule{Eliminators}\AgdaSpace{}%
\AgdaKeyword{where}\<%
\\
\\[\AgdaEmptyExtraSkip]%
\>[0][@{}l@{\AgdaIndent{0}}]%
\>[2]\AgdaKeyword{open}\AgdaSpace{}%
\AgdaModule{Introduction}\AgdaSpace{}%
\AgdaKeyword{using}\AgdaSpace{}%
\AgdaSymbol{(}\AgdaDatatype{Erased}\AgdaSymbol{;}\AgdaSpace{}%
\AgdaOperator{\AgdaInductiveConstructor{[\AgdaUnderscore{}]}}\AgdaSymbol{;}\AgdaSpace{}%
\AgdaFunction{subst\ensuremath{{}^{\mkern2mu\mathrm{E}}}}\AgdaSymbol{)}\<%
\\
\\[\AgdaEmptyExtraSkip]%
\>[2]\AgdaOperator{\AgdaFunction{\AgdaUnderscore{}≡\AgdaUnderscore{}}}\AgdaSpace{}%
\AgdaSymbol{:}\AgdaSpace{}%
\AgdaSymbol{\{}\AgdaBound{A}\AgdaSpace{}%
\AgdaSymbol{:}\AgdaSpace{}%
\AgdaPrimitive{Type}\AgdaSpace{}%
\AgdaGeneralizable{a}\AgdaSymbol{\}}\AgdaSpace{}%
\AgdaSymbol{→}\AgdaSpace{}%
\AgdaBound{A}\AgdaSpace{}%
\AgdaSymbol{→}\AgdaSpace{}%
\AgdaBound{A}\AgdaSpace{}%
\AgdaSymbol{→}\AgdaSpace{}%
\AgdaPrimitive{Type}\AgdaSpace{}%
\AgdaGeneralizable{a}\<%
\\
\>[2]\AgdaBound{x}\AgdaSpace{}%
\AgdaOperator{\AgdaFunction{≡}}\AgdaSpace{}%
\AgdaBound{y}\AgdaSpace{}%
\AgdaSymbol{=}\AgdaSpace{}%
\AgdaDatatype{Id}\AgdaSpace{}%
\AgdaSymbol{\AgdaUnderscore{}}\AgdaSpace{}%
\AgdaBound{x}\AgdaSpace{}%
\AgdaBound{y}\<%
\end{code}

Let us now return to the type \AgdaDatatype{Vec\ensuremath{{}^{\mkern1mu\mathrm{D}}}} from
§ \ref{sec:erased-postulates-and-erased-matches}.
This is an inductive family, where the two constructors have different
codomains.
In the absence of primitive inductive families one can try to encode
them using the technique of \emph{fording}, i.e.\ using regular data
types and identity proofs
\citep{chapman-et-al-2010,pujet-leray-tabareau-2025}.
Let us encode \AgdaDatatype{Vec\ensuremath{{}^{\mkern1mu\mathrm{D}}}} in this way, but using \emph{erased}
identity proofs:
\begin{code}%
\>[2]\AgdaKeyword{data}\AgdaSpace{}%
\AgdaDatatype{Vec\ensuremath{{}^{\mkern1mu\mathrm{F}}}}\AgdaSpace{}%
\AgdaSymbol{(}\AgdaBound{A}\AgdaSpace{}%
\AgdaSymbol{:}\AgdaSpace{}%
\AgdaPrimitive{Type}\AgdaSpace{}%
\AgdaGeneralizable{a}\AgdaSymbol{)}\AgdaSpace{}%
\AgdaSymbol{(}\AgdaSymbol{@0}\AgdaSpace{}%
\AgdaBound{n}\AgdaSpace{}%
\AgdaSymbol{:}\AgdaSpace{}%
\AgdaDatatype{ℕ}\AgdaSymbol{)}\AgdaSpace{}%
\AgdaSymbol{:}\AgdaSpace{}%
\AgdaPrimitive{Type}\AgdaSpace{}%
\AgdaBound{a}\AgdaSpace{}%
\AgdaKeyword{where}\<%
\\
\>[2][@{}l@{\AgdaIndent{0}}]%
\>[4]\AgdaInductiveConstructor{nil}%
\>[10]\AgdaSymbol{:}\AgdaSpace{}%
\AgdaSymbol{@0}\AgdaSpace{}%
\AgdaDatatype{Id}\AgdaSpace{}%
\AgdaDatatype{ℕ}\AgdaSpace{}%
\AgdaInductiveConstructor{zero}\AgdaSpace{}%
\AgdaBound{n}\AgdaSpace{}%
\AgdaSymbol{→}\AgdaSpace{}%
\AgdaDatatype{Vec\ensuremath{{}^{\mkern1mu\mathrm{F}}}}\AgdaSpace{}%
\AgdaBound{A}\AgdaSpace{}%
\AgdaBound{n}\<%
\\
\>[4]\AgdaInductiveConstructor{cons}%
\>[10]\AgdaSymbol{:}\AgdaSpace{}%
\AgdaSymbol{\{}\AgdaSymbol{@0}\AgdaSpace{}%
\AgdaBound{m}\AgdaSpace{}%
\AgdaSymbol{:}\AgdaSpace{}%
\AgdaDatatype{ℕ}\AgdaSymbol{\}}\AgdaSpace{}%
\AgdaSymbol{→}\AgdaSpace{}%
\AgdaBound{A}\AgdaSpace{}%
\AgdaSymbol{→}\AgdaSpace{}%
\AgdaDatatype{Vec\ensuremath{{}^{\mkern1mu\mathrm{F}}}}\AgdaSpace{}%
\AgdaBound{A}\AgdaSpace{}%
\AgdaBound{m}\AgdaSpace{}%
\AgdaSymbol{→}\AgdaSpace{}%
\AgdaSymbol{@0}\AgdaSpace{}%
\AgdaDatatype{Id}\AgdaSpace{}%
\AgdaDatatype{ℕ}\AgdaSpace{}%
\AgdaSymbol{(}\AgdaInductiveConstructor{suc}\AgdaSpace{}%
\AgdaBound{m}\AgdaSymbol{)}\AgdaSpace{}%
\AgdaBound{n}\AgdaSpace{}%
\AgdaSymbol{→}\AgdaSpace{}%
\AgdaDatatype{Vec\ensuremath{{}^{\mkern1mu\mathrm{F}}}}\AgdaSpace{}%
\AgdaBound{A}\AgdaSpace{}%
\AgdaBound{n}\<%
\end{code}
For \AgdaDatatype{Vec\ensuremath{{}^{\mkern1mu\mathrm{F}}}} we get the following eliminator (recursion
principle):
\begin{code}%
\>[2]\AgdaFunction{Vec^{\mkern1mu\mathrm{F}}\mkern-2mu{}‐elim}\AgdaSpace{}%
\AgdaSymbol{:}%
\>[1547I]\AgdaSymbol{\{}\AgdaSymbol{@}\AgdaNumber{0}\AgdaSpace{}%
\AgdaBound{A}\AgdaSpace{}%
\AgdaSymbol{:}\AgdaSpace{}%
\AgdaPrimitive{Type}\AgdaSpace{}%
\AgdaGeneralizable{a}\AgdaSymbol{\}}\AgdaSpace{}%
\AgdaSymbol{(}\AgdaSymbol{@}\AgdaNumber{0}\AgdaSpace{}%
\AgdaBound{P}\AgdaSpace{}%
\AgdaSymbol{:}\AgdaSpace{}%
\AgdaSymbol{\{}\AgdaBound{n}\AgdaSpace{}%
\AgdaSymbol{:}\AgdaSpace{}%
\AgdaDatatype{ℕ}\AgdaSymbol{\}}\AgdaSpace{}%
\AgdaSymbol{→}\AgdaSpace{}%
\AgdaDatatype{Vec\ensuremath{{}^{\mkern1mu\mathrm{F}}}}\AgdaSpace{}%
\AgdaBound{A}\AgdaSpace{}%
\AgdaBound{n}\AgdaSpace{}%
\AgdaSymbol{→}\AgdaSpace{}%
\AgdaPrimitive{Type}\AgdaSpace{}%
\AgdaGeneralizable{p}\AgdaSymbol{)}\AgdaSpace{}%
\AgdaSymbol{→}\<%
\\
\>[.][@{}l@{}]\<[1547I]%
\>[14]\AgdaSymbol{(\{}\AgdaSymbol{@}\AgdaNumber{0}\AgdaSpace{}%
\AgdaBound{n}\AgdaSpace{}%
\AgdaSymbol{:}\AgdaSpace{}%
\AgdaDatatype{ℕ}\AgdaSymbol{\}}\AgdaSpace{}%
\AgdaSymbol{(}\AgdaSymbol{@}\AgdaNumber{0}\AgdaSpace{}%
\AgdaBound{eq}\AgdaSpace{}%
\AgdaSymbol{:}\AgdaSpace{}%
\AgdaDatatype{Id}\AgdaSpace{}%
\AgdaDatatype{ℕ}\AgdaSpace{}%
\AgdaInductiveConstructor{zero}\AgdaSpace{}%
\AgdaBound{n}\AgdaSymbol{)}\AgdaSpace{}%
\AgdaSymbol{→}\AgdaSpace{}%
\AgdaBound{P}\AgdaSpace{}%
\AgdaSymbol{(}\AgdaInductiveConstructor{nil}\AgdaSpace{}%
\AgdaBound{eq}\AgdaSymbol{))}\AgdaSpace{}%
\AgdaSymbol{→}\<%
\\
\>[14]\AgdaSymbol{(\{}\AgdaSymbol{@}\AgdaNumber{0}\AgdaSpace{}%
\AgdaBound{m}\AgdaSpace{}%
\AgdaBound{n}\AgdaSpace{}%
\AgdaSymbol{:}\AgdaSpace{}%
\AgdaDatatype{ℕ}\AgdaSymbol{\}}\AgdaSpace{}%
\AgdaSymbol{(}\AgdaBound{x}\AgdaSpace{}%
\AgdaSymbol{:}\AgdaSpace{}%
\AgdaBound{A}\AgdaSymbol{)}\AgdaSpace{}%
\AgdaSymbol{(}\AgdaBound{xs}\AgdaSpace{}%
\AgdaSymbol{:}\AgdaSpace{}%
\AgdaDatatype{Vec\ensuremath{{}^{\mkern1mu\mathrm{F}}}}\AgdaSpace{}%
\AgdaBound{A}\AgdaSpace{}%
\AgdaBound{m}\AgdaSymbol{)}\AgdaSpace{}%
\AgdaSymbol{(}\AgdaSymbol{@}\AgdaNumber{0}\AgdaSpace{}%
\AgdaBound{eq}\AgdaSpace{}%
\AgdaSymbol{:}\AgdaSpace{}%
\AgdaDatatype{Id}\AgdaSpace{}%
\AgdaDatatype{ℕ}\AgdaSpace{}%
\AgdaSymbol{(}\AgdaInductiveConstructor{suc}\AgdaSpace{}%
\AgdaBound{m}\AgdaSymbol{)}\AgdaSpace{}%
\AgdaBound{n}\AgdaSymbol{)}\AgdaSpace{}%
\AgdaSymbol{→}\AgdaSpace{}%
\AgdaBound{P}\AgdaSpace{}%
\AgdaBound{xs}\AgdaSpace{}%
\AgdaSymbol{→}\<%
\\
\>[14][@{}l@{\AgdaIndent{0}}]%
\>[15]\AgdaBound{P}\AgdaSpace{}%
\AgdaSymbol{(}\AgdaInductiveConstructor{cons}\AgdaSpace{}%
\AgdaBound{x}\AgdaSpace{}%
\AgdaBound{xs}\AgdaSpace{}%
\AgdaBound{eq}\AgdaSymbol{))}\AgdaSpace{}%
\AgdaSymbol{→}\<%
\\
\>[14]\AgdaSymbol{\{}\AgdaSymbol{@}\AgdaNumber{0}\AgdaSpace{}%
\AgdaBound{n}\AgdaSpace{}%
\AgdaSymbol{:}\AgdaSpace{}%
\AgdaDatatype{ℕ}\AgdaSymbol{\}}\AgdaSpace{}%
\AgdaSymbol{(}\AgdaBound{xs}\AgdaSpace{}%
\AgdaSymbol{:}\AgdaSpace{}%
\AgdaDatatype{Vec\ensuremath{{}^{\mkern1mu\mathrm{F}}}}\AgdaSpace{}%
\AgdaBound{A}\AgdaSpace{}%
\AgdaBound{n}\AgdaSymbol{)}\AgdaSpace{}%
\AgdaSymbol{→}\AgdaSpace{}%
\AgdaBound{P}\AgdaSpace{}%
\AgdaBound{xs}\<%
\\
\>[2]\AgdaFunction{Vec^{\mkern1mu\mathrm{F}}\mkern-2mu{}‐elim}\AgdaSpace{}%
\AgdaBound{P}\AgdaSpace{}%
\AgdaBound{n}\AgdaSpace{}%
\AgdaBound{c}\AgdaSpace{}%
\AgdaSymbol{(}\AgdaInductiveConstructor{nil}\AgdaSpace{}%
\AgdaBound{eq}\AgdaSymbol{)}%
\>[34]\AgdaSymbol{=}\AgdaSpace{}%
\AgdaBound{n}\AgdaSpace{}%
\AgdaBound{eq}\<%
\\
\>[2]\AgdaFunction{Vec^{\mkern1mu\mathrm{F}}\mkern-2mu{}‐elim}\AgdaSpace{}%
\AgdaBound{P}\AgdaSpace{}%
\AgdaBound{n}\AgdaSpace{}%
\AgdaBound{c}\AgdaSpace{}%
\AgdaSymbol{(}\AgdaInductiveConstructor{cons}\AgdaSpace{}%
\AgdaBound{x}\AgdaSpace{}%
\AgdaBound{xs}\AgdaSpace{}%
\AgdaBound{eq}\AgdaSymbol{)}%
\>[34]\AgdaSymbol{=}\AgdaSpace{}%
\AgdaBound{c}\AgdaSpace{}%
\AgdaBound{x}\AgdaSpace{}%
\AgdaBound{xs}\AgdaSpace{}%
\AgdaBound{eq}\AgdaSpace{}%
\AgdaSymbol{(}\AgdaFunction{Vec^{\mkern1mu\mathrm{F}}\mkern-2mu{}‐elim}\AgdaSpace{}%
\AgdaBound{P}\AgdaSpace{}%
\AgdaBound{n}\AgdaSpace{}%
\AgdaBound{c}\AgdaSpace{}%
\AgdaBound{xs}\AgdaSymbol{)}\<%
\end{code}
This eliminator is somewhat low-level, in the sense that the arguments
corresponding to \AgdaInductiveConstructor{nil} and
\AgdaInductiveConstructor{cons} take explicit identity proofs.
Can we define a more high-level eliminator, closer to what one might
expect for \AgdaDatatype{Vec\ensuremath{{}^{\mkern1mu\mathrm{D}}}}?

\newcommand{\JType}{%
  \begin{code}[inline]%
\>[2]\AgdaFunction{J}\AgdaSpace{}%
\AgdaSymbol{:}\AgdaSpace{}%
\AgdaSymbol{(}\AgdaBound{P}\AgdaSpace{}%
\AgdaSymbol{:}\AgdaSpace{}%
\AgdaSymbol{\{}\AgdaBound{y}\AgdaSpace{}%
\AgdaSymbol{:}\AgdaSpace{}%
\AgdaGeneralizable{A}\AgdaSymbol{\}}\AgdaSpace{}%
\AgdaSymbol{→}\AgdaSpace{}%
\AgdaDatatype{Id}\AgdaSpace{}%
\AgdaGeneralizable{A}\AgdaSpace{}%
\AgdaGeneralizable{x}\AgdaSpace{}%
\AgdaBound{y}\AgdaSpace{}%
\AgdaSymbol{→}\AgdaSpace{}%
\AgdaPrimitive{Type}\AgdaSpace{}%
\AgdaGeneralizable{p}\AgdaSymbol{)}\AgdaSpace{}%
\AgdaSymbol{→}\AgdaSpace{}%
\AgdaBound{P}\AgdaSpace{}%
\AgdaInductiveConstructor{refl}\AgdaSpace{}%
\AgdaSymbol{→}\AgdaSpace{}%
\AgdaSymbol{(}\AgdaBound{eq}\AgdaSpace{}%
\AgdaSymbol{:}\AgdaSpace{}%
\AgdaDatatype{Id}\AgdaSpace{}%
\AgdaGeneralizable{A}\AgdaSpace{}%
\AgdaGeneralizable{x}\AgdaSpace{}%
\AgdaGeneralizable{y}\AgdaSymbol{)}\AgdaSpace{}%
\AgdaSymbol{→}\AgdaSpace{}%
\AgdaBound{P}\AgdaSpace{}%
\AgdaBound{eq}\<%
\end{code}}
\begin{code}[hide]%
\>[2]\AgdaFunction{J}\AgdaSpace{}%
\AgdaSymbol{\AgdaUnderscore{}}\AgdaSpace{}%
\AgdaBound{p}\AgdaSpace{}%
\AgdaInductiveConstructor{refl}\AgdaSpace{}%
\AgdaSymbol{=}\AgdaSpace{}%
\AgdaBound{p}\<%
\end{code}

\newcommand{\JEType}{%
\begin{code}%
\>[2]\AgdaFunction{J\ensuremath{{}^{\mkern2mu\mathrm{E}}}}\AgdaSpace{}%
\AgdaSymbol{:}%
\>[1677I]\AgdaSymbol{\{}\AgdaSymbol{@}\AgdaNumber{0}\AgdaSpace{}%
\AgdaBound{A}\AgdaSpace{}%
\AgdaSymbol{:}\AgdaSpace{}%
\AgdaPrimitive{Type}\AgdaSpace{}%
\AgdaGeneralizable{a}\AgdaSymbol{\}}\AgdaSpace{}%
\AgdaSymbol{\{}\AgdaSymbol{@}\AgdaNumber{0}\AgdaSpace{}%
\AgdaBound{x}\AgdaSpace{}%
\AgdaBound{y}\AgdaSpace{}%
\AgdaSymbol{:}\AgdaSpace{}%
\AgdaBound{A}\AgdaSymbol{\}}\AgdaSpace{}%
\AgdaSymbol{(}\AgdaBound{P}\AgdaSpace{}%
\AgdaSymbol{:}\AgdaSpace{}%
\AgdaSymbol{\{}\AgdaSymbol{@}\AgdaNumber{0}\AgdaSpace{}%
\AgdaBound{y}\AgdaSpace{}%
\AgdaSymbol{:}\AgdaSpace{}%
\AgdaBound{A}\AgdaSymbol{\}}\AgdaSpace{}%
\AgdaSymbol{→}\AgdaSpace{}%
\AgdaSymbol{@}\AgdaNumber{0}\AgdaSpace{}%
\AgdaDatatype{Id}\AgdaSpace{}%
\AgdaBound{A}\AgdaSpace{}%
\AgdaBound{x}\AgdaSpace{}%
\AgdaBound{y}\AgdaSpace{}%
\AgdaSymbol{→}\AgdaSpace{}%
\AgdaPrimitive{Type}\AgdaSpace{}%
\AgdaGeneralizable{p}\AgdaSymbol{)}\AgdaSpace{}%
\AgdaSymbol{→}\<%
\\
\>[.][@{}l@{}]\<[1677I]%
\>[7]\AgdaBound{P}\AgdaSpace{}%
\AgdaInductiveConstructor{refl}\AgdaSpace{}%
\AgdaSymbol{→}\AgdaSpace{}%
\AgdaSymbol{(}\AgdaSymbol{@}\AgdaNumber{0}\AgdaSpace{}%
\AgdaBound{eq}\AgdaSpace{}%
\AgdaSymbol{:}\AgdaSpace{}%
\AgdaDatatype{Id}\AgdaSpace{}%
\AgdaBound{A}\AgdaSpace{}%
\AgdaBound{x}\AgdaSpace{}%
\AgdaBound{y}\AgdaSymbol{)}\AgdaSpace{}%
\AgdaSymbol{→}\AgdaSpace{}%
\AgdaBound{P}\AgdaSpace{}%
\AgdaBound{eq}\<%
\end{code}}

\newcommand{\JEDef}{%
\begin{code}%
\>[2]\AgdaFunction{J\ensuremath{{}^{\mkern2mu\mathrm{E}}}}\AgdaSpace{}%
\AgdaSymbol{\{}\AgdaBound{A}\AgdaSymbol{\}}\AgdaSpace{}%
\AgdaSymbol{\{}\AgdaBound{x}\AgdaSymbol{\}}\AgdaSpace{}%
\AgdaBound{P}\AgdaSpace{}%
\AgdaBound{p}\AgdaSpace{}%
\AgdaBound{eq}\AgdaSpace{}%
\AgdaSymbol{=}\<%
\\
\>[2][@{}l@{\AgdaIndent{0}}]%
\>[4]\AgdaFunction{subst\ensuremath{{}^{\mkern2mu\mathrm{E}}}}\AgdaSpace{}%
\AgdaSymbol{(λ}\AgdaSpace{}%
\AgdaBound{p}\AgdaSpace{}%
\AgdaSymbol{→}\AgdaSpace{}%
\AgdaBound{P}\AgdaSpace{}%
\AgdaSymbol{(}\AgdaField{snd}\AgdaSpace{}%
\AgdaBound{p}\AgdaSymbol{))}\AgdaSpace{}%
\AgdaSymbol{(}\AgdaFunction{J}\AgdaSpace{}%
\AgdaSymbol{(λ}\AgdaSpace{}%
\AgdaSymbol{\{}\AgdaBound{y}\AgdaSymbol{\}}\AgdaSpace{}%
\AgdaBound{eq}\AgdaSpace{}%
\AgdaSymbol{→}\AgdaSpace{}%
\AgdaDatatype{Id}\AgdaSpace{}%
\AgdaSymbol{(}\AgdaSymbol{(}%
\AgdaBound{y}\AgdaSpace{}%
\AgdaSymbol{:}\AgdaSpace{}%
\AgdaBound{A}\AgdaSymbol{)}\AgdaSpace{}%
\AgdaFunction{×}\AgdaSpace{}%
\AgdaDatatype{Id}\AgdaSpace{}%
\AgdaBound{A}\AgdaSpace{}%
\AgdaBound{x}\AgdaSpace{}%
\AgdaBound{y}\AgdaSymbol{)}\AgdaSpace{}%
\AgdaSymbol{(}\AgdaBound{x}\AgdaSpace{}%
\AgdaOperator{\AgdaInductiveConstructor{,}}\AgdaSpace{}%
\AgdaInductiveConstructor{refl}\AgdaSymbol{)}\AgdaSpace{}%
\AgdaSymbol{(}\AgdaBound{y}\AgdaSpace{}%
\AgdaOperator{\AgdaInductiveConstructor{,}}\AgdaSpace{}%
\AgdaBound{eq}\AgdaSymbol{))}\AgdaSpace{}%
\AgdaInductiveConstructor{refl}\AgdaSpace{}%
\AgdaBound{eq}\AgdaSymbol{)}\AgdaSpace{}%
\AgdaBound{p}\<%
\end{code}}

Let us first define variants of \AgdaInductiveConstructor{nil} and
\AgdaInductiveConstructor{cons} that do not take identity
proofs as arguments:\MCbreak{}
\begin{code}[hide]%
\>[2]\AgdaKeyword{infixr}\AgdaSpace{}%
\AgdaNumber{5}\AgdaSpace{}%
\AgdaOperator{\AgdaFunction{\AgdaUnderscore{}∷\AgdaUnderscore{}}}\<%
\end{code}
\begin{minipage}[t]{0.49\linewidth}
\begin{code}%
\>[2]\AgdaFunction{[\ensuremath{\mkern1.5mu}]}\AgdaSpace{}%
\AgdaSymbol{:}\AgdaSpace{}%
\AgdaSymbol{\{}\AgdaSymbol{@}\AgdaNumber{0}\AgdaSpace{}%
\AgdaBound{A}\AgdaSpace{}%
\AgdaSymbol{:}\AgdaSpace{}%
\AgdaPrimitive{Type}\AgdaSpace{}%
\AgdaGeneralizable{a}\AgdaSymbol{\}}\AgdaSpace{}%
\AgdaSymbol{→}\AgdaSpace{}%
\AgdaDatatype{Vec\ensuremath{{}^{\mkern1mu\mathrm{F}}}}\AgdaSpace{}%
\AgdaBound{A}\AgdaSpace{}%
\AgdaInductiveConstructor{zero}\<%
\\
\>[2]\AgdaFunction{[\ensuremath{\mkern1.5mu}]}\AgdaSpace{}%
\AgdaSymbol{=}\AgdaSpace{}%
\AgdaInductiveConstructor{nil}\AgdaSpace{}%
\AgdaInductiveConstructor{refl}\<%
\end{code}
\end{minipage}
\begin{minipage}[t]{0.49\linewidth}
\begin{code}%
\>[2]\AgdaOperator{\AgdaFunction{\AgdaUnderscore{}∷\AgdaUnderscore{}}}\AgdaSpace{}%
\AgdaSymbol{:}%
\>[1768I]\AgdaSymbol{\{}\AgdaSymbol{@}\AgdaNumber{0}\AgdaSpace{}%
\AgdaBound{A}\AgdaSpace{}%
\AgdaSymbol{:}\AgdaSpace{}%
\AgdaPrimitive{Type}\AgdaSpace{}%
\AgdaGeneralizable{a}\AgdaSymbol{\}}\AgdaSpace{}%
\AgdaSymbol{\{}\AgdaSymbol{@}\AgdaNumber{0}\AgdaSpace{}%
\AgdaBound{n}\AgdaSpace{}%
\AgdaSymbol{:}\AgdaSpace{}%
\AgdaDatatype{ℕ}\AgdaSymbol{\}}\AgdaSpace{}%
\AgdaSymbol{→}\<%
\\
\>[.][@{}l@{}]\<[1768I]%
\>[8]\AgdaBound{A}\AgdaSpace{}%
\AgdaSymbol{→}\AgdaSpace{}%
\AgdaDatatype{Vec\ensuremath{{}^{\mkern1mu\mathrm{F}}}}\AgdaSpace{}%
\AgdaBound{A}\AgdaSpace{}%
\AgdaBound{n}\AgdaSpace{}%
\AgdaSymbol{→}\AgdaSpace{}%
\AgdaDatatype{Vec\ensuremath{{}^{\mkern1mu\mathrm{F}}}}\AgdaSpace{}%
\AgdaBound{A}\AgdaSpace{}%
\AgdaSymbol{(}\AgdaInductiveConstructor{suc}\AgdaSpace{}%
\AgdaBound{n}\AgdaSymbol{)}\<%
\\
\>[2]\AgdaBound{x}\AgdaSpace{}%
\AgdaOperator{\AgdaFunction{∷}}\AgdaSpace{}%
\AgdaBound{xs}\AgdaSpace{}%
\AgdaSymbol{=}\AgdaSpace{}%
\AgdaInductiveConstructor{cons}\AgdaSpace{}%
\AgdaBound{x}\AgdaSpace{}%
\AgdaBound{xs}\AgdaSpace{}%
\AgdaInductiveConstructor{refl}\<%
\end{code}
\end{minipage}\\
In the presence of \boxcong{} we can then define a high-level
eliminator:
\begin{code}%
\>[2]\AgdaFunction{Vec^{\mkern1mu\mathrm{F}}\mkern-2mu{}‐elim′}\AgdaSpace{}%
\AgdaSymbol{:}%
\>[1795I]\AgdaSymbol{\{}\AgdaSymbol{@}\AgdaNumber{0}\AgdaSpace{}%
\AgdaBound{A}\AgdaSpace{}%
\AgdaSymbol{:}\AgdaSpace{}%
\AgdaPrimitive{Type}\AgdaSpace{}%
\AgdaGeneralizable{a}\AgdaSymbol{\}}\AgdaSpace{}%
\AgdaSymbol{(}\AgdaBound{P}\AgdaSpace{}%
\AgdaSymbol{:}\AgdaSpace{}%
\AgdaSymbol{\{}\AgdaSymbol{@}\AgdaNumber{0}\AgdaSpace{}%
\AgdaBound{n}\AgdaSpace{}%
\AgdaSymbol{:}\AgdaSpace{}%
\AgdaDatatype{ℕ}\AgdaSymbol{\}}\AgdaSpace{}%
\AgdaSymbol{→}\AgdaSpace{}%
\AgdaDatatype{Vec\ensuremath{{}^{\mkern1mu\mathrm{F}}}}\AgdaSpace{}%
\AgdaBound{A}\AgdaSpace{}%
\AgdaBound{n}\AgdaSpace{}%
\AgdaSymbol{→}\AgdaSpace{}%
\AgdaPrimitive{Type}\AgdaSpace{}%
\AgdaGeneralizable{p}\AgdaSymbol{)}\AgdaSpace{}%
\AgdaSymbol{→}\<%
\\
\>[.][@{}l@{}]\<[1795I]%
\>[15]\AgdaBound{P}\AgdaSpace{}%
\AgdaFunction{[\ensuremath{\mkern1.5mu}]}\AgdaSpace{}%
\AgdaSymbol{→}\AgdaSpace{}%
\AgdaSymbol{(\{}\AgdaSymbol{@}\AgdaNumber{0}\AgdaSpace{}%
\AgdaBound{n}\AgdaSpace{}%
\AgdaSymbol{:}\AgdaSpace{}%
\AgdaDatatype{ℕ}\AgdaSymbol{\}}\AgdaSpace{}%
\AgdaSymbol{(}\AgdaBound{x}\AgdaSpace{}%
\AgdaSymbol{:}\AgdaSpace{}%
\AgdaBound{A}\AgdaSymbol{)}\AgdaSpace{}%
\AgdaSymbol{(}\AgdaBound{xs}\AgdaSpace{}%
\AgdaSymbol{:}\AgdaSpace{}%
\AgdaDatatype{Vec\ensuremath{{}^{\mkern1mu\mathrm{F}}}}\AgdaSpace{}%
\AgdaBound{A}\AgdaSpace{}%
\AgdaBound{n}\AgdaSymbol{)}\AgdaSpace{}%
\AgdaSymbol{→}\AgdaSpace{}%
\AgdaBound{P}\AgdaSpace{}%
\AgdaBound{xs}\AgdaSpace{}%
\AgdaSymbol{→}\AgdaSpace{}%
\AgdaBound{P}\AgdaSpace{}%
\AgdaSymbol{(}\AgdaBound{x}\AgdaSpace{}%
\AgdaOperator{\AgdaFunction{∷}}\AgdaSpace{}%
\AgdaBound{xs}\AgdaSymbol{))}\AgdaSpace{}%
\AgdaSymbol{→}\<%
\\
\>[15]\AgdaSymbol{\{}\AgdaSymbol{@}\AgdaNumber{0}\AgdaSpace{}%
\AgdaBound{n}\AgdaSpace{}%
\AgdaSymbol{:}\AgdaSpace{}%
\AgdaDatatype{ℕ}\AgdaSymbol{\}}\AgdaSpace{}%
\AgdaSymbol{(}\AgdaBound{xs}\AgdaSpace{}%
\AgdaSymbol{:}\AgdaSpace{}%
\AgdaDatatype{Vec\ensuremath{{}^{\mkern1mu\mathrm{F}}}}\AgdaSpace{}%
\AgdaBound{A}\AgdaSpace{}%
\AgdaBound{n}\AgdaSymbol{)}\AgdaSpace{}%
\AgdaSymbol{→}\AgdaSpace{}%
\AgdaBound{P}\AgdaSpace{}%
\AgdaBound{xs}\<%
\\
\>[2]\AgdaFunction{Vec^{\mkern1mu\mathrm{F}}\mkern-2mu{}‐elim′}\AgdaSpace{}%
\AgdaBound{P}\AgdaSpace{}%
\AgdaBound{n}\AgdaSpace{}%
\AgdaBound{c}\AgdaSpace{}%
\AgdaSymbol{=}%
\>[1852I]\AgdaFunction{Vec^{\mkern1mu\mathrm{F}}\mkern-2mu{}‐elim}\AgdaSpace{}%
\AgdaBound{P}\AgdaSpace{}%
\AgdaSymbol{(λ}\AgdaSpace{}%
\AgdaBound{eq}\AgdaSpace{}%
\AgdaSymbol{→}\AgdaSpace{}%
\AgdaFunction{J\ensuremath{{}^{\mkern2mu\mathrm{E}}}}\AgdaSpace{}%
\AgdaSymbol{(λ}\AgdaSpace{}%
\AgdaBound{eq}\AgdaSpace{}%
\AgdaSymbol{→}\AgdaSpace{}%
\AgdaBound{P}\AgdaSpace{}%
\AgdaSymbol{(}\AgdaInductiveConstructor{nil}\AgdaSpace{}%
\AgdaBound{eq}\AgdaSymbol{))}\AgdaSpace{}%
\AgdaBound{n}\AgdaSpace{}%
\AgdaBound{eq}\AgdaSymbol{)}\<%
\\
\>[1852I][@{}l@{\AgdaIndent{0}}]%
\>[23]\AgdaSymbol{(λ}\AgdaSpace{}%
\AgdaBound{x}\AgdaSpace{}%
\AgdaBound{xs}\AgdaSpace{}%
\AgdaBound{eq}\AgdaSpace{}%
\AgdaBound{p}\AgdaSpace{}%
\AgdaSymbol{→}\AgdaSpace{}%
\AgdaFunction{J\ensuremath{{}^{\mkern2mu\mathrm{E}}}}\AgdaSpace{}%
\AgdaSymbol{(λ}\AgdaSpace{}%
\AgdaBound{eq}\AgdaSpace{}%
\AgdaSymbol{→}\AgdaSpace{}%
\AgdaBound{P}\AgdaSpace{}%
\AgdaSymbol{(}\AgdaInductiveConstructor{cons}\AgdaSpace{}%
\AgdaBound{x}\AgdaSpace{}%
\AgdaBound{xs}\AgdaSpace{}%
\AgdaBound{eq}\AgdaSymbol{))}\AgdaSpace{}%
\AgdaSymbol{(}\AgdaBound{c}\AgdaSpace{}%
\AgdaBound{x}\AgdaSpace{}%
\AgdaBound{xs}\AgdaSpace{}%
\AgdaBound{p}\AgdaSymbol{)}\AgdaSpace{}%
\AgdaBound{eq}\AgdaSymbol{)}\<%
\end{code}
This implementation uses the following variant of the J rule:
\JEType{}%
For instance, in the \AgdaInductiveConstructor{nil} case we should
return something of type
\begin{code}[hide]%
\>[2]\AgdaKeyword{postulate}\<%
\\
\>[2][@{}l@{\AgdaIndent{0}}]%
\>[4]\AgdaPostulate{\AgdaUnderscore{}}%
\>[1885I]\AgdaSymbol{:}\<%
\end{code}
\begin{code}[inline]%
\>[.][@{}l@{}]\<[1885I]%
\>[6]\AgdaGeneralizable{P}\AgdaSpace{}%
\AgdaSymbol{(}\AgdaInductiveConstructor{nil}\AgdaSpace{}%
\AgdaGeneralizable{eq}\AgdaSymbol{)}\<%
\end{code},
but \AgdaBound{n} has type
\begin{code}[hide]%
\>[2]\AgdaKeyword{postulate}\<%
\\
\>[2][@{}l@{\AgdaIndent{0}}]%
\>[4]\AgdaPostulate{\AgdaUnderscore{}}%
\>[1888I]\AgdaSymbol{:}\<%
\end{code}
\mbox{\begin{code}[inline]%
\>[.][@{}l@{}]\<[1888I]%
\>[6]\AgdaGeneralizable{P}\AgdaSpace{}%
\AgdaSymbol{(}\AgdaInductiveConstructor{nil}\AgdaSpace{}%
\AgdaInductiveConstructor{refl}\AgdaSymbol{)}\<%
\end{code}},
and \AgdaFunction{J\ensuremath{{}^{\mkern2mu\mathrm{E}}}} is used to cast from the latter type to the
former.
The identity proof \AgdaBound{eq} is erased, so it is important that
the final argument of \AgdaFunction{J\ensuremath{{}^{\mkern2mu\mathrm{E}}}} is erased.
Note that the motive \AgdaBound{P} of \AgdaFunction{Vec^{\mkern1mu\mathrm{F}}\mkern-2mu{}‐elim′} is not
erased, but takes an erased natural number argument: this ensures that
the two uses of \AgdaFunction{J\ensuremath{{}^{\mkern2mu\mathrm{E}}}} are well-formed.

The function \AgdaFunction{J\ensuremath{{}^{\mkern2mu\mathrm{E}}}} is implemented using the function
\AgdaFunction{subst\ensuremath{{}^{\mkern2mu\mathrm{E}}}} from § \ref{sec:box-cong}:
\JEDef{}%
This code uses \JType{}.
The implicit argument \AgdaBound{A} of \AgdaFunction{subst\ensuremath{{}^{\mkern2mu\mathrm{E}}}} is here
the pair type
\begin{code}[hide]%
\>[2]\AgdaKeyword{postulate}\<%
\\
\>[2][@{}l@{\AgdaIndent{0}}]%
\>[4]\AgdaPostulate{\AgdaUnderscore{}}%
\>[1891I]\AgdaSymbol{:}\<%
\end{code}
\mbox{\begin{code}[inline]%
\>[.][@{}l@{}]\<[1891I]%
\>[6]\AgdaSymbol{(}%
\AgdaBound{y}\AgdaSpace{}%
\AgdaSymbol{:}\AgdaSpace{}%
\AgdaGeneralizable{A}\AgdaSymbol{)}\AgdaSpace{}%
\AgdaFunction{×}\AgdaSpace{}%
\AgdaDatatype{Id}\AgdaSpace{}%
\AgdaGeneralizable{A}\AgdaSpace{}%
\AgdaGeneralizable{x}\AgdaSpace{}%
\AgdaBound{y}\<%
\end{code}},
and the second explicit argument has type
\begin{code}[hide]%
\>[2]\AgdaKeyword{postulate}\<%
\\
\>[2][@{}l@{\AgdaIndent{0}}]%
\>[4]\AgdaPostulate{\AgdaUnderscore{}}%
\>[1901I]\AgdaSymbol{:}\AgdaSpace{}%
\AgdaSymbol{(}\AgdaBound{eq}\AgdaSpace{}%
\AgdaSymbol{:}\AgdaSpace{}%
\AgdaDatatype{Id}\AgdaSpace{}%
\AgdaGeneralizable{A}\AgdaSpace{}%
\AgdaGeneralizable{x}\AgdaSpace{}%
\AgdaGeneralizable{y}\AgdaSymbol{)}\AgdaSpace{}%
\AgdaSymbol{→}\<%
\end{code}
\begin{code}[inline]%
\>[.][@{}l@{}]\<[1901I]%
\>[6]\AgdaDatatype{Id}\AgdaSpace{}%
\AgdaSymbol{(}\AgdaSymbol{(}%
\AgdaBound{y}\AgdaSpace{}%
\AgdaSymbol{:}\AgdaSpace{}%
\AgdaGeneralizable{A}\AgdaSymbol{)}\AgdaSpace{}%
\AgdaFunction{×}\AgdaSpace{}%
\AgdaDatatype{Id}\AgdaSpace{}%
\AgdaGeneralizable{A}\AgdaSpace{}%
\AgdaGeneralizable{x}\AgdaSpace{}%
\AgdaBound{y}\AgdaSymbol{)}\AgdaSpace{}%
\AgdaSymbol{(}\AgdaGeneralizable{x}\AgdaSpace{}%
\AgdaOperator{\AgdaInductiveConstructor{,}}\AgdaSpace{}%
\AgdaInductiveConstructor{refl}\AgdaSymbol{)}\AgdaSpace{}%
\AgdaSymbol{(}\AgdaGeneralizable{y}\AgdaSpace{}%
\AgdaOperator{\AgdaInductiveConstructor{,}}\AgdaSpace{}%
\AgdaBound{eq}\AgdaSymbol{)}\<%
\end{code}.

The function \AgdaFunction{Vec^{\mkern1mu\mathrm{F}}\mkern-2mu{}‐elim′} computes ``as expected'':
\begin{code}[hide]%
\>[2]\AgdaKeyword{postulate}\<%
\\
\>[2][@{}l@{\AgdaIndent{0}}]%
\>[4]\AgdaPostulate{\AgdaUnderscore{}}%
\>[1925I]\AgdaSymbol{:}%
\>[1926I]\AgdaSymbol{(}\AgdaBound{P}\AgdaSpace{}%
\AgdaSymbol{:}\AgdaSpace{}%
\AgdaSymbol{∀}\AgdaSpace{}%
\AgdaSymbol{\{}\AgdaSymbol{@0}\AgdaSpace{}%
\AgdaBound{n}\AgdaSymbol{\}}\AgdaSpace{}%
\AgdaSymbol{→}\AgdaSpace{}%
\AgdaDatatype{Vec\ensuremath{{}^{\mkern1mu\mathrm{F}}}}\AgdaSpace{}%
\AgdaGeneralizable{A}\AgdaSpace{}%
\AgdaBound{n}\AgdaSpace{}%
\AgdaSymbol{→}\AgdaSpace{}%
\AgdaPrimitive{Type}\AgdaSpace{}%
\AgdaGeneralizable{p}\AgdaSymbol{)}\<%
\\
\>[.][@{}l@{}]\<[1926I]%
\>[8]\AgdaSymbol{(}\AgdaBound{n}\AgdaSpace{}%
\AgdaSymbol{:}\AgdaSpace{}%
\AgdaBound{P}\AgdaSpace{}%
\AgdaFunction{[\ensuremath{\mkern1.5mu}]}\AgdaSymbol{)}\<%
\\
\>[8]\AgdaSymbol{(}\AgdaBound{c}\AgdaSpace{}%
\AgdaSymbol{:}\AgdaSpace{}%
\AgdaSymbol{∀}\AgdaSpace{}%
\AgdaSymbol{\{}\AgdaSymbol{@0}\AgdaSpace{}%
\AgdaBound{n}\AgdaSymbol{\}}\AgdaSpace{}%
\AgdaSymbol{(}\AgdaBound{x}\AgdaSpace{}%
\AgdaSymbol{:}\AgdaSpace{}%
\AgdaGeneralizable{A}\AgdaSymbol{)}\AgdaSpace{}%
\AgdaSymbol{(}\AgdaBound{xs}\AgdaSpace{}%
\AgdaSymbol{:}\AgdaSpace{}%
\AgdaDatatype{Vec\ensuremath{{}^{\mkern1mu\mathrm{F}}}}\AgdaSpace{}%
\AgdaGeneralizable{A}\AgdaSpace{}%
\AgdaBound{n}\AgdaSymbol{)}\AgdaSpace{}%
\AgdaSymbol{→}\AgdaSpace{}%
\AgdaBound{P}\AgdaSpace{}%
\AgdaBound{xs}\AgdaSpace{}%
\AgdaSymbol{→}\AgdaSpace{}%
\AgdaBound{P}\AgdaSpace{}%
\AgdaSymbol{(}\AgdaBound{x}\AgdaSpace{}%
\AgdaOperator{\AgdaFunction{∷}}\AgdaSpace{}%
\AgdaBound{xs}\AgdaSymbol{))}\AgdaSpace{}%
\AgdaSymbol{→}\<%
\end{code}
\begin{code}[inline*]%
\>[.][@{}l@{}]\<[1925I]%
\>[6]\AgdaFunction{Vec^{\mkern1mu\mathrm{F}}\mkern-2mu{}‐elim′}\AgdaSpace{}%
\AgdaBound{P}\AgdaSpace{}%
\AgdaBound{n}\AgdaSpace{}%
\AgdaBound{c}\AgdaSpace{}%
\AgdaFunction{[\ensuremath{\mkern1.5mu}]}\<%
\end{code}
is definitionally equal to
\begin{code}[hide]%
\>[6]\AgdaOperator{\AgdaFunction{≡}}\<%
\end{code}
\begin{code}[inline]%
\>[6]\AgdaBound{n}\<%
\end{code},
and
\begin{code}[hide]%
\>[2]\AgdaKeyword{postulate}\<%
\\
\>[2][@{}l@{\AgdaIndent{0}}]%
\>[4]\AgdaPostulate{\AgdaUnderscore{}}%
\>[1966I]\AgdaSymbol{:}%
\>[1967I]\AgdaSymbol{(}\AgdaBound{P}\AgdaSpace{}%
\AgdaSymbol{:}\AgdaSpace{}%
\AgdaSymbol{∀}\AgdaSpace{}%
\AgdaSymbol{\{}\AgdaSymbol{@0}\AgdaSpace{}%
\AgdaBound{n}\AgdaSymbol{\}}\AgdaSpace{}%
\AgdaSymbol{→}\AgdaSpace{}%
\AgdaDatatype{Vec\ensuremath{{}^{\mkern1mu\mathrm{F}}}}\AgdaSpace{}%
\AgdaGeneralizable{A}\AgdaSpace{}%
\AgdaBound{n}\AgdaSpace{}%
\AgdaSymbol{→}\AgdaSpace{}%
\AgdaPrimitive{Type}\AgdaSpace{}%
\AgdaGeneralizable{p}\AgdaSymbol{)}\<%
\\
\>[.][@{}l@{}]\<[1967I]%
\>[8]\AgdaSymbol{(}\AgdaBound{n}\AgdaSpace{}%
\AgdaSymbol{:}\AgdaSpace{}%
\AgdaBound{P}\AgdaSpace{}%
\AgdaFunction{[\ensuremath{\mkern1.5mu}]}\AgdaSymbol{)}\<%
\\
\>[8]\AgdaSymbol{(}\AgdaBound{c}\AgdaSpace{}%
\AgdaSymbol{:}\AgdaSpace{}%
\AgdaSymbol{∀}\AgdaSpace{}%
\AgdaSymbol{\{}\AgdaSymbol{@0}\AgdaSpace{}%
\AgdaBound{n}\AgdaSymbol{\}}\AgdaSpace{}%
\AgdaSymbol{(}\AgdaBound{x}\AgdaSpace{}%
\AgdaSymbol{:}\AgdaSpace{}%
\AgdaGeneralizable{A}\AgdaSymbol{)}\AgdaSpace{}%
\AgdaSymbol{(}\AgdaBound{xs}\AgdaSpace{}%
\AgdaSymbol{:}\AgdaSpace{}%
\AgdaDatatype{Vec\ensuremath{{}^{\mkern1mu\mathrm{F}}}}\AgdaSpace{}%
\AgdaGeneralizable{A}\AgdaSpace{}%
\AgdaBound{n}\AgdaSymbol{)}\AgdaSpace{}%
\AgdaSymbol{→}\AgdaSpace{}%
\AgdaBound{P}\AgdaSpace{}%
\AgdaBound{xs}\AgdaSpace{}%
\AgdaSymbol{→}\AgdaSpace{}%
\AgdaBound{P}\AgdaSpace{}%
\AgdaSymbol{(}\AgdaBound{x}\AgdaSpace{}%
\AgdaOperator{\AgdaFunction{∷}}\AgdaSpace{}%
\AgdaBound{xs}\AgdaSymbol{))}\AgdaSpace{}%
\AgdaSymbol{→}\<%
\end{code}
\begin{code}[inline*]%
\>[.][@{}l@{}]\<[1966I]%
\>[6]\AgdaFunction{Vec^{\mkern1mu\mathrm{F}}\mkern-2mu{}‐elim′}\AgdaSpace{}%
\AgdaBound{P}\AgdaSpace{}%
\AgdaBound{n}\AgdaSpace{}%
\AgdaBound{c}\AgdaSpace{}%
\AgdaSymbol{(}\AgdaGeneralizable{x}\AgdaSpace{}%
\AgdaOperator{\AgdaFunction{∷}}\AgdaSpace{}%
\AgdaGeneralizable{xs}\AgdaSymbol{)}\<%
\end{code}
is definitionally equal to
\begin{code}[hide]%
\>[6]\AgdaOperator{\AgdaFunction{≡}}\<%
\end{code}
\mbox{\begin{code}[inline]%
\>[6]\AgdaBound{c}\AgdaSpace{}%
\AgdaGeneralizable{x}\AgdaSpace{}%
\AgdaGeneralizable{xs}\AgdaSpace{}%
\AgdaSymbol{(}\AgdaFunction{Vec^{\mkern1mu\mathrm{F}}\mkern-2mu{}‐elim′}\AgdaSpace{}%
\AgdaBound{P}\AgdaSpace{}%
\AgdaBound{n}\AgdaSpace{}%
\AgdaBound{c}\AgdaSpace{}%
\AgdaGeneralizable{xs}\AgdaSymbol{)}\<%
\end{code}}.
Thus \AgdaDatatype{Vec\ensuremath{{}^{\mkern1mu\mathrm{F}}}}, \AgdaInductiveConstructor{[\ensuremath{\mkern1.5mu}]},
\AgdaInductiveConstructor{\AgdaUnderscore{}∷\AgdaUnderscore{}} and
\AgdaFunction{Vec^{\mkern1mu\mathrm{F}}\mkern-2mu{}‐elim′} provide an encoding of \AgdaDatatype{Vec\ensuremath{{}^{\mkern1mu\mathrm{D}}}}
that is, in some sense, reasonable.

If unrestricted erased matches are allowed for identity types, then
one can implement a variant of \AgdaFunction{Vec^{\mkern1mu\mathrm{F}}\mkern-2mu{}‐elim′} with an
erased motive \AgdaBound{P}.
On the other hand, if no erased matches are allowed, then
\AgdaFunction{Vec^{\mkern1mu\mathrm{F}}\mkern-2mu{}‐elim′} cannot be implemented (in certain type
theories) – using \AgdaFunction{Vec^{\mkern1mu\mathrm{F}}\mkern-2mu{}‐elim′} one can implement a
limited variant of \boxcong{} that, for certain type theories, cannot
be implemented in the absence of erased matches (see
Theorem~\ref{thm:no-box-cong}):
\begin{code}%
\>[2]\AgdaFunction{[\ensuremath{\mkern1.5mu}]‐cong\ensuremath{{}_{\mathrm{0}}}}\AgdaSpace{}%
\AgdaSymbol{:}\AgdaSpace{}%
\AgdaSymbol{\{}\AgdaSymbol{@}\AgdaNumber{0}\AgdaSpace{}%
\AgdaBound{n}\AgdaSpace{}%
\AgdaSymbol{:}\AgdaSpace{}%
\AgdaDatatype{ℕ}\AgdaSymbol{\}}\AgdaSpace{}%
\AgdaSymbol{→}\AgdaSpace{}%
\AgdaSymbol{@}\AgdaNumber{0}\AgdaSpace{}%
\AgdaDatatype{Id}\AgdaSpace{}%
\AgdaDatatype{ℕ}\AgdaSpace{}%
\AgdaInductiveConstructor{zero}\AgdaSpace{}%
\AgdaBound{n}\AgdaSpace{}%
\AgdaSymbol{→}\AgdaSpace{}%
\AgdaDatatype{Id}\AgdaSpace{}%
\AgdaSymbol{(}\AgdaDatatype{Erased}\AgdaSpace{}%
\AgdaDatatype{ℕ}\AgdaSymbol{)}\AgdaSpace{}%
\AgdaOperator{\AgdaInductiveConstructor{[}}\AgdaSpace{}%
\AgdaInductiveConstructor{zero}\AgdaSpace{}%
\AgdaOperator{\AgdaInductiveConstructor{]}}\AgdaSpace{}%
\AgdaOperator{\AgdaInductiveConstructor{[}}\AgdaSpace{}%
\AgdaBound{n}\AgdaSpace{}%
\AgdaOperator{\AgdaInductiveConstructor{]}}\<%
\\
\>[2]\AgdaFunction{[\ensuremath{\mkern1.5mu}]‐cong\ensuremath{{}_{\mathrm{0}}}}\AgdaSpace{}%
\AgdaBound{eq}\AgdaSpace{}%
\AgdaSymbol{=}\AgdaSpace{}%
\AgdaFunction{Vec^{\mkern1mu\mathrm{F}}\mkern-2mu{}‐elim′}\<%
\\
\>[2][@{}l@{\AgdaIndent{0}}]%
\>[4]\AgdaSymbol{(}\AgdaFunction{Vec^{\mkern1mu\mathrm{F}}\mkern-2mu{}‐elim}\AgdaSpace{}%
\AgdaSymbol{(λ}\AgdaSpace{}%
\AgdaBound{\AgdaUnderscore{}}\AgdaSpace{}%
\AgdaSymbol{→}\AgdaSpace{}%
\AgdaPrimitive{Type}\AgdaSymbol{)}\AgdaSpace{}%
\AgdaSymbol{(λ}\AgdaSpace{}%
\AgdaSymbol{\{}\AgdaBound{n}\AgdaSymbol{\}}\AgdaSpace{}%
\AgdaBound{\AgdaUnderscore{}}\AgdaSpace{}%
\AgdaSymbol{→}\AgdaSpace{}%
\AgdaDatatype{Id}\AgdaSpace{}%
\AgdaSymbol{(}\AgdaDatatype{Erased}\AgdaSpace{}%
\AgdaDatatype{ℕ}\AgdaSymbol{)}\AgdaSpace{}%
\AgdaOperator{\AgdaInductiveConstructor{[}}\AgdaSpace{}%
\AgdaInductiveConstructor{zero}\AgdaSpace{}%
\AgdaOperator{\AgdaInductiveConstructor{]}}\AgdaSpace{}%
\AgdaOperator{\AgdaInductiveConstructor{[}}\AgdaSpace{}%
\AgdaBound{n}\AgdaSpace{}%
\AgdaOperator{\AgdaInductiveConstructor{]}}\AgdaSymbol{)}\AgdaSpace{}%
\AgdaSymbol{(λ}\AgdaSpace{}%
\AgdaBound{\AgdaUnderscore{}}\AgdaSpace{}%
\AgdaBound{\AgdaUnderscore{}}\AgdaSpace{}%
\AgdaBound{\AgdaUnderscore{}}\AgdaSpace{}%
\AgdaBound{\AgdaUnderscore{}}\AgdaSpace{}%
\AgdaSymbol{→}\AgdaSpace{}%
\AgdaRecord{⊤}\AgdaSymbol{))}\<%
\\
\>[4]\AgdaInductiveConstructor{refl}\AgdaSpace{}%
\AgdaSymbol{(λ}\AgdaSpace{}%
\AgdaBound{\AgdaUnderscore{}}\AgdaSpace{}%
\AgdaBound{\AgdaUnderscore{}}\AgdaSpace{}%
\AgdaBound{\AgdaUnderscore{}}\AgdaSpace{}%
\AgdaSymbol{→}\AgdaSpace{}%
\AgdaInductiveConstructor{tt}\AgdaSymbol{)}\AgdaSpace{}%
\AgdaSymbol{(}\AgdaInductiveConstructor{nil}\AgdaSpace{}%
\AgdaSymbol{\{}\AgdaArgument{A}\AgdaSpace{}%
\AgdaSymbol{=}\AgdaSpace{}%
\AgdaRecord{⊤}\AgdaSymbol{\}}\AgdaSpace{}%
\AgdaBound{eq}\AgdaSymbol{)}\<%
\end{code}
(Here \AgdaRecord{⊤} is a unit type with constructor
\AgdaInductiveConstructor{tt}.
The notation
\begin{code}[hide]%
\>[2]\AgdaKeyword{postulate}\<%
\\
\>[2][@{}l@{\AgdaIndent{0}}]%
\>[4]\AgdaPostulate{\AgdaUnderscore{}}%
\>[2075I]\AgdaSymbol{:}\AgdaSpace{}%
\AgdaOperator{\AgdaGeneralizable{Has‐type}}\<%
\end{code}
\begin{code}[inline*]%
\>[.][@{}l@{}]\<[2075I]%
\>[6]\AgdaInductiveConstructor{nil}\AgdaSpace{}%
\AgdaSymbol{\{}\AgdaArgument{A}\AgdaSpace{}%
\AgdaSymbol{=}\AgdaSpace{}%
\AgdaGeneralizable{B}\AgdaSymbol{\}}\<%
\end{code}
means ``\AgdaInductiveConstructor{nil}, with the implicit argument
\AgdaArgument{A} instantiated to \AgdaBound{B}''.)
One could try to use the type \AgdaFunction{Vec\ensuremath{{}^{\mkern2mu\mathrm{L}}}} from
§ \ref{sec:erased-postulates-and-erased-matches} instead of
\AgdaFunction{Vec\ensuremath{{}^{\mkern1mu\mathrm{F}}}}, but we get a similar situation:
\AgdaFunction{Vec\ensuremath{{}^{\mkern2mu\mathrm{L}}}‐elim′} can be implemented (with propositional
``computation'' rules) using \boxcong{}, and \AgdaFunction{[\ensuremath{\mkern1.5mu}]‐cong\ensuremath{{}_{\mathrm{0}}}}
can be implemented using \AgdaFunction{Vec\ensuremath{{}^{\mkern2mu\mathrm{L}}}‐elim′}.
Thus, if a designer of a dependently typed programming language wants
to support \AgdaDatatype{Vec\ensuremath{{}^{\mkern1mu\mathrm{D}}}} in the surface language, but only more
basic data types in the core language, then it may make sense to allow
erased matches (of some kind) for identity types in the core language.

Here we have only looked at implementations of vectors: I expect that
similar ideas work for a larger class of inductive families, but that
remains to be investigated.

\section{Type Theory with Identity Types, Quotients and Erasure}
\label{sec:typing}

Let me now present a family of type theories with support for erasure.

The Agda formalisation supports identity types, set quotients,
Π-types, Σ-types with and without η-equality, unit types with and
without η-equality, an empty type, natural numbers, an infinite
hierarchy of universes, optional first-class universe levels, lift
types, and top-level, possibly opaque definitions.
However, here I will focus on identity types, quotients and Π-types.
Full typing rules (excluding rules for first-class universe levels and
top-level definitions) can be found in
\appendixA{}\ifAppendicesIncluded{}{ in the supplementary material}.

The formalisation also supports various \emph{modality semirings},
following \citet{mcbride-2016}, \citet{atkey-2018},
\citet{abel-et-al-2023} and others.
However, here I present the type system specialised to the erasure
semiring,\footnote{The definition of the general rules for identity
  types, with support for arbitrary modality semirings, was influenced
  by input from \ifAnonymous{[\emph{name omitted to make it harder to
      deanonymise the paper}]}{Oskar Eriksson}. Those rules are not
  presented here.} with the two grades $\zeroG$ and $\omegaG$:
$\zeroG$ stands for ``not used'', and $\omegaG$ for ``used an
arbitrary number of times''.
The semiring comes with a multiplication operator with $\zeroG$ as a
zero and $\omegaG$ as an identity.
There is also a total order with $\leqG{\omegaG}{\zeroG}$.

The formalisation uses well-scoped syntax, but such details have been
omitted here: the definitions should be read as if everything were
well-scoped.
(There are also other, minor differences between the text and the
formalisation.)
There are a number of judgements that are defined using mutual
induction:
\begin{itemize}
\item Well-typed and well-resourced terms: $\wrTm{γ}{Γ}{t}{p}{A}$.
  Here $γ$ is a grade context, $Γ$ a type context of equal length, and
  $p$ a grade.
  One can think of $γ$ as specifying how each variable is allowed to
  be used: variables with grade $\zeroG$ must only be used in
  \emph{erased contexts}, i.e.\ when the \emph{mode} $p$ is $\zeroG$.
  The notation $\wfTm{Γ}{t}{A}$ is an abbreviation for
  $\wrTm{\zeroGC}{Γ}{t}{\zeroG}{A}$, where $\zeroGC$ is a grade
  context where every entry is $\zeroG$ ($\wrTm{γ}{Γ}{t}{p}{A}$
  provably implies $\wfTm{Γ}{t}{A}$).
\item Well-formed and well-resourced types: $\wrTy{γ}{Γ}{p}{A}$.
  The notation $\wfTy{Γ}{A}$ is an abbreviation for
  $\wrTy{\zeroGC}{Γ}{\zeroG}{A}$ (and $\wrTy{γ}{Γ}{p}{A}$ implies
  $\wfTy{Γ}{A}$).
\item Well-formed contexts, $\wfCtxt{Γ}$, judgemental equality of
  types, $\eqTy{Γ}{A}{B}$, and judgemental equality of terms,
  $\eqTm{Γ}{t}{u}{A}$.
  These relations are not defined with reference to grade contexts or
  modes.
\end{itemize}

Here are the conversion rule and the typing rule for variables (de
Bruijn indices):
\begin{mathpar}
  \inferrule{\wrTm{γ}{Γ}{t}{p}{A} \\ \eqTy{Γ}{A}{B}}{%
    \wrTm{γ}{Γ}{t}{p}{B}}
  \and
  \inferrule{\wfCtxt{Γ} \\ \varTy{x}{A}{Γ} \\ \leqG{\lookup{γ}{x}}{p}}{%
    \wrTm{γ}{Γ}{x}{p}{A}}
\end{mathpar}
The judgemental equality rules have mostly been relegated to
\appendixA{}.
They are rather standard, and include references to $\wfTy{Γ}{A}$ and
$\wfTm{Γ}{t}{A}$, but not other instances of $\wrTy{γ}{Γ}{p}{A}$ and
$\wrTm{γ}{Γ}{t}{p}{A}$.
For the variable rule, if the mode is $p$, then the variable $x$ is
well-resourced if its grade is at most $p$: if the mode is $\omegaG$,
then the grade of the variable must be $\omegaG$, but in erased
contexts the grade can be anything.
(The definition of $\varTy{x}{A}{Γ}$ can be found in \appendixA{}.)

For Π-types we have the following rules, following
\citet{abel-et-al-2021} and \citet{danielsson-et-al-2026}:
\begin{mathpar}
  \inferrule{\wrTy{γ}{Γ}{\mul{q}{p}}{A} \\
    \wrTy{\consGC{γ}{p}}{\consC{Γ}{A}}{q}{B}}{%
    \wrTy{γ}{Γ}{q}{\PiT{p}{A}{B}}}
  \and
  \inferrule{%
    \wrTm{γ}{Γ}{A}{\mul{q}{p}}{\U{l}} \\
    \wrTm{\consGC{γ}{p}}{\consC{Γ}{A}}{B}{q}{\U{l}}}{%
    \wrTm{γ}{Γ}{\PiT{p}{A}{B}}{q}{\U{l}}}
  \and
  \inferrule{\wrTm{\consGC{γ}{p}}{\consC{Γ}{A}}{t}{q}{B}}{%
    \wrTm{γ}{Γ}{\lam{p}{t}}{q}{\PiT{p}{A}{B}}}
  \and
  \inferrule{\wrTm{γ}{Γ}{t}{q}{\PiT{p}{A}{B}} \\
    \wrTm{γ}{Γ}{u}{\mul{q}{p}}{A}}{%
    \wrTm{γ}{Γ}{\app{t}{p}{u}}{q}{\substZ{B}{u}}}
\end{mathpar}
The Π-type $\PiT{p}{A}{B}$ comes with a grade $p$.
If this grade is $\zeroG$, then the argument may only be used in
erased contexts in the codomain $B$, and furthermore the mode is
$\zeroG$ for the antecedents $\wrTy{γ}{Γ}{\mul{q}{p}}{A}$ and
$\wrTm{γ}{Γ}{A}{\mul{q}{p}}{\U{l}}$.
The first rule gives a condition for $\PiT{p}{A}{B}$ being a type, and
the second rule gives a condition for $\PiT{p}{A}{B}$ being a code in
the universe $\U{l}$.
Lambdas and applications are annotated with grades, and these grades
have to match the grade of the Π-type.
If the grade is $\zeroG$, then the bound variable may only be used in
erased contexts in the bodies of lambdas, and correspondingly the
argument of an application is in an erased context.

In some previous work introduction rules for the universe are only
given for the case where the mode is $\zeroG$
\citep{mcbride-2016,atkey-2018}.
Here I use rules based on those given by \citet{abel-et-al-2021}.
As mentioned above I do not prove correctness of erasure in the
presence of both erased, postulated univalence and \boxcong{}, but I
base the design on that of \citeauthor{abel-et-al-2021} in the hope
that such a proof will be possible.
This design is also used by Agda.

The rules for Σ and unit types are omitted, but note that one can
choose whether or not to allow erased matches for the weak variants of
those types (i.e.\ the variants with pattern matching instead of
projections and η-equality).
The rules for the empty type are also omitted.
Again one can choose whether or not to allow erased matches for this
type.
The statement of Theorem~\ref{thm:soundness-of-erasure} below contains
assumptions related to erased matches for these types.

Let us now consider identity types.
One can choose between the following two formation rules (and
corresponding rules for type formation):
\begin{mathpar}
  \inferrule{\wrTm{γ}{Γ}{A}{p}{\U{l}} \\\\ \wrTm{γ}{Γ}{t}{p}{A} \\
    \wrTm{γ}{Γ}{u}{p}{A}}{%
    \wrTm{γ}{Γ}{\IdT{A}{t}{u}}{p}{\U{l}}}
  \and
  \inferrule{\wfTm{Γ}{A}{\U{l}} \\ \wfTm{Γ}{t}{A} \\ \wfTm{Γ}{u}{A}}{%
    \wrTm{γ}{Γ}{\IdT{A}{t}{u}}{p}{\U{l}}}
\end{mathpar}
Cubical Agda allows something akin to the first rule with
$\wfTm{Γ}{A}{\U{l}}$ instead of \mbox{$\wrTm{γ}{Γ}{A}{p}{\U{l}}$}.
Here I have instead chosen a rule based on the one given for paths by
\citet{abel-et-al-2021}.
The second rule above is not compatible with erased univalence in the
presence of \boxcong{} (see §~\ref{sec:box-cong}).
However, all results presented below hold with either rule.

The following rules are used for reflexivity and equality reflection:
\begin{mathpar}
  \inferrule{\wfTm{Γ}{t}{A}}{\wrTm{γ}{Γ}{\rfl}{p}{\IdT{A}{t}{t}}}
  \and
  \inferrule{\wfTm{Γ}{v}{\IdT{A}{t}{u}}}{\eqTm{Γ}{t}{u}{A}}
\end{mathpar}
The reflexivity constructor does not take any arguments.
Equality reflection, which makes univalence inconsistent, is optional.

For the J rule – the main eliminator for the identity type – one can
choose between three different sets of rules.
In all cases the same rule, presented below, is used when the mode is
$\zeroG$.
If no erased matches are allowed for J, then the following rule is
used when the mode is $\omegaG$:
\begin{mathpar}
  \inferrule{\wfTy{Γ}{A} \\ \wrTm{γ}{Γ}{t}{\omegaG}{A} \\
    \wrTy{\consGC{\consGC{γ}{p}}{q}}{%
      \consC{\consC{Γ}{A}}{\IdT{(\wkO{A})}{(\wkO{t})}{\varZ}}}{%
      \omegaG}{B} \\
    \wrTm{γ}{Γ}{u}{\omegaG}{\substO{B}{t}{\rfl}} \\
    \wrTm{γ}{Γ}{v}{\omegaG}{A} \\
    \wrTm{γ}{Γ}{w}{\omegaG}{\IdT{A}{t}{v}}}{%
    \wrTm{γ}{Γ}{\J{p}{q}{A}{t}{B}{u}{v}{w}}{\omegaG}{\substO{B}{v}{w}}}
\end{mathpar}
The eliminator $\JName$ is annotated with two grades.
Those specify how the bound arguments are allowed to be used in the
motive \AgdaBound{B}.
(The notation $\wkO{A}$ stands for $A$, weakened one step.)
Note that the grade context $γ$ and the mode $\omegaG$ are used for
all antecedents except for the first one.

One can also allow limited forms of erased matches.
In that case there are two rules for the mode $\omegaG$ – the one
above, restricted to the case where at least one of $p$ and $q$ is
$\omegaG$, and the following one:
\begin{mathpar}
  \inferrule{\wfTy{Γ}{A} \\ \wfTm{Γ}{t}{A} \\
    \wrTy{\consGC{\consGC{γ}{\zeroG}}{\zeroG}}{%
      \consC{\consC{Γ}{A}}{\IdT{(\wkO{A})}{(\wkO{t})}{\varZ}}}{%
      \omegaG}{B} \\
    \wrTm{γ}{Γ}{u}{\omegaG}{\substO{B}{t}{\rfl}} \\
    \wfTm{Γ}{v}{A} \\
    \wfTm{Γ}{w}{\IdT{A}{t}{v}}}{%
    \wrTm{γ}{Γ}{\J{\zeroG}{\zeroG}{A}{t}{B}{u}{v}{w}}{\omegaG}{%
      \substO{B}{v}{w}}}
\end{mathpar}
Note that the grade context $γ$ and the mode $\omegaG$ are now only
used for two antecedents, those for $B$ and $u$.
This set of rules is equally expressive as the previous one plus
\boxcong{}, see § \ref{sec:no-box-cong}.

Finally one can allow unrestricted erased matches for J.
This rule is also used, in all cases, when the mode $r$ is $\zeroG$:
\begin{mathpar}
  \inferrule{\wfTy{Γ}{A} \\ \wfTm{Γ}{t}{A} \\
    \wfTy{\consC{\consC{Γ}{A}}{\IdT{(\wkO{A})}{(\wkO{t})}{\varZ}}}{B} \\
    \wrTm{γ}{Γ}{u}{r}{\substO{B}{t}{\rfl}} \\
    \wfTm{Γ}{v}{A} \\
    \wfTm{Γ}{w}{\IdT{A}{t}{v}}}{%
    \wrTm{γ}{Γ}{\J{p}{q}{A}{t}{B}{u}{v}{w}}{r}{%
      \substO{B}{v}{w}}}
\end{mathpar}
Here $γ$ and $r$ are only used for one antecedent, the one for $u$,
and $p$ and $q$ are ignored (for simplicity the same syntax is used
for all the rules).
The intuition here is that, because we have a judgemental equality
between $\J{p}{q}{A}{t}{B}{u}{t}{\rfl}$ and $u$, the only
computationally relevant argument is $u$.
However, this rule allows one to implement functions like
\AgdaFunction{subst′} and \AgdaFunction{subst″} from
§~\ref{sec:erased-postulates-and-erased-matches} and
§~\ref{sec:box-cong}, so the type theory we get with this rule is not
compatible with erased, postulated univalence.

The formalisation also has optional support for the K rule – optional
because the K rule makes univalence inconsistent – with three sets of
typing rules akin to the ones presented for J above.
Those rules can be found in \appendixA{}.

The formalisation supports the definition of a type like \Erased{}
using a graded Σ-type \citep{abel-et-al-2023}, a unit type and a lift
type.
Instead of presenting rules for Σ, unit and lift types I present some
admissible typing rules for a type $\EName$.
In the formalisation one can choose between two variants of $\EName$,
with or without η-equality; the following rules are admissible in
either case:
\begin{mathpar}
  \inferrule{\wfTm{Γ}{A}{\U{l}}}{\wrTm{γ}{Γ}{\E{l}{A}}{p}{\U{l}}}
  \and
  \inferrule{\wfTm{Γ}{t}{A}}{\wrTm{γ}{Γ}{\bx{t}}{p}{\E{l}{A}}}
  \and
  \inferrule{\wfTm{Γ}{t}{\E{l}{A}}}{\wfTm{Γ}{\erased{A}{t}}{A}}
\end{mathpar}
Note that the projection $\erasedName$ can only be used in erased
contexts, and that the argument $t$ of $\bx{t}$ is in an erased
context.
The application $\erased{A}{\bx{t}}$ is judgementally equal to $t$
(given $\wfTm{Γ}{t}{A}$).

The type theory has \emph{optional} support for \boxcong{}.
The typing rules for $\bcName$ refer to $\EName$ (in the formalisation
one can choose to support one or both of two variants of $\bcName$,
one for $\EName$ without η-equality, and one for $\EName$ with
η-equality):
\begin{mathpar}
  \inferrule{\wfTy{Γ}{A} \\ \wfTm{Γ}{t}{A} \\ \wfTm{Γ}{u}{A} \\
    \wfTm{Γ}{v}{\IdT{A}{t}{u}}}{%
    \wrTm{γ}{Γ}{\bc{l}{A}{t}{u}{v}}{p}{\IdT{(\E{l}{A})}{\bx{t}}{\bx{u}}}}
  \and
  \inferrule{\wfTm{Γ}{t}{A}}{%
    \eqTm{Γ}{\bc{l}{A}{t}{t}{\rfl}}{\rfl}{\IdT{(\E{l}{A})}{\bx{t}}{\bx{t}}}}
\end{mathpar}
The congruence rule has been omitted.
Note that all arguments of $\bcWithLevel{l}$ are in erased contexts.

The rules for quotients are based on the interface presented in
§~\ref{sec:quotients} and omitted; these rules can be found in
\appendixA{}.
However, note that quotient types are optional, and that the rules use
binders instead of Π-types.
The point constructor $\className$ only takes one argument.
The erased higher constructors are included as term formers
($\resp{A}{R}{t}{u}{v}$ and $\set{A}{R}{t}{u}{v}{w}$).
In the absence of equality reflection they are treated as neutral, but
in the presence of equality reflection well-typed applications of
these constructors reduce to $\rfl$: note that, in the presence of
equality reflection, every identity proof is judgementally equal to
$\rfl$.

There are six reduction judgements:
\begin{itemize}
\item First weak head reduction for terms is defined:
  $\redTm{Γ}{t}{u}{A}$ means that, under the context $Γ$, $t$ reduces
  in a single step to $u$ at type $A$.
  The rules of this judgement can be found in \appendixA{}.
  It is deterministic, and compatible with judgemental equality:
  \mbox{$\redTm{Γ}{t}{u}{A}$} implies $\eqTm{Γ}{t}{u}{A}$.
  A corresponding relation allowing zero or more steps is denoted by
  $\redsTm{Γ}{t}{u}{A}$.
\item Then weak head reduction for types is defined: $\redTy{Γ}{A}{B}$
  holds if $\redTm{Γ}{A}{B}{\U{l}}$, and $\redsTy{Γ}{A}{B}$ is the
  corresponding multi-step relation.
\item There is also a reduction relation that allows evaluation under
  successor constructors:
  \begin{mathpar}
    \inferrule{\redTm{Γ}{t}{u}{\Nat}}{\redSuc{Γ}{t}{u}}
    \and
    \inferrule{\redSuc{Γ}{t}{u}}{\redSuc{Γ}{\suc{t}}{\suc{u}}}
  \end{mathpar}
  The corresponding multi-step relation is denoted by
  $\redsSuc{Γ}{t}{u}$.
  This relation is used in the statement of
  Theorem~\ref{thm:soundness-of-erasure} below.
\end{itemize}

In § \ref{sec:soundness-erasure-postulates} and
§ \ref{sec:no-box-cong} a well-formed substitution relation is used:
$\wrSubst{δ}{Δ}{σ}{p}{Γ}$ means that $σ$ is a well-formed parallel
substitution from $Γ$ to $Δ$, where all terms in $σ$ are
well-resourced with respect to $δ$ and $p$.
The notation $\wfSubst{Δ}{σ}{Γ}$ stands for
$\wrSubst{\zeroGC}{Δ}{σ}{\zeroG}{Γ}$.
In § \ref{sec:no-box-cong} \emph{weakenings} are also used.
$\wfWk{Δ}{ρ}{Γ}$ means that $ρ$ is a well-formed weakening from $Γ$ to
$Δ$.

The family of type theories presented above includes a number of
options.
Some of them affect results presented below.
They are summarised here:
\begin{itemize}
\item Equality reflection, K, \boxcong{}, and quotient types are
  optional.
\item Erased matches for the empty type, weak Σ-types, and the weak
  unit type are optional.
\item For J one can choose between no erased matches, limited erased
  matches, and unrestricted erased matches, and similarly for K.
\end{itemize}

\section{Basic Meta-Theory}
\label{sec:meta-theory}

Let us now look at some basic meta-theory for this family of type
theories.
(I do not claim that the results about identity types are particularly
novel. See, for instance, \citet{adjedj-et-al-2024} for formalised
meta-theory of Martin-Löf type theory with intensional identity
types.)

\citet{abel-et-al-2017,abel-et-al-2023} prove normalisation,
consistency etc.\ using a logical relations argument.
Extending this argument to support identity types is straightforward
(in the absence of equality reflection).
I do not present a full argument here, but a key part of the
formalised proof is to define two terms as reducibly equal at a type
that reduces to $\IdT{A}{t}{u}$ if (roughly) they both reduce to
$\rfl$ and $t$ and $u$ are reducibly equal at type $A$, or if both
reduce to neutral, judgementally equal terms.\footnote{This definition
  is based on a suggestion from \ifAnonymous{[\emph{name omitted to
      make it harder to deanonymise the paper}]}{Andreas Abel}.}
In the absence of equality reflection the higher quotient constructors
are treated as neutral, and two terms are reducibly equal at a type
that reduces to a quotient type if (roughly) they both reduce to
neutral, judgementally equal terms, or if they both reduce to
applications of the point constructor and the arguments are reducibly
equal (\citet{pujet-tabareau-2022} use a similar definition).

The formalisation also supports type theories with equality
reflection.
For such theories the logical relations omit cases for neutral terms.
For quotients the point constructor case mentioned above is extended
with an option for the point constructor arguments to be related by
the symmetric, transitive closure of (roughly) ``two terms are related
if they are reducible and there is a reducible proof showing that they
are related by the quotient relation''.
If the quotient relation had been required to be an equivalence
relation, then one could presumably have skipped taking the symmetric,
transitive closure.

We obtain a number of meta-theoretical properties, including the
following ones related to empty contexts ($\emptyC$):
\begin{itemize}
\item Canonicity for $\IdName$ does not hold in the presence of the
  ``postulated'' higher quotient constructors, unless equality
  reflection is allowed, in which case well-typed applications of
  these constructors reduce to $\rfl$.
  We have that, if quotient types are not allowed or equality
  reflection is allowed, and $\wfTm{\emptyC}{v}{\IdT{A}{t}{u}}$, then
  $\redsTm{\emptyC}{v}{\rfl}{\IdT{A}{t}{u}}$ and
  $\eqTm{\emptyC}{t}{u}{A}$.
\item Similarly, if quotient types are not allowed or equality
  reflection is allowed, and $\wfTm{\emptyC}{t}{\Nat}$, then there is
  a natural number $n$ such that
  $\eqTm{\emptyC}{t}{\numeralAsTerm{n}}{\Nat}$, where
  $\numeralAsTerm{n}$ stands for $\numeralAsTermExplanation{n}$.
\item Consistency: there is no term $t$ such that
  $\wfTm{\emptyC}{t}{\Empty}$.
  From the logical relation we get that such a term must reduce to a
  closed, neutral term.
  In the absence of quotient types there are no closed, neutral terms,
  but if quotient types are allowed and equality reflection is
  disallowed, then the higher quotient constructors are treated as
  neutral.
  However, if $\wfTm{\emptyC}{t}{\Empty}$ holds in the absence of
  equality reflection, then it holds also when equality reflection is
  turned on, so consistency holds unconditionally.
\end{itemize}
The following properties hold in the absence of equality reflection,
or if the context $Γ$ is empty:
\begin{itemize}
\item Weak head normalisation for types and terms: if $\wfTy{Γ}{A}$
  then there is some weak head normal form (WHNF) $B$ such that
  $\redsTy{Γ}{A}{B}$, and if $\wfTm{Γ}{t}{A}$ then there is some WHNF
  $u$ such that $\redsTm{Γ}{t}{u}{A}$.
  If $\wfTm{Γ}{v}{\IdT{A}{t}{u}}$ then
  $\redsTm{Γ}{v}{w}{\IdT{A}{t}{u}}$ for some $w$ that is either $\rfl$
  (in which case $\eqTm{Γ}{t}{u}{A}$) or neutral.
\item If $\eqTy{Γ}{\IdT{A}{t}{u}}{B}$, then $B$ is not neutral and not
  $\U{l}$, a lift type, the empty type, a unit type, a Π-type, a
  Σ-type, $\Nat$, or a quotient type.
  A corresponding property holds for $\Quot{A}{B}$.
\item Injectivity of $\IdName$: If
  $\eqTy{Γ}{\IdT{A\ensuremath{{}_{\mathrm{1}}}}{t\ensuremath{{}_{\mathrm{1}}}}{u\ensuremath{{}_{\mathrm{1}}}}}{\IdT{A\ensuremath{{}_{\mathrm{2}}}}{t\ensuremath{{}_{\mathrm{2}}}}{u\ensuremath{{}_{\mathrm{2}}}}}$, then
  $\eqTy{Γ}{A\ensuremath{{}_{\mathrm{1}}}}{A\ensuremath{{}_{\mathrm{2}}}}$, $\eqTm{Γ}{t\ensuremath{{}_{\mathrm{1}}}}{t\ensuremath{{}_{\mathrm{2}}}}{A\ensuremath{{}_{\mathrm{1}}}}$ and
  $\eqTm{Γ}{u\ensuremath{{}_{\mathrm{1}}}}{u\ensuremath{{}_{\mathrm{2}}}}{A\ensuremath{{}_{\mathrm{1}}}}$.
\end{itemize}
The following properties hold in the absence of equality reflection:
\begin{itemize}
\item Injectivity of $\QuotName$: if
  $\eqTy{Γ}{\Quot{A\ensuremath{{}_{\mathrm{1}}}}{B\ensuremath{{}_{\mathrm{1}}}}}{\Quot{A\ensuremath{{}_{\mathrm{2}}}}{B\ensuremath{{}_{\mathrm{2}}}}}$, then $\eqTy{Γ}{A\ensuremath{{}_{\mathrm{1}}}}{A\ensuremath{{}_{\mathrm{2}}}}$
  and $\eqTy{\consC{\consC{Γ}{A\ensuremath{{}_{\mathrm{1}}}}}{\wkO{A\ensuremath{{}_{\mathrm{1}}}}}}{B\ensuremath{{}_{\mathrm{1}}}}{B\ensuremath{{}_{\mathrm{2}}}}$.
\item Decidability of judgemental equality of types and terms: if
  $\wfTy{Γ}{A}$ and $\wfTy{Γ}{B}$, then it is decidable whether
  $\eqTy{Γ}{A}{B}$ holds, and if $\wfTm{Γ}{t}{A}$ and
  $\wfTm{Γ}{u}{A}$, then it is decidable whether $\eqTm{Γ}{t}{u}{A}$
  holds.
\item Type-checking is decidable for a fragment of the language that
  excludes, for instance, β-redexes (because lambdas and applications
  are not annotated with types).
\end{itemize}

\section{Soundness of Erasure}
\label{sec:soundness-erasure}

\citet{abel-et-al-2023} define a translation to an untyped target
language.
This translation replaces erased arguments with a dummy value, and
this dummy value is also used for some other things.
\citeauthor{abel-et-al-2023} prove that the translation is correct in
the following sense: given certain assumptions a term $t$ of type
$\Nat$ reduces to a numeral $n$, and the translated term reduces to
the same numeral.
\citet{danielsson-et-al-2026} extend the translation with support for
call-by-value in addition to call-by-name, and in the case of
call-by-name erased arguments are removed entirely instead of being
replaced by a dummy value.

Here the translation of \citeauthor{danielsson-et-al-2026} is extended
with support for identity types and quotients.
For simplicity only the call-by-value translation is presented in the
text: the call-by-name translation can be found in the formalisation.
Note that the soundness theorem is stated in a slightly different way
when the call-by-name translation is used
\citep{danielsson-et-al-2026}.

The target language does not have any term formers directly related to
identity types or quotients.
The syntax of the target language is defined in the following way:
$t, u ∷= \dummy ~|~ \varT{x} ~|~ \lamT{t} ~|$\linebreak
$\appT{t}{u} ~|~ \zeroName ~|~ \sucT{t} ~|~ \ldots$ (one of the term
formers related to $\Nat$, and all of the term formers related to Σ
and unit types, have been omitted; they can be found in
\appendixB{}\ifAppendicesIncluded{}{ in the supplementary material}).
The formalisation uses well-scoped syntax, but details related to that
are omitted here.
Variables $x$ are natural number literals (de Bruijn indices).

The target language comes with an inductively defined call-by-value
small-step semantics, $\step{t}{u}$, without evaluation under
constructors (the relation can be found in \appendixB{}).
The statement of the soundness theorem below
(Theorem~\ref{thm:soundness-of-erasure}) uses the reflexive,
transitive closure of this relation, denoted by $\stepsName$.
For a natural number $n$ the notation $\numeralAsTerm{n}$ stands for
the corresponding term, $\numeralAsTermExplanation{n}$ (in the source
language as well as the target language).

\begin{figure}[t]
  \centering
  \begin{equation*}
    \begin{pmboxed}
      \>[][@{}l@{}]\extract{\var{x}}                      \>[][@{}l@{}]= \varT{x}                                           \>[][@{}l@{\quad}]
      \>[][@{}l@{}]\extract{(\set{A}{B}{t}{u}{v}{w})}     \>[][@{}l@{}]= \loopT                                             \\
      \>           \extract{(\IdT{A}{t}{u})}              \>           = \dummy                                             \>
      \>           \extract{(\qrec{C}{t}{u}{v}{w})}       \>           = \appT{(\lamT{(\extract{t})})}{(\extract{w})}       \\
      \>           \extract{\rfl}                         \>           = \dummy                                             \>
      \>           \extract{(\PiT{p}{A}{B})}              \>           = \dummy                                             \\
      \>           \extract{(\J{p}{q}{A}{t}{B}{u}{v}{w})} \>           = \extract{u}                                        \>
      \>           \extract{(\lam{p}{t})}                 \>           = \lamT{(\extract{t})}                               \\
      \>           \extract{(\bc{l}{A}{t}{u}{v})}         \>           = \dummy                                             \>
      \>           \extract{(\app{t}{\omegaG}{u})}        \>           = \appT{(\extract{t})}{(\extract{u})}                \\
      \>           \extract{(\Quot{A}{B})}                \>           = \dummy                                             \>
      \>           \extract{(\app{t}{\zeroG}{u})}         \>           = \appT{(\extract{t})}{\dummy}                       \\
      \>           \extract{(\class{t})}                  \>           = \extract{t}                                        \>
      \>           \extract{\zeroName}                    \>           = \zeroName                                          \\
      \>           \extract{(\resp{A}{B}{t}{u}{v})}       \>           = \loopT                                             \>
      \>           \extract{(\suc{t})}                    \>           = \appT{(\lamT{(\sucT{(\varT{0})})})}{(\extract{t})}
    \end{pmboxed}
  \end{equation*}
  \caption{The extraction function.}
  \label{fig:extraction}
\end{figure}

Let me now present the extraction function.
It is defined recursively in Figure \ref{fig:extraction} (some cases
have been omitted, see \appendixB{} for more details).
All type constructors are replaced by the dummy term $\dummy$, as are
erased function arguments.
The identity type primitives $\rfl$ and $\bcName$ are also replaced by
$\dummy$: no attempt is made to compute the values of identity proofs.
The higher quotient constructors are instead replaced by $\loopT$ –
which is a non-terminating term – to illustrate that their extracted
terms will not be reduced to WHNF.
For J all arguments except for one are thrown away (and K is treated
similarly).
There is no representation of $\className$ in the target language, and
applications of $\qrecName$ are turned into β-redexes.
The application $\suc{t}$ is translated to
``$(λx.\sucT{x})(\extract{t})$'': as mentioned above the semantics
does not include evaluation under constructors, and this construction
ensures that $\extract{t}$ is reduced before the successor constructor
is applied to the term.

\citet{abel-et-al-2023} and \citet{danielsson-et-al-2026} prove
soundness of erasure using logical relations.
Adapting this kind of argument to the present setting is
straightforward.
A key step in my proof is to define the logical relation in the
following way for identity types and quotients:
\begin{itemize}
\item A source term and a target term are related at a type that
  reduces to $\IdT{A}{t}{u}$ if the source term reduces to $\rfl$ and
  the target term reduces to $\dummy$.
\item A source term $t$ and a target term $v$ are related at a type
  that reduces to $\Quot{A}{B}$ if $t$ reduces to $\class{t′}$ and
  $t′$ is related to $v$ at type $A$.
\end{itemize}
We obtain the following theorem (``$Γ$ is consistent'' means that
there is no $t$ such that $\wfTm{Γ}{t}{\Empty}$):
\begin{theorem}[Soundness of erasure]
  \label{thm:soundness-of-erasure}
  If $\wrTm{\zeroGC}{Γ}{t}{\omegaG}{\Nat}$; $Γ$ is consistent or
  erased matches are disallowed for the empty type; and either
  \begin{itemize}
  \item $Γ$ is empty and either quotient types are disallowed or
    equality reflection is allowed, or
  \item $\bcName$ is disallowed, erased matches are disallowed for J
    and K as well as weak Σ and unit types, and equality reflection is
    disallowed,
  \end{itemize}
  then there is a natural number $n$ such that
  $\redsSuc{Γ}{t}{\numeralAsTerm{n}}$ and
  $\steps{\extract{t}}{\numeralAsTerm{n}}$.
\end{theorem}
\noindent
The theorem would fail, even for closed terms, if quotient types were
allowed, equality reflection were disallowed, and (say) limited erased
matches were allowed for J, because then one could construct a term
$t$ for which source-level reduction gets stuck.
However, $\wrTm{\emptyGC}{\emptyC}{t}{\omegaG}{\Nat}$ suffices to
ensure that the extracted term does not get stuck: this follows from
Theorem~\ref{thm:soundness-reflection}.

The soundness theorem allows \emph{open} terms, as long as all free
variables are erased (have grade $\zeroG$), erased matches are
disallowed for identity types and weak Σ and unit types, $\bcName$ is
disallowed, equality reflection is disallowed, and either erased
matches are disallowed for the empty type or the context is
consistent.
The free variables can be seen as erased postulates.
One could for instance postulate function extensionality, univalence,
excluded middle and a choice principle, and use these assumptions in
erased contexts.
If one can prove that these postulates are mutually consistent (or if
erased matches for the empty type are disallowed), then one gets a
guarantee that programs of type $\Nat$ reduce to numerals, and that
the corresponding extracted terms reduce to equal numerals.
However, this guarantee does not hold if for instance $\bcName$ is
allowed: the term $\bc{l}{A}{t}{u}{\var{x}}$ is stuck, and it is
well-resourced even if $x$ has grade $\zeroG$ (and it is easy to
construct a similar example of type $\Nat$).

Before we leave this section, let me present a result about how it is
not in general possible to ``resurrect'' erased data, i.e.\ produce a
value that is equal to an erased value.
In the interest of readability the following definition uses a
notation with names instead of de Bruijn indices (such a notation is
also sometimes used below):
\begin{definition}
  The type $A$ is \emph{resurrectable} at universe level $l$ if there
  is a term $t$ such that
  $\wrTm{\emptyGC}{\emptyC}{t}{\omegaG}{\PiT{\zeroG}{(\hastype{x}{A})}{%
      (\SigmaT{\strongS}{\omegaG}{(\hastype{y}{A})}{%
        (\E{l}{(\IdT{A}{y}{x})})})}}$, where
  $\SigmaWith{\strongS}{\omegaG}$ stands for a Σ-type with η-equality
  (a strong Σ-type) and a non-erased first component.
\end{definition}
\begin{proposition} Unit types are resurrectable at level $\zeroL$,
  and if erased matches are allowed for the empty type, then it is
  also resurrectable at level $\zeroL$.
  However, the type of natural numbers is not resurrectable at any
  level $l$.
\end{proposition}
\begin{proof}[Proof sketch]
  The first two cases are easy to prove.
  For the last case one can use the fundamental lemma of the logical
  relation used to prove Theorem~\ref{thm:soundness-of-erasure}.
\end{proof}

\section{Soundness of Erasure in the Presence of Erased Postulates and Erased Matches}
\label{sec:soundness-erasure-postulates}

Theorem \ref{thm:soundness-of-erasure} does not allow erased
postulates (i.e.\ a non-empty context) and erased matches at the same
time, with the exception of erased matches for the empty type.
If any of the conditions related to erased matches were removed, then
the theorem would not hold: in those cases one can construct terms $t$
for which it is not the case that $t$ reduces to a numeral.

Would it be possible to prove the theorem without some or all of the
restrictions related to erased matches if
``$\redsSuc{Γ}{t}{\numeralAsTerm{n}}$'' in the theorem's conclusion
were replaced by something else?
One idea is to replace the reduction with the statement that there is
a term $u$ such that $\wfTm{Γ}{u}{\IdT{\Nat}{t}{\numeralAsTerm{n}}}$.
If one also included suitable extensionality assumptions, then there
is perhaps some hope that the theorem could be proved without any
restrictions on \boxcong{}.
However, I have found it hard to prove something like this.
If $t$ reduces to $u$, or if $t$ is judgementally equal to $u$, then
one can freely replace $t$ with $u$: the type $P\,t$ is
interchangeable with $P\,u$.
This does not follow if the identity type $\IdT{A}{t}{u}$ is inhabited
(in the absence of something like equality reflection): one can use
$\JName$ to replace $t$ with $u$, but a term that has type $P\,t$
might not also have type $P\,u$.

Because I find it easier to work with judgemental equality than with
identity types I have taken a different approach: I replace
$\redsSuc{Γ}{t}{\numeralAsTerm{n}}$ with a judgemental equality in an
\emph{extended} type theory, in which the postulates can be
implemented.
For that type theory it suffices to prove soundness of erasure for
closed terms.
Note that, for closed terms, Theorem \ref{thm:soundness-of-erasure}
does not impose any restrictions related to erased matches or quotient
types if equality reflection is allowed.

Let us fix one type theory from the family of type theories from
§ \ref{sec:typing}.
\begin{definition}
  An \emph{extension} of the type theory is a type of terms
  $T$\footnote{The formalisation uses well-scoped syntax. Those
    details are elided here.} closed under application of parallel
  substitutions $σ$ containing terms from $T$ (application is denoted
  by $\appSubstExt{t}{σ}$), a translation $\trName$ taking regular
  terms to terms in $T$, an extraction function $\extractExtName$ from
  $T$ to the target language from §~\ref{sec:soundness-erasure},
  a typing relation $\wrTmExt{γ}{Γ}{t}{p}{A}$ (where $t, A : T$ and
  $Γ$ uses types from $T$), a judgemental equality relation
  $\eqTmExt{Γ}{t}{u}{A}$, and a well-formed substitution relation
  $\wrSubstExt{δ}{Δ}{σ}{p}{Γ}$, satisfying the following properties:
  \begin{itemize}
  \item Translation is type-preserving: $\wrTm{γ}{Γ}{t}{p}{A}$ implies
    $\wrTmExt{γ}{\tr{Γ}}{\tr{t}}{p}{\tr{A}}$ (where $\tr{Γ}$ is the
    pointwise application of $\trName$ to the types in $Γ$).
  \item A substitution lemma: $\wrTmExt{\zeroGC}{Γ}{t}{p}{A}$ and
    $\wrSubstExt{\emptyGC}{\emptyC}{σ}{\zeroG}{Γ}$ imply
    $\wrTmExt{\emptyGC}{\emptyC}{\appSubstExt{t}{σ}}{p}{\appSubstExt{A}{σ}}$.
  \item $\wrTmExt{\zeroGC}{Γ}{t}{\omegaG}{A}$ implies that
    $\extractExt{(\appSubstExt{t}{σ})}$ is equal to $\extractExt{t}$
    for closing substitutions $σ$.
  \item The target term $\extractExt{(\tr{t})}$ is equal to
    $\extract{t}$, and $\appSubstExt{(\tr{\Nat})}{σ}$ is equal to
    $\tr{\Nat}$.
  \item Soundness of erasure for closed terms: if
    $\wrTmExt{\emptyGC}{\emptyC}{t}{\omegaG}{\tr{\Nat}}$, then there
    is a natural number $n$ such that
    $\eqTmExt{\emptyC}{t}{\tr{\numeralAsTerm{n}}}{\tr{\Nat}}$ and
    $\steps{\extractExt{t}}{\numeralAsTerm{n}}$.
  \end{itemize}
\end{definition}
\noindent
Given an extension of the type theory it is easy to prove the
following theorem:
\begin{theorem}
  \label{thm:meta}
  Let $σ$ be a closing substitution satisfying
  $\wrSubstExt{\emptyGC}{\emptyC}{σ}{\zeroG}{\tr{Γ}}$.
  If $\wrTm{\zeroGC}{Γ}{t}{\omegaG}{\Nat}$, then there is a natural
  number $n$ such that
  $\eqTmExt{\emptyC}{\appSubstExt{(\tr{t})}{σ}}{\tr{\numeralAsTerm{n}}}{\tr{\Nat}}$
  and $\steps{\extract{t}}{\numeralAsTerm{n}}$.
\end{theorem}

Let us now instantiate the theorem.
Every type theory from § \ref{sec:typing} can be extended by
turning on equality reflection and making no other changes, letting
$\trName$ be the identity function.
Thus we obtain the following theorem, where the relations marked with
R refer to the extended type theory with equality reflection:
\begin{theorem}
  \label{thm:soundness-reflection}
  Let $σ$ be a closing substitution satisfying
  $\wrSubstR{\emptyGC}{\emptyC}{σ}{\zeroG}{Γ}$.
  If $\wrTm{\zeroGC}{Γ}{t}{\omegaG}{\Nat}$, then there is a natural
  number $n$ such that
  $\eqTmR{\emptyC}{\appSubstExt{t}{σ}}{\numeralAsTerm{n}}{\Nat}$ and
  $\steps{\extract{t}}{\numeralAsTerm{n}}$.
\end{theorem}
\noindent
This theorem can for instance be used to show that a program that uses
(only) postulated erased function extensionality will not get stuck,
even if erased matches and quotient types are used.

It is my hope that (some variant of) Theorem \ref{thm:meta} can also
be instantiated with a suitable cubical type theory, thus showing that
Computational Book HoTT is well-behaved even in the presence of
\boxcong{}, as discussed in §~\ref{sec:computational-book-hott}.
However, that is left for future work.

Note that propositional extensionality cannot be implemented as a
closed term in the type theories with equality reflection presented
here: propext would allow one to construct a closed term that does not
reduce to WHNF, contradicting the weak head normalisation result from
§~\ref{sec:meta-theory}.
However, propext can be implemented in cubical type theory (it follows
from univalence).

One might think that another way to show that erased propext is safe
is to instantiate Theorem \ref{thm:meta} with some kind of
observational type theory with support for quotients and propext.
\citet{pujet-tabareau-2022} discuss such a type theory.
However, their notion of propext is restricted to \emph{strict}
propositions in a specific definitionally proof-irrelevant universe,
unlike the notion of propext discussed here, which works for all types
that are propositions semantically (as specified by
\AgdaFunction{Is‐prop}).
The semantic notion of propext is presumably\footnote{The presentation
  of one of the reduction rules (\textsc{Eq-Π}) appears to contain
  some typos \citep[Figure 4]{pujet-tabareau-2022}.} inconsistent with
the type theory of \citeauthor{pujet-tabareau-2022}: observationally
equal proof-relevant Π-types necessarily have equal domains, so I
expect that this variant of propext could be used to prove that the
unit type is observationally equal to the empty type.
One can discuss which variant is ``better'':
\citet{sterling-et-al-2022} state that `short of adding equality
reflection, we must conclude that the weak notion of proposition is
the ``correct'' one, and the strict one is not particularly useful for
mathematics in an environment without equality reflection' (``the weak
notion'' refers to the semantic one).

%
%
%
%

\section{Box-Cong Can Sometimes, but Not Always, Be Defined}
\label{sec:no-box-cong}

\begin{code}[hide]%
\>[0]\AgdaKeyword{module}\AgdaSpace{}%
\AgdaModule{Sometimes}\AgdaSpace{}%
\AgdaKeyword{where}\<%
\\
\>[0][@{}l@{\AgdaIndent{0}}]%
\>[2]\AgdaKeyword{open}\AgdaSpace{}%
\AgdaModule{Introduction}\AgdaSpace{}%
\AgdaKeyword{using}\AgdaSpace{}%
\AgdaSymbol{(}\AgdaDatatype{Erased}\AgdaSymbol{;}\AgdaSpace{}%
\AgdaOperator{\AgdaInductiveConstructor{[\AgdaUnderscore{}]}}\AgdaSymbol{;}\AgdaSpace{}%
\AgdaFunction{erased}\AgdaSymbol{;}\AgdaSpace{}%
\AgdaFunction{subst\ensuremath{{}^{\mkern2mu\mathrm{E}}}}\AgdaSymbol{)}\<%
\end{code}

Let us now investigate under what circumstances one can define
something like $\bcName$.
``Something like'' is made precise in the following way:
\begin{definition}
  \label{def:boxcong}
  \emph{\boxcong{} is supported for the context $Γ$, the level $l$,
    the mode $p$, and the grades $q\ensuremath{{}_{\mathrm{1}}}$, $q\ensuremath{{}_{\mathrm{2}}}$ and $q\ensuremath{{}_{\mathrm{3}}}$} if
  $\wrTm{\zeroGC}{Γ}{\bcVar}{p}{%
    \PiT{q\ensuremath{{}_{\mathrm{1}}}}{(\hastype{A}{\U{l}})}{%
      \PiT{q\ensuremath{{}_{\mathrm{2}}}}{(\hastype{x}{A})}{%
        \PiT{q\ensuremath{{}_{\mathrm{3}}}}{(\hastype{y}{A})}{%
          \PiT{\zeroG}{(\IdT{A}{x}{y})}{%
            (\IdT{(\E{l}{A})}{\bx{x}}{\bx{y}})}}}}}$
  holds for some term $\bcVar$.
  \emph{\boxcong{} is supported computationally for the context $Γ$,
    the level $l$, the mode $p$, and the grades $q\ensuremath{{}_{\mathrm{1}}}$, $q\ensuremath{{}_{\mathrm{2}}}$ and $q\ensuremath{{}_{\mathrm{3}}}$}
  if additionally the following equality rule is admissible (note that
  $ρ$ is a weakening):
  \begin{mathpar}
    \inferrule{\wfWk{Δ}{ρ}{Γ} \\ \wfTm{Δ}{A}{\U{l}} \\ \wfTm{Δ}{t}{A}}{%
      \eqTm{Δ}{%
        \app{\app{\app{\app{(\wk{\bcVar}{ρ})}{q\ensuremath{{}_{\mathrm{1}}}}{A}}{q\ensuremath{{}_{\mathrm{2}}}}{t}}{q\ensuremath{{}_{\mathrm{3}}}}{t}}{%
          \zeroG}{\rfl}}{%
        \rfl}{\IdT{(\E{l}{A})}{\bx{t}}{\bx{t}}}}
  \end{mathpar}
\end{definition}
\noindent

It may appear as if $\bcVar$ is more useful if $q\ensuremath{{}_{\mathrm{1}}}$, $q\ensuremath{{}_{\mathrm{2}}}$ and $q\ensuremath{{}_{\mathrm{3}}}$ are
all $\zeroG$, but such a variant of $\bcVar$ can be defined using any
other variant:
\begin{proposition}\leavevmode
  \label{prop:boxcong-weak-boxcong}
  \boxcong{} is supported (respectively supported computationally) for
  the context $Γ$, the level $l$, the mode $p$, and the grades
  $\zeroG$, $\zeroG$ and $\zeroG$ if and only if \boxcong{} is
  supported (respectively supported computationally) for the context
  $Γ$, the level $l$, the mode $p$, and the grades $q\ensuremath{{}_{\mathrm{1}}}$, $q\ensuremath{{}_{\mathrm{2}}}$ and
  $q\ensuremath{{}_{\mathrm{3}}}$.
\end{proposition}
\begin{proof}[Proof sketch]
  For the left-to-right direction this should be easy to see.
  For the other direction one can define \boxcong{} with erased
  arguments in terms of the other variant of \boxcong{} using code
  that is similar to the following Agda code:
  \begin{code}[hide]%
\>[2]\AgdaKeyword{module}\AgdaSpace{}%
\AgdaModule{\AgdaUnderscore{}}\<%
\\
\>[2][@{}l@{\AgdaIndent{0}}]%
\>[4]\AgdaSymbol{(}\AgdaBound{[\ensuremath{\mkern1.5mu}]‐cong}\AgdaSpace{}%
\AgdaSymbol{:}\AgdaSpace{}%
\AgdaSymbol{∀}\AgdaSpace{}%
\AgdaSymbol{\{}\AgdaBound{a}\AgdaSymbol{\}}\AgdaSpace{}%
\AgdaSymbol{(}\AgdaBound{A}\AgdaSpace{}%
\AgdaSymbol{:}\AgdaSpace{}%
\AgdaPrimitive{Type}\AgdaSpace{}%
\AgdaBound{a}\AgdaSymbol{)}\AgdaSpace{}%
\AgdaSymbol{(}\AgdaBound{x}\AgdaSpace{}%
\AgdaBound{y}\AgdaSpace{}%
\AgdaSymbol{:}\AgdaSpace{}%
\AgdaBound{A}\AgdaSymbol{)}\AgdaSpace{}%
\AgdaSymbol{→}\AgdaSpace{}%
\AgdaSymbol{@0}\AgdaSpace{}%
\AgdaDatatype{Id}\AgdaSpace{}%
\AgdaBound{A}\AgdaSpace{}%
\AgdaBound{x}\AgdaSpace{}%
\AgdaBound{y}\AgdaSpace{}%
\AgdaSymbol{→}\AgdaSpace{}%
\AgdaDatatype{Id}\AgdaSpace{}%
\AgdaSymbol{(}\AgdaDatatype{Erased}\AgdaSpace{}%
\AgdaBound{A}\AgdaSymbol{)}\AgdaSpace{}%
\AgdaOperator{\AgdaInductiveConstructor{[}}\AgdaSpace{}%
\AgdaBound{x}\AgdaSpace{}%
\AgdaOperator{\AgdaInductiveConstructor{]}}\AgdaSpace{}%
\AgdaOperator{\AgdaInductiveConstructor{[}}\AgdaSpace{}%
\AgdaBound{y}\AgdaSpace{}%
\AgdaOperator{\AgdaInductiveConstructor{]}}\AgdaSymbol{)}\<%
\\
\>[4]\AgdaKeyword{where}\<%
\end{code}
  \begin{code}%
\>[4]\AgdaFunction{[\ensuremath{\mkern1.5mu}]‐cong′}\AgdaSpace{}%
\AgdaSymbol{:}\AgdaSpace{}%
\AgdaSymbol{(}\AgdaSymbol{@}\AgdaNumber{0}\AgdaSpace{}%
\AgdaBound{A}\AgdaSpace{}%
\AgdaSymbol{:}\AgdaSpace{}%
\AgdaPrimitive{Type}\AgdaSpace{}%
\AgdaGeneralizable{a}\AgdaSymbol{)}\AgdaSpace{}%
\AgdaSymbol{(}\AgdaSymbol{@}\AgdaNumber{0}\AgdaSpace{}%
\AgdaBound{x}\AgdaSpace{}%
\AgdaBound{y}\AgdaSpace{}%
\AgdaSymbol{:}\AgdaSpace{}%
\AgdaBound{A}\AgdaSymbol{)}\AgdaSpace{}%
\AgdaSymbol{→}\AgdaSpace{}%
\AgdaSymbol{@}\AgdaNumber{0}\AgdaSpace{}%
\AgdaDatatype{Id}\AgdaSpace{}%
\AgdaBound{A}\AgdaSpace{}%
\AgdaBound{x}\AgdaSpace{}%
\AgdaBound{y}\AgdaSpace{}%
\AgdaSymbol{→}\AgdaSpace{}%
\AgdaDatatype{Id}\AgdaSpace{}%
\AgdaSymbol{(}\AgdaDatatype{Erased}\AgdaSpace{}%
\AgdaBound{A}\AgdaSymbol{)}\AgdaSpace{}%
\AgdaOperator{\AgdaInductiveConstructor{[}}\AgdaSpace{}%
\AgdaBound{x}\AgdaSpace{}%
\AgdaOperator{\AgdaInductiveConstructor{]}}\AgdaSpace{}%
\AgdaOperator{\AgdaInductiveConstructor{[}}\AgdaSpace{}%
\AgdaBound{y}\AgdaSpace{}%
\AgdaOperator{\AgdaInductiveConstructor{]}}\<%
\\
\>[4]\AgdaFunction{[\ensuremath{\mkern1.5mu}]‐cong′}\AgdaSpace{}%
\AgdaBound{A}\AgdaSpace{}%
\AgdaBound{x}\AgdaSpace{}%
\AgdaBound{y}\AgdaSpace{}%
\AgdaBound{eq}\AgdaSpace{}%
\AgdaSymbol{=}\<%
\\
\>[4][@{}l@{\AgdaIndent{0}}]%
\>[6]\AgdaFunction{cong}\AgdaSpace{}%
\AgdaSymbol{(λ}\AgdaSpace{}%
\AgdaBound{x}\AgdaSpace{}%
\AgdaSymbol{→}\AgdaSpace{}%
\AgdaOperator{\AgdaInductiveConstructor{[}}\AgdaSpace{}%
\AgdaFunction{erased}\AgdaSpace{}%
\AgdaSymbol{(}\AgdaFunction{erased}\AgdaSpace{}%
\AgdaBound{x}\AgdaSymbol{)}\AgdaSpace{}%
\AgdaOperator{\AgdaInductiveConstructor{]}}\AgdaSymbol{)}\AgdaSpace{}%
\AgdaSymbol{(}\AgdaBound{[\ensuremath{\mkern1.5mu}]‐cong}\AgdaSpace{}%
\AgdaSymbol{(}\AgdaDatatype{Erased}\AgdaSpace{}%
\AgdaBound{A}\AgdaSymbol{)}\AgdaSpace{}%
\AgdaOperator{\AgdaInductiveConstructor{[}}\AgdaSpace{}%
\AgdaBound{x}\AgdaSpace{}%
\AgdaOperator{\AgdaInductiveConstructor{]}}\AgdaSpace{}%
\AgdaOperator{\AgdaInductiveConstructor{[}}\AgdaSpace{}%
\AgdaBound{y}\AgdaSpace{}%
\AgdaOperator{\AgdaInductiveConstructor{]}}\AgdaSpace{}%
\AgdaSymbol{(}\AgdaFunction{cong}\AgdaSpace{}%
\AgdaSymbol{(λ}\AgdaSpace{}%
\AgdaBound{x}\AgdaSpace{}%
\AgdaSymbol{→}\AgdaSpace{}%
\AgdaOperator{\AgdaInductiveConstructor{[}}\AgdaSpace{}%
\AgdaBound{x}\AgdaSpace{}%
\AgdaOperator{\AgdaInductiveConstructor{]}}\AgdaSymbol{)}\AgdaSpace{}%
\AgdaBound{eq}\AgdaSymbol{))}\<%
\end{code}
  This code makes use of the fact that the type function $\Erased$
  takes an erased argument.
  It also uses the function \congType{}.
\end{proof}
\noindent
Below ``\boxcong{} is supported (computationally) for $Γ$, $l$ and
$p$'' stands for ``\boxcong{} is supported (computationally) for the
context $Γ$, the level $l$, the mode $p$, and the grades $\zeroG$,
$\zeroG$ and $\zeroG$''.

In some cases \boxcong{} is supported (see also
§~\ref{sec:box-cong-from-funext}):
\begin{proposition}
  \label{prop:boxcong-supported}
  \boxcong{} is supported computationally for $Γ$, $l$ and $p$ if\/
  $\wfCtxt{Γ}$ and one of the following four conditions hold: the
  $\bcName$ primitive is allowed; $p$ is $\zeroG$; equality reflection
  is allowed; or erased matches (limited or unrestricted) are allowed
  for J.
\end{proposition}
\begin{proof}[Proof sketch]
  If the $\bcName$ primitive is allowed, then it is easy to show that
  \boxcong{} is supported computationally.
  If $p$ is $\zeroG$, then it is also easy, because erased variables
  can be used freely in erased contexts.
  If equality reflection is allowed, then one can let $\bcVar$ be
  $\lam{\zeroG}{\lam{\zeroG}{\lam{\zeroG}{\lam{\zeroG}{\rfl}}}}$.
  Finally, if erased matches are allowed for J, then one can define
  $\bcVar$ in roughly the following way:
  \begin{code}[hide]%
\>[2]\AgdaKeyword{postulate}\<%
\\
\>[2][@{}l@{\AgdaIndent{0}}]%
\>[4]\AgdaPostulate{subst\ensuremath{{}^{\mkern2mu\mathrm{E}'}}}\AgdaSpace{}%
\AgdaSymbol{:}%
\>[2174I]\AgdaSymbol{\{}\AgdaSymbol{@0}\AgdaSpace{}%
\AgdaBound{A}\AgdaSpace{}%
\AgdaSymbol{:}\AgdaSpace{}%
\AgdaPrimitive{Type}\AgdaSpace{}%
\AgdaGeneralizable{a}\AgdaSymbol{\}}\AgdaSpace{}%
\AgdaSymbol{\{}\AgdaSymbol{@0}\AgdaSpace{}%
\AgdaBound{x}\AgdaSpace{}%
\AgdaBound{y}\AgdaSpace{}%
\AgdaSymbol{:}\AgdaSpace{}%
\AgdaBound{A}\AgdaSymbol{\}}\<%
\\
\>[.][@{}l@{}]\<[2174I]%
\>[14]\AgdaSymbol{(}\AgdaBound{P}\AgdaSpace{}%
\AgdaSymbol{:}\AgdaSpace{}%
\AgdaSymbol{@0}\AgdaSpace{}%
\AgdaBound{A}\AgdaSpace{}%
\AgdaSymbol{→}\AgdaSpace{}%
\AgdaPrimitive{Type}\AgdaSpace{}%
\AgdaGeneralizable{p}\AgdaSymbol{)}\AgdaSpace{}%
\AgdaSymbol{→}\AgdaSpace{}%
\AgdaSymbol{@0}\AgdaSpace{}%
\AgdaDatatype{Id}\AgdaSpace{}%
\AgdaBound{A}\AgdaSpace{}%
\AgdaBound{x}\AgdaSpace{}%
\AgdaBound{y}\AgdaSpace{}%
\AgdaSymbol{→}\AgdaSpace{}%
\AgdaBound{P}\AgdaSpace{}%
\AgdaBound{x}\AgdaSpace{}%
\AgdaSymbol{→}\AgdaSpace{}%
\AgdaBound{P}\AgdaSpace{}%
\AgdaBound{y}\<%
\end{code}
  \begin{code}%
\>[2]\AgdaFunction{[\ensuremath{\mkern1.5mu}]‐cong}\AgdaSpace{}%
\AgdaSymbol{:}\AgdaSpace{}%
\AgdaSymbol{(}\AgdaSymbol{@}\AgdaNumber{0}\AgdaSpace{}%
\AgdaBound{A}\AgdaSpace{}%
\AgdaSymbol{:}\AgdaSpace{}%
\AgdaPrimitive{Type}\AgdaSpace{}%
\AgdaGeneralizable{a}\AgdaSymbol{)}\AgdaSpace{}%
\AgdaSymbol{(}\AgdaSymbol{@}\AgdaNumber{0}\AgdaSpace{}%
\AgdaBound{x}\AgdaSpace{}%
\AgdaBound{y}\AgdaSpace{}%
\AgdaSymbol{:}\AgdaSpace{}%
\AgdaBound{A}\AgdaSymbol{)}\AgdaSpace{}%
\AgdaSymbol{→}\AgdaSpace{}%
\AgdaSymbol{@}\AgdaNumber{0}\AgdaSpace{}%
\AgdaDatatype{Id}\AgdaSpace{}%
\AgdaBound{A}\AgdaSpace{}%
\AgdaBound{x}\AgdaSpace{}%
\AgdaBound{y}\AgdaSpace{}%
\AgdaSymbol{→}\AgdaSpace{}%
\AgdaDatatype{Id}\AgdaSpace{}%
\AgdaSymbol{(}\AgdaDatatype{Erased}\AgdaSpace{}%
\AgdaBound{A}\AgdaSymbol{)}\AgdaSpace{}%
\AgdaOperator{\AgdaInductiveConstructor{[}}\AgdaSpace{}%
\AgdaBound{x}\AgdaSpace{}%
\AgdaOperator{\AgdaInductiveConstructor{]}}\AgdaSpace{}%
\AgdaOperator{\AgdaInductiveConstructor{[}}\AgdaSpace{}%
\AgdaBound{y}\AgdaSpace{}%
\AgdaOperator{\AgdaInductiveConstructor{]}}\<%
\\
\>[2]\AgdaFunction{[\ensuremath{\mkern1.5mu}]‐cong}\AgdaSpace{}%
\AgdaBound{A}\AgdaSpace{}%
\AgdaBound{x}\AgdaSpace{}%
\AgdaBound{y}\AgdaSpace{}%
\AgdaBound{eq}\AgdaSpace{}%
\AgdaSymbol{=}\AgdaSpace{}%
\AgdaPostulate{subst\ensuremath{{}^{\mkern2mu\mathrm{E}'}}}\AgdaSpace{}%
\AgdaSymbol{(λ}\AgdaSpace{}%
\AgdaBound{y}\AgdaSpace{}%
\AgdaSymbol{→}\AgdaSpace{}%
\AgdaDatatype{Id}\AgdaSpace{}%
\AgdaSymbol{(}\AgdaDatatype{Erased}\AgdaSpace{}%
\AgdaBound{A}\AgdaSymbol{)}\AgdaSpace{}%
\AgdaOperator{\AgdaInductiveConstructor{[}}\AgdaSpace{}%
\AgdaBound{x}\AgdaSpace{}%
\AgdaOperator{\AgdaInductiveConstructor{]}}\AgdaSpace{}%
\AgdaOperator{\AgdaInductiveConstructor{[}}\AgdaSpace{}%
\AgdaBound{y}\AgdaSpace{}%
\AgdaOperator{\AgdaInductiveConstructor{]}}\AgdaSymbol{)}\AgdaSpace{}%
\AgdaBound{eq}\AgdaSpace{}%
\AgdaInductiveConstructor{refl}\<%
\end{code}
  Note the use of \AgdaFunction{subst\ensuremath{{}^{\mkern2mu\mathrm{E}'}}}, a function with the same
  type as the function \AgdaFunction{subst\ensuremath{{}^{\mkern2mu\mathrm{E}}}} from
  § \ref{sec:box-cong}.
  This variant can be defined using J if (limited or unrestricted)
  erased matches are allowed for J.
\end{proof}

Above it is shown that \boxcong{} is supported if erased matches are
allowed for J.
In the presence of the $\bcName$ primitive one can also define a term
former $\JEName$ that is a drop-in replacement for
$\JWithGrades{0}{0}$ with (only) limited erased matches; the
definition is similar to the implementation of \AgdaFunction{J\ensuremath{{}^{\mkern2mu\mathrm{E}}}} in
§~\ref{sec:fording}:
\begin{proposition}
  \label{prop:JE-supported}
  If the $\bcName$ primitive is allowed, then there is a term former
  $\JEName$ for which the following rule is admissible:
  \begin{mathpar}
    \inferrule{%
      \wrTy{\consGC{\consGC{γ}{\zeroG}}{\zeroG}}{%
        \consC{\consC{Γ}{A}}{\IdT{(\wkO{A})}{(\wkO{t})}{\varZ}}}{p}{B} \\
      \wrTm{γ}{Γ}{u}{p}{\substO{B}{t}{\rfl}} \\
      \wfTm{Γ}{w}{\IdT{A}{t}{v}}}{%
      \wrTm{γ}{Γ}{\JE{A}{t}{B}{u}{v}{w}}{p}{\substO{B}{v}{w}}}
  \end{mathpar}
  This term former also satisfies the substitution equation, the
  congruence rule, and the β rule of $\JWithGrades{0}{0}$ (those
  equations/rules were omitted above).
\end{proposition}

Using ideas from Propositions \ref{prop:boxcong-supported}
and \ref{prop:JE-supported} one can show that limited erased matches
for $\JName$ have the same expressive power as $\bcName$ by
constructing two type-preserving translations: one that replaces
$\bcName$ with a term former implemented using $\JName$, and one that
replaces $\JWithGrades{\zeroG}{\zeroG}$ with a term former implemented
using $\bcName$.
(The translations do not affect extraction.)

Even if something like \boxcong{} can sometimes be defined, this is
not always the case.
Before we prove this, let me state a lemma that can be proved by
induction.
Let ``erased matches are disallowed for $Γ$'' mean that the $\bcName$
primitive is disallowed, erased matches are disallowed for J and K as
well as weak Σ and unit types, and either erased matches are
disallowed for the empty type or\/ $Γ$ is consistent:
\begin{lemma}[No neutral terms]
  \label{lem:no-neutral-terms}
  If erased matches are disallowed for $Γ$ and $t$ is neutral, then it
  is not the case that $\wrTm{\zeroGC}{Γ}{t}{\omegaG}{A}$.
\end{lemma}

\begin{theorem}[\boxcong{} cannot be defined]
  \label{thm:no-box-cong}
  If erased matches are disallowed for $Γ$ and equality reflection is
  disallowed, then \boxcong{} is not supported for\/ $Γ$, $l$ and
  $\omegaG$, and there is no term $\bcVarZero$ such that
  $\wrTm{\zeroGC}{Γ}{\bcVarZero}{\omegaG}{%
    \PiT{\zeroG}{(\hastype{n}{\Nat})}{%
      \PiT{\zeroG}{(\IdT{\Nat}{\zero}{n})}{%
        (\IdT{(\E{\zeroL}{\Nat})}{\bx{\zero}}{\bx{n}})}}}$.
\end{theorem}
\noindent
Note that the empty context is consistent.
I do not see why \boxcong{} would be supported if erased matches were
allowed for weak Σ or unit types, but my proof makes use of
Lemma~\ref{lem:no-neutral-terms}.
The second part of the conclusion establishes that functions like
\AgdaFunction{[\ensuremath{\mkern1.5mu}]‐cong\ensuremath{{}_{\mathrm{0}}}} from §~\ref{sec:fording} cannot be
implemented.
\begin{proof}[Proof sketch]
  Let us focus on the first part of the conclusion: the second part
  can be proved in a similar way.
  If \boxcong{} is supported for $Γ$, $l$ and $\omegaG$, then there is
  some term $\bcVar$ such that\linebreak
  $\wrTm{\zeroGC}{Γ}{\bcVar}{\omegaG}{%
    \PiT{\zeroG}{(\hastype{A}{\U{l}})}{%
      \PiT{\zeroG}{(\hastype{x}{A})}{%
        \PiT{\zeroG}{(\hastype{y}{A})}{%
          \PiT{\zeroG}{(\IdT{A}{x}{y})}{%
            (\IdT{(\E{l}{A})}{\bx{x}}{\bx{y}})}}}}}$.
  We get that
  \mbox{$\wrTm{\zeroGC}{Δ}{t}{\omegaG}{\IdT{(\E{l}{A})}{\bx{x}}{\bx{y}}}$}
  holds for
  $Δ =
  \consC{\consC{\consC{\consC{Γ}{\hastype{A}{\U{l}}}}{\hastype{x}{A}}}{%
      \hastype{y}{A}}}{\hastype{\identifier{eq}}{\IdT{A}{x}{y}}}$
  and
  $t =
  \app{\app{\app{\app{(\wk{\bcVar}{ρ})}{\zeroG}{A}}{\zeroG}{x}}{%
      \zeroG}{y}}{%
    \zeroG}{\identifier{eq}}$
  (for a suitable weakening $ρ$).
  Equality reflection is disallowed, so the term $t$ must have a WHNF,
  and that WHNF must be either $\rfl$ or a neutral term (see
  §~\ref{sec:meta-theory}):
  \begin{itemize}
  \item If it is $\rfl$, then we must have
    $\eqTm{Δ}{\bx{x}}{\bx{y}}{\E{l}{A}}$, and thus
    $\eqTm{Δ}{x}{y}{A}$, but two distinct variables are not
    judgementally equal at a neutral type (this can be proved using
    the formalisation's algorithmic equality judgement
    \citep{abel-et-al-2017}).
  \item The WHNF cannot be neutral by
    Lemma~\ref{lem:no-neutral-terms}: if $Γ$ is consistent, then $Δ$
    is consistent because there is a substitution $σ$ satisfying
    $\wfSubst{Γ}{σ}{Δ}$.\qedhere
  \end{itemize}
\end{proof}

\boxcong{} can be defined using function extensionality (see
§~\ref{sec:box-cong-from-funext}).
However, \boxcong{} cannot be defined using (only) postulated
\emph{erased} function extensionality: a context with a single
variable with the type of function extensionality is consistent, so
the theorem above applies.

\section{Modalities}
\label{sec:modalities}

\begin{code}[hide]%
\>[2]\AgdaKeyword{open}\AgdaSpace{}%
\AgdaModule{Introduction}\AgdaSpace{}%
\AgdaKeyword{using}\AgdaSpace{}%
\AgdaSymbol{(}\AgdaFunction{Is‐prop}\AgdaSymbol{)}\<%
\\
\>[2]\AgdaKeyword{open}\AgdaSpace{}%
\AgdaModule{Eliminators}\AgdaSpace{}%
\AgdaKeyword{using}\AgdaSpace{}%
\AgdaSymbol{(}\AgdaOperator{\AgdaFunction{\AgdaUnderscore{}≡\AgdaUnderscore{}}}\AgdaSymbol{)}\<%
\\
\\[\AgdaEmptyExtraSkip]%
\>[2]\AgdaKeyword{postulate}\<%
\\
\>[2][@{}l@{\AgdaIndent{0}}]%
\>[4]\AgdaPostulate{Funext}%
\>[22]\AgdaSymbol{:}\AgdaSpace{}%
\AgdaSymbol{(}\AgdaBound{a}\AgdaSpace{}%
\AgdaBound{p}\AgdaSpace{}%
\AgdaSymbol{:}\AgdaSpace{}%
\AgdaPostulate{Level}\AgdaSymbol{)}\AgdaSpace{}%
\AgdaSymbol{→}\AgdaSpace{}%
\AgdaPrimitive{Type}\AgdaSpace{}%
\AgdaSymbol{(}\AgdaPrimitive{lsuc}\AgdaSpace{}%
\AgdaSymbol{(}\AgdaBound{a}\AgdaSpace{}%
\AgdaOperator{\AgdaPrimitive{⊔}}\AgdaSpace{}%
\AgdaBound{p}\AgdaSymbol{))}\<%
\\
\>[4]\AgdaOperator{\AgdaPostulate{\AgdaUnderscore{}≃\AgdaUnderscore{}}}\AgdaSpace{}%
\AgdaOperator{\AgdaPostulate{\AgdaUnderscore{}↠\AgdaUnderscore{}}}\AgdaSpace{}%
\AgdaPostulate{Embedding}\AgdaSpace{}%
\AgdaSymbol{:}\AgdaSpace{}%
\AgdaPrimitive{Type}\AgdaSpace{}%
\AgdaGeneralizable{a}\AgdaSpace{}%
\AgdaSymbol{→}\AgdaSpace{}%
\AgdaPrimitive{Type}\AgdaSpace{}%
\AgdaGeneralizable{b}\AgdaSpace{}%
\AgdaSymbol{→}\AgdaSpace{}%
\AgdaPrimitive{Type}\AgdaSpace{}%
\AgdaSymbol{(}\AgdaGeneralizable{a}\AgdaSpace{}%
\AgdaOperator{\AgdaPrimitive{⊔}}\AgdaSpace{}%
\AgdaGeneralizable{b}\AgdaSymbol{)}\<%
\end{code}

Let us now discuss idempotent, monadic \emph{modalities} in the style
of \citet{rijke-et-al-2020}.
This discussion is included because
\begin{code}[hide]%
\>[2]\AgdaKeyword{postulate}\<%
\\
\>[2][@{}l@{\AgdaIndent{0}}]%
\>[4]\AgdaPostulate{\AgdaUnderscore{}}%
\>[2278I]\AgdaSymbol{:}\AgdaSpace{}%
\AgdaOperator{\AgdaGeneralizable{Has‐type[}}\AgdaSpace{}%
\AgdaSymbol{(∀}\AgdaSpace{}%
\AgdaSymbol{(}\AgdaSymbol{@ω}\AgdaSpace{}%
\AgdaBound{\AgdaUnderscore{}}\AgdaSymbol{)}\AgdaSpace{}%
\AgdaSymbol{→}\AgdaSpace{}%
\AgdaSymbol{\AgdaUnderscore{})}\AgdaSpace{}%
\AgdaOperator{\AgdaGeneralizable{]}}\<%
\end{code}
\begin{code}[inline*]%
\>[.][@{}l@{}]\<[2278I]%
\>[6]\AgdaSymbol{λ}\AgdaSpace{}%
\AgdaBound{A}\AgdaSpace{}%
\AgdaSymbol{→}\AgdaSpace{}%
\AgdaDatatype{Erased}\AgdaSpace{}%
\AgdaBound{A}\<%
\end{code}
and
\begin{code}[hide]%
\>[4]\AgdaPostulate{\AgdaUnderscore{}}%
\>[2290I]\AgdaSymbol{:}\AgdaSpace{}%
\AgdaOperator{\AgdaGeneralizable{Has‐type[}}\AgdaSpace{}%
\AgdaSymbol{((}\AgdaSymbol{@ω}\AgdaSpace{}%
\AgdaBound{\AgdaUnderscore{}}\AgdaSpace{}%
\AgdaSymbol{:}\AgdaSpace{}%
\AgdaGeneralizable{A}\AgdaSymbol{)}\AgdaSpace{}%
\AgdaSymbol{→}\AgdaSpace{}%
\AgdaSymbol{\AgdaUnderscore{})}\AgdaSpace{}%
\AgdaOperator{\AgdaGeneralizable{]}}\<%
\end{code}
\begin{code}[inline*]%
\>[.][@{}l@{}]\<[2290I]%
\>[6]\AgdaSymbol{λ}\AgdaSpace{}%
\AgdaBound{x}\AgdaSpace{}%
\AgdaSymbol{→}\AgdaSpace{}%
\AgdaOperator{\AgdaInductiveConstructor{[}}\AgdaSpace{}%
\AgdaBound{x}\AgdaSpace{}%
\AgdaOperator{\AgdaInductiveConstructor{]}}\<%
\end{code}
form a modality if and only if \boxcong{} can be defined,\footnote{The
  functions are η-expanded to give them suitable types:
  \begin{code}[hide]%
\>[4]\AgdaPostulate{\AgdaUnderscore{}}%
\>[2304I]\AgdaSymbol{:}\<%
\end{code}
  \begin{code}[inline*]%
\>[.][@{}l@{}]\<[2304I]%
\>[6]\AgdaSymbol{@0}\AgdaSpace{}%
\AgdaPrimitive{Type}\AgdaSpace{}%
\AgdaGeneralizable{a}\AgdaSpace{}%
\AgdaSymbol{→}\AgdaSpace{}%
\AgdaPrimitive{Type}\AgdaSpace{}%
\AgdaGeneralizable{a}\<%
\end{code}
  is not the same type as
  \begin{code}[hide]%
\>[4]\AgdaPostulate{\AgdaUnderscore{}}%
\>[2310I]\AgdaSymbol{:}\<%
\end{code}
  \begin{code}[inline]%
\>[.][@{}l@{}]\<[2310I]%
\>[6]\AgdaPrimitive{Type}\AgdaSpace{}%
\AgdaGeneralizable{a}\AgdaSpace{}%
\AgdaSymbol{→}\AgdaSpace{}%
\AgdaPrimitive{Type}\AgdaSpace{}%
\AgdaGeneralizable{a}\<%
\end{code}.}
and furthermore \citet{rijke-et-al-2020} present a number of
properties that hold for all modalities: thus, if \boxcong{} can be
defined, then these properties hold for \Erased{} as well (with the
caveat that \citeauthor{rijke-et-al-2020} make use of function
extensionality, univalence and higher inductive types).

In this and the following sections I switch from the type theory of
the formalisation to Agda (with the K rule turned off).
I also mostly write
\begin{code}[hide]%
\>[2]\AgdaKeyword{postulate}\<%
\\
\>[2][@{}l@{\AgdaIndent{0}}]%
\>[4]\AgdaPostulate{\AgdaUnderscore{}}%
\>[2315I]\AgdaSymbol{:}\<%
\end{code}
\begin{code}[inline*]%
\>[.][@{}l@{}]\<[2315I]%
\>[6]\AgdaGeneralizable{x}\AgdaSpace{}%
\AgdaOperator{\AgdaFunction{≡}}\AgdaSpace{}%
\AgdaGeneralizable{y}\<%
\end{code}
instead of
\begin{code}[hide]%
\>[2]\AgdaKeyword{postulate}\<%
\\
\>[2][@{}l@{\AgdaIndent{0}}]%
\>[4]\AgdaPostulate{\AgdaUnderscore{}}%
\>[2318I]\AgdaSymbol{:}\<%
\end{code}
\begin{code}[inline]%
\>[.][@{}l@{}]\<[2318I]%
\>[6]\AgdaDatatype{Id}\AgdaSpace{}%
\AgdaGeneralizable{A}\AgdaSpace{}%
\AgdaGeneralizable{x}\AgdaSpace{}%
\AgdaGeneralizable{y}\<%
\end{code}.
Some definitions from homotopy type theory \citep{hott-2013} are used.
A type is \emph{contractible} if there is one element that is equal to
every other:
\begin{code}[hide]%
\>[2]\AgdaFunction{Contractible}\AgdaSpace{}%
\AgdaSymbol{:}\AgdaSpace{}%
\AgdaPrimitive{Type}\AgdaSpace{}%
\AgdaGeneralizable{a}\AgdaSpace{}%
\AgdaSymbol{→}\AgdaSpace{}%
\AgdaPrimitive{Type}\AgdaSpace{}%
\AgdaGeneralizable{a}\<%
\end{code}
\begin{code}[inline]%
\>[2]\AgdaFunction{Contractible}\AgdaSpace{}%
\AgdaBound{A}\AgdaSpace{}%
\AgdaSymbol{=}\AgdaSpace{}%
\AgdaSymbol{(}%
\AgdaBound{x}\AgdaSpace{}%
\AgdaSymbol{:}\AgdaSpace{}%
\AgdaBound{A}\AgdaSymbol{)}\AgdaSpace{}%
\AgdaFunction{×}\AgdaSpace{}%
\AgdaSymbol{((}\AgdaBound{y}\AgdaSpace{}%
\AgdaSymbol{:}\AgdaSpace{}%
\AgdaBound{A}\AgdaSymbol{)}\AgdaSpace{}%
\AgdaSymbol{→}\AgdaSpace{}%
\AgdaBound{x}\AgdaSpace{}%
\AgdaOperator{\AgdaFunction{≡}}\AgdaSpace{}%
\AgdaBound{y}\AgdaSymbol{)}\<%
\end{code}.
A function
\begin{code}[hide]%
\>[2]\AgdaKeyword{postulate}\<%
\\
\>[2][@{}l@{\AgdaIndent{0}}]%
\>[4]\AgdaPostulate{\AgdaUnderscore{}}%
\>[2343I]\AgdaSymbol{:}%
\>[2344I]\AgdaSymbol{(}\<%
\end{code}
\mbox{\begin{code}[inline]%
\>[2344I][@{}l@{\AgdaIndent{1}}]%
\>[10]\AgdaBound{f\,}\AgdaSpace{}%
\AgdaSymbol{:}\AgdaSpace{}%
\AgdaGeneralizable{A}\AgdaSpace{}%
\AgdaSymbol{→}\AgdaSpace{}%
\AgdaGeneralizable{B}\<%
\end{code}}
is an \emph{equivalence} if and only if it is bijective, and the type
\begin{code}[hide]%
\>[.][@{}l@{}]\<[2344I]%
\>[8]\AgdaSymbol{)}\AgdaSpace{}%
\AgdaSymbol{→}\<%
\end{code}
\begin{code}[inline*]%
\>[.][@{}l@{}]\<[2343I]%
\>[6]\AgdaFunction{Is‐equivalence}\AgdaSpace{}%
\AgdaBound{f\,}\<%
\end{code}
is – given funext – necessarily a proposition.
If there is an equivalence from \AgdaBound{A} to \AgdaBound{B}, then
the types are \emph{equivalent}, written
\begin{code}[hide]%
\>[2]\AgdaKeyword{postulate}\<%
\\
\>[2][@{}l@{\AgdaIndent{0}}]%
\>[4]\AgdaPostulate{\AgdaUnderscore{}}%
\>[2351I]\AgdaSymbol{:}\<%
\end{code}
\begin{code}[inline]%
\>[2351I][@{}l@{\AgdaIndent{1}}]%
\>[8]\AgdaGeneralizable{A}\AgdaSpace{}%
\AgdaOperator{\AgdaPostulate{≃}}\AgdaSpace{}%
\AgdaGeneralizable{B}\<%
\end{code}.
\begin{code}[hide]%
\>[2]\AgdaKeyword{postulate}\<%
\\
\>[2][@{}l@{\AgdaIndent{0}}]%
\>[4]\AgdaPostulate{\AgdaUnderscore{}}%
\>[2354I]\AgdaSymbol{:}\<%
\end{code}
\begin{code}[inline*]%
\>[2354I][@{}l@{\AgdaIndent{1}}]%
\>[8]\AgdaPostulate{Funext}\AgdaSpace{}%
\AgdaGeneralizable{a}\AgdaSpace{}%
\AgdaGeneralizable{p}\<%
\end{code}
is funext for functions of type
\begin{code}[hide]%
\>[2]\AgdaKeyword{postulate}\<%
\\
\>[2][@{}l@{\AgdaIndent{0}}]%
\>[4]\AgdaPostulate{\AgdaUnderscore{}}%
\>[2357I]\AgdaSymbol{:}\<%
\end{code}
\begin{code}[inline]%
\>[2357I][@{}l@{\AgdaIndent{1}}]%
\>[8]\AgdaSymbol{(}\AgdaBound{x}\AgdaSpace{}%
\AgdaSymbol{:}\AgdaSpace{}%
\AgdaGeneralizable{A}\AgdaSymbol{)}\AgdaSpace{}%
\AgdaSymbol{→}\AgdaSpace{}%
\AgdaGeneralizable{P}\AgdaSpace{}%
\AgdaBound{x}\<%
\end{code},
where
\begin{code}[hide]%
\>[2]\AgdaKeyword{postulate}\<%
\\
\>[2][@{}l@{\AgdaIndent{0}}]%
\>[4]\AgdaPostulate{\AgdaUnderscore{}}\AgdaSpace{}%
\AgdaSymbol{:}%
\>[2364I]\AgdaSymbol{(}\<%
\end{code}
\begin{code}[inline*]%
\>[.][@{}l@{}]\<[2364I]%
\>[8]\AgdaBound{A}\AgdaSpace{}%
\AgdaSymbol{:}\AgdaSpace{}%
\AgdaPrimitive{Type}\AgdaSpace{}%
\AgdaGeneralizable{a}\<%
\end{code}
and
\begin{code}[hide]%
\>[8]\AgdaSymbol{)}\AgdaSpace{}%
\AgdaSymbol{(}\<%
\end{code}
\begin{code}[inline]%
\>[8]\AgdaBound{P}\AgdaSpace{}%
\AgdaSymbol{:}\AgdaSpace{}%
\AgdaBound{A}\AgdaSpace{}%
\AgdaSymbol{→}\AgdaSpace{}%
\AgdaPrimitive{Type}\AgdaSpace{}%
\AgdaGeneralizable{p}\<%
\end{code},
\begin{code}[hide]%
\>[8]\AgdaSymbol{)}\AgdaSpace{}%
\AgdaSymbol{→}\AgdaSpace{}%
\AgdaPrimitive{Type}\<%
\end{code}
expressed as ``certain functions of type
\begin{code}[hide]%
\>[2]\AgdaKeyword{postulate}\<%
\\
\>[2][@{}l@{\AgdaIndent{0}}]%
\>[4]\AgdaPostulate{\AgdaUnderscore{}}\AgdaSpace{}%
\AgdaSymbol{:}%
\>[2377I]\AgdaSymbol{\{}\AgdaBound{f\,}\AgdaSpace{}%
\AgdaBound{g}\AgdaSpace{}%
\AgdaSymbol{:}\AgdaSpace{}%
\AgdaSymbol{(}\AgdaBound{x}\AgdaSpace{}%
\AgdaSymbol{:}\AgdaSpace{}%
\AgdaGeneralizable{A}\AgdaSymbol{)}\AgdaSpace{}%
\AgdaSymbol{→}\AgdaSpace{}%
\AgdaGeneralizable{P}\AgdaSpace{}%
\AgdaBound{x}\AgdaSymbol{\}}\AgdaSpace{}%
\AgdaSymbol{→}\<%
\end{code}
\begin{code}[inline*]%
\>[.][@{}l@{}]\<[2377I]%
\>[8]\AgdaBound{f\,}\AgdaSpace{}%
\AgdaOperator{\AgdaFunction{≡}}\AgdaSpace{}%
\AgdaBound{g}\AgdaSpace{}%
\AgdaSymbol{→}\AgdaSpace{}%
\AgdaSymbol{(∀}\AgdaSpace{}%
\AgdaBound{x}\AgdaSpace{}%
\AgdaSymbol{→}\AgdaSpace{}%
\AgdaBound{f\,}\AgdaSpace{}%
\AgdaBound{x}\AgdaSpace{}%
\AgdaOperator{\AgdaFunction{≡}}\AgdaSpace{}%
\AgdaBound{g}\AgdaSpace{}%
\AgdaBound{x}\AgdaSymbol{)}\<%
\end{code}
are equivalences''.

The definition of modality that I use is not taken from the text of
\citet{rijke-et-al-2020}, but is based on a definition in the
associated formalisation.
The definition is aimed at being usable also in the absence of funext.
One property is stated in the form ``if funext holds, then…'', because
the ``then'' part is often proved using funext.\footnote{See Mike
  Shulman's commit
  \url{https://github.com/HoTT/Coq-HoTT/commit/a4ac2269f389be21f2c2385b99b3be467c0c4484}.}
Another part uses the notion of being \emph{∞-extendable along a
  function} \citep{shulman-2014}.
\begin{code}[hide]%
\>[2]\AgdaOperator{\AgdaFunction{Is‐[\AgdaUnderscore{}]\mkern-2mu{}‐extendable‐along‐[\AgdaUnderscore{}]}}\AgdaSpace{}%
\AgdaSymbol{:}\<%
\\
\>[2][@{}l@{\AgdaIndent{0}}]%
\>[4]\AgdaSymbol{\{}\AgdaBound{A}\AgdaSpace{}%
\AgdaSymbol{:}\AgdaSpace{}%
\AgdaPrimitive{Type}\AgdaSpace{}%
\AgdaGeneralizable{a}\AgdaSymbol{\}}\AgdaSpace{}%
\AgdaSymbol{\{}\AgdaBound{B}\AgdaSpace{}%
\AgdaSymbol{:}\AgdaSpace{}%
\AgdaPrimitive{Type}\AgdaSpace{}%
\AgdaGeneralizable{b}\AgdaSymbol{\}}\AgdaSpace{}%
\AgdaSymbol{→}\AgdaSpace{}%
\AgdaDatatype{ℕ}\AgdaSpace{}%
\AgdaSymbol{→}\AgdaSpace{}%
\AgdaSymbol{(}\AgdaBound{A}\AgdaSpace{}%
\AgdaSymbol{→}\AgdaSpace{}%
\AgdaBound{B}\AgdaSymbol{)}\AgdaSpace{}%
\AgdaSymbol{→}\AgdaSpace{}%
\AgdaSymbol{(}\AgdaBound{B}\AgdaSpace{}%
\AgdaSymbol{→}\AgdaSpace{}%
\AgdaPrimitive{Type}\AgdaSpace{}%
\AgdaGeneralizable{c}\AgdaSymbol{)}\AgdaSpace{}%
\AgdaSymbol{→}\AgdaSpace{}%
\AgdaPrimitive{Type}\AgdaSpace{}%
\AgdaSymbol{(}\AgdaGeneralizable{a}\AgdaSpace{}%
\AgdaOperator{\AgdaPrimitive{⊔}}\AgdaSpace{}%
\AgdaGeneralizable{b}\AgdaSpace{}%
\AgdaOperator{\AgdaPrimitive{⊔}}\AgdaSpace{}%
\AgdaGeneralizable{c}\AgdaSymbol{)}\<%
\\
\>[2]\AgdaOperator{\AgdaFunction{Is‐[}}\AgdaSpace{}%
\AgdaInductiveConstructor{zero}%
\>[14]\AgdaOperator{\AgdaFunction{]‐extendable‐along‐[}}\AgdaSpace{}%
\AgdaBound{f\,}\AgdaSpace{}%
\AgdaOperator{\AgdaFunction{]}}\AgdaSpace{}%
\AgdaBound{P}\AgdaSpace{}%
\AgdaSymbol{=}\AgdaSpace{}%
\AgdaFunction{Lift}\AgdaSpace{}%
\AgdaRecord{⊤}\<%
\\
\>[2]\AgdaOperator{\AgdaFunction{Is‐[}}\AgdaSpace{}%
\AgdaInductiveConstructor{suc}\AgdaSpace{}%
\AgdaBound{n}%
\>[14]\AgdaOperator{\AgdaFunction{]‐extendable‐along‐[}}\AgdaSpace{}%
\AgdaBound{f\,}\AgdaSpace{}%
\AgdaOperator{\AgdaFunction{]}}\AgdaSpace{}%
\AgdaBound{P}\AgdaSpace{}%
\AgdaSymbol{=}\<%
\\
\>[2][@{}l@{\AgdaIndent{0}}]%
\>[4]\AgdaSymbol{((}\AgdaBound{g}\AgdaSpace{}%
\AgdaSymbol{:}\AgdaSpace{}%
\AgdaSymbol{∀}\AgdaSpace{}%
\AgdaBound{x}\AgdaSpace{}%
\AgdaSymbol{→}\AgdaSpace{}%
\AgdaBound{P}\AgdaSpace{}%
\AgdaSymbol{(}\AgdaBound{f\,}\AgdaSpace{}%
\AgdaBound{x}\AgdaSymbol{))}\AgdaSpace{}%
\AgdaSymbol{→}\AgdaSpace{}%
\AgdaSymbol{(}%
\AgdaBound{h}\AgdaSpace{}%
\AgdaSymbol{:}\AgdaSpace{}%
\AgdaSymbol{∀}\AgdaSpace{}%
\AgdaBound{x}\AgdaSpace{}%
\AgdaSymbol{→}\AgdaSpace{}%
\AgdaBound{P}\AgdaSpace{}%
\AgdaBound{x}\AgdaSymbol{)}\AgdaSpace{}%
\AgdaFunction{×}\AgdaSpace{}%
\AgdaSymbol{(∀}\AgdaSpace{}%
\AgdaBound{x}\AgdaSpace{}%
\AgdaSymbol{→}\AgdaSpace{}%
\AgdaBound{h}\AgdaSpace{}%
\AgdaSymbol{(}\AgdaBound{f\,}\AgdaSpace{}%
\AgdaBound{x}\AgdaSymbol{)}\AgdaSpace{}%
\AgdaOperator{\AgdaFunction{≡}}\AgdaSpace{}%
\AgdaBound{g}\AgdaSpace{}%
\AgdaBound{x}\AgdaSymbol{))}\AgdaSpace{}%
\AgdaOperator{\AgdaFunction{×}}\<%
\\
\>[4]\AgdaSymbol{((}\AgdaBound{g}\AgdaSpace{}%
\AgdaBound{h}\AgdaSpace{}%
\AgdaSymbol{:}\AgdaSpace{}%
\AgdaSymbol{∀}\AgdaSpace{}%
\AgdaBound{x}\AgdaSpace{}%
\AgdaSymbol{→}\AgdaSpace{}%
\AgdaBound{P}\AgdaSpace{}%
\AgdaBound{x}\AgdaSymbol{)}\AgdaSpace{}%
\AgdaSymbol{→}\AgdaSpace{}%
\AgdaOperator{\AgdaFunction{Is‐[}}\AgdaSpace{}%
\AgdaBound{n}\AgdaSpace{}%
\AgdaOperator{\AgdaFunction{]‐extendable‐along‐[}}\AgdaSpace{}%
\AgdaBound{f\,}\AgdaSpace{}%
\AgdaOperator{\AgdaFunction{]}}\AgdaSpace{}%
\AgdaSymbol{(λ}\AgdaSpace{}%
\AgdaBound{x}\AgdaSpace{}%
\AgdaSymbol{→}\AgdaSpace{}%
\AgdaBound{g}\AgdaSpace{}%
\AgdaBound{x}\AgdaSpace{}%
\AgdaOperator{\AgdaFunction{≡}}\AgdaSpace{}%
\AgdaBound{h}\AgdaSpace{}%
\AgdaBound{x}\AgdaSymbol{))}\<%
\\
\\[\AgdaEmptyExtraSkip]%
\>[2]\AgdaOperator{\AgdaFunction{Is‐∞‐extendable‐along‐[\AgdaUnderscore{}]}}\AgdaSpace{}%
\AgdaSymbol{:}\<%
\\
\>[2][@{}l@{\AgdaIndent{0}}]%
\>[4]\AgdaSymbol{\{}\AgdaBound{A}\AgdaSpace{}%
\AgdaSymbol{:}\AgdaSpace{}%
\AgdaPrimitive{Type}\AgdaSpace{}%
\AgdaGeneralizable{a}\AgdaSymbol{\}}\AgdaSpace{}%
\AgdaSymbol{\{}\AgdaBound{B}\AgdaSpace{}%
\AgdaSymbol{:}\AgdaSpace{}%
\AgdaPrimitive{Type}\AgdaSpace{}%
\AgdaGeneralizable{b}\AgdaSymbol{\}}\AgdaSpace{}%
\AgdaSymbol{→}\AgdaSpace{}%
\AgdaSymbol{(}\AgdaBound{A}\AgdaSpace{}%
\AgdaSymbol{→}\AgdaSpace{}%
\AgdaBound{B}\AgdaSymbol{)}\AgdaSpace{}%
\AgdaSymbol{→}\AgdaSpace{}%
\AgdaSymbol{(}\AgdaBound{B}\AgdaSpace{}%
\AgdaSymbol{→}\AgdaSpace{}%
\AgdaPrimitive{Type}\AgdaSpace{}%
\AgdaGeneralizable{c}\AgdaSymbol{)}\AgdaSpace{}%
\AgdaSymbol{→}\AgdaSpace{}%
\AgdaPrimitive{Type}\AgdaSpace{}%
\AgdaSymbol{(}\AgdaGeneralizable{a}\AgdaSpace{}%
\AgdaOperator{\AgdaPrimitive{⊔}}\AgdaSpace{}%
\AgdaGeneralizable{b}\AgdaSpace{}%
\AgdaOperator{\AgdaPrimitive{⊔}}\AgdaSpace{}%
\AgdaGeneralizable{c}\AgdaSymbol{)}\<%
\\
\>[2]\AgdaOperator{\AgdaFunction{Is‐∞‐extendable‐along‐[}}\AgdaSpace{}%
\AgdaBound{f\,}\AgdaSpace{}%
\AgdaOperator{\AgdaFunction{]}}\AgdaSpace{}%
\AgdaBound{P}\AgdaSpace{}%
\AgdaSymbol{=}\AgdaSpace{}%
\AgdaSymbol{∀}\AgdaSpace{}%
\AgdaBound{n}\AgdaSpace{}%
\AgdaSymbol{→}\AgdaSpace{}%
\AgdaOperator{\AgdaFunction{Is‐[}}\AgdaSpace{}%
\AgdaBound{n}\AgdaSpace{}%
\AgdaOperator{\AgdaFunction{]‐extendable‐along‐[}}\AgdaSpace{}%
\AgdaBound{f\,}\AgdaSpace{}%
\AgdaOperator{\AgdaFunction{]}}\AgdaSpace{}%
\AgdaBound{P}\<%
\end{code}
This definition is rather technical and omitted.
In the presence of funext ``\AgdaBound{P} is ∞-extendable along
\AgdaBound{f}\,'' is equivalent to ``precomposition with \AgdaBound{f\,}
for functions of type
\begin{code}[hide]%
\>[2]\AgdaKeyword{postulate}\<%
\\
\>[2][@{}l@{\AgdaIndent{0}}]%
\>[4]\AgdaPostulate{\AgdaUnderscore{}}%
\>[2523I]\AgdaSymbol{:}\<%
\end{code}
\mbox{\begin{code}[inline]%
\>[.][@{}l@{}]\<[2523I]%
\>[6]\AgdaSymbol{∀}\AgdaSpace{}%
\AgdaBound{x}\AgdaSpace{}%
\AgdaSymbol{→}\AgdaSpace{}%
\AgdaGeneralizable{P}\AgdaSpace{}%
\AgdaBound{x}\<%
\end{code}}
is an equivalence'':
\begin{code}[hide]%
\>[2]\AgdaKeyword{postulate}\<%
\\
\>[2][@{}l@{\AgdaIndent{0}}]%
\>[4]\AgdaPostulate{\AgdaUnderscore{}}%
\>[2528I]\AgdaSymbol{:}\<%
\end{code}
\begin{code}[inline]%
\>[.][@{}l@{}]\<[2528I]%
\>[6]\AgdaOperator{\AgdaFunction{Is‐∞‐extendable‐along‐[}}\AgdaSpace{}%
\AgdaGeneralizable{f\,}\AgdaSpace{}%
\AgdaOperator{\AgdaFunction{]}}\AgdaSpace{}%
\AgdaGeneralizable{P}\AgdaSpace{}%
\AgdaOperator{\AgdaPostulate{≃}}\AgdaSpace{}%
\AgdaFunction{Is‐equivalence}\AgdaSpace{}%
\AgdaSymbol{(λ}\AgdaSpace{}%
\AgdaSymbol{(}\AgdaBound{g}\AgdaSpace{}%
\AgdaSymbol{:}\AgdaSpace{}%
\AgdaSymbol{∀}\AgdaSpace{}%
\AgdaBound{x}\AgdaSpace{}%
\AgdaSymbol{→}\AgdaSpace{}%
\AgdaGeneralizable{P}\AgdaSpace{}%
\AgdaBound{x}\AgdaSymbol{)}\AgdaSpace{}%
\AgdaSymbol{→}\AgdaSpace{}%
\AgdaBound{g}\AgdaSpace{}%
\AgdaOperator{\AgdaFunction{∘}}\AgdaSpace{}%
\AgdaGeneralizable{f\,}\AgdaSymbol{)}\<%
\end{code}.

\begin{code}[hide]%
\>[2]\AgdaKeyword{postulate}\<%
\\
\>[2][@{}l@{\AgdaIndent{0}}]%
\>[4]\AgdaPostulate{\AgdaUnderscore{}}%
\>[2546I]\AgdaSymbol{:}\AgdaSpace{}%
\AgdaSymbol{(}\AgdaBound{◯}\AgdaSpace{}%
\AgdaSymbol{:}\AgdaSpace{}%
\AgdaPrimitive{Type}\AgdaSpace{}%
\AgdaGeneralizable{a}\AgdaSpace{}%
\AgdaSymbol{→}\AgdaSpace{}%
\AgdaPrimitive{Type}\AgdaSpace{}%
\AgdaGeneralizable{a}\AgdaSymbol{)}\AgdaSpace{}%
\AgdaSymbol{→}\<%
\end{code}
\newcommand{\OA}{%
\begin{code}[inline]%
\>[.][@{}l@{}]\<[2546I]%
\>[6]\AgdaBound{◯}\AgdaSpace{}%
\AgdaGeneralizable{A}\<%
\end{code}}

\begin{code}[hide]%
\>[2]\AgdaKeyword{postulate}\<%
\\
\>[2][@{}l@{\AgdaIndent{0}}]%
\>[4]\AgdaPostulate{Is‐modality}\AgdaSpace{}%
\AgdaSymbol{:}\<%
\\
\>[4][@{}l@{\AgdaIndent{0}}]%
\>[6]\AgdaSymbol{∀}\AgdaSpace{}%
\AgdaBound{a}\AgdaSpace{}%
\AgdaSymbol{→}\AgdaSpace{}%
\AgdaSymbol{(}\AgdaBound{◯}\AgdaSpace{}%
\AgdaSymbol{:}\AgdaSpace{}%
\AgdaPrimitive{Type}\AgdaSpace{}%
\AgdaBound{a}\AgdaSpace{}%
\AgdaSymbol{→}\AgdaSpace{}%
\AgdaPrimitive{Type}\AgdaSpace{}%
\AgdaBound{a}\AgdaSymbol{)}\AgdaSpace{}%
\AgdaSymbol{→}\AgdaSpace{}%
\AgdaSymbol{(\{}\AgdaBound{A}\AgdaSpace{}%
\AgdaSymbol{:}\AgdaSpace{}%
\AgdaPrimitive{Type}\AgdaSpace{}%
\AgdaBound{a}\AgdaSymbol{\}}\AgdaSpace{}%
\AgdaSymbol{→}\AgdaSpace{}%
\AgdaBound{A}\AgdaSpace{}%
\AgdaSymbol{→}\AgdaSpace{}%
\AgdaBound{◯}\AgdaSpace{}%
\AgdaBound{A}\AgdaSymbol{)}\AgdaSpace{}%
\AgdaSymbol{→}\AgdaSpace{}%
\AgdaPrimitive{Type}\AgdaSpace{}%
\AgdaSymbol{(}\AgdaPrimitive{lsuc}\AgdaSpace{}%
\AgdaBound{a}\AgdaSymbol{)}\<%
\\
\>[4]\AgdaOperator{\AgdaPostulate{\AgdaUnderscore{}-connected\AgdaUnderscore{}}}\AgdaSpace{}%
\AgdaSymbol{:}\<%
\\
\>[4][@{}l@{\AgdaIndent{0}}]%
\>[6]\AgdaSymbol{(}\AgdaPrimitive{Type}\AgdaSpace{}%
\AgdaGeneralizable{a}\AgdaSpace{}%
\AgdaSymbol{→}\AgdaSpace{}%
\AgdaPrimitive{Type}\AgdaSpace{}%
\AgdaGeneralizable{a}\AgdaSymbol{)}\AgdaSpace{}%
\AgdaSymbol{→}\AgdaSpace{}%
\AgdaPrimitive{Type}\AgdaSpace{}%
\AgdaGeneralizable{a}\AgdaSpace{}%
\AgdaSymbol{→}\AgdaSpace{}%
\AgdaPrimitive{Type}\AgdaSpace{}%
\AgdaGeneralizable{a}\<%
\end{code}

Now we can define what a modality is:
\begin{definition}
  \begin{code}[hide]%
\>[2]\AgdaKeyword{record}\AgdaSpace{}%
\AgdaRecord{Modality‐for}\<%
\\
\>[2][@{}l@{\AgdaIndent{0}}]%
\>[4]\AgdaSymbol{\{}\<%
\end{code}
  The functions
  \mbox{\begin{code}[inline*]%
\>[4][@{}l@{\AgdaIndent{1}}]%
\>[6]\AgdaBound{◯}\AgdaSpace{}%
\AgdaSymbol{:}\AgdaSpace{}%
\AgdaPrimitive{Type}\AgdaSpace{}%
\AgdaGeneralizable{a}\AgdaSpace{}%
\AgdaSymbol{→}\AgdaSpace{}%
\AgdaPrimitive{Type}\AgdaSpace{}%
\AgdaGeneralizable{a}\<%
\end{code}}
  \begin{code}[hide]%
\>[4]\AgdaSymbol{\}}%
\>[2598I]\AgdaSymbol{(}\<%
\end{code}
  and
  \begin{code}[inline*]%
\>[.][@{}l@{}]\<[2598I]%
\>[6]\AgdaBound{η}\AgdaSpace{}%
\AgdaSymbol{:}\AgdaSpace{}%
\AgdaSymbol{\{}\AgdaBound{A}\AgdaSpace{}%
\AgdaSymbol{:}\AgdaSpace{}%
\AgdaPrimitive{Type}\AgdaSpace{}%
\AgdaGeneralizable{a}\AgdaSymbol{\}}\AgdaSpace{}%
\AgdaSymbol{→}\AgdaSpace{}%
\AgdaBound{A}\AgdaSpace{}%
\AgdaSymbol{→}\AgdaSpace{}%
\AgdaBound{◯}\AgdaSpace{}%
\AgdaBound{A}\<%
\end{code}
  \begin{code}[hide]%
\>[4]\AgdaSymbol{)}\AgdaSpace{}%
\AgdaSymbol{:}\AgdaSpace{}%
\AgdaPrimitive{Type}\AgdaSpace{}%
\AgdaSymbol{(}\AgdaPrimitive{lsuc}\AgdaSpace{}%
\AgdaBound{a}\AgdaSymbol{)}\AgdaSpace{}%
\AgdaKeyword{where}\<%
\\
\>[4]\AgdaKeyword{field}\<%
\end{code}
  \emph{form a modality} for the level \AgdaBound{a} if they
  satisfy
  \begin{code}[hide]%
\>[4][@{}l@{\AgdaIndent{1}}]%
\>[6]\AgdaField{dummy}\AgdaSpace{}%
\AgdaSymbol{:}\<%
\end{code}
  \begin{code}[inline]%
\>[6][@{}l@{\AgdaIndent{1}}]%
\>[8]\AgdaPostulate{Is‐modality}\AgdaSpace{}%
\AgdaBound{a}\AgdaSpace{}%
\AgdaBound{◯}\AgdaSpace{}%
\AgdaBound{η}\<%
\end{code},
  a type with the following components:
  \begin{itemize}
  \item A function
    \begin{code}[inline*]%
\>[6]\AgdaField{Modal}\AgdaSpace{}%
\AgdaSymbol{:}\AgdaSpace{}%
\AgdaPrimitive{Type}\AgdaSpace{}%
\AgdaBound{a}\AgdaSpace{}%
\AgdaSymbol{→}\AgdaSpace{}%
\AgdaPrimitive{Type}\AgdaSpace{}%
\AgdaBound{a}\<%
\end{code}
    that classifies modal types.
  \item Funext (for \AgdaBound{a}) implies that \AgdaField{Modal} is
    propositional:
    \begin{code}[hide]%
\>[6]\AgdaField{Modal-propositional}\AgdaSpace{}%
\AgdaSymbol{:}\<%
\end{code}
    \begin{code}[inline]%
\>[6][@{}l@{\AgdaIndent{1}}]%
\>[8]\AgdaPostulate{Funext}\AgdaSpace{}%
\AgdaBound{a}\AgdaSpace{}%
\AgdaBound{a}\AgdaSpace{}%
\AgdaSymbol{→}\AgdaSpace{}%
\AgdaFunction{Is‐prop}\AgdaSpace{}%
\AgdaSymbol{(}\AgdaField{Modal}\AgdaSpace{}%
\AgdaGeneralizable{A}\AgdaSymbol{)}\<%
\end{code}.
  \item%
    \begin{code}[hide]%
\>[6]\AgdaField{Modal-◯}\AgdaSpace{}%
\AgdaSymbol{:}\<%
\end{code}
    \mbox{\begin{code}[inline]%
\>[6][@{}l@{\AgdaIndent{1}}]%
\>[8]\AgdaField{Modal}\AgdaSpace{}%
\AgdaSymbol{(}\AgdaBound{◯}\AgdaSpace{}%
\AgdaGeneralizable{A}\AgdaSymbol{)}\<%
\end{code}}
    always holds.
  \item \AgdaField{Modal} respects equivalences:
    \begin{code}[hide]%
\>[6]\AgdaField{Modal-respects-≃}\AgdaSpace{}%
\AgdaSymbol{:}\<%
\end{code}
    \mbox{\begin{code}[inline]%
\>[6][@{}l@{\AgdaIndent{1}}]%
\>[8]\AgdaGeneralizable{A}\AgdaSpace{}%
\AgdaOperator{\AgdaPostulate{≃}}\AgdaSpace{}%
\AgdaGeneralizable{B}\AgdaSpace{}%
\AgdaSymbol{→}\AgdaSpace{}%
\AgdaField{Modal}\AgdaSpace{}%
\AgdaGeneralizable{A}\AgdaSpace{}%
\AgdaSymbol{→}\AgdaSpace{}%
\AgdaField{Modal}\AgdaSpace{}%
\AgdaGeneralizable{B}\<%
\end{code}}.
  \item Every pointwise modal family of type
    \begin{code}[hide]%
\>[4]\AgdaFunction{\AgdaUnderscore{}}\AgdaSpace{}%
\AgdaSymbol{:}\AgdaSpace{}%
\AgdaPrimitive{Type}\AgdaSpace{}%
\AgdaSymbol{(}\AgdaPrimitive{lsuc}\AgdaSpace{}%
\AgdaBound{a}\AgdaSymbol{)}\<%
\\
\>[4]\AgdaSymbol{\AgdaUnderscore{}}\AgdaSpace{}%
\AgdaSymbol{=}\AgdaSpace{}%
\AgdaSymbol{∀}\AgdaSpace{}%
\AgdaBound{A}%
\>[2650I]\AgdaSymbol{→}\<%
\end{code}
    \begin{code}[inline*]%
\>[2650I][@{}l@{\AgdaIndent{1}}]%
\>[16]\AgdaBound{◯}\AgdaSpace{}%
\AgdaBound{A}\AgdaSpace{}%
\AgdaSymbol{→}\AgdaSpace{}%
\AgdaPrimitive{Type}\AgdaSpace{}%
\AgdaBound{a}\<%
\end{code}
    \begin{code}[hide]%
\>[4]\AgdaKeyword{field}\<%
\end{code}
    is ∞-extendable along~\AgdaBound{η}:
    \begin{code}[hide]%
\>[4][@{}l@{\AgdaIndent{1}}]%
\>[6]\AgdaField{extendable-along-η}\AgdaSpace{}%
\AgdaSymbol{:}\<%
\end{code}
    \begin{code}%
\>[6][@{}l@{\AgdaIndent{1}}]%
\>[8]\AgdaSymbol{\{}\AgdaBound{P}\AgdaSpace{}%
\AgdaSymbol{:}\AgdaSpace{}%
\AgdaBound{◯}\AgdaSpace{}%
\AgdaGeneralizable{A}\AgdaSpace{}%
\AgdaSymbol{→}\AgdaSpace{}%
\AgdaPrimitive{Type}\AgdaSpace{}%
\AgdaBound{a}\AgdaSymbol{\}}\AgdaSpace{}%
\AgdaSymbol{→}\AgdaSpace{}%
\AgdaSymbol{(∀}\AgdaSpace{}%
\AgdaBound{x}\AgdaSpace{}%
\AgdaSymbol{→}\AgdaSpace{}%
\AgdaField{Modal}\AgdaSpace{}%
\AgdaSymbol{(}\AgdaBound{P}\AgdaSpace{}%
\AgdaBound{x}\AgdaSymbol{))}\AgdaSpace{}%
\AgdaSymbol{→}\AgdaSpace{}%
\AgdaOperator{\AgdaFunction{Is‐∞‐extendable‐along‐[}}\AgdaSpace{}%
\AgdaBound{η}\AgdaSpace{}%
\AgdaOperator{\AgdaFunction{]}}\AgdaSpace{}%
\AgdaBound{P}\<%
\end{code}
  \end{itemize}
  A \emph{modality} for the level \AgdaBound{a} consists of two
  functions \AgdaBound{◯} and \AgdaBound{η} and a proof of
  \begin{code}[hide]%
\>[4]\AgdaKeyword{postulate}\<%
\\
\>[4][@{}l@{\AgdaIndent{0}}]%
\>[6]\AgdaPostulate{\AgdaUnderscore{}}%
\>[2674I]\AgdaSymbol{:}\<%
\end{code}
  \begin{code}[inline]%
\>[.][@{}l@{}]\<[2674I]%
\>[8]\AgdaPostulate{Is‐modality}\AgdaSpace{}%
\AgdaBound{a}\AgdaSpace{}%
\AgdaBound{◯}\AgdaSpace{}%
\AgdaBound{η}\<%
\end{code}.
  The modality is \emph{left exact} if
  \begin{code}[hide]%
\>[4]\AgdaKeyword{postulate}\<%
\\
\>[4][@{}l@{\AgdaIndent{0}}]%
\>[6]\AgdaPostulate{\AgdaUnderscore{}}%
\>[2678I]\AgdaSymbol{:}\<%
\end{code}
  \begin{code}[inline]%
\>[.][@{}l@{}]\<[2678I]%
\>[8]\AgdaSymbol{\{}\AgdaBound{A}\AgdaSpace{}%
\AgdaSymbol{:}\AgdaSpace{}%
\AgdaPrimitive{Type}\AgdaSpace{}%
\AgdaBound{a}\AgdaSymbol{\}}\AgdaSpace{}%
\AgdaSymbol{\{}\AgdaBound{x}\AgdaSpace{}%
\AgdaBound{y}\AgdaSpace{}%
\AgdaSymbol{:}\AgdaSpace{}%
\AgdaBound{A}\AgdaSymbol{\}}\AgdaSpace{}%
\AgdaSymbol{→}\AgdaSpace{}%
\AgdaFunction{Contractible}\AgdaSpace{}%
\AgdaSymbol{(}\AgdaBound{◯}\AgdaSpace{}%
\AgdaBound{A}\AgdaSymbol{)}\AgdaSpace{}%
\AgdaSymbol{→}\AgdaSpace{}%
\AgdaFunction{Contractible}\AgdaSpace{}%
\AgdaSymbol{(}\AgdaBound{◯}\AgdaSpace{}%
\AgdaSymbol{(}\AgdaBound{x}\AgdaSpace{}%
\AgdaOperator{\AgdaFunction{≡}}\AgdaSpace{}%
\AgdaBound{y}\AgdaSymbol{))}\<%
\end{code}.
\end{definition}

\begin{code}[hide]%
\>[4]\AgdaKeyword{postulate}\<%
\\
\>[4][@{}l@{\AgdaIndent{0}}]%
\>[6]\AgdaPostulate{\AgdaUnderscore{}}%
\>[2696I]\AgdaSymbol{:}\<%
\end{code}
\newcommand{\VeryModal}{%
  \begin{code}[inline]%
\>[2696I][@{}l@{\AgdaIndent{1}}]%
\>[10]\AgdaSymbol{\{}\AgdaBound{A}\AgdaSpace{}%
\AgdaSymbol{:}\AgdaSpace{}%
\AgdaPrimitive{Type}\AgdaSpace{}%
\AgdaBound{a}\AgdaSymbol{\}}\AgdaSpace{}%
\AgdaSymbol{→}\AgdaSpace{}%
\AgdaBound{◯}\AgdaSpace{}%
\AgdaSymbol{(}\AgdaField{Modal}\AgdaSpace{}%
\AgdaBound{A}\AgdaSymbol{)}\<%
\end{code}}

One example of a modality is propositional truncation.
The variant of propositional truncation that can be defined using the
quotients discussed in this text is also an example.
In the presence of funext double-negation is a modality.

\citet{rijke-et-al-2020} present a number of properties that hold for
every modality (or left exact modality).
For every modality there is a function
\begin{code}[hide]%
\>[4]\AgdaKeyword{postulate}\<%
\end{code}
\begin{code}[inline*]%
\>[4][@{}l@{\AgdaIndent{1}}]%
\>[6]\AgdaPostulate{η‐cong}\AgdaSpace{}%
\AgdaSymbol{:}\AgdaSpace{}%
\AgdaBound{◯}\AgdaSpace{}%
\AgdaSymbol{(}\AgdaGeneralizable{x}\AgdaSpace{}%
\AgdaOperator{\AgdaFunction{≡}}\AgdaSpace{}%
\AgdaGeneralizable{y}\AgdaSymbol{)}\AgdaSpace{}%
\AgdaSymbol{→}\AgdaSpace{}%
\AgdaBound{η}\AgdaSpace{}%
\AgdaGeneralizable{x}\AgdaSpace{}%
\AgdaOperator{\AgdaFunction{≡}}\AgdaSpace{}%
\AgdaBound{η}\AgdaSpace{}%
\AgdaGeneralizable{y}\<%
\end{code}
satisfying
\begin{code}[hide]%
\>[6]\AgdaPostulate{\AgdaUnderscore{}}%
\>[2715I]\AgdaSymbol{:}\<%
\end{code}
\mbox{\begin{code}[inline]%
\>[.][@{}l@{}]\<[2715I]%
\>[8]\AgdaPostulate{η‐cong}\AgdaSpace{}%
\AgdaSymbol{(}\AgdaBound{η}\AgdaSpace{}%
\AgdaGeneralizable{eq}\AgdaSymbol{)}\AgdaSpace{}%
\AgdaOperator{\AgdaFunction{≡}}\AgdaSpace{}%
\AgdaFunction{cong}\AgdaSpace{}%
\AgdaBound{η}\AgdaSpace{}%
\AgdaGeneralizable{eq}\<%
\end{code}}.
Notice the similarity to \boxcong{}: if
\begin{code}[hide]%
\>[6]\AgdaPostulate{\AgdaUnderscore{}}%
\>[2722I]\AgdaSymbol{:}\AgdaSpace{}%
\AgdaOperator{\AgdaGeneralizable{Has‐type[}}\AgdaSpace{}%
\AgdaSymbol{(∀}\AgdaSpace{}%
\AgdaSymbol{(}\AgdaSymbol{@ω}\AgdaSpace{}%
\AgdaBound{\AgdaUnderscore{}}\AgdaSymbol{)}\AgdaSpace{}%
\AgdaSymbol{→}\AgdaSpace{}%
\AgdaSymbol{\AgdaUnderscore{})}\AgdaSpace{}%
\AgdaOperator{\AgdaGeneralizable{]}}\<%
\end{code}
\begin{code}[inline*]%
\>[.][@{}l@{}]\<[2722I]%
\>[8]\AgdaSymbol{λ}\AgdaSpace{}%
\AgdaBound{A}\AgdaSpace{}%
\AgdaSymbol{→}\AgdaSpace{}%
\AgdaDatatype{Erased}\AgdaSpace{}%
\AgdaBound{A}\<%
\end{code}
and
\begin{code}[hide]%
\>[6]\AgdaPostulate{\AgdaUnderscore{}}%
\>[2734I]\AgdaSymbol{:}\AgdaSpace{}%
\AgdaOperator{\AgdaGeneralizable{Has‐type[}}\AgdaSpace{}%
\AgdaSymbol{((}\AgdaSymbol{@ω}\AgdaSpace{}%
\AgdaBound{\AgdaUnderscore{}}\AgdaSpace{}%
\AgdaSymbol{:}\AgdaSpace{}%
\AgdaGeneralizable{A}\AgdaSymbol{)}\AgdaSpace{}%
\AgdaSymbol{→}\AgdaSpace{}%
\AgdaSymbol{\AgdaUnderscore{})}\AgdaSpace{}%
\AgdaOperator{\AgdaGeneralizable{]}}\<%
\end{code}
\begin{code}[inline*]%
\>[.][@{}l@{}]\<[2734I]%
\>[8]\AgdaSymbol{λ}\AgdaSpace{}%
\AgdaBound{x}\AgdaSpace{}%
\AgdaSymbol{→}\AgdaSpace{}%
\AgdaOperator{\AgdaInductiveConstructor{[}}\AgdaSpace{}%
\AgdaBound{x}\AgdaSpace{}%
\AgdaOperator{\AgdaInductiveConstructor{]}}\<%
\end{code}
form a modality, then \boxcong{} can be defined.

A type \AgdaBound{A} is modal if and only if
\begin{code}[hide]%
\>[6]\AgdaPostulate{\AgdaUnderscore{}}%
\>[2748I]\AgdaSymbol{:}\AgdaSpace{}%
\AgdaOperator{\AgdaGeneralizable{Has‐type}}\<%
\end{code}
\begin{code}[inline*]%
\>[.][@{}l@{}]\<[2748I]%
\>[8]\AgdaBound{η}\AgdaSpace{}%
\AgdaSymbol{\{}\AgdaArgument{A}\AgdaSpace{}%
\AgdaSymbol{=}\AgdaSpace{}%
\AgdaGeneralizable{A}\AgdaSymbol{\}}\<%
\end{code}
is an equivalence \citep{rijke-et-al-2020}.
A type is \emph{separated} \citep{christensen-et-al-2020} if all its
identity types are modal:
\begin{code}[hide]%
\>[4]\AgdaFunction{Separated}\AgdaSpace{}%
\AgdaSymbol{:}\AgdaSpace{}%
\AgdaPrimitive{Type}\AgdaSpace{}%
\AgdaBound{a}\AgdaSpace{}%
\AgdaSymbol{→}\AgdaSpace{}%
\AgdaPrimitive{Type}\AgdaSpace{}%
\AgdaBound{a}\<%
\end{code}
\begin{code}[inline]%
\>[4]\AgdaFunction{Separated}\AgdaSpace{}%
\AgdaBound{A}\AgdaSpace{}%
\AgdaSymbol{=}\AgdaSpace{}%
\AgdaSymbol{(}\AgdaBound{x}\AgdaSpace{}%
\AgdaBound{y}\AgdaSpace{}%
\AgdaSymbol{:}\AgdaSpace{}%
\AgdaBound{A}\AgdaSymbol{)}\AgdaSpace{}%
\AgdaSymbol{→}\AgdaSpace{}%
\AgdaField{Modal}\AgdaSpace{}%
\AgdaSymbol{(}\AgdaBound{x}\AgdaSpace{}%
\AgdaOperator{\AgdaFunction{≡}}\AgdaSpace{}%
\AgdaBound{y}\AgdaSymbol{)}\<%
\end{code}.
In light of this, let me use the following definitions:\MCbreak{}
\begin{minipage}[t]{0.49\linewidth}
\begin{code}%
\>[2]\AgdaFunction{Modal\ensuremath{{}^{\mkern2mu\mathrm{E}}}}\AgdaSpace{}%
\AgdaSymbol{:}\AgdaSpace{}%
\AgdaPrimitive{Type}\AgdaSpace{}%
\AgdaGeneralizable{a}\AgdaSpace{}%
\AgdaSymbol{→}\AgdaSpace{}%
\AgdaPrimitive{Type}\AgdaSpace{}%
\AgdaGeneralizable{a}\<%
\\
\>[2]\AgdaFunction{Modal\ensuremath{{}^{\mkern2mu\mathrm{E}}}}\AgdaSpace{}%
\AgdaBound{A}\AgdaSpace{}%
\AgdaSymbol{=}\AgdaSpace{}%
\AgdaFunction{Is‐equivalence}\AgdaSpace{}%
\AgdaSymbol{(λ}\AgdaSpace{}%
\AgdaSymbol{(}\AgdaBound{x}\AgdaSpace{}%
\AgdaSymbol{:}\AgdaSpace{}%
\AgdaBound{A}\AgdaSymbol{)}\AgdaSpace{}%
\AgdaSymbol{→}\AgdaSpace{}%
\AgdaOperator{\AgdaInductiveConstructor{[}}\AgdaSpace{}%
\AgdaBound{x}\AgdaSpace{}%
\AgdaOperator{\AgdaInductiveConstructor{]}}\AgdaSymbol{)}\<%
\end{code}
\end{minipage}
\begin{minipage}[t]{0.49\linewidth}
\begin{code}%
\>[2]\AgdaFunction{Separated\ensuremath{{}^{\mkern2mu\mathrm{E}}}}\AgdaSpace{}%
\AgdaSymbol{:}\AgdaSpace{}%
\AgdaPrimitive{Type}\AgdaSpace{}%
\AgdaGeneralizable{a}\AgdaSpace{}%
\AgdaSymbol{→}\AgdaSpace{}%
\AgdaPrimitive{Type}\AgdaSpace{}%
\AgdaGeneralizable{a}\<%
\\
\>[2]\AgdaFunction{Separated\ensuremath{{}^{\mkern2mu\mathrm{E}}}}\AgdaSpace{}%
\AgdaBound{A}\AgdaSpace{}%
\AgdaSymbol{=}\AgdaSpace{}%
\AgdaSymbol{(}\AgdaBound{x}\AgdaSpace{}%
\AgdaBound{y}\AgdaSpace{}%
\AgdaSymbol{:}\AgdaSpace{}%
\AgdaBound{A}\AgdaSymbol{)}\AgdaSpace{}%
\AgdaSymbol{→}\AgdaSpace{}%
\AgdaFunction{Modal\ensuremath{{}^{\mkern2mu\mathrm{E}}}}\AgdaSpace{}%
\AgdaSymbol{(}\AgdaBound{x}\AgdaSpace{}%
\AgdaOperator{\AgdaFunction{≡}}\AgdaSpace{}%
\AgdaBound{y}\AgdaSymbol{)}\<%
\end{code}
\end{minipage}\\
I also use
\begin{code}[hide]%
\>[2]\AgdaFunction{Stable\ensuremath{{}^{\mkern2mu\mathrm{E}}}}\AgdaSpace{}%
\AgdaSymbol{:}\AgdaSpace{}%
\AgdaPrimitive{Type}\AgdaSpace{}%
\AgdaGeneralizable{a}\AgdaSpace{}%
\AgdaSymbol{→}\AgdaSpace{}%
\AgdaPrimitive{Type}\AgdaSpace{}%
\AgdaGeneralizable{a}\<%
\end{code}
\begin{code}[inline*]%
\>[2]\AgdaFunction{Stable\ensuremath{{}^{\mkern2mu\mathrm{E}}}}\AgdaSpace{}%
\AgdaBound{A}\AgdaSpace{}%
\AgdaSymbol{=}\AgdaSpace{}%
\AgdaDatatype{Erased}\AgdaSpace{}%
\AgdaBound{A}\AgdaSpace{}%
\AgdaSymbol{→}\AgdaSpace{}%
\AgdaBound{A}\<%
\end{code}
and
\begin{code}[hide]%
\>[2]\AgdaFunction{Stable‐equality\ensuremath{{}^{\mkern2mu\mathrm{E}}}}\AgdaSpace{}%
\AgdaSymbol{:}\AgdaSpace{}%
\AgdaPrimitive{Type}\AgdaSpace{}%
\AgdaGeneralizable{a}\AgdaSpace{}%
\AgdaSymbol{→}\AgdaSpace{}%
\AgdaPrimitive{Type}\AgdaSpace{}%
\AgdaGeneralizable{a}\<%
\end{code}
\begin{code}[inline]%
\>[2]\AgdaFunction{Stable‐equality\ensuremath{{}^{\mkern2mu\mathrm{E}}}}\AgdaSpace{}%
\AgdaBound{A}\AgdaSpace{}%
\AgdaSymbol{=}\AgdaSpace{}%
\AgdaSymbol{(}\AgdaBound{x}\AgdaSpace{}%
\AgdaBound{y}\AgdaSpace{}%
\AgdaSymbol{:}\AgdaSpace{}%
\AgdaBound{A}\AgdaSymbol{)}\AgdaSpace{}%
\AgdaSymbol{→}\AgdaSpace{}%
\AgdaFunction{Stable\ensuremath{{}^{\mkern2mu\mathrm{E}}}}\AgdaSpace{}%
\AgdaSymbol{(}\AgdaBound{x}\AgdaSpace{}%
\AgdaOperator{\AgdaFunction{≡}}\AgdaSpace{}%
\AgdaBound{y}\AgdaSymbol{)}\<%
\end{code}.
Note that if
\begin{code}[hide]%
\>[2]\AgdaKeyword{postulate}\<%
\\
\>[2][@{}l@{\AgdaIndent{0}}]%
\>[4]\AgdaPostulate{\AgdaUnderscore{}}%
\>[2833I]\AgdaSymbol{:}\<%
\end{code}
\begin{code}[inline*]%
\>[.][@{}l@{}]\<[2833I]%
\>[6]\AgdaFunction{Stable\ensuremath{{}^{\mkern2mu\mathrm{E}}}}\AgdaSpace{}%
\AgdaGeneralizable{A}\<%
\end{code}
holds, then one can convert from
\ErasedA{} to \AgdaBound{A}, and if
\begin{code}[hide]%
\>[2]\AgdaKeyword{postulate}\<%
\\
\>[2][@{}l@{\AgdaIndent{0}}]%
\>[4]\AgdaPostulate{\AgdaUnderscore{}}%
\>[2835I]\AgdaSymbol{:}\<%
\end{code}
\begin{code}[inline*]%
\>[.][@{}l@{}]\<[2835I]%
\>[6]\AgdaFunction{Modal\ensuremath{{}^{\mkern2mu\mathrm{E}}}}\AgdaSpace{}%
\AgdaGeneralizable{A}\<%
\end{code}
holds, then there is an equivalence between \ErasedA{} and
\AgdaBound{A}: for modal types one can resurrect erased values.

\section{Characterisations of \texorpdfstring{``Box-Cong Can Be
    Defined''}{"Box-Cong Can Be Defined"}}
\label{sec:box-cong-iff}

Let me now present a number of logically equivalent characterisations
of ``\boxcong{} can be defined for the universe level
\AgdaBound{a}''.
Due to lack of space the proof is omitted, but note that it employs a
technique used by \citet{altenkirch-et-al-2017}.
The proof can be found in the accompanying code, along with even more
equivalent characterisations.
\begin{theorem}
  \label{thm:equivalent-1}
  \setlength{\mathindent}{0em}%
  \AgdaNoSpaceAroundCode{}%
  For a universe level \AgdaBound{a} the following definitions are
  logically equivalent.
  In the presence of function extensionality they are propositions,
  and thus equivalent, and in light of
  Theorem~\ref{thm:box-cong-from-funext} even contractible.
  Every definition given below has the type
  \begin{code}[hide]%
\>[2]\AgdaKeyword{module}\AgdaSpace{}%
\AgdaModule{\AgdaUnderscore{}}\AgdaSpace{}%
\AgdaSymbol{\{}\AgdaBound{a}\AgdaSpace{}%
\AgdaSymbol{:}\AgdaSpace{}%
\AgdaPostulate{Level}\AgdaSymbol{\}}\AgdaSpace{}%
\AgdaKeyword{where}\<%
\\
\>[2][@{}l@{\AgdaIndent{0}}]%
\>[4]\AgdaFunction{[\ensuremath{\mkern1.5mu}]‐cong‐ax\ensuremath{{}_{\mathrm{1}}}}\<%
\\
\>[4][@{}l@{\AgdaIndent{0}}]%
\>[6]\AgdaComment{--\ []-cong-ax\ensuremath{{}_{\mathrm{2}}}}\<%
\\
\>[6]\AgdaFunction{[\ensuremath{\mkern1.5mu}]‐cong‐ax\ensuremath{{}_{\mathrm{3}}}}\<%
\\
\>[6]\AgdaComment{--\ []-cong-ax\ensuremath{{}_{\mathrm{4}}}}\<%
\\
\>[6]\AgdaFunction{[\ensuremath{\mkern1.5mu}]‐cong‐ax\ensuremath{{}_{\mathrm{5}}}}\<%
\\
\>[6]\AgdaComment{--\ []-cong-ax\ensuremath{{}_{\mathrm{6}}}}\<%
\\
\>[6]\AgdaFunction{[\ensuremath{\mkern1.5mu}]‐cong\ensuremath{{}^{\mkern1mu\mathrm{-1}}}‐ax\ensuremath{{}_{\mathrm{1}}}}\<%
\\
\>[6]\AgdaFunction{[\ensuremath{\mkern1.5mu}]‐cong\ensuremath{{}^{\mkern1mu\mathrm{-1}}}‐ax\ensuremath{{}_{\mathrm{2}}}}\<%
\\
\>[6]\AgdaFunction{Erased‐Separated\ensuremath{{}^{\mkern2mu\mathrm{E}}}‐ax}\<%
\\
\>[6]\AgdaFunction{Modal\ensuremath{{}^{\mkern2mu\mathrm{E}}}‐Erased×Modal\ensuremath{{}^{\mkern2mu\mathrm{E}}}→Separated\ensuremath{{}^{\mkern2mu\mathrm{E}}}‐ax}\<%
\\
\>[6]\AgdaFunction{Subst\ensuremath{{}^{\mkern2mu\mathrm{E}}}‐ax}\<%
\\
\>[6]\AgdaComment{--\ J\ensuremath{{}^{\mkern2mu\mathrm{E}}}-ax\ensuremath{{}_{\mathrm{1}}}}\<%
\\
\>[6]\AgdaFunction{J\ensuremath{{}^{\mkern2mu\mathrm{E}}}‐ax\ensuremath{{}_{\mathrm{2}}}}\<%
\\
\>[6]\AgdaComment{--\ 2-ext-ax}\<%
\\
\>[6]\AgdaFunction{∞‐ext‐ax}\AgdaSpace{}%
\AgdaSymbol{:}\<%
\end{code}
  \begin{code}[inline]%
\>[6]\AgdaPrimitive{Type}\AgdaSpace{}%
\AgdaSymbol{(}\AgdaPrimitive{lsuc}\AgdaSpace{}%
\AgdaBound{a}\AgdaSymbol{)}\<%
\end{code}.

  The first definitions are variations of
  Definition~\ref{def:boxcong}, using propositional instead of
  judgemental equality (there are also variants with
  \begin{code}[hide]%
\>[4]\AgdaKeyword{postulate}\<%
\\
\>[4][@{}l@{\AgdaIndent{0}}]%
\>[6]\AgdaPostulate{\AgdaUnderscore{}}%
\>[2845I]\AgdaSymbol{:}\<%
\end{code}
  \begin{code}[inline*]%
\>[.][@{}l@{}]\<[2845I]%
\>[8]\AgdaDatatype{Erased}\AgdaSpace{}%
\AgdaSymbol{(}\AgdaGeneralizable{x}\AgdaSpace{}%
\AgdaOperator{\AgdaFunction{≡}}\AgdaSpace{}%
\AgdaGeneralizable{y}\AgdaSymbol{)}\<%
\end{code}
  instead of
  \begin{code}[hide]%
\>[4]\AgdaKeyword{postulate}\<%
\\
\>[4][@{}l@{\AgdaIndent{0}}]%
\>[6]\AgdaPostulate{\AgdaUnderscore{}}%
\>[2849I]\AgdaSymbol{:}\<%
\end{code}
  \begin{code}[inline]%
\>[.][@{}l@{}]\<[2849I]%
\>[8]\AgdaSymbol{@0}\AgdaSpace{}%
\AgdaGeneralizable{x}\AgdaSpace{}%
\AgdaOperator{\AgdaFunction{≡}}\AgdaSpace{}%
\AgdaGeneralizable{y}\<%
\end{code};
  \begin{code}[hide]%
\>[8]\AgdaSymbol{→}\AgdaSpace{}%
\AgdaPrimitive{Type}\<%
\\
\\[\AgdaEmptyExtraSkip]%
\>[4]\AgdaFunction{\AgdaUnderscore{}}\AgdaSpace{}%
\AgdaSymbol{:}\AgdaSpace{}%
\AgdaSymbol{(}\AgdaBound{x}\AgdaSpace{}%
\AgdaSymbol{:}\AgdaSpace{}%
\AgdaGeneralizable{A}\AgdaSymbol{)}\AgdaSpace{}%
\AgdaSymbol{→}\AgdaSpace{}%
\AgdaBound{x}\AgdaSpace{}%
\AgdaOperator{\AgdaFunction{≡}}\AgdaSpace{}%
\AgdaBound{x}\<%
\\
\>[4]\AgdaSymbol{\AgdaUnderscore{}}%
\>[2862I]\AgdaSymbol{=}\AgdaSpace{}%
\AgdaSymbol{λ}\AgdaSpace{}%
\AgdaBound{x}\AgdaSpace{}%
\AgdaSymbol{→}\<%
\end{code}
  \begin{code}[inline*]%
\>[.][@{}l@{}]\<[2862I]%
\>[6]\AgdaInductiveConstructor{refl}\AgdaSpace{}%
\AgdaSymbol{\{}\AgdaArgument{x}\AgdaSpace{}%
\AgdaSymbol{=}\AgdaSpace{}%
\AgdaBound{x}\AgdaSymbol{\}}\<%
\end{code}
  has type
  \begin{code}[hide]%
\>[4]\AgdaKeyword{postulate}\<%
\\
\>[4][@{}l@{\AgdaIndent{0}}]%
\>[6]\AgdaPostulate{\AgdaUnderscore{}}%
\>[2869I]\AgdaSymbol{:}\<%
\end{code}
  \begin{code}[inline]%
\>[.][@{}l@{}]\<[2869I]%
\>[8]\AgdaGeneralizable{x}\AgdaSpace{}%
\AgdaOperator{\AgdaFunction{≡}}\AgdaSpace{}%
\AgdaGeneralizable{x}\<%
\end{code}):
  \begin{code}[hide]%
\>[4]\AgdaFunction{[\ensuremath{\mkern1.5mu}]‐cong‐ax\ensuremath{{}_{\mathrm{1}}}}\AgdaSpace{}%
\AgdaSymbol{=}\<%
\end{code}
  \begin{enumerate}
  \item
    \label{ax:1}
    \begin{code}%
\>[4][@{}l@{\AgdaIndent{1}}]%
\>[6]\AgdaSymbol{(}%
\AgdaBound{[\ensuremath{\mkern1.5mu}]‐cong}\AgdaSpace{}%
\AgdaSymbol{:}\AgdaSpace{}%
\AgdaSymbol{\{}\AgdaSymbol{@0}\AgdaSpace{}%
\AgdaBound{A}\AgdaSpace{}%
\AgdaSymbol{:}\AgdaSpace{}%
\AgdaPrimitive{Type}\AgdaSpace{}%
\AgdaBound{a}\AgdaSymbol{\}}\AgdaSpace{}%
\AgdaSymbol{\{}\AgdaSymbol{@0}\AgdaSpace{}%
\AgdaBound{x}\AgdaSpace{}%
\AgdaBound{y}\AgdaSpace{}%
\AgdaSymbol{:}\AgdaSpace{}%
\AgdaBound{A}\AgdaSymbol{\}}\AgdaSpace{}%
\AgdaSymbol{→}\AgdaSpace{}%
\AgdaSymbol{@0}\AgdaSpace{}%
\AgdaBound{x}\AgdaSpace{}%
\AgdaOperator{\AgdaFunction{≡}}\AgdaSpace{}%
\AgdaBound{y}\AgdaSpace{}%
\AgdaSymbol{→}\AgdaSpace{}%
\AgdaOperator{\AgdaInductiveConstructor{[}}\AgdaSpace{}%
\AgdaBound{x}\AgdaSpace{}%
\AgdaOperator{\AgdaInductiveConstructor{]}}\AgdaSpace{}%
\AgdaOperator{\AgdaFunction{≡}}\AgdaSpace{}%
\AgdaOperator{\AgdaInductiveConstructor{[}}\AgdaSpace{}%
\AgdaBound{y}\AgdaSpace{}%
\AgdaOperator{\AgdaInductiveConstructor{]}}\AgdaSymbol{)}\AgdaSpace{}%
\AgdaFunction{×}\<%
\\
\>[6]\AgdaSymbol{(\{}\AgdaSymbol{@0}\AgdaSpace{}%
\AgdaBound{A}\AgdaSpace{}%
\AgdaSymbol{:}\AgdaSpace{}%
\AgdaPrimitive{Type}\AgdaSpace{}%
\AgdaBound{a}\AgdaSymbol{\}}\AgdaSpace{}%
\AgdaSymbol{\{}\AgdaSymbol{@0}\AgdaSpace{}%
\AgdaBound{x}\AgdaSpace{}%
\AgdaSymbol{:}\AgdaSpace{}%
\AgdaBound{A}\AgdaSymbol{\}}\AgdaSpace{}%
\AgdaSymbol{→}\AgdaSpace{}%
\AgdaBound{[\ensuremath{\mkern1.5mu}]‐cong}\AgdaSpace{}%
\AgdaSymbol{(}\AgdaInductiveConstructor{refl}\AgdaSpace{}%
\AgdaSymbol{\{}\AgdaArgument{x}\AgdaSpace{}%
\AgdaSymbol{=}\AgdaSpace{}%
\AgdaBound{x}\AgdaSymbol{\})}\AgdaSpace{}%
\AgdaOperator{\AgdaFunction{≡}}\AgdaSpace{}%
\AgdaInductiveConstructor{refl}\AgdaSymbol{)}\<%
\end{code}
    \begin{code}[hide]%
\>[4]\AgdaFunction{[\ensuremath{\mkern1.5mu}]‐cong‐ax\ensuremath{{}_{\mathrm{3}}}}\AgdaSpace{}%
\AgdaSymbol{=}\<%
\end{code}
  \item
    \label{ax:3}
    \begin{code}%
\>[4][@{}l@{\AgdaIndent{1}}]%
\>[6]\AgdaSymbol{(}%
\AgdaBound{[\ensuremath{\mkern1.5mu}]‐cong}\AgdaSpace{}%
\AgdaSymbol{:}\AgdaSpace{}%
\AgdaSymbol{\{}\AgdaBound{A}\AgdaSpace{}%
\AgdaSymbol{:}\AgdaSpace{}%
\AgdaPrimitive{Type}\AgdaSpace{}%
\AgdaBound{a}\AgdaSymbol{\}}\AgdaSpace{}%
\AgdaSymbol{\{}\AgdaBound{x}\AgdaSpace{}%
\AgdaBound{y}\AgdaSpace{}%
\AgdaSymbol{:}\AgdaSpace{}%
\AgdaBound{A}\AgdaSymbol{\}}\AgdaSpace{}%
\AgdaSymbol{→}\AgdaSpace{}%
\AgdaSymbol{@0}\AgdaSpace{}%
\AgdaBound{x}\AgdaSpace{}%
\AgdaOperator{\AgdaFunction{≡}}\AgdaSpace{}%
\AgdaBound{y}\AgdaSpace{}%
\AgdaSymbol{→}\AgdaSpace{}%
\AgdaOperator{\AgdaInductiveConstructor{[}}\AgdaSpace{}%
\AgdaBound{x}\AgdaSpace{}%
\AgdaOperator{\AgdaInductiveConstructor{]}}\AgdaSpace{}%
\AgdaOperator{\AgdaFunction{≡}}\AgdaSpace{}%
\AgdaOperator{\AgdaInductiveConstructor{[}}\AgdaSpace{}%
\AgdaBound{y}\AgdaSpace{}%
\AgdaOperator{\AgdaInductiveConstructor{]}}\AgdaSymbol{)}\AgdaSpace{}%
\AgdaFunction{×}\<%
\\
\>[6]\AgdaSymbol{(\{}\AgdaBound{A}\AgdaSpace{}%
\AgdaSymbol{:}\AgdaSpace{}%
\AgdaPrimitive{Type}\AgdaSpace{}%
\AgdaBound{a}\AgdaSymbol{\}}\AgdaSpace{}%
\AgdaSymbol{\{}\AgdaBound{x}\AgdaSpace{}%
\AgdaSymbol{:}\AgdaSpace{}%
\AgdaBound{A}\AgdaSymbol{\}}\AgdaSpace{}%
\AgdaSymbol{→}\AgdaSpace{}%
\AgdaBound{[\ensuremath{\mkern1.5mu}]‐cong}\AgdaSpace{}%
\AgdaSymbol{(}\AgdaInductiveConstructor{refl}\AgdaSpace{}%
\AgdaSymbol{\{}\AgdaArgument{x}\AgdaSpace{}%
\AgdaSymbol{=}\AgdaSpace{}%
\AgdaBound{x}\AgdaSymbol{\})}\AgdaSpace{}%
\AgdaOperator{\AgdaFunction{≡}}\AgdaSpace{}%
\AgdaInductiveConstructor{refl}\AgdaSymbol{)}\<%
\end{code}
  \end{enumerate}
  We also have the following variant, stated using
  \AgdaFunction{Stable‐equality\ensuremath{{}^{\mkern2mu\mathrm{E}}}}:
  \begin{code}[hide]%
\>[4]\AgdaFunction{[\ensuremath{\mkern1.5mu}]‐cong‐ax\ensuremath{{}_{\mathrm{5}}}}\AgdaSpace{}%
\AgdaSymbol{=}\<%
\end{code}
  \begin{enumerate}
    \setcounter{enumi}{2}
  \item
    \label{ax:5}
    \begin{code}%
\>[4][@{}l@{\AgdaIndent{1}}]%
\>[6]\AgdaSymbol{(}%
\AgdaBound{f\,}\AgdaSpace{}%
\AgdaSymbol{:}\AgdaSpace{}%
\AgdaSymbol{\{}\AgdaSymbol{@0}\AgdaSpace{}%
\AgdaBound{A}\AgdaSpace{}%
\AgdaSymbol{:}\AgdaSpace{}%
\AgdaPrimitive{Type}\AgdaSpace{}%
\AgdaBound{a}\AgdaSymbol{\}}\AgdaSpace{}%
\AgdaSymbol{→}\AgdaSpace{}%
\AgdaFunction{Stable‐equality\ensuremath{{}^{\mkern2mu\mathrm{E}}}}\AgdaSpace{}%
\AgdaSymbol{(}\AgdaDatatype{Erased}\AgdaSpace{}%
\AgdaBound{A}\AgdaSymbol{))}\AgdaSymbol{)}\AgdaSpace{}%
\AgdaFunction{×}\<%
\\
\>[6]\AgdaSymbol{(\{}\AgdaSymbol{@0}\AgdaSpace{}%
\AgdaBound{A}\AgdaSpace{}%
\AgdaSymbol{:}\AgdaSpace{}%
\AgdaPrimitive{Type}\AgdaSpace{}%
\AgdaBound{a}\AgdaSymbol{\}}\AgdaSpace{}%
\AgdaSymbol{\{}\AgdaBound{x}\AgdaSpace{}%
\AgdaSymbol{:}\AgdaSpace{}%
\AgdaDatatype{Erased}\AgdaSpace{}%
\AgdaBound{A}\AgdaSymbol{\}}\AgdaSpace{}%
\AgdaSymbol{→}\AgdaSpace{}%
\AgdaBound{f\,}\AgdaSpace{}%
\AgdaBound{x}\AgdaSpace{}%
\AgdaBound{x}\AgdaSpace{}%
\AgdaOperator{\AgdaInductiveConstructor{[}}\AgdaSpace{}%
\AgdaInductiveConstructor{refl}\AgdaSpace{}%
\AgdaOperator{\AgdaInductiveConstructor{]}}\AgdaSpace{}%
\AgdaOperator{\AgdaFunction{≡}}\AgdaSpace{}%
\AgdaInductiveConstructor{refl}\AgdaSymbol{)}\<%
\end{code}
  \end{enumerate}
  Then we have some definitions of the form ``certain functions are
  equivalences'':
  \begin{code}[hide]%
\>[4]\AgdaFunction{[\ensuremath{\mkern1.5mu}]‐cong\ensuremath{{}^{\mkern1mu\mathrm{-1}}}‐ax\ensuremath{{}_{\mathrm{1}}}}\AgdaSpace{}%
\AgdaSymbol{=}\<%
\end{code}
  \begin{enumerate}
    \setcounter{enumi}{3}
  \item
    \label{ax:7}
    \begin{code}%
\>[4][@{}l@{\AgdaIndent{1}}]%
\>[6]\AgdaSymbol{\{}\AgdaSymbol{@0}\AgdaSpace{}%
\AgdaBound{A}\AgdaSpace{}%
\AgdaSymbol{:}\AgdaSpace{}%
\AgdaPrimitive{Type}\AgdaSpace{}%
\AgdaBound{a}\AgdaSymbol{\}}\AgdaSpace{}%
\AgdaSymbol{\{}\AgdaSymbol{@0}\AgdaSpace{}%
\AgdaBound{x}\AgdaSpace{}%
\AgdaBound{y}\AgdaSpace{}%
\AgdaSymbol{:}\AgdaSpace{}%
\AgdaBound{A}\AgdaSymbol{\}}\AgdaSpace{}%
\AgdaSymbol{→}\<%
\\
\>[6]\AgdaFunction{Is‐equivalence}\AgdaSpace{}%
\AgdaSymbol{(λ}\AgdaSpace{}%
\AgdaSymbol{(}\AgdaBound{eq}\AgdaSpace{}%
\AgdaSymbol{:}\AgdaSpace{}%
\AgdaOperator{\AgdaInductiveConstructor{[}}\AgdaSpace{}%
\AgdaBound{x}\AgdaSpace{}%
\AgdaOperator{\AgdaInductiveConstructor{]}}\AgdaSpace{}%
\AgdaOperator{\AgdaFunction{≡}}\AgdaSpace{}%
\AgdaOperator{\AgdaInductiveConstructor{[}}\AgdaSpace{}%
\AgdaBound{y}\AgdaSpace{}%
\AgdaOperator{\AgdaInductiveConstructor{]}}\AgdaSymbol{)}\AgdaSpace{}%
\AgdaSymbol{→}\AgdaSpace{}%
\AgdaOperator{\AgdaInductiveConstructor{[}}\AgdaSpace{}%
\AgdaFunction{cong}\AgdaSpace{}%
\AgdaFunction{erased}\AgdaSpace{}%
\AgdaBound{eq}\AgdaSpace{}%
\AgdaOperator{\AgdaInductiveConstructor{]}}\AgdaSymbol{)}\<%
\end{code}
    \begin{code}[hide]%
\>[4]\AgdaFunction{[\ensuremath{\mkern1.5mu}]‐cong\ensuremath{{}^{\mkern1mu\mathrm{-1}}}‐ax\ensuremath{{}_{\mathrm{2}}}}\AgdaSpace{}%
\AgdaSymbol{=}\<%
\end{code}
  \item
    \label{ax:8}
    \begin{code}%
\>[4][@{}l@{\AgdaIndent{1}}]%
\>[6]\AgdaSymbol{\{}\AgdaSymbol{@0}\AgdaSpace{}%
\AgdaBound{A}\AgdaSpace{}%
\AgdaSymbol{:}\AgdaSpace{}%
\AgdaPrimitive{Type}\AgdaSpace{}%
\AgdaBound{a}\AgdaSymbol{\}}\AgdaSpace{}%
\AgdaSymbol{\{}\AgdaBound{x}\AgdaSpace{}%
\AgdaBound{y}\AgdaSpace{}%
\AgdaSymbol{:}\AgdaSpace{}%
\AgdaDatatype{Erased}\AgdaSpace{}%
\AgdaBound{A}\AgdaSymbol{\}}\AgdaSpace{}%
\AgdaSymbol{→}\AgdaSpace{}%
\AgdaFunction{Is‐equivalence}\AgdaSpace{}%
\AgdaSymbol{(λ}\AgdaSpace{}%
\AgdaSymbol{(}\AgdaBound{eq}\AgdaSpace{}%
\AgdaSymbol{:}\AgdaSpace{}%
\AgdaBound{x}\AgdaSpace{}%
\AgdaOperator{\AgdaFunction{≡}}\AgdaSpace{}%
\AgdaBound{y}\AgdaSymbol{)}\AgdaSpace{}%
\AgdaSymbol{→}\AgdaSpace{}%
\AgdaOperator{\AgdaInductiveConstructor{[}}\AgdaSpace{}%
\AgdaFunction{cong}\AgdaSpace{}%
\AgdaFunction{erased}\AgdaSpace{}%
\AgdaBound{eq}\AgdaSpace{}%
\AgdaOperator{\AgdaInductiveConstructor{]}}\AgdaSymbol{)}\<%
\end{code}
    \begin{code}[hide]%
\>[4]\AgdaFunction{Erased‐Separated\ensuremath{{}^{\mkern2mu\mathrm{E}}}‐ax}\AgdaSpace{}%
\AgdaSymbol{=}\<%
\end{code}
  \item
    \label{ax:9}
    \begin{code}%
\>[4][@{}l@{\AgdaIndent{1}}]%
\>[6]\AgdaSymbol{\{}\AgdaSymbol{@0}\AgdaSpace{}%
\AgdaBound{A}\AgdaSpace{}%
\AgdaSymbol{:}\AgdaSpace{}%
\AgdaPrimitive{Type}\AgdaSpace{}%
\AgdaBound{a}\AgdaSymbol{\}}\AgdaSpace{}%
\AgdaSymbol{→}\AgdaSpace{}%
\AgdaFunction{Separated\ensuremath{{}^{\mkern2mu\mathrm{E}}}}\AgdaSpace{}%
\AgdaSymbol{(}\AgdaDatatype{Erased}\AgdaSpace{}%
\AgdaBound{A}\AgdaSymbol{)}\<%
\end{code}
    \begin{code}[hide]%
\>[4]\AgdaFunction{Modal\ensuremath{{}^{\mkern2mu\mathrm{E}}}‐Erased×Modal\ensuremath{{}^{\mkern2mu\mathrm{E}}}→Separated\ensuremath{{}^{\mkern2mu\mathrm{E}}}‐ax}\AgdaSpace{}%
\AgdaSymbol{=}\<%
\end{code}
  \item
    \label{ax:9b}
    \begin{code}%
\>[4][@{}l@{\AgdaIndent{1}}]%
\>[6]\AgdaSymbol{(\{}\AgdaSymbol{@0}\AgdaSpace{}%
\AgdaBound{A}\AgdaSpace{}%
\AgdaSymbol{:}\AgdaSpace{}%
\AgdaPrimitive{Type}\AgdaSpace{}%
\AgdaBound{a}\AgdaSymbol{\}}\AgdaSpace{}%
\AgdaSymbol{→}\AgdaSpace{}%
\AgdaFunction{Modal\ensuremath{{}^{\mkern2mu\mathrm{E}}}}\AgdaSpace{}%
\AgdaSymbol{(}\AgdaDatatype{Erased}\AgdaSpace{}%
\AgdaBound{A}\AgdaSymbol{))}\AgdaSpace{}%
\AgdaOperator{\AgdaFunction{×}}\AgdaSpace{}%
\AgdaSymbol{(\{}\AgdaBound{A}\AgdaSpace{}%
\AgdaSymbol{:}\AgdaSpace{}%
\AgdaPrimitive{Type}\AgdaSpace{}%
\AgdaBound{a}\AgdaSymbol{\}}\AgdaSpace{}%
\AgdaSymbol{→}\AgdaSpace{}%
\AgdaFunction{Modal\ensuremath{{}^{\mkern2mu\mathrm{E}}}}\AgdaSpace{}%
\AgdaBound{A}\AgdaSpace{}%
\AgdaSymbol{→}\AgdaSpace{}%
\AgdaFunction{Separated\ensuremath{{}^{\mkern2mu\mathrm{E}}}}\AgdaSpace{}%
\AgdaBound{A}\AgdaSymbol{)}\<%
\end{code}
  \end{enumerate}
  The next batch are of the form
  ``\AgdaFunction{subst\ensuremath{{}^{\mkern2mu\mathrm{E}}}}/\AgdaFunction{J\ensuremath{{}^{\mkern2mu\mathrm{E}}}} can be defined'':
  \begin{code}[hide]%
\>[4]\AgdaFunction{Subst\ensuremath{{}^{\mkern2mu\mathrm{E}}}‐ax}\AgdaSpace{}%
\AgdaSymbol{=}\<%
\end{code}
  \begin{enumerate}
    \setcounter{enumi}{7}
  \item
    \label{ax:10}
    \begin{code}%
\>[4][@{}l@{\AgdaIndent{1}}]%
\>[6]\AgdaSymbol{(}%
\AgdaBound{subst\ensuremath{{}^{\mkern2mu\mathrm{E}}}}\AgdaSpace{}%
\AgdaSymbol{:}\AgdaSpace{}%
\>[3071I]\AgdaSymbol{\{}\AgdaSymbol{@0}\AgdaSpace{}%
\AgdaBound{A}\AgdaSpace{}%
\AgdaSymbol{:}\AgdaSpace{}%
\AgdaPrimitive{Type}\AgdaSpace{}%
\AgdaBound{a}\AgdaSymbol{\}}\AgdaSpace{}%
\AgdaSymbol{\{}\AgdaSymbol{@0}\AgdaSpace{}%
\AgdaBound{x}\AgdaSpace{}%
\AgdaBound{y}\AgdaSpace{}%
\AgdaSymbol{:}\AgdaSpace{}%
\AgdaBound{A}\AgdaSymbol{\}}\AgdaSpace{}%
\AgdaSymbol{(}\AgdaBound{P}\AgdaSpace{}%
\AgdaSymbol{:}\AgdaSpace{}%
\AgdaSymbol{@0}\AgdaSpace{}%
\AgdaBound{A}\AgdaSpace{}%
\AgdaSymbol{→}\AgdaSpace{}%
\AgdaPrimitive{Type}\AgdaSpace{}%
\AgdaBound{a}\AgdaSymbol{)}\AgdaSpace{}%
\AgdaSymbol{→}\<%
\\
\>[.][@{}l@{}]\<[3071I]%
\>[19]\AgdaSymbol{@0}\AgdaSpace{}%
\AgdaBound{x}\AgdaSpace{}%
\AgdaOperator{\AgdaFunction{≡}}\AgdaSpace{}%
\AgdaBound{y}\AgdaSpace{}%
\AgdaSymbol{→}\AgdaSpace{}%
\AgdaBound{P}\AgdaSpace{}%
\AgdaBound{x}\AgdaSpace{}%
\AgdaSymbol{→}\AgdaSpace{}%
\AgdaBound{P}\AgdaSpace{}%
\AgdaBound{y}\AgdaSymbol{)}\AgdaSpace{}%
\AgdaFunction{×}\<%
\\
\>[6]\AgdaSymbol{(\{}\AgdaSymbol{@0}\AgdaSpace{}%
\AgdaBound{A}\AgdaSpace{}%
\AgdaSymbol{:}\AgdaSpace{}%
\AgdaPrimitive{Type}\AgdaSpace{}%
\AgdaBound{a}\AgdaSymbol{\}}\AgdaSpace{}%
\AgdaSymbol{\{}\AgdaSymbol{@0}\AgdaSpace{}%
\AgdaBound{x}\AgdaSpace{}%
\AgdaSymbol{:}\AgdaSpace{}%
\AgdaBound{A}\AgdaSymbol{\}}\AgdaSpace{}%
\AgdaSymbol{\{}\AgdaBound{P}\AgdaSpace{}%
\AgdaSymbol{:}\AgdaSpace{}%
\AgdaSymbol{@0}\AgdaSpace{}%
\AgdaBound{A}\AgdaSpace{}%
\AgdaSymbol{→}\AgdaSpace{}%
\AgdaPrimitive{Type}\AgdaSpace{}%
\AgdaBound{a}\AgdaSymbol{\}}\AgdaSpace{}%
\AgdaSymbol{\{}\AgdaBound{p}\AgdaSpace{}%
\AgdaSymbol{:}\AgdaSpace{}%
\AgdaBound{P}\AgdaSpace{}%
\AgdaBound{x}\AgdaSymbol{\}}\AgdaSpace{}%
\AgdaSymbol{→}\AgdaSpace{}%
\AgdaBound{subst\ensuremath{{}^{\mkern2mu\mathrm{E}}}}\AgdaSpace{}%
\AgdaBound{P}\AgdaSpace{}%
\AgdaInductiveConstructor{refl}\AgdaSpace{}%
\AgdaBound{p}\AgdaSpace{}%
\AgdaOperator{\AgdaFunction{≡}}\AgdaSpace{}%
\AgdaBound{p}\AgdaSymbol{)}\<%
\end{code}
    \begin{code}[hide]%
\>[4]\AgdaFunction{J\ensuremath{{}^{\mkern2mu\mathrm{E}}}‐ax\ensuremath{{}_{\mathrm{2}}}}\AgdaSpace{}%
\AgdaSymbol{=}\<%
\end{code}
  \item
    \label{ax:12}
    \begin{code}%
\>[4][@{}l@{\AgdaIndent{1}}]%
\>[6]\AgdaSymbol{(}%
\AgdaBound{J\ensuremath{{}^{\mkern2mu\mathrm{E}}}}\AgdaSpace{}%
\AgdaSymbol{:}\AgdaSpace{}%
\>[3130I]\AgdaSymbol{\{}\AgdaSymbol{@0}\AgdaSpace{}%
\AgdaBound{A}\AgdaSpace{}%
\AgdaSymbol{:}\AgdaSpace{}%
\AgdaPrimitive{Type}\AgdaSpace{}%
\AgdaBound{a}\AgdaSymbol{\}}\AgdaSpace{}%
\AgdaSymbol{\{}\AgdaSymbol{@0}\AgdaSpace{}%
\AgdaBound{x}\AgdaSpace{}%
\AgdaBound{y}\AgdaSpace{}%
\AgdaSymbol{:}\AgdaSpace{}%
\AgdaBound{A}\AgdaSymbol{\}}\AgdaSpace{}%
\AgdaSymbol{(}\AgdaBound{P}\AgdaSpace{}%
\AgdaSymbol{:}\AgdaSpace{}%
\AgdaSymbol{\{}\AgdaSymbol{@0}\AgdaSpace{}%
\AgdaBound{y}\AgdaSpace{}%
\AgdaSymbol{:}\AgdaSpace{}%
\AgdaBound{A}\AgdaSymbol{\}}\AgdaSpace{}%
\AgdaSymbol{→}\AgdaSpace{}%
\AgdaSymbol{@0}\AgdaSpace{}%
\AgdaBound{x}\AgdaSpace{}%
\AgdaOperator{\AgdaFunction{≡}}\AgdaSpace{}%
\AgdaBound{y}\AgdaSpace{}%
\AgdaSymbol{→}\AgdaSpace{}%
\AgdaPrimitive{Type}\AgdaSpace{}%
\AgdaBound{a}\AgdaSymbol{)}\AgdaSpace{}%
\AgdaSymbol{→}\<%
\\
\>[.][@{}l@{}]\<[3130I]%
\>[15]\AgdaBound{P}\AgdaSpace{}%
\AgdaInductiveConstructor{refl}\AgdaSpace{}%
\AgdaSymbol{→}\AgdaSpace{}%
\AgdaSymbol{(}\AgdaSymbol{@0}\AgdaSpace{}%
\AgdaBound{eq}\AgdaSpace{}%
\AgdaSymbol{:}\AgdaSpace{}%
\AgdaBound{x}\AgdaSpace{}%
\AgdaOperator{\AgdaFunction{≡}}\AgdaSpace{}%
\AgdaBound{y}\AgdaSymbol{)}\AgdaSpace{}%
\AgdaSymbol{→}\AgdaSpace{}%
\AgdaBound{P}\AgdaSpace{}%
\AgdaBound{eq}\AgdaSymbol{)}\AgdaSpace{}%
\AgdaFunction{×}\<%
\\
\>[6]\AgdaSymbol{(}%
\>[3168I]\AgdaSymbol{\{}\AgdaSymbol{@0}\AgdaSpace{}%
\AgdaBound{A}\AgdaSpace{}%
\AgdaSymbol{:}\AgdaSpace{}%
\AgdaPrimitive{Type}\AgdaSpace{}%
\AgdaBound{a}\AgdaSymbol{\}}\AgdaSpace{}%
\AgdaSymbol{\{}\AgdaSymbol{@0}\AgdaSpace{}%
\AgdaBound{x}\AgdaSpace{}%
\AgdaSymbol{:}\AgdaSpace{}%
\AgdaBound{A}\AgdaSymbol{\}}\AgdaSpace{}%
\AgdaSymbol{\{}\AgdaBound{P}\AgdaSpace{}%
\AgdaSymbol{:}\AgdaSpace{}%
\AgdaSymbol{\{}\AgdaSymbol{@0}\AgdaSpace{}%
\AgdaBound{y}\AgdaSpace{}%
\AgdaSymbol{:}\AgdaSpace{}%
\AgdaBound{A}\AgdaSymbol{\}}\AgdaSpace{}%
\AgdaSymbol{→}\AgdaSpace{}%
\AgdaSymbol{@0}\AgdaSpace{}%
\AgdaBound{x}\AgdaSpace{}%
\AgdaOperator{\AgdaFunction{≡}}\AgdaSpace{}%
\AgdaBound{y}\AgdaSpace{}%
\AgdaSymbol{→}\AgdaSpace{}%
\AgdaPrimitive{Type}\AgdaSpace{}%
\AgdaBound{a}\AgdaSymbol{\}}\<%
\\
\>[.][@{}l@{}]\<[3168I]%
\>[8]\AgdaSymbol{(}\AgdaBound{r}\AgdaSpace{}%
\AgdaSymbol{:}\AgdaSpace{}%
\AgdaBound{P}\AgdaSpace{}%
\AgdaInductiveConstructor{refl}\AgdaSymbol{)}\AgdaSpace{}%
\AgdaSymbol{→}\AgdaSpace{}%
\AgdaBound{J\ensuremath{{}^{\mkern2mu\mathrm{E}}}}\AgdaSpace{}%
\AgdaBound{P}\AgdaSpace{}%
\AgdaBound{r}\AgdaSpace{}%
\AgdaInductiveConstructor{refl}\AgdaSpace{}%
\AgdaOperator{\AgdaFunction{≡}}\AgdaSpace{}%
\AgdaBound{r}\AgdaSymbol{)}\<%
\end{code}
  \end{enumerate}
  Finally we have a definition involving extendability:
  \begin{enumerate}
    \setcounter{enumi}{9}
    \begin{code}[hide]%
\>[4]\AgdaFunction{∞‐ext‐ax}\AgdaSpace{}%
\AgdaSymbol{=}\<%
\end{code}
  \item
    \label{ax:14}
    \begin{code}%
\>[4][@{}l@{\AgdaIndent{1}}]%
\>[6]\AgdaSymbol{\{}\AgdaBound{A}\AgdaSpace{}%
\AgdaBound{B}\AgdaSpace{}%
\AgdaSymbol{:}\AgdaSpace{}%
\AgdaPrimitive{Type}\AgdaSpace{}%
\AgdaBound{a}\AgdaSymbol{\}}\AgdaSpace{}%
\AgdaSymbol{→}\AgdaSpace{}%
\AgdaOperator{\AgdaFunction{Is‐∞‐extendable‐along‐[}}\AgdaSpace{}%
\AgdaSymbol{(λ}\AgdaSpace{}%
\AgdaBound{x}\AgdaSpace{}%
\AgdaSymbol{→}\AgdaSpace{}%
\AgdaOperator{\AgdaInductiveConstructor{[}}\AgdaSpace{}%
\AgdaBound{x}\AgdaSpace{}%
\AgdaOperator{\AgdaInductiveConstructor{]}}\AgdaSymbol{)}\AgdaSpace{}%
\AgdaOperator{\AgdaFunction{]}}\AgdaSpace{}%
\AgdaSymbol{(λ}\AgdaSpace{}%
\AgdaSymbol{(}\AgdaBound{\AgdaUnderscore{}}\AgdaSpace{}%
\AgdaSymbol{:}\AgdaSpace{}%
\AgdaDatatype{Erased}\AgdaSpace{}%
\AgdaBound{A}\AgdaSymbol{)}\AgdaSpace{}%
\AgdaSymbol{→}\AgdaSpace{}%
\AgdaDatatype{Erased}\AgdaSpace{}%
\AgdaBound{B}\AgdaSymbol{)}\<%
\end{code}
  \end{enumerate}
\end{theorem}

\begin{code}[hide]%
\>[2]\AgdaKeyword{postulate}\<%
\\
\>[2][@{}l@{\AgdaIndent{0}}]%
\>[4]\AgdaPostulate{[\ensuremath{\mkern1.5mu}]‐cong‐axiomatisation}\AgdaSpace{}%
\AgdaSymbol{:}\AgdaSpace{}%
\AgdaSymbol{(}\AgdaBound{a}\AgdaSpace{}%
\AgdaSymbol{:}\AgdaSpace{}%
\AgdaPostulate{Level}\AgdaSymbol{)}\AgdaSpace{}%
\AgdaSymbol{→}\AgdaSpace{}%
\AgdaPrimitive{Type}\AgdaSpace{}%
\AgdaSymbol{(}\AgdaPrimitive{lsuc}\AgdaSpace{}%
\AgdaBound{a}\AgdaSymbol{)}\<%
\\
\>[4]\AgdaPostulate{\AgdaUnderscore{}}%
\>[3231I]\AgdaSymbol{:}\<%
\end{code}
\newcommand{\boxcongax}{%
\begin{code}[inline]%
\>[.][@{}l@{}]\<[3231I]%
\>[6]\AgdaPostulate{[\ensuremath{\mkern1.5mu}]‐cong‐axiomatisation}\AgdaSpace{}%
\AgdaGeneralizable{a}\<%
\end{code}}

Let \boxcongax{} denote any of the logically equivalent definitions
from Theorem \ref{thm:equivalent-1}.
If one of these definitions is inhabited for a given universe level,
then it is also inhabited for smaller universe levels, and all
definitions are inhabited in erased contexts.
Another definition is given in the following theorem:
\begin{theorem}
  \label{thm:equivalent-2}
  \boxcongax{} is logically equivalent to
  \begin{code}[hide]%
\>[2]\AgdaFunction{Is‐modality‐ax}\AgdaSpace{}%
\AgdaSymbol{:}\AgdaSpace{}%
\AgdaSymbol{(}\AgdaBound{a}\AgdaSpace{}%
\AgdaSymbol{:}\AgdaSpace{}%
\AgdaPostulate{Level}\AgdaSymbol{)}\AgdaSpace{}%
\AgdaSymbol{→}\AgdaSpace{}%
\AgdaPrimitive{Type}\AgdaSpace{}%
\AgdaSymbol{(}\AgdaPrimitive{lsuc}\AgdaSpace{}%
\AgdaBound{a}\AgdaSymbol{)}\<%
\\
\>[2]\AgdaFunction{Is‐modality‐ax}\AgdaSpace{}%
\AgdaBound{a}\AgdaSpace{}%
\AgdaSymbol{=}\<%
\end{code}
  \mbox{\begin{code}[inline]%
\>[2][@{}l@{\AgdaIndent{1}}]%
\>[4]\AgdaPostulate{Is‐modality}\AgdaSpace{}%
\AgdaBound{a}\AgdaSpace{}%
\AgdaSymbol{(λ}\AgdaSpace{}%
\AgdaBound{A}\AgdaSpace{}%
\AgdaSymbol{→}\AgdaSpace{}%
\AgdaDatatype{Erased}\AgdaSpace{}%
\AgdaBound{A}\AgdaSymbol{)}\<%
\end{code}}
  \begin{code}[inline]%
\>[4][@{}l@{\AgdaIndent{1}}]%
\>[6]\AgdaSymbol{(λ}\AgdaSpace{}%
\AgdaBound{x}\AgdaSpace{}%
\AgdaSymbol{→}\AgdaSpace{}%
\AgdaOperator{\AgdaInductiveConstructor{[}}\AgdaSpace{}%
\AgdaBound{x}\AgdaSpace{}%
\AgdaOperator{\AgdaInductiveConstructor{]}}\AgdaSymbol{)}\<%
\end{code}.
  In the presence of funext and propext this type is a proposition,
  and thus equivalent to the other definitions.
\end{theorem}
\begin{proof}[Proof sketch]
  Let us start with the last part.
  Using funext one can prove that, with \AgdaField{Modal} fixed, the
  remaining parts of the type form a proposition.
  Thus it suffices to prove that any two implementations of
  \AgdaField{Modal} are equal.
  In the presence of funext and propext this follows because
  \mbox{$\AgdaField{Modal} \AgdaBound{A}$} is a proposition that is
  logically equivalent to
  \begin{code}[hide]%
\>[2]\AgdaKeyword{postulate}\<%
\\
\>[2][@{}l@{\AgdaIndent{0}}]%
\>[4]\AgdaPostulate{\AgdaUnderscore{}}%
\>[3254I]\AgdaSymbol{:}\<%
\end{code}
  \begin{code}[inline]%
\>[.][@{}l@{}]\<[3254I]%
\>[6]\AgdaFunction{Is‐equivalence}\AgdaSpace{}%
\AgdaSymbol{(λ}\AgdaSpace{}%
\AgdaSymbol{(}\AgdaBound{x}\AgdaSpace{}%
\AgdaSymbol{:}\AgdaSpace{}%
\AgdaGeneralizable{A}\AgdaSymbol{)}\AgdaSpace{}%
\AgdaSymbol{→}\AgdaSpace{}%
\AgdaOperator{\AgdaInductiveConstructor{[}}\AgdaSpace{}%
\AgdaBound{x}\AgdaSpace{}%
\AgdaOperator{\AgdaInductiveConstructor{]}}\AgdaSymbol{)}\<%
\end{code},
  see §~\ref{sec:modalities}.

  If \AgdaDatatype{Erased} and \boxop{} (η-expanded) form a modality,
  then \boxcong{} can be defined (see §~\ref{sec:modalities}).

  For the other direction we can proceed in the following way:
  Let \AgdaField{Modal} be \AgdaFunction{Modal\ensuremath{{}^{\mkern2mu\mathrm{E}}}}.
  Given funext this definition is propositional.
  \ErasedA{} is modal in the presence of \boxcong{}, see
  Theorem~\ref{thm:equivalent-1}.
  \AgdaFunction{Modal\ensuremath{{}^{\mkern2mu\mathrm{E}}}} also respects equivalences: this can be proved
  using \boxcong{}.
  Finally ∞-extendability holds (this part of the proof is
  omitted).
\end{proof}

\section{The Erasure Modality}
\label{sec:erasure-modality}

Let us call the modality defined in (the proof of)
Theorem~\ref{thm:equivalent-2} the \emph{erasure modality}.
This modality, which is defined under the assumption that \boxcongax{}
holds, is \emph{left exact}, and thereby satisfies all properties that
left exact modalities satisfy \citep{rijke-et-al-2020}.

\begin{code}[hide]%
\>[2]\AgdaKeyword{module}\AgdaSpace{}%
\AgdaModule{\AgdaUnderscore{}}\AgdaSpace{}%
\AgdaSymbol{(}\AgdaBound{◯}\AgdaSpace{}%
\AgdaSymbol{:}\AgdaSpace{}%
\AgdaPrimitive{Type}\AgdaSpace{}%
\AgdaGeneralizable{a}\AgdaSpace{}%
\AgdaSymbol{→}\AgdaSpace{}%
\AgdaPrimitive{Type}\AgdaSpace{}%
\AgdaGeneralizable{a}\AgdaSymbol{)}\AgdaSpace{}%
\AgdaKeyword{where}\<%
\end{code}

There are also properties that hold for \Erased{} that do not hold for
every left exact modality.
\begin{code}[hide]%
\>[2][@{}l@{\AgdaIndent{1}}]%
\>[4]\AgdaKeyword{postulate}\<%
\\
\>[4][@{}l@{\AgdaIndent{0}}]%
\>[6]\AgdaPostulate{\AgdaUnderscore{}}%
\>[3272I]\AgdaSymbol{:}\<%
\end{code}
\mbox{\begin{code}[inline]%
\>[.][@{}l@{}]\<[3272I]%
\>[8]\AgdaFunction{Modal\ensuremath{{}^{\mkern2mu\mathrm{E}}}}\AgdaSpace{}%
\AgdaFunction{⊥}\<%
\end{code}}
holds (also in the absence of \boxcong{}), but the empty type is not
modal for every left exact (or even ``topological''
\citep{rijke-et-al-2020}) modality: it is not modal for the (left
exact and topological) zero modality, for which the modal operator
\AgdaBound{◯} maps every type to a unit type.
For every modality for which the empty type is modal one can prove
properties like the following ones:\MCbreak{}
\begin{code}[hide]%
\>[4]\AgdaKeyword{postulate}\<%
\\
\>[4][@{}l@{\AgdaIndent{0}}]%
\>[6]\AgdaPostulate{Separated}\AgdaSpace{}%
\AgdaSymbol{:}\AgdaSpace{}%
\AgdaPrimitive{Type}\AgdaSpace{}%
\AgdaBound{a}\AgdaSpace{}%
\AgdaSymbol{→}\AgdaSpace{}%
\AgdaPrimitive{Type}\AgdaSpace{}%
\AgdaBound{a}\<%
\end{code}
\begin{minipage}[t]{0.57\linewidth}
\begin{AgdaMultiCode}
  \begin{code}[hide]%
\>[4]\AgdaKeyword{postulate}\<%
\\
\>[4][@{}l@{\AgdaIndent{0}}]%
\>[6]\AgdaPostulate{\AgdaUnderscore{}}%
\>[3280I]\AgdaSymbol{:}\<%
\end{code}
  \begin{code}%
\>[3280I][@{}l@{\AgdaIndent{1}}]%
\>[10]\AgdaPostulate{Separated}\AgdaSpace{}%
\AgdaGeneralizable{A}\AgdaSpace{}%
\AgdaSymbol{→}\AgdaSpace{}%
\AgdaPostulate{Separated}\AgdaSpace{}%
\AgdaGeneralizable{B}\AgdaSpace{}%
\AgdaSymbol{→}\AgdaSpace{}%
\AgdaPostulate{Separated}\AgdaSpace{}%
\AgdaSymbol{(}\AgdaGeneralizable{A}\AgdaSpace{}%
\AgdaOperator{\AgdaDatatype{⊎}}\AgdaSpace{}%
\AgdaGeneralizable{B}\AgdaSymbol{)}\<%
\end{code}
  \begin{code}[hide]%
\>[4]\AgdaKeyword{postulate}\<%
\\
\>[4][@{}l@{\AgdaIndent{0}}]%
\>[6]\AgdaPostulate{\AgdaUnderscore{}}%
\>[3290I]\AgdaSymbol{:}\<%
\end{code}
  \begin{code}%
\>[3290I][@{}l@{\AgdaIndent{1}}]%
\>[10]\AgdaPostulate{Separated}\AgdaSpace{}%
\AgdaGeneralizable{A}\AgdaSpace{}%
\AgdaSymbol{→}\AgdaSpace{}%
\AgdaPostulate{Separated}\AgdaSpace{}%
\AgdaSymbol{(}\AgdaDatatype{List}\AgdaSpace{}%
\AgdaGeneralizable{A}\AgdaSymbol{)}\<%
\end{code}
\end{AgdaMultiCode}
\end{minipage}
\begin{minipage}[t]{0.425\linewidth}
  \begin{code}[hide]%
\>[4]\AgdaKeyword{postulate}\<%
\\
\>[4][@{}l@{\AgdaIndent{0}}]%
\>[6]\AgdaPostulate{\AgdaUnderscore{}}%
\>[3296I]\AgdaSymbol{:}\<%
\end{code}
  \begin{code}%
\>[3296I][@{}l@{\AgdaIndent{1}}]%
\>[10]\AgdaFunction{Decidable‐equality}\AgdaSpace{}%
\AgdaGeneralizable{A}\AgdaSpace{}%
\AgdaSymbol{→}\AgdaSpace{}%
\AgdaPostulate{Separated}\AgdaSpace{}%
\AgdaGeneralizable{A}\<%
\end{code}
\end{minipage}
This implies, for instance, that the type of lists of natural numbers
is separated.

As another example of a property that holds for \Erased{} that does
not hold for every left exact modality, take
\begin{code}[hide]%
\>[4]\AgdaKeyword{postulate}\<%
\\
\>[4][@{}l@{\AgdaIndent{0}}]%
\>[6]\AgdaPostulate{\AgdaUnderscore{}}%
\>[3301I]\AgdaSymbol{:}\<%
\end{code}
\begin{code}[inline]%
\>[.][@{}l@{}]\<[3301I]%
\>[8]\AgdaSymbol{\{}\AgdaSymbol{@0}\AgdaSpace{}%
\AgdaBound{A}\AgdaSpace{}%
\AgdaSymbol{:}\AgdaSpace{}%
\AgdaPrimitive{Type}\AgdaSpace{}%
\AgdaBound{a}\AgdaSymbol{\}}\AgdaSpace{}%
\AgdaSymbol{→}\AgdaSpace{}%
\AgdaDatatype{Erased}\AgdaSpace{}%
\AgdaSymbol{(}\AgdaFunction{Modal\ensuremath{{}^{\mkern2mu\mathrm{E}}}}\AgdaSpace{}%
\AgdaBound{A}\AgdaSymbol{)}\<%
\end{code},
which can be proved without the use of \boxcong{}.
The property \VeryModal{} holds for the (left exact and topological)
closed modality associated to a proposition \AgdaBound{A}
\citep{rijke-et-al-2020} – which can be defined in Cubical Agda – if
and only if
\begin{code}[hide]%
\>[4]\AgdaKeyword{postulate}\<%
\\
\>[4][@{}l@{\AgdaIndent{0}}]%
\>[6]\AgdaPostulate{\AgdaUnderscore{}}%
\>[3310I]\AgdaSymbol{:}\<%
\end{code}
\begin{code}[inline*]%
\>[3310I][@{}l@{\AgdaIndent{1}}]%
\>[10]\AgdaGeneralizable{A}\AgdaSpace{}%
\AgdaOperator{\AgdaDatatype{⊎}}\AgdaSpace{}%
\AgdaOperator{\AgdaFunction{¬}}\AgdaSpace{}%
\AgdaGeneralizable{A}\<%
\end{code}
holds, and it holds for the closed modalities associated to every
proposition
\begin{code}[hide]%
\>[4]\AgdaKeyword{postulate}\<%
\\
\>[4][@{}l@{\AgdaIndent{0}}]%
\>[6]\AgdaPostulate{\AgdaUnderscore{}}\AgdaSpace{}%
\AgdaSymbol{:}%
\>[3315I]\AgdaSymbol{(}\<%
\end{code}
\begin{code}[inline*]%
\>[3315I][@{}l@{\AgdaIndent{1}}]%
\>[12]\AgdaBound{A}\AgdaSpace{}%
\AgdaSymbol{:}\AgdaSpace{}%
\AgdaPrimitive{Type}\AgdaSpace{}%
\AgdaBound{a}\<%
\end{code}
\begin{code}[hide]%
\>[.][@{}l@{}]\<[3315I]%
\>[10]\AgdaSymbol{)}\AgdaSpace{}%
\AgdaSymbol{→}\AgdaSpace{}%
\AgdaPrimitive{Type}\<%
\end{code}
if and only if excluded middle holds for \AgdaBound{a}.\footnote{This
  observation related to closed modalities is due to Christian Sattler
  and David Wärn.}

For every modality one can define a ``map'' function
\begin{code}[hide]%
\>[4]\AgdaKeyword{postulate}\<%
\end{code}
\begin{code}[inline]%
\>[4][@{}l@{\AgdaIndent{1}}]%
\>[6]\AgdaPostulate{◯‐map}\AgdaSpace{}%
\AgdaSymbol{:}\AgdaSpace{}%
\AgdaSymbol{(}\AgdaGeneralizable{A}\AgdaSpace{}%
\AgdaSymbol{→}\AgdaSpace{}%
\AgdaGeneralizable{B}\AgdaSymbol{)}\AgdaSpace{}%
\AgdaSymbol{→}\AgdaSpace{}%
\AgdaBound{◯}\AgdaSpace{}%
\AgdaGeneralizable{A}\AgdaSpace{}%
\AgdaSymbol{→}\AgdaSpace{}%
\AgdaBound{◯}\AgdaSpace{}%
\AgdaGeneralizable{B}\<%
\end{code}.
For any (necessarily left exact) modality that satisfies \VeryModal{}
one can prove the following logical equivalences (and given funext
they can be strengthened to equivalences):\MCbreak{}
\begin{minipage}[t]{0.35\linewidth}
\begin{AgdaMultiCode}
\begin{code}[hide]%
\>[4]\AgdaKeyword{postulate}\<%
\\
\>[4][@{}l@{\AgdaIndent{0}}]%
\>[6]\AgdaPostulate{\AgdaUnderscore{}}%
\>[3331I]\AgdaSymbol{:}\<%
\end{code}
\begin{code}%
\>[3331I][@{}l@{\AgdaIndent{1}}]%
\>[10]\AgdaBound{◯}\AgdaSpace{}%
\AgdaSymbol{(}\AgdaGeneralizable{A}\AgdaSpace{}%
\AgdaSymbol{→}\AgdaSpace{}%
\AgdaGeneralizable{B}\AgdaSymbol{)}%
\>[34]\AgdaOperator{\AgdaRecord{⇔}}\AgdaSpace{}%
\AgdaSymbol{(}\AgdaBound{◯}\AgdaSpace{}%
\AgdaGeneralizable{A}\AgdaSpace{}%
\AgdaSymbol{→}\AgdaSpace{}%
\AgdaBound{◯}\AgdaSpace{}%
\AgdaGeneralizable{B}\AgdaSymbol{)}\<%
\end{code}
\begin{code}[hide]%
\>[4]\AgdaKeyword{postulate}\<%
\\
\>[4][@{}l@{\AgdaIndent{0}}]%
\>[6]\AgdaPostulate{\AgdaUnderscore{}}%
\>[3340I]\AgdaSymbol{:}\<%
\end{code}
\begin{code}%
\>[3340I][@{}l@{\AgdaIndent{1}}]%
\>[10]\AgdaBound{◯}\AgdaSpace{}%
\AgdaSymbol{(}\AgdaGeneralizable{A}\AgdaSpace{}%
\AgdaOperator{\AgdaPostulate{≃}}\AgdaSpace{}%
\AgdaGeneralizable{B}\AgdaSymbol{)}%
\>[34]\AgdaOperator{\AgdaRecord{⇔}}\AgdaSpace{}%
\AgdaSymbol{(}\AgdaBound{◯}\AgdaSpace{}%
\AgdaGeneralizable{A}\AgdaSpace{}%
\AgdaOperator{\AgdaPostulate{≃}}\AgdaSpace{}%
\AgdaBound{◯}\AgdaSpace{}%
\AgdaGeneralizable{B}\AgdaSymbol{)}\<%
\end{code}
\end{AgdaMultiCode}
\end{minipage}
\begin{minipage}[t]{0.6\linewidth}
\begin{AgdaMultiCode}
\begin{code}[hide]%
\>[4]\AgdaKeyword{postulate}\<%
\\
\>[4][@{}l@{\AgdaIndent{0}}]%
\>[6]\AgdaPostulate{\AgdaUnderscore{}}%
\>[3349I]\AgdaSymbol{:}\<%
\end{code}
\begin{code}%
\>[3349I][@{}l@{\AgdaIndent{1}}]%
\>[10]\AgdaBound{◯}\AgdaSpace{}%
\AgdaSymbol{(}\AgdaFunction{Is‐equivalence}\AgdaSpace{}%
\AgdaGeneralizable{f\,}\AgdaSymbol{)}%
\>[34]\AgdaOperator{\AgdaRecord{⇔}}\AgdaSpace{}%
\AgdaFunction{Is‐equivalence}\AgdaSpace{}%
\AgdaSymbol{(}\AgdaPostulate{◯‐map}\AgdaSpace{}%
\AgdaGeneralizable{f\,}\AgdaSymbol{)}\<%
\end{code}
\begin{code}[hide]%
\>[4]\AgdaKeyword{postulate}\<%
\\
\>[4][@{}l@{\AgdaIndent{0}}]%
\>[6]\AgdaPostulate{\AgdaUnderscore{}}%
\>[3355I]\AgdaSymbol{:}\<%
\end{code}
\begin{code}%
\>[3355I][@{}l@{\AgdaIndent{1}}]%
\>[10]\AgdaBound{◯}\AgdaSpace{}%
\AgdaSymbol{(}\AgdaFunction{H‐level}\AgdaSpace{}%
\AgdaGeneralizable{n}\AgdaSpace{}%
\AgdaGeneralizable{A}\AgdaSymbol{)}%
\>[34]\AgdaOperator{\AgdaRecord{⇔}}\AgdaSpace{}%
\AgdaFunction{H‐level}\AgdaSpace{}%
\AgdaGeneralizable{n}\AgdaSpace{}%
\AgdaSymbol{(}\AgdaBound{◯}\AgdaSpace{}%
\AgdaGeneralizable{A}\AgdaSymbol{)}\<%
\end{code}
\end{AgdaMultiCode}
\end{minipage}\\
\AgdaFunction{H‐level} generalises \AgdaFunction{Contractible},
\AgdaFunction{Is‐prop} and \AgdaFunction{Is‐set}, so we get that
\ErasedA{} is a proposition (or set) if there is an erased proof
showing that \AgdaBound{A} is a proposition (set).
We also get that an erased equivalence between \AgdaBound{A} and
\AgdaBound{B} can be turned into a non-erased equivalence between
\ErasedA{} and
\begin{code}[hide]%
\>[2]\AgdaKeyword{postulate}\<%
\\
\>[2][@{}l@{\AgdaIndent{0}}]%
\>[4]\AgdaPostulate{\AgdaUnderscore{}}%
\>[3363I]\AgdaSymbol{:}\<%
\end{code}
\begin{code}[inline]%
\>[3363I][@{}l@{\AgdaIndent{1}}]%
\>[8]\AgdaDatatype{Erased}\AgdaSpace{}%
\AgdaGeneralizable{B}\<%
\end{code}.

\section{Box-Cong Can Be Defined Using Function Extensionality}
\label{sec:box-cong-from-funext}

Let us now see how one can use function extensionality to define
\boxcong{}.
First we have a lemma:
\begin{lemma}
  \label{lem:uniquely-eliminating}
  If function extensionality holds, then a certain form of
  precomposition with \boxop{} is an equivalence (the notation
  \AgdaSymbol{@ω} ensures that the argument is non-erased):
  \begin{code}[hide]%
\>[2]\AgdaKeyword{postulate}\<%
\\
\>[2][@{}l@{\AgdaIndent{0}}]%
\>[4]\AgdaPostulate{\AgdaUnderscore{}}%
\>[3365I]\AgdaSymbol{:}\<%
\end{code}
  \setlength{\belowdisplayskip}{0pt}%
  \begin{code}%
\>[.][@{}l@{}]\<[3365I]%
\>[6]\AgdaSymbol{\{}\AgdaBound{A}\AgdaSpace{}%
\AgdaSymbol{:}\AgdaSpace{}%
\AgdaPrimitive{Type}\AgdaSpace{}%
\AgdaGeneralizable{a}\AgdaSymbol{\}}\AgdaSpace{}%
\AgdaSymbol{\{}\AgdaSymbol{@0}\AgdaSpace{}%
\AgdaBound{P}\AgdaSpace{}%
\AgdaSymbol{:}\AgdaSpace{}%
\AgdaDatatype{Erased}\AgdaSpace{}%
\AgdaBound{A}\AgdaSpace{}%
\AgdaSymbol{→}\AgdaSpace{}%
\AgdaPrimitive{Type}\AgdaSpace{}%
\AgdaGeneralizable{p}\AgdaSymbol{\}}\AgdaSpace{}%
\AgdaSymbol{→}\AgdaSpace{}%
\AgdaPostulate{Funext}\AgdaSpace{}%
\AgdaGeneralizable{a}\AgdaSpace{}%
\AgdaGeneralizable{p}\AgdaSpace{}%
\AgdaSymbol{→}\<%
\\
\>[6]\AgdaFunction{Is‐equivalence}\AgdaSpace{}%
\AgdaSymbol{(λ}\AgdaSpace{}%
\AgdaSymbol{(}\AgdaBound{f\,}\AgdaSpace{}%
\AgdaSymbol{:}\AgdaSpace{}%
\AgdaSymbol{(}\AgdaBound{x}\AgdaSpace{}%
\AgdaSymbol{:}\AgdaSpace{}%
\AgdaDatatype{Erased}\AgdaSpace{}%
\AgdaBound{A}\AgdaSymbol{)}\AgdaSpace{}%
\AgdaSymbol{→}\AgdaSpace{}%
\AgdaDatatype{Erased}\AgdaSpace{}%
\AgdaSymbol{(}\AgdaBound{P}\AgdaSpace{}%
\AgdaBound{x}\AgdaSymbol{))}\AgdaSpace{}%
\AgdaSymbol{→}\AgdaSpace{}%
\AgdaSymbol{λ}\AgdaSpace{}%
\AgdaSymbol{(}\AgdaSymbol{@ω}\AgdaSpace{}%
\AgdaBound{x}\AgdaSpace{}%
\AgdaSymbol{:}\AgdaSpace{}%
\AgdaBound{A}\AgdaSymbol{)}\AgdaSpace{}%
\AgdaSymbol{→}\AgdaSpace{}%
\AgdaBound{f\,}\AgdaSpace{}%
\AgdaOperator{\AgdaInductiveConstructor{[}}\AgdaSpace{}%
\AgdaBound{x}\AgdaSpace{}%
\AgdaOperator{\AgdaInductiveConstructor{]}}\AgdaSymbol{)}\<%
\end{code}
\end{lemma}
\noindent
\citeauthor{rijke-et-al-2020} use (roughly) ``a function \AgdaBound{◯}
and a function \AgdaBound{η} that satisfy the conclusion of
Lemma~\ref{lem:uniquely-eliminating}'' as one of their definitions of
``modality'' \citep[Definition~1.2]{rijke-et-al-2020}.
If \Erased{} is defined as a record type with η-equality, then one can
prove Lemma~\ref{lem:uniquely-eliminating} without the assumption of
function extensionality.
However, note that Theorem~\ref{thm:no-box-cong} (``\boxcong{} cannot
be defined'') holds also if a variant of $\EName$ with η-equality is
used (see the accompanying code), and that – as discussed in
§~\ref{sec:modalities} – the definition of modality used in this text
is aimed at being usable in the absence of function extensionality.

\citeauthor{rijke-et-al-2020} prove that modal types are separated for
any \emph{reflective subuniverse}
\citep[Lemma~1.25]{rijke-et-al-2020}.
(A reflective subuniverse for which \AgdaField{Modal} is closed under
Σ is a modality.
Note that \AgdaFunction{Modal\ensuremath{{}^{\mkern2mu\mathrm{E}}}} is closed under Σ also in the absence
of \boxcong{}.)
My proof below is based on their lemma and the code that accompanies
their paper:
\begin{theorem}
  \label{thm:box-cong-from-funext}
  The following type is inhabited:
  \begin{code}[hide]%
\>[2]\AgdaKeyword{postulate}\<%
\\
\>[2][@{}l@{\AgdaIndent{0}}]%
\>[4]\AgdaPostulate{\AgdaUnderscore{}}%
\>[3404I]\AgdaSymbol{:}\<%
\end{code}
  \begin{code}[inline]%
\>[.][@{}l@{}]\<[3404I]%
\>[6]\AgdaPostulate{Funext}\AgdaSpace{}%
\AgdaGeneralizable{a}\AgdaSpace{}%
\AgdaGeneralizable{a}\AgdaSpace{}%
\AgdaSymbol{→}\AgdaSpace{}%
\AgdaPostulate{[\ensuremath{\mkern1.5mu}]‐cong‐axiomatisation}\AgdaSpace{}%
\AgdaGeneralizable{a}\<%
\end{code}.
\end{theorem}
\begin{proof}[Proof sketch]
  One can prove that variant~\ref{ax:5} from
  Theorem~\ref{thm:equivalent-1} is inhabited.
  Let us first define a function with the type
  \begin{code}[hide]%
\>[2]\AgdaKeyword{postulate}\<%
\\
\>[2][@{}l@{\AgdaIndent{0}}]%
\>[4]\AgdaPostulate{\AgdaUnderscore{}}%
\>[3410I]\AgdaSymbol{:}\<%
\end{code}
  \begin{code}[inline]%
\>[.][@{}l@{}]\<[3410I]%
\>[6]\AgdaSymbol{\{}\AgdaSymbol{@0}\AgdaSpace{}%
\AgdaBound{A}\AgdaSpace{}%
\AgdaSymbol{:}\AgdaSpace{}%
\AgdaPrimitive{Type}\AgdaSpace{}%
\AgdaGeneralizable{a}\AgdaSymbol{\}}\AgdaSpace{}%
\AgdaSymbol{→}\AgdaSpace{}%
\AgdaFunction{Stable‐equality\ensuremath{{}^{\mkern2mu\mathrm{E}}}}\AgdaSpace{}%
\AgdaSymbol{(}\AgdaDatatype{Erased}\AgdaSpace{}%
\AgdaBound{A}\AgdaSymbol{)}\<%
\end{code}.
  We are given two values
  \begin{code}[hide]%
\>[2]\AgdaKeyword{postulate}\<%
\\
\>[2][@{}l@{\AgdaIndent{0}}]%
\>[4]\AgdaPostulate{\AgdaUnderscore{}}%
\>[3419I]\AgdaSymbol{:}\<%
\end{code}
  \begin{code}[inline*]%
\>[3419I][@{}l@{\AgdaIndent{1}}]%
\>[8]\AgdaGeneralizable{x}\<%
\end{code}
  and
  \begin{code}[hide]%
\>[.][@{}l@{}]\<[3419I]%
\>[6]\AgdaOperator{\AgdaFunction{≡}}\<%
\end{code}
  \begin{code}[inline*]%
\>[6][@{}l@{\AgdaIndent{1}}]%
\>[8]\AgdaGeneralizable{y}\<%
\end{code}
  of type \ErasedA{}, as well as a proof of
  \begin{code}[hide]%
\>[2]\AgdaKeyword{postulate}\<%
\\
\>[2][@{}l@{\AgdaIndent{0}}]%
\>[4]\AgdaPostulate{\AgdaUnderscore{}}%
\>[3420I]\AgdaSymbol{:}\<%
\end{code}
  \mbox{\begin{code}[inline]%
\>[3420I][@{}l@{\AgdaIndent{1}}]%
\>[8]\AgdaDatatype{Erased}\AgdaSpace{}%
\AgdaSymbol{(}\AgdaGeneralizable{x}\AgdaSpace{}%
\AgdaOperator{\AgdaFunction{≡}}\AgdaSpace{}%
\AgdaGeneralizable{y}\AgdaSymbol{)}\<%
\end{code}},
  and the goal is to prove
  \begin{code}[hide]%
\>[2]\AgdaKeyword{postulate}\<%
\\
\>[2][@{}l@{\AgdaIndent{0}}]%
\>[4]\AgdaPostulate{\AgdaUnderscore{}}%
\>[3424I]\AgdaSymbol{:}\<%
\end{code}
  \mbox{\begin{code}[inline]%
\>[3424I][@{}l@{\AgdaIndent{1}}]%
\>[8]\AgdaGeneralizable{x}\AgdaSpace{}%
\AgdaOperator{\AgdaFunction{≡}}\AgdaSpace{}%
\AgdaGeneralizable{y}\<%
\end{code}}.
  This follows if we can prove that
  \begin{code}[hide]%
\>[2]\AgdaKeyword{postulate}\<%
\\
\>[2][@{}l@{\AgdaIndent{0}}]%
\>[4]\AgdaPostulate{\AgdaUnderscore{}}\AgdaSpace{}%
\AgdaSymbol{:}%
\>[3428I]\AgdaOperator{\AgdaFunction{\AgdaUnderscore{}≡\AgdaUnderscore{}}}%
\>[3429I]\AgdaSymbol{\{}\AgdaArgument{A}\AgdaSpace{}%
\AgdaSymbol{=}\AgdaSpace{}%
\AgdaSymbol{∀}\AgdaSpace{}%
\AgdaSymbol{(}\AgdaSymbol{@ω}\AgdaSpace{}%
\AgdaBound{\AgdaUnderscore{}}\AgdaSymbol{)}\AgdaSpace{}%
\AgdaSymbol{→}\AgdaSpace{}%
\AgdaSymbol{\AgdaUnderscore{}\}}\AgdaSpace{}%
\AgdaSymbol{(}\<%
\end{code}
  \begin{code}[inline*]%
\>[3429I][@{}l@{\AgdaIndent{1}}]%
\>[14]\AgdaSymbol{λ}\AgdaSpace{}%
\AgdaSymbol{(}\AgdaBound{\AgdaUnderscore{}}\AgdaSpace{}%
\AgdaSymbol{:}\AgdaSpace{}%
\AgdaDatatype{Erased}\AgdaSpace{}%
\AgdaSymbol{(}\AgdaGeneralizable{x}\AgdaSpace{}%
\AgdaOperator{\AgdaFunction{≡}}\AgdaSpace{}%
\AgdaGeneralizable{y}\AgdaSymbol{))}\AgdaSpace{}%
\AgdaSymbol{→}\AgdaSpace{}%
\AgdaGeneralizable{x}\<%
\end{code}
  is equal to
  \begin{code}[hide]%
\>[3428I][@{}l@{\AgdaIndent{2}}]%
\>[10]\AgdaSymbol{)}\<%
\end{code}
  \mbox{\begin{code}[inline]%
\>[10][@{}l@{\AgdaIndent{1}}]%
\>[14]\AgdaSymbol{λ}\AgdaSpace{}%
\AgdaBound{\AgdaUnderscore{}}\AgdaSpace{}%
\AgdaSymbol{→}\AgdaSpace{}%
\AgdaGeneralizable{y}\<%
\end{code}}.
  By Lemma~\ref{lem:uniquely-eliminating} we know that
  \begin{code}[hide]%
\>[2]\AgdaKeyword{postulate}\<%
\\
\>[2][@{}l@{\AgdaIndent{0}}]%
\>[4]\AgdaPostulate{\AgdaUnderscore{}}\AgdaSpace{}%
\AgdaSymbol{:}%
\>[3449I]\AgdaOperator{\AgdaGeneralizable{Has‐type[}}\AgdaSpace{}%
\AgdaSymbol{(∀}\AgdaSpace{}%
\AgdaSymbol{(}\AgdaSymbol{@ω}\AgdaSpace{}%
\AgdaBound{\AgdaUnderscore{}}\AgdaSpace{}%
\AgdaBound{\AgdaUnderscore{}}\AgdaSymbol{)}\AgdaSpace{}%
\AgdaSymbol{→}\AgdaSpace{}%
\AgdaSymbol{\AgdaUnderscore{})}\AgdaSpace{}%
\AgdaOperator{\AgdaGeneralizable{]}}\<%
\end{code}
  \begin{code}[inline*]%
\>[3449I][@{}l@{\AgdaIndent{1}}]%
\>[10]\AgdaSymbol{λ}\AgdaSpace{}%
\AgdaSymbol{(}\AgdaBound{f\,}\AgdaSpace{}%
\AgdaSymbol{:}\AgdaSpace{}%
\AgdaDatatype{Erased}\AgdaSpace{}%
\AgdaSymbol{(}\AgdaGeneralizable{x}\AgdaSpace{}%
\AgdaOperator{\AgdaFunction{≡}}\AgdaSpace{}%
\AgdaGeneralizable{y}\AgdaSymbol{)}\AgdaSpace{}%
\AgdaSymbol{→}\AgdaSpace{}%
\AgdaDatatype{Erased}\AgdaSpace{}%
\AgdaGeneralizable{A}\AgdaSymbol{)}\AgdaSpace{}%
\AgdaSymbol{→}\AgdaSpace{}%
\AgdaSymbol{λ}\AgdaSpace{}%
\AgdaSymbol{(}\AgdaBound{eq}\AgdaSpace{}%
\AgdaSymbol{:}\AgdaSpace{}%
\AgdaGeneralizable{x}\AgdaSpace{}%
\AgdaOperator{\AgdaFunction{≡}}\AgdaSpace{}%
\AgdaGeneralizable{y}\AgdaSymbol{)}\AgdaSpace{}%
\AgdaSymbol{→}\AgdaSpace{}%
\AgdaBound{f\,}\AgdaSpace{}%
\AgdaOperator{\AgdaInductiveConstructor{[}}\AgdaSpace{}%
\AgdaBound{eq}\AgdaSpace{}%
\AgdaOperator{\AgdaInductiveConstructor{]}}\<%
\end{code}
  is an equivalence and thus injective, so it suffices to prove that
  \begin{code}[hide]%
\>[2]\AgdaKeyword{postulate}\<%
\\
\>[2][@{}l@{\AgdaIndent{0}}]%
\>[4]\AgdaPostulate{\AgdaUnderscore{}}\AgdaSpace{}%
\AgdaSymbol{:}\AgdaSpace{}%
\AgdaOperator{\AgdaFunction{\AgdaUnderscore{}≡\AgdaUnderscore{}}}%
\>[3480I]\AgdaSymbol{\{}\AgdaArgument{A}\AgdaSpace{}%
\AgdaSymbol{=}\AgdaSpace{}%
\AgdaSymbol{∀}\AgdaSpace{}%
\AgdaSymbol{(}\AgdaSymbol{@ω}\AgdaSpace{}%
\AgdaBound{\AgdaUnderscore{}}\AgdaSymbol{)}\AgdaSpace{}%
\AgdaSymbol{→}\AgdaSpace{}%
\AgdaSymbol{\AgdaUnderscore{}\}}\AgdaSpace{}%
\AgdaSymbol{(}\<%
\end{code}
  \mbox{\begin{code}[inline]%
\>[3480I][@{}l@{\AgdaIndent{1}}]%
\>[14]\AgdaSymbol{λ}\AgdaSpace{}%
\AgdaSymbol{(}\AgdaBound{\AgdaUnderscore{}}\AgdaSpace{}%
\AgdaSymbol{:}\AgdaSpace{}%
\AgdaGeneralizable{x}\AgdaSpace{}%
\AgdaOperator{\AgdaFunction{≡}}\AgdaSpace{}%
\AgdaGeneralizable{y}\AgdaSymbol{)}\AgdaSpace{}%
\AgdaSymbol{→}\AgdaSpace{}%
\AgdaGeneralizable{x}\<%
\end{code}}
  is equal to
  \begin{code}[hide]%
\>[.][@{}l@{}]\<[3480I]%
\>[12]\AgdaSymbol{)}\<%
\end{code}
  \mbox{\begin{code}[inline]%
\>[12][@{}l@{\AgdaIndent{1}}]%
\>[14]\AgdaSymbol{λ}\AgdaSpace{}%
\AgdaBound{\AgdaUnderscore{}}\AgdaSpace{}%
\AgdaSymbol{→}\AgdaSpace{}%
\AgdaGeneralizable{y}\<%
\end{code}}.
  By function extensionality these two functions are equal if the
  trivial property
  \begin{code}[hide]%
\>[2]\AgdaKeyword{postulate}\<%
\\
\>[2][@{}l@{\AgdaIndent{0}}]%
\>[4]\AgdaPostulate{\AgdaUnderscore{}}%
\>[3498I]\AgdaSymbol{:}\<%
\end{code}
  \mbox{\begin{code}[inline]%
\>[3498I][@{}l@{\AgdaIndent{1}}]%
\>[8]\AgdaGeneralizable{x}\AgdaSpace{}%
\AgdaOperator{\AgdaFunction{≡}}\AgdaSpace{}%
\AgdaGeneralizable{y}\AgdaSpace{}%
\AgdaSymbol{→}\AgdaSpace{}%
\AgdaGeneralizable{x}\AgdaSpace{}%
\AgdaOperator{\AgdaFunction{≡}}\AgdaSpace{}%
\AgdaGeneralizable{y}\<%
\end{code}}
  holds.

  One should also prove that the constructed function ``computes'' in
  the correct way.
  For that part of the proof I refer to the accompanying code.
\end{proof}

\section{Related Work}
\label{sec:related}

\citet{coquand-et-al-2017} discuss how postulating negated types does
not lead to loss of canonicity.

There has been quite a bit of work on erasure for dependently typed
languages.
One can erase types \citep{coquand-huet-1988,augustsson-1998} or other
things that lack computational content \citep{hayashi-nakano-1988}, or
let the system figure out things that can be erased automatically
\citep{brady-et-al-2004,brady-2005,mishra-linger-2008,
  fredriksson-gustafsson-2011,tejiscak-2020}.
One can also let the programmer mark parts that should be erased, and
let the type-checker check that erased parts cannot influence the
results of run-time computations
\citep{paulin-mohring-1989,paulin-mohring-werner-1993,
  raamsdonk-severi-2002,letouzey-2003,fernandez-et-al-2003,
  barras-bernardo-2008,mishra-linger-sheard-2008,mishra-linger-2008,
  gundry-mcbride-2013,bernardy-moulin-2013,gundry-2013,mcbride-2016,
  weirich-et-al-2017,atkey-2018,choudhury-et-al-2021,moon-et-al-2021,
  abel-et-al-2021,choudhury-et-al-2022,abel-et-al-2023,
  liu-weirich-2023,danielsson-et-al-2026}.

\citet{dijkstra-2013} discusses erasure in the context of homotopy
type theory.
The concept of an erased constructor seems to have been introduced by
\citet{abel-et-al-2021}, who discuss uses for both erased point
constructors and erased higher constructors.
\citet{theocaris-brady-2026} also mention erased constructors, but
mostly leave this concept to future work.
Except for the draft paper of \citet{abel-et-al-2021}, which has been
discussed above, I am not aware of any work that discusses
\boxcong{}.

\citet[Theorem 9]{letouzey-2004} proves a form of correctness of
erasure for open terms in the presence of singleton elimination, a
feature that is similar to the erased matches discussed here.
The correctness statement is, very roughly, ``for every type-correct
term there is a proof term showing that the extracted term is
correct''.
However, I am not convinced that the correctness statement is
correctly formulated.
The proof term uses a transformed context, and there is no proof that,
if the original context is consistent, then the transformed context is
consistent.
The type
$\mathit{Id}\,\mathit{Set}\,\mathit{Unit\ensuremath{{}_{\mathrm{1}}}}\,\mathit{Unit\ensuremath{{}_{\mathrm{2}}}}:\mathit{Set}$,
where $\mathit{Unit\ensuremath{{}_{\mathrm{1}}}}, \mathit{Unit\ensuremath{{}_{\mathrm{2}}}} : \mathit{Set}$ are two distinct
unit types (data types with single nullary constructors), is perhaps
consistent with \citeauthor{letouzey-2004}'s type theory – a variant
of the Calculus of Inductive Constructions.
However, the transformation of this type appears to be inconsistent
(the type theory is not fully specified, and I may have misinterpreted
something).

\citeauthor{atkey-2018}'s Quantitative Type Theory (QTT) does not
include an identity type \citep{atkey-2018}.
Later work by \citet{nakov-forsberg-2022} and \citet{atkey-2024}
includes identity types.
In the case of \citeauthor{atkey-2024}'s theory equality reflection is
available in the erased fragment of the language, but not in the
non-erased fragment.
\citeauthor{atkey-2024} also states that there is an η-rule for the
identity type, but does not state whether that rule is available in
the non-erased fragment or not, so it is unclear to me exactly what
rules the type theory has, and what one can do with equalities in the
non-erased fragment.
Equality reflection is also available in the theory discussed by
\citeauthor{nakov-forsberg-2022}, but some rules were omitted from
their paper, so it is unclear to me if equality reflection is
restricted in some way.

\citet{sozeau-et-al-2025} present a mechanised proof of correctness of
erasure for a type theory – a ``slightly simplified version of Coq's
core language'' – that includes intensional identity types (identity
types without equality reflection).
Their correctness statements have the form ``if a term has a value,
then…'' (roughly), unlike Theorem~\ref{thm:soundness-of-erasure},
which has the form ``every well-typed term has a value and…''
(roughly).

\citet{winterhalter-2024} describes a type theory with a type
constructor $\mathsf{erased}$ that takes a type to a \emph{ghost
  type}.
There is also a family of definitionally proof-irrelevant identity
types, but only for ghost types.
If one has such an identity between $t$ and $u$, then one can
transport from $P\,t$ to $P\,u$, given certain assumptions.
The conversion rule compares terms after such transports have been
removed.
There is also an empty type with something like erased matches.
\citeauthor{winterhalter-2024} proves relative consistency, and his
paper is accompanied by machine-checked proofs.

\citet{liu-et-al-2024,liu-et-al-2025} present type theories with
support for both compile-time and run-time erasure.
These theories come with identity types that are indexed by
\emph{observer levels}: if $t$ is equal to $u$ at a low level, then
``high parts'' of $t$ and $u$ are not necessarily equal.
\citet{liu-et-al-2025} prove things like consistency and
normalisation, and their work is accompanied by machine-checked
proofs.
Their paper does not contain anything like
Theorem~\ref{thm:soundness-of-erasure}: there is no ``erasure
function''.
However, there is a \emph{simulation lemma} that states that if $t$
and $u$ are definitionally equal for observers at level $p$, and $t$
reduces to $t′$, then $u$ reduces to some term $u′$ that is
definitionally equal to $t′$ for observers at level $p$.
The theory includes an empty type with support for something like
erased matches.

\citet{altenkirch-et-al-2007} present observational type theory.
The type theory does not use identity types but rather
proof-irrelevant observational equality.
One can cast between two equal, proof-relevant types: that feature is
reminiscent of the erased matches discussed here.
\citeauthor{altenkirch-et-al-2007} state that ``we may
extend the language of proofs with whatever propositional laws we
like, as long as they are consistent, and yet retain canonicity''.
They prove consistency by embedding into extensional type theory.
They also discuss erasure, and plans for quotient types.

\citet{pujet-tabareau-2022} present observational type theory with
quotients and identity types.
Propositional extensionality is built-in, but not for all types that
are propositions semantically, only for strict propositions.
Parts of the meta-theory are mechanised, but not the results about
quotients.
\citeauthor{pujet-tabareau-2022} use a logical relation that, for
quotients, is similar to the one discussed for type theory without
equality reflection in §~\ref{sec:meta-theory}.
They state that ``we could postulate any consistent proof-irrelevant
axiom and the normalization proof would carry through''.

\citet{pujet-leray-tabareau-2025} discuss fording using observational
equality.
In that setting it can be tricky to get the right computation rules.
In the setting of this paper, with a J eliminator that computes in the
usual way when applied to reflexivity, that is not an issue.

In concurrent work \citet{felicissimo-et-al-2026} present an
observational type theory with a proof-irrelevant accessibility
predicate with singleton elimination.
The type theory supports postulates, restricted to a proof-irrelevant
universe.
\citeauthor{felicissimo-et-al-2026} prove that, if the postulates are
validated in a certain set-theoretic model (and satisfy another
assumption), then every closed term $t$ of type $ℕ$ reduces to a
correct numeral.
The proof is partly mechanised.
Erasure/extraction is left for future work.

\citet{felicissimo-et-al-2026-2} – who cite a previous version of the
present paper – present a type theory with definitionally
proof-irrelevant equality and an enhanced judgemental equality rule
for transports (which allow one to cast between proof-relevant types).
Postulates in a proof-irrelevant universe are allowed.
It is proved that, if the postulates include function extensionality
as well as propositional injectivity for type formers of
proof-relevant types, and one can prove that propositionally equal
proof-relevant types have equal heads, then a program obtained by
extraction from a closed term of type $ℕ$ reduces to a correct numeral
in the call-by-name target language.

\section{Discussion}
\label{sec:discussion}

I have presented a family of type theories with erasure, identity
types with optional equality reflection, optional quotients, and
optional support for erased matches/\boxcong{}.
Theorem \ref{thm:soundness-of-erasure} shows that (consistent) erased
postulates are safe in the absence of erased matches/equality
reflection, and that erased matches are safe in the absence of
postulates and quotients.
Unlike a number of pieces of related work discussed above these
results are not restricted to proof-irrelevant identity types.
Note that univalence is stated using \emph{proof-relevant} identity
types \citep{hott-2013}, and that identity types with J become
proof-irrelevant in the presence of equality
reflection.

Theorem \ref{thm:soundness-reflection} shows that it is safe to
combine erased matches with erased postulates if the postulates can be
implemented in an extension of the type theory with equality
reflection, and Theorem \ref{thm:meta} provides a framework for
proving similar results for other kinds of extensions.
Instantiating Theorem \ref{thm:meta} with cubical type theories, thus
hopefully showing that Computational Book HoTT is well-behaved in the
presence of \boxcong{}, is left for future work.

I also tried to motivate having support for \boxcong{} (or,
equivalently, limited erased matches for J).
I showed that \boxcong{} cannot be defined in the absence of erased
matches/equality reflection (Theorem \ref{thm:no-box-cong}), I gave a
number of logically equivalent characterisations of ``\boxcong{} can
be defined'' (Theorems \ref{thm:equivalent-1}
and \ref{thm:equivalent-2}), and I adapted the technique of fording
for use with erased identity proofs, and the adapted technique uses
\boxcong{} (or other forms of erased matches for identity
types).

\begin{acks}
  I would like to thank Andreas Abel, Jesper Cockx, Oskar Eriksson,
  Thiago Felicissimo, Loïc Pujet, Christian Sattler, Mike Shulman,
  Andrea Vezzosi and David Wärn for discussions about topics related
  to this paper.
  Specifically, the investigation of fording reported in
  § \ref{sec:fording} was prompted by a discussion with Jesper Cockx
  about pattern matching in Agda in the presence of erasure.
  I would also like to thank some anonymous reviewers.

  I acknowledge financial support from
  \grantsponsor{vr}{Vetenskapsrådet}{https://doi.org/10.13039/501100004359}
  (\grantnum{vr}{2023-04538}).
\end{acks}

\bibliography{References}

\ifAppendicesIncluded{\appendix

\section{Full Typing and Reduction Rules}
\label{sec:full-typing-rules}

This section contains full typing and reduction rules matching the
formalisation, except that first-class and transfinite universe levels
(which are optional in the formalisation) and top-level definitions
have been omitted, and some other smaller differences.
The rules have been specialised to the erasure semiring; such a
specialisation is also available in the formalisation.
The formalisation uses well-scoped syntax, but such details have been
omitted here: the definitions should be read as if everything were
well-scoped.

As mentioned in the main text the formalism is parametrised, and one
can choose whether or not to allow, for instance, $\bcName$.
This is expressed using conditions like ``$\emptyrecWithLevel{p}$ is
allowed when the mode is $q$''.
Note that $\emptyrecName$, $\unitrecName$ and $\prodrecName$ cannot be
disallowed when the mode is $\zeroG$.

The notation $\wkn{n}{t}$ stands for $t$, weakened $n$ steps, and
$\wklift{n}{x}{t}$ stands for $t$, with variable $x$ and higher
weakened $n$ steps.
The notation $\substZ{t}{u}$ stands for $t$ with $u$ substituted for
variable $\var{0}$ (and the other free variables decreased by one),
and $\substO{t}{u}{v}$ stands for $t$ with $u$ substituted for
variable $\var{1}$ and $v$ substituted for variable $\var{0}$ (and the
other free variables decreased by two).
The notation $\substZUpn{1+n}{t}{u}$ stands for $t$, with $u$
substituted for variable $\var{0}$, and the other free variables
increased by $n$.

Universe level literals $l$ are natural numbers.

\subsection{Typing}

Recall that the notation $\wfTm{Γ}{t}{A}$ is an abbreviation for
$\wrTm{\zeroGC}{Γ}{t}{\zeroG}{A}$, and that $\wfTy{Γ}{A}$ is an
abbreviation for $\wrTy{\zeroGC}{Γ}{\zeroG}{A}$.

Rules for contexts:
\begin{mathparpagebreakable}
  \inferrule{ }{\wfCtxt{\emptyC}}
  \and
  \inferrule{\wfTy{Γ}{A}}{\wfCtxt{\consC{Γ}{A}}}
\end{mathparpagebreakable}

The conversion rule:
\begin{mathparpagebreakable}
  \inferrule{\wrTm{γ}{Γ}{t}{p}{A} \\ \eqTy{Γ}{A}{B}}{%
    \wrTm{γ}{Γ}{t}{p}{B}}
\end{mathparpagebreakable}

Reflexivity, symmetry and transitivity:
\begin{mathparpagebreakable}
  \inferrule{\wfTy{Γ}{A}}{\eqTy{Γ}{A}{A}}
  \and
  \inferrule{\eqTy{Γ}{A₁}{A₂}}{\eqTy{Γ}{A₂}{A₁}}
  \and
  \inferrule{\eqTy{Γ}{A₁}{A₂} \\ \eqTy{Γ}{A₂}{A₃}}{\eqTy{Γ}{A₁}{A₃}}
\ifAppendicesIncluded{\\}{
\end{mathparpagebreakable}
\begin{mathparpagebreakable}}
  \inferrule{\wfTm{Γ}{t}{A}}{\eqTm{Γ}{t}{t}{A}}
  \and
  \inferrule{\eqTm{Γ}{t₁}{t₂}{A}}{\eqTm{Γ}{t₂}{t₁}{A}}
  \and
  \inferrule{\eqTm{Γ}{t₁}{t₂}{A} \\ \eqTm{Γ}{t₂}{t₃}{A}}{
    \eqTm{Γ}{t₁}{t₃}{A}}
\end{mathparpagebreakable}

Rules for variables:
\begin{mathparpagebreakable}
  \inferrule{ }{\varTy{\varZ}{\wkO{A}}{\consC{Γ}{A}}}
  \and
  \inferrule{\varTy{x}{A}{Γ}}{\varTy{\varSuc{x}}{\wkO{A}}{\consC{Γ}{B}}}
  \and
  \inferrule{\wfCtxt{Γ} \\ \varTy{x}{A}{Γ} \\ \leqG{\lookup{γ}{x}}{p}}{%
    \wrTm{γ}{Γ}{x}{p}{A}}
\end{mathparpagebreakable}

\ifAnonymous{\newpage}{}

Rules for $\UName$:
\begin{mathparpagebreakable}
  \inferrule{\wrTm{γ}{Γ}{A}{p}{\U{l}}}{\wrTy{γ}{Γ}{p}{A}}
  \and
  \inferrule{\eqTm{Γ}{A₁}{A₂}{\U{l}}}{\eqTy{Γ}{A₁}{A₂}}
  \and
  \inferrule{\wfCtxt{Γ}}{\wrTm{γ}{Γ}{\U{l}}{p}{\U{\sucL{l}}}}
\end{mathparpagebreakable}

Rules for lift types.
The type constructor $\LiftWith{l₂}$ lifts types in $\U{l₁}$ to
$\U{\maxL{l₁}{l₂}}$, where $\maxL{l₁}{l₂}$ is the maximum of $l₁$ and
$l₂$:
\begin{mathparpagebreakable}
  \inferrule{\wrTy{γ}{Γ}{p}{A}}{\wrTy{γ}{Γ}{p}{\LiftT{l}{A}}}
  \and
  \inferrule{\eqTy{Γ}{A₁}{A₂}}{\eqTy{Γ}{\LiftT{l}{A₁}}{\LiftT{l}{A₂}}}
  \\
  \inferrule{\wrTm{γ}{Γ}{A}{p}{\U{l₁}}}{
    \wrTm{γ}{Γ}{\LiftT{l₂}{A}}{p}{\U{\maxL{l₁}{l₂}}}}
  \and
  \inferrule{\eqTm{Γ}{A₁}{A₂}{\U{l₁}}}{
    \eqTm{Γ}{\LiftT{l₂}{A₁}}{\LiftT{l₂}{A₂}}{\U{\maxL{l₁}{l₂}}}}
  \and
  \inferrule{\wrTm{γ}{Γ}{t}{p}{A}}{
    \wrTm{γ}{Γ}{\lift{t}}{p}{\LiftT{l}{A}}}
  \\
  \inferrule{\wrTm{γ}{Γ}{t}{p}{\LiftT{l}{A}}}{
    \wrTm{γ}{Γ}{\lowerTm{t}}{p}{A}}
  \and
  \inferrule{\eqTm{Γ}{t₁}{t₂}{\LiftT{l}{A}}}{
    \eqTm{Γ}{\lowerTm{t₁}}{\lowerTm{t₂}}{A}}
  \and
  \inferrule{\wfTm{Γ}{t}{A}}{\eqTm{Γ}{\lowerTm{(\lift{t})}}{t}{A}}
  \and
  \inferrule{\wfTm{Γ}{t₁}{\LiftT{l}{A}} \\
    \wfTm{Γ}{t₂}{\LiftT{l}{A}} \\
    \eqTm{Γ}{\lowerTm{t₁}}{\lowerTm{t₂}}{A}}{
    \eqTm{Γ}{t₁}{t₂}{A}}
\end{mathparpagebreakable}

Rules for the empty type.
An application of $\emptyrecWithLevel{p}$ is an erased match if $p$ is
$\zeroG$ and the mode is $\omegaG$:
\begin{mathparpagebreakable}
  \inferrule{\wfCtxt{Γ}}{\wrTm{γ}{Γ}{\Empty}{p}{\U{\zeroL}}}
  \and
  \inferrule{
    \text{$\emptyrecWithLevel{p}$ is allowed when the mode is $q$} \\\\
    \wfTy{Γ}{A} \\ \wrTm{γ}{Γ}{t}{\mul{q}{p}}{\Empty}}{
    \wrTm{γ}{Γ}{\emptyrec{p}{A}{t}}{q}{A}}
  \and
  \inferrule{\eqTy{Γ}{A₁}{A₂} \\ \eqTm{Γ}{t₁}{t₂}{\Empty}}{
    \eqTm{Γ}{\emptyrec{p}{A₁}{t₁}}{\emptyrec{p}{A₂}{t₂}}{A₁}}
\end{mathparpagebreakable}

Rules for unit types.
An application of $\unitrecWithLevel{p}$ is an erased match if $p$ is
$\zeroG$ and the mode is $\omegaG$.
The parameter $s$ is a \emph{strength}: $\strongS$ (``strong'') for
the unit type with η-equality, and $\weak$ (``weak'') for the unit
type without η-equality:
\begin{mathparpagebreakable}
  \inferrule{\excludeAllowed{\text{$\Unit{s}$ is allowed} \\}
    \wfCtxt{Γ}}{
    \wrTm{γ}{Γ}{\Unit{s}}{p}{\U{\zeroL}}}
  \and
  \inferrule{\excludeAllowed{\text{$\Unit{s}$ is allowed} \\}
    \wfCtxt{Γ}}{
    \wrTm{γ}{Γ}{\starTm{s}}{p}{\Unit{s}}}
  \and
  \inferrule{\wfTm{Γ}{t₁}{\Unit{\strongS}} \\
    \wfTm{Γ}{t₂}{\Unit{\strongS}}}{
    \eqTm{Γ}{t₁}{t₂}{\Unit{\strongS}}}
  \\
  \inferrule{
    \text{$\unitrecWithLevel{p}$ is allowed when the mode is $q$} \\
    \wfTy{\consC{Γ}{\Unit{\weak}}}{A} \\
    \wrTm{γ}{Γ}{t}{\mul{q}{p}}{\Unit{\weak}} \\
    \wrTm{γ}{Γ}{u}{q}{\substZ{A}{\starTm{\weak}}}}{
    \wrTm{γ}{Γ}{\unitrec{p}{A}{t}{u}}{q}{\substZ{A}{t}}}
  \\
  \inferrule{\eqTy{\consC{Γ}{\Unit{\weak}}}{A₁}{A₂} \\
    \eqTm{Γ}{t₁}{t₂}{\Unit{\weak}} \\
    \eqTm{Γ}{u₁}{u₂}{\substZ{A₁}{\starTm{\weak}}}}{
    \eqTm{Γ}{\unitrec{p}{A₁}{t₁}{u₁}}{\unitrec{p}{A₂}{t₂}{u₂}}{
      \substZ{A₁}{t₁}}}
  \\
  \inferrule{\wfTy{\consC{Γ}{\Unit{\weak}}}{A} \\
    \wfTm{Γ}{t}{\substZ{A}{\starTm{\weak}}}}{
    \eqTm{Γ}{\unitrec{p}{A}{\starTm{\weak}}{t}}{t}{
      \substZ{A}{\starTm{\weak}}}}
\end{mathparpagebreakable}

\ifAnonymous{\newpage}{}

Rules for Π-types:
\begin{mathparpagebreakable}
  \inferrule{
    \excludeAllowed{\text{$\PiWithLevel{p}$ is allowed} \\\\}
    \wrTy{γ}{Γ}{\mul{q}{p}}{A} \\
    \wrTy{\consGC{γ}{p}}{\consC{Γ}{A}}{q}{B}}{%
    \wrTy{γ}{Γ}{q}{\PiT{p}{A}{B}}}
  \and
  \inferrule{
    \excludeAllowed{\text{$\PiWithLevel{p}$ is allowed} \\\\}
    \eqTy{Γ}{A₁}{A₂} \\ \eqTy{\consC{Γ}{A₁}}{B₁}{B₂}}{%
    \eqTy{Γ}{\PiT{p}{A₁}{B₁}}{\PiT{p}{A₂}{B₂}}}
\ifAppendicesIncluded{
\end{mathparpagebreakable}
\begin{mathparpagebreakable}}{\\}
  \inferrule{
    \excludeAllowed{\text{$\PiWithLevel{p}$ is allowed} \\\\}
    \wrTm{γ}{Γ}{A}{\mul{q}{p}}{\U{l}} \\
    \wrTm{\consGC{γ}{p}}{\consC{Γ}{A}}{B}{q}{\U{l}}}{%
    \wrTm{γ}{Γ}{\PiT{p}{A}{B}}{q}{\U{l}}}
  \and
  \inferrule{
    \excludeAllowed{\text{$\PiWithLevel{p}$ is allowed} \\\\}
    \eqTm{Γ}{A₁}{A₂}{\U{l}} \\ \eqTm{\consC{Γ}{A₁}}{B₁}{B₂}{\U{l}}}{%
    \eqTm{Γ}{\PiT{p}{A₁}{B₁}}{\PiT{p}{A₂}{B₂}}{\U{l}}}
  \and
  \inferrule{
    \excludeAllowed{\text{$\PiWithLevel{p}$ is allowed} \\}
    \wrTm{\consGC{γ}{p}}{\consC{Γ}{A}}{t}{q}{B}}{%
    \wrTm{γ}{Γ}{\lam{p}{t}}{q}{\PiT{p}{A}{B}}}
  \and
  \inferrule{\wrTm{γ}{Γ}{t}{q}{\PiT{p}{A}{B}} \\
    \wrTm{γ}{Γ}{u}{\mul{q}{p}}{A}}{%
    \wrTm{γ}{Γ}{\app{t}{p}{u}}{q}{\substZ{B}{u}}}
  \and
  \inferrule{\eqTm{Γ}{t₁}{t₂}{\PiT{p}{A}{B}} \\
    \eqTm{Γ}{u₁}{u₂}{A}}{%
    \eqTm{Γ}{\app{t₁}{p}{u₁}}{\app{t₂}{p}{u₂}}{\substZ{B}{u₁}}}
  \and
  \inferrule{
    \excludeAllowed{\text{$\PiWithLevel{p}$ is allowed} \\}
    \wfTm{\consC{Γ}{A}}{t}{B} \\ \wfTm{Γ}{u}{A}}{%
    \eqTm{Γ}{\app{(\lam{p}{t})}{p}{u}}{\substZ{t}{u}}{\substZ{B}{u}}}
  \and
  \inferrule{\wfTm{Γ}{t₁}{\PiT{p}{A}{B}} \\
    \wfTm{Γ}{t₂}{\PiT{p}{A}{B}} \\
    \eqTm{\consC{Γ}{A}}{\app{(\wkO{t₁})}{p}{\varZ}}{
      \app{(\wkO{t₂})}{p}{\varZ}}{B}}{%
    \eqTm{Γ}{t₁}{t₂}{\PiT{p}{A}{B}}}
\end{mathparpagebreakable}

Rules for graded Σ-types.
The first component of a value of type $\SigmaT{s}{p}{A}{B}$ is erased
if $p$ is $\zeroG$.
An application of $\prodrecWith{p}{q}$ is an erased match if $q$ is
$\zeroG$ and the mode is $\omegaG$.
Note that $\fstWith{p}$ can only be used if $\leqG{p}{q}$, where $q$
is the mode – erased first components can only be accessed in erased
contexts:
\begin{mathparpagebreakable}
  \inferrule{
    \excludeAllowed{\text{$\SigmaWith{s}{p}$ is allowed} \\\\}
    \wrTy{γ}{Γ}{\mul{q}{p}}{A} \\
    \wrTy{\consGC{γ}{p}}{\consC{Γ}{A}}{q}{B}}{%
    \wrTy{γ}{Γ}{q}{\SigmaT{s}{p}{A}{B}}}
  \and
  \inferrule{
    \excludeAllowed{\text{$\SigmaWith{s}{p}$ is allowed} \\\\}
    \eqTy{Γ}{A₁}{A₂} \\ \eqTy{\consC{Γ}{A₁}}{B₁}{B₂}}{%
    \eqTy{Γ}{\SigmaT{s}{p}{A₁}{B₁}}{\SigmaT{s}{p}{A₂}{B₂}}}
  \and
  \inferrule{
    \excludeAllowed{\text{$\SigmaWith{s}{p}$ is allowed} \\\\}
    \wrTm{γ}{Γ}{A}{\mul{q}{p}}{\U{l}} \\
    \wrTm{\consGC{γ}{p}}{\consC{Γ}{A}}{B}{q}{\U{l}}}{%
    \wrTm{γ}{Γ}{\SigmaT{s}{p}{A}{B}}{q}{\U{l}}}
  \and
  \inferrule{
    \excludeAllowed{\text{$\SigmaWith{s}{p}$ is allowed} \\\\}
    \eqTm{Γ}{A₁}{A₂}{\U{l}} \\ \eqTm{\consC{Γ}{A₁}}{B₁}{B₂}{\U{l}}}{%
    \eqTm{Γ}{\SigmaT{s}{p}{A₁}{B₁}}{\SigmaT{s}{p}{A₂}{B₂}}{\U{l}}}
  \and
  \inferrule{
    \excludeAllowed{\text{$\SigmaWith{s}{p}$ is allowed} \\}
    \wfTy{\consC{Γ}{A}}{B} \\\\
    \wrTm{γ}{Γ}{t}{\mul{q}{p}}{A} \\
    \wrTm{γ}{Γ}{u}{q}{\substZ{B}{t}}}{%
    \wrTm{γ}{Γ}{\prodTm{s}{p}{t}{u}}{q}{\SigmaT{s}{p}{A}{B}}}
  \and
  \inferrule{
    \excludeAllowed{\text{$\SigmaWith{s}{p}$ is allowed} \\}
    \wfTy{\consC{Γ}{A}}{B} \\\\
    \eqTm{Γ}{t₁}{t₂}{A} \\
    \eqTm{Γ}{u₁}{u₂}{\substZ{B}{t₁}}}{%
    \eqTm{Γ}{\prodTm{s}{p}{t₁}{u₁}}{\prodTm{s}{p}{t₂}{u₂}}{
      \SigmaT{s}{p}{A}{B}}}
  \and
  \inferrule{\wrTm{γ}{Γ}{t}{q}{\SigmaT{\strongS}{p}{A}{B}} \\
    \leqG{p}{q}}{%
    \wrTm{γ}{Γ}{\fst{p}{t}}{q}{A}}
  \and
  \inferrule{\eqTm{Γ}{t₁}{t₂}{\SigmaT{\strongS}{p}{A}{B}}}{%
    \eqTm{Γ}{\fst{p}{t₁}}{\fst{p}{t₂}}{A}}
  \and
  \inferrule{
    \excludeAllowed{\text{$\SigmaWith{\strongS}{p}$ is allowed} \\}
    \wfTy{\consC{Γ}{A}}{B} \\\\
    \wfTm{Γ}{t}{A} \\ \wfTm{Γ}{u}{\substZ{B}{t}}}{%
    \eqTm{Γ}{\fst{p}{\prodTm{\strongS}{p}{t}{u}}}{t}{A}}
  \and
  \inferrule{\wrTm{γ}{Γ}{t}{q}{\SigmaT{\strongS}{p}{A}{B}}}{%
    \wrTm{γ}{Γ}{\snd{p}{t}}{q}{\substZ{B}{\fst{p}{t}}}}
  \and
  \inferrule{\eqTm{Γ}{t₁}{t₂}{\SigmaT{\strongS}{p}{A}{B}}}{%
    \eqTm{Γ}{\snd{p}{t₁}}{\snd{p}{t₂}}{\substZ{B}{\fst{p}{t₁}}}}
  \and
  \inferrule{
    \excludeAllowed{\text{$\SigmaWith{\strongS}{p}$ is allowed} \\}
    \wfTy{\consC{Γ}{A}}{B} \\\\
    \wfTm{Γ}{t}{A} \\ \wfTm{Γ}{u}{\substZ{B}{t}}}{%
    \eqTm{Γ}{\snd{p}{\prodTm{\strongS}{p}{t}{u}}}{u}{\substZ{B}{t}}}
  \and
  \inferrule{\wfTm{Γ}{t₁}{\SigmaT{\strongS}{p}{A}{B}} \\
    \wfTm{Γ}{t₂}{\SigmaT{\strongS}{p}{A}{B}} \\\\
    \eqTm{Γ}{\fst{p}{t₁}}{\fst{p}{t₂}}{A} \\\\
    \eqTm{Γ}{\snd{p}{t₁}}{\snd{p}{t₂}}{\substZ{B}{\fst{p}{t₁}}}}{%
    \eqTm{Γ}{t₁}{t₂}{\SigmaT{\strongS}{p}{A}{B}}}
  \and
  \inferrule{
    \text{$\prodrecWith{p}{q}$ is allowed when the mode is $r$} \\
    \wfTy{\consC{Γ}{\SigmaT{\weak}{p}{A}{B}}}{C} \\\\
    \wrTm{γ}{Γ}{t}{\mul{r}{q}}{\SigmaT{\weak}{p}{A}{B}} \\
    \wrTm{\consGC{\consGC{γ}{\mul{q}{p}}}{q}}{\consC{\consC{Γ}{A}}{B}}{
      u}{r}{\substZUpT{C}{\prodTm{\weak}{p}{\varO}{\varZ}}}}{%
    \wrTm{γ}{Γ}{\prodrec{p}{q}{C}{t}{u}}{r}{\substZ{C}{t}}}
  \and
  \inferrule{\eqTy{\consC{Γ}{\SigmaT{\weak}{p}{A}{B}}}{C₁}{C₂} \\
    \eqTm{Γ}{t₁}{t₂}{\SigmaT{\weak}{p}{A}{B}} \\
    \eqTm{\consC{\consC{Γ}{A}}{B}}{u₁}{u₂}{
      \substZUpT{C₁}{\prodTm{\weak}{p}{\varO}{\varZ}}}}{%
    \eqTm{Γ}{\prodrec{p}{q}{C₁}{t₁}{u₁}}{\prodrec{p}{q}{C₂}{t₂}{u₂}}{
      \substZ{C₁}{t₁}}}
  \and
  \inferrule{\wfTy{\consC{Γ}{\SigmaT{\weak}{p}{A}{B}}}{C} \\
    \wfTm{Γ}{t}{A} \\ \wfTm{Γ}{u}{\substZ{B}{t}} \\
    \wfTm{\consC{\consC{Γ}{A}}{B}}{v}{
      \substZUpT{C}{\prodTm{\weak}{p}{\varO}{\varZ}}}}{%
    \eqTm{Γ}{\prodrec{p}{q}{C}{\prodTm{\weak}{p}{t}{u}}{v}}{
      \substO{v}{t}{u}}{\substZ{C}{\prodTm{\weak}{p}{t}{u}}}}
\end{mathparpagebreakable}
One can use unit types and graded Σ-types to encode $\EName$.
$\EName$ and the box constructor, which are used in the typing rules
for $\bcName$ below, are encoded in the following way:
\begin{equation*}
  \begin{pmboxed}
    \>[][@{}l@{}]\ES{s}{l}{A} \>[][@{}l@{}]= \SigmaT{s}{\zeroG}{A}{(\LiftT{l}{\Unit{s}})}\\
    \>           \bxS{s}{t}   \>           = \prodTm{s}{\zeroG}{t}{\lift{\starTm{s}}}
  \end{pmboxed}
\end{equation*}
Note that the first component of the Σ-type is erased.
Unlike in the main body of the text the encoding is parametrised by a
strength, so that one can choose whether or not to allow η-equality
for $\EName$.
(Note that η-equality is always allowed for lift types.)

Rules for natural numbers:
\begin{mathparpagebreakable}
  \inferrule{\wfCtxt{Γ}}{\wrTm{γ}{Γ}{\Nat}{p}{\U{\zeroL}}}
  \and
  \inferrule{\wfCtxt{Γ}}{\wrTm{γ}{Γ}{\zero}{p}{\Nat}}
  \and
  \inferrule{\wrTm{γ}{Γ}{t}{p}{\Nat}}{\wrTm{γ}{Γ}{\suc{t}}{p}{\Nat}}
  \and
  \inferrule{\eqTm{Γ}{t₁}{t₂}{\Nat}}{\eqTm{Γ}{\suc{t₁}}{\suc{t₂}}{\Nat}}
  \and
  \inferrule{\wfTy{\consC{Γ}{\Nat}}{A} \\
    \wrTm{γ}{Γ}{t}{r}{\substZ{A}{\zero}} \\\\
    \wrTm{\consGC{\consGC{γ}{p}}{q}}{\consC{\consC{Γ}{\Nat}}{A}}{
      u}{r}{\substZUpT{A}{\suc{\varO}}} \\
    \wrTm{γ}{Γ}{v}{r}{\Nat}}{%
    \wrTm{γ}{Γ}{\natrec{p}{q}{A}{t}{u}{v}}{r}{\substZ{A}{v}}}
  \and
  \inferrule{\eqTy{\consC{Γ}{\Nat}}{A₁}{A₂} \\
    \eqTm{Γ}{t₁}{t₂}{\substZ{A₁}{\zero}} \\\\
    \eqTm{\consC{\consC{Γ}{\Nat}}{A₁}}{u₁}{u₂}{
      \substZUpT{A₁}{\suc{\varO}}} \\
    \eqTm{Γ}{v₁}{v₂}{\Nat}}{%
    \eqTm{Γ}{\natrec{p}{q}{A₁}{t₁}{u₁}{v₁}}{
      \natrec{p}{q}{A₂}{t₂}{u₂}{v₂}}{\substZ{A₁}{v₁}}}
  \and
  \inferrule{\wfTm{Γ}{t}{\substZ{A}{\zero}} \\
    \wfTm{\consC{\consC{Γ}{\Nat}}{A}}{u}{\substZUpT{A}{\suc{\varO}}}}{%
    \eqTm{Γ}{\natrec{p}{q}{A}{t}{u}{\zero}}{t}{\substZ{A}{\zero}}}
  \and
  \inferrule{\wfTm{Γ}{t}{\substZ{A}{\zero}} \\
    \wfTm{\consC{\consC{Γ}{\Nat}}{A}}{u}{\substZUpT{A}{\suc{\varO}}} \\
    \wfTm{Γ}{v}{\Nat}}{%
    \eqTm{Γ}{\natrec{p}{q}{A}{t}{u}{(\suc{v})}}{
      \substO{u}{v}{\natrec{p}{q}{A}{t}{u}{v}}}{\substZ{A}{\suc{v}}}}
\end{mathparpagebreakable}

Rules for identity types.
As noted in the main text one can choose between two sets of typing
rules for $\IdName$.
Below the antecedent $\IdRules{1}$ is satisfied if the first set has
been chosen, and $\IdRules{2}$ is satisfied if the second set has been
chosen.
For J and K there are three sets of rules.
The antecedent $\JRules{n}$ is satisfied if the $n$-th set of rules
for J has been chosen, and similarly for K\excludeAllowed{.
If $\bcSWith{s}$ is allowed, then $\Unit{s}$ and
$\SigmaWith{s}{\zeroG}$ must be allowed}:
\begin{mathparpagebreakable}
  \inferrule{\IdRules{1} \\ \wrTy{γ}{Γ}{p}{A} \\\\
    \wrTm{γ}{Γ}{t}{p}{A} \\ \wrTm{γ}{Γ}{u}{p}{A}}{%
    \wrTy{γ}{Γ}{p}{\IdT{A}{t}{u}}}
  \and
  \inferrule{\IdRules{1} \\ \wrTm{γ}{Γ}{A}{p}{\U{l}} \\\\
    \wrTm{γ}{Γ}{t}{p}{A} \\ \wrTm{γ}{Γ}{u}{p}{A}}{%
    \wrTm{γ}{Γ}{\IdT{A}{t}{u}}{p}{\U{l}}}
  \and
  \inferrule{\IdRules{2} \\ \wfTy{Γ}{A} \\\\
    \wfTm{Γ}{t}{A} \\ \wfTm{Γ}{u}{A}}{%
    \wrTy{γ}{Γ}{p}{\IdT{A}{t}{u}}}
  \and
  \inferrule{\IdRules{2} \\ \wfTm{Γ}{A}{\U{l}} \\\\
    \wfTm{Γ}{t}{A} \\ \wfTm{Γ}{u}{A}}{%
    \wrTm{γ}{Γ}{\IdT{A}{t}{u}}{p}{\U{l}}}
  \and
  \inferrule{\eqTy{Γ}{A₁}{A₂} \\\\ \eqTm{Γ}{t₁}{t₂}{A₁} \\
    \eqTm{Γ}{u₁}{u₂}{A₁}}{%
    \eqTy{Γ}{\IdT{A₁}{t₁}{u₁}}{\IdT{A₂}{t₂}{u₂}}}
  \and
  \inferrule{\eqTm{Γ}{A₁}{A₂}{\U{l}} \\\\ \eqTm{Γ}{t₁}{t₂}{A₁} \\
    \eqTm{Γ}{u₁}{u₂}{A₁}}{%
    \eqTm{Γ}{\IdT{A₁}{t₁}{u₁}}{\IdT{A₂}{t₂}{u₂}}{\U{l}}}
  \and
  \inferrule{\wfTm{Γ}{t}{A}}{\wrTm{γ}{Γ}{\rfl}{p}{\IdT{A}{t}{t}}}
  \and
  \inferrule{
    \text{$\JRules{1}$ or, if $p$ or $q$ is $\omegaG$, $\JRules{2}$} \\
    \wfTy{Γ}{A} \\ \wrTm{γ}{Γ}{t}{\omegaG}{A} \\
    \wrTy{\consGC{\consGC{γ}{p}}{q}}{%
      \consC{\consC{Γ}{A}}{\IdT{(\wkO{A})}{(\wkO{t})}{\varZ}}}{%
      \omegaG}{B} \\\\
    \wrTm{γ}{Γ}{u}{\omegaG}{\substO{B}{t}{\rfl}} \\
    \wrTm{γ}{Γ}{v}{\omegaG}{A} \\
    \wrTm{γ}{Γ}{w}{\omegaG}{\IdT{A}{t}{v}}}{%
    \wrTm{γ}{Γ}{\J{p}{q}{A}{t}{B}{u}{v}{w}}{\omegaG}{\substO{B}{v}{w}}}
  \and
  \inferrule{\JRules{2} \\ \wfTy{Γ}{A} \\ \wfTm{Γ}{t}{A} \\
    \wrTy{\consGC{\consGC{γ}{\zeroG}}{\zeroG}}{%
      \consC{\consC{Γ}{A}}{\IdT{(\wkO{A})}{(\wkO{t})}{\varZ}}}{%
      \omegaG}{B} \\\\
    \wrTm{γ}{Γ}{u}{\omegaG}{\substO{B}{t}{\rfl}} \\
    \wfTm{Γ}{v}{A} \\
    \wfTm{Γ}{w}{\IdT{A}{t}{v}}}{%
    \wrTm{γ}{Γ}{\J{\zeroG}{\zeroG}{A}{t}{B}{u}{v}{w}}{\omegaG}{%
      \substO{B}{v}{w}}}
  \and
  \inferrule{\text{If $r=\omegaG$, $\JRules{3}$} \\ \wfTy{Γ}{A} \\
    \wfTm{Γ}{t}{A} \\
    \wfTy{\consC{\consC{Γ}{A}}{\IdT{(\wkO{A})}{(\wkO{t})}{\varZ}}}{B} \\\\
    \wrTm{γ}{Γ}{u}{r}{\substO{B}{t}{\rfl}} \\
    \wfTm{Γ}{v}{A} \\
    \wfTm{Γ}{w}{\IdT{A}{t}{v}}}{%
    \wrTm{γ}{Γ}{\J{p}{q}{A}{t}{B}{u}{v}{w}}{r}{%
      \substO{B}{v}{w}}}
  \and
  \inferrule{\eqTy{Γ}{A₁}{A₂} \\ \eqTm{Γ}{t₁}{t₂}{A₁} \\
    \eqTy{\consC{\consC{Γ}{A₁}}{\IdT{(\wkO{A₁})}{(\wkO{t₁})}{\varZ}}}{B₁}{
      B₂} \\\\
    \eqTm{Γ}{u₁}{u₂}{\substO{B₁}{t₁}{\rfl}} \\
    \eqTm{Γ}{v₁}{v₂}{A₁} \\
    \eqTm{Γ}{w₁}{w₂}{\IdT{A₁}{t₁}{v₁}}}{%
    \eqTm{Γ}{\J{p}{q}{A₁}{t₁}{B₁}{u₁}{v₁}{w₁}}{
      \J{p}{q}{A₂}{t₂}{B₂}{u₂}{v₂}{w₂}}{\substO{B₁}{v₁}{w₁}}}
  \and
  \inferrule{\wfTm{Γ}{t}{A} \\
    \wfTy{\consC{\consC{Γ}{A}}{\IdT{(\wkO{A})}{(\wkO{t})}{\varZ}}}{B} \\
    \wfTm{Γ}{u}{\substO{B}{t}{\rfl}}}{
    \eqTm{Γ}{\J{p}{q}{A}{t}{B}{u}{t}{\rfl}}{u}{\substO{B}{t}{\rfl}}}
  \and
  \inferrule{\text{K is allowed} \\
    \text{$\KRules{1}$ or, if $p$ is $\omegaG$, $\KRules{2}$} \\
    \wfTy{Γ}{A} \\ \wrTm{γ}{Γ}{t}{\omegaG}{A} \\\\
    \wrTy{\consGC{γ}{p}}{\consC{Γ}{\IdT{A}{t}{t}}}{\omegaG}{B} \\
    \wrTm{γ}{Γ}{u}{\omegaG}{\substZ{B}{\rfl}} \\
    \wrTm{γ}{Γ}{v}{\omegaG}{\IdT{A}{t}{t}}}{%
    \wrTm{γ}{Γ}{\K{p}{A}{t}{B}{u}{v}}{\omegaG}{\substZ{B}{v}}}
  \and
  \inferrule{\text{K is allowed} \\ \KRules{2} \\
    \wfTy{Γ}{A} \\ \wfTm{Γ}{t}{A} \\\\
    \wrTy{\consGC{γ}{\zeroG}}{\consC{Γ}{\IdT{A}{t}{t}}}{\omegaG}{B} \\
    \wrTm{γ}{Γ}{u}{\omegaG}{\substZ{B}{\rfl}} \\
    \wfTm{Γ}{v}{\IdT{A}{t}{t}}}{%
    \wrTm{γ}{Γ}{\K{\zeroG}{A}{t}{B}{u}{v}}{\omegaG}{\substZ{B}{v}}}
  \and
  \inferrule{\text{K is allowed} \\
    \text{If $q=\omegaG$, $\KRules{3}$} \\
    \wfTy{Γ}{A} \\ \wfTm{Γ}{t}{A} \\\\
    \wfTy{\consC{Γ}{\IdT{A}{t}{t}}}{B} \\
    \wrTm{γ}{Γ}{u}{q}{\substZ{B}{\rfl}} \\
    \wfTm{Γ}{v}{\IdT{A}{t}{t}}}{%
    \wrTm{γ}{Γ}{\K{p}{A}{t}{B}{u}{v}}{q}{\substZ{B}{v}}}
  \and
  \inferrule{\text{K is allowed} \\
    \eqTy{Γ}{A₁}{A₂} \\ \eqTm{Γ}{t₁}{t₂}{A₁} \\\\
    \eqTy{\consC{Γ}{\IdT{A₁}{t₁}{t₁}}}{B₁}{B₂} \\
    \eqTm{Γ}{u₁}{u₂}{\substZ{B₁}{\rfl}} \\
    \eqTm{Γ}{v₁}{v₂}{\IdT{A₁}{t₁}{t₁}}}{%
    \eqTm{Γ}{\K{p}{A₁}{t₁}{B₁}{u₁}{v₁}}{\K{p}{A₂}{t₂}{B₂}{u₂}{v₂}}{
      \substZ{B₁}{v₁}}}
  \and
  \inferrule{\text{K is allowed} \\
    \wfTy{\consC{Γ}{\IdT{A}{t}{t}}}{B} \\
    \wfTm{Γ}{u}{\substZ{B}{\rfl}}}{
    \eqTm{Γ}{\K{p}{A}{t}{B}{u}{\rfl}}{u}{\substZ{B}{\rfl}}}
  \and
  \inferrule{\text{$\bcSWith{s}$ is allowed} \\
    \wfTy{Γ}{A} \\ \wfTm{Γ}{t}{A} \\ \wfTm{Γ}{u}{A} \\
    \wfTm{Γ}{v}{\IdT{A}{t}{u}}}{%
    \wrTm{γ}{Γ}{\bcS{s}{l}{A}{t}{u}{v}}{p}{
      \IdT{(\ES{s}{l}{A})}{\bxS{s}{t}}{\bxS{s}{u}}}}
  \and
  \inferrule{\text{$\bcSWith{s}$ is allowed} \\
    \eqTy{Γ}{A₁}{A₂} \\ \eqTm{Γ}{t₁}{t₂}{A₁} \\
    \eqTm{Γ}{u₁}{u₂}{A₁} \\ \eqTm{Γ}{v₁}{v₂}{\IdT{A₁}{t₁}{u₁}}}{%
    \eqTm{Γ}{\bcS{s}{l}{A₁}{t₁}{u₁}{v₁}}{\bcS{s}{l}{A₂}{t₂}{u₂}{v₂}}{%
      \IdT{(\ES{s}{l}{A₁})}{\bxS{s}{t₁}}{\bxS{s}{u₁}}}}
  \and
  \inferrule{\text{$\bcSWith{s}$ is allowed} \\ \wfTm{Γ}{t}{A}}{%
    \eqTm{Γ}{\bcS{s}{l}{A}{t}{t}{\rfl}}{\rfl}{
      \IdT{(\ES{s}{l}{A})}{\bxS{s}{t}}{\bxS{s}{t}}}}
  \and
  \inferrule{\text{Equality reflection is allowed} \\\\
    \wfTm{Γ}{v}{\IdT{A}{t}{u}}}{
    \eqTm{Γ}{t}{u}{A}}
\end{mathparpagebreakable}

Rules for set quotients.
The point constructor is denoted by $\className$.
In the typing rules for $\qrecName$ the abbreviation
$\respType{A}{B}{C}{t}$ stands for the following term:
\begin{equation*}
  \begin{pmboxed}
    \>\IdT{(\substZUpn{3}{C}{\var{1}})}{}{}\\
    \>\quad\>(\J{\omegaG}{\zeroG}{(\wkn{3}{(\Quot{A}{B})})}{(\class{\var{2}})}{%
      (\wkO{(\substZUpn{4}{C}{\var{0}})})}{(\wkT{t})}{(\class{\var{1}})}{%
      (\resp{(\wkn{3}{A})}{(\wklift{3}{2}{B})}{\var{2}}{\var{1}}{\var{0}})})\\
    \>\>(\substZUpn{3}{t}{\var{1}})
  \end{pmboxed}
\end{equation*}
Note that one can choose whether or not to allow quotient types:
\begin{mathparpagebreakable}
  \inferrule{\text{Quotient types are allowed} \\\\ \wrTy{γ}{Γ}{p}{A} \\
    \wfTy{\consC{\consC{Γ}{A}}{\wkO{A}}}{B}}{%
    \wrTy{γ}{Γ}{p}{\Quot{A}{B}}}
  \and
  \inferrule{\text{Quotient types are allowed} \\\\
    \wrTm{γ}{Γ}{A}{p}{\U{l}} \\
    \wfTm{\consC{\consC{Γ}{A}}{\wkO{A}}}{B}{\U{l}}}{%
    \wrTm{γ}{Γ}{\Quot{A}{B}}{p}{\U{l}}}
  \and
  \inferrule{\text{Quotient types are allowed} \\\\ \eqTy{Γ}{A₁}{A₂} \\
    \eqTy{\consC{\consC{Γ}{A₁}}{\wkO{A₁}}}{B₁}{B₂}}{%
    \eqTy{Γ}{\Quot{A₁}{B₁}}{\Quot{A₂}{B₂}}}
  \and
  \inferrule{\text{Quotient types are allowed} \\\\
    \eqTm{Γ}{A₁}{A₂}{\U{l}} \\
    \eqTm{\consC{\consC{Γ}{A₁}}{\wkO{A₁}}}{B₁}{B₂}{\U{l}}}{%
    \eqTm{Γ}{\Quot{A₁}{B₁}}{\Quot{A₂}{B₂}}{\U{l}}}
  \and
  \inferrule{\text{Quotient types are allowed} \\\\
    \wfTy{Γ}{\Quot{A}{B}} \\ \wrTm{γ}{Γ}{t}{p}{A}}{%
    \wrTm{γ}{Γ}{\class{t}}{p}{\Quot{A}{B}}}
  \and
  \inferrule{\wfTy{Γ}{\Quot{A}{B}} \\ \eqTm{Γ}{t₁}{t₂}{A}}{%
    \eqTm{Γ}{\class{t₁}}{\class{t₂}}{\Quot{A}{B}}}
  \and
  \inferrule{\wfTy{Γ}{\Quot{A}{B}} \\ \wfTm{Γ}{t}{A} \\
    \wfTm{Γ}{u}{A} \\ \wfTm{Γ}{v}{\substO{B}{t}{u}}}{%
    \wrTm{γ}{Γ}{\resp{A}{B}{t}{u}{v}}{\zeroG}{%
      \IdT{(\Quot{A}{B})}{(\class{t})}{(\class{u})}}}
  \and
  \inferrule{\text{Quotient types are allowed} \\ \eqTy{Γ}{A₁}{A₂} \\
    \eqTy{\consC{\consC{Γ}{A₁}}{\wkO{A₁}}}{B₁}{B₂} \\
    \eqTm{Γ}{t₁}{t₂}{A₁} \\ \eqTm{Γ}{u₁}{u₂}{A₁} \\
    \eqTm{Γ}{v₁}{v₂}{\substO{B₁}{t₁}{u₁}}}{%
    \eqTm{Γ}{\resp{A₁}{B₁}{t₁}{u₁}{v₁}}{\resp{A₂}{B₂}{t₂}{u₂}{v₂}}{%
      \IdT{(\Quot{A₁}{B₁})}{(\class{t₁})}{(\class{u₁})}}}
  \and
  \inferrule{\wfTy{Γ}{\Quot{A}{B}} \\ \wfTm{Γ}{t}{\Quot{A}{B}} \\
    \wfTm{Γ}{u}{\Quot{A}{B}} \\ \wfTm{Γ}{v}{\IdT{(\Quot{A}{B})}{t}{u}} \\
    \wfTm{Γ}{w}{\IdT{(\Quot{A}{B})}{t}{u}}}{%
    \wrTm{γ}{Γ}{\set{A}{B}{t}{u}{v}{w}}{\zeroG}{%
      \IdT{(\IdT{(\Quot{A}{B})}{t}{u})}{v}{w}}}
  \and
  \inferrule{\eqTy{Γ}{A₁}{A₂} \\
    \eqTy{\consC{\consC{Γ}{A₁}}{\wkO{A₁}}}{B₁}{B₂} \\
    \eqTm{Γ}{t₁}{t₂}{\Quot{A₁}{B₁}} \\ \eqTm{Γ}{u₁}{u₂}{\Quot{A₁}{B₁}} \\
    \eqTm{Γ}{v₁}{v₂}{\IdT{(\Quot{A₁}{B₁})}{t₁}{u₁}} \\
    \eqTm{Γ}{w₁}{w₂}{\IdT{(\Quot{A₁}{B₁})}{t₁}{u₁}}}{%
    \eqTm{Γ}{\set{A₁}{B₁}{t₁}{u₁}{v₁}{w₁}}{\set{A₂}{B₂}{t₂}{u₂}{v₂}{w₂}}{%
      \IdT{(\IdT{(\Quot{A₁}{B₁})}{t₁}{u₁})}{v₁}{w₁}}}
  \and
  \inferrule{\text{Quotient types are allowed} \\
    \wrTy{\consGC{γ}{\omegaG}}{\consC{Γ}{\Quot{A}{B}}}{p}{C} \\\\
    \wrTm{\consGC{γ}{\omegaG}}{\consC{Γ}{A}}{t}{p}{%
      \substZUp{C}{\class{\var{0}}}} \\
    \wfTm{\consC{\consC{\consC{Γ}{A}}{\wkO{A}}}{B}}{u}{%
      \respType{A}{B}{C}{t}} \\
    \wfTm{\consC{\consC{\consC{\consC{\consC{Γ}{\Quot{A}{B}}}{C}}{%
            \wkO{C}}}{%
          \IdT{(\wkT{C})}{\var{1}}{\var{0}}}}{%
        \IdT{(\wkn{3}{C})}{\var{2}}{\var{1}}}}{v}{%
      \IdT{(\IdT{(\wkn{4}{C})}{\var{3}}{\var{2}})}{\var{1}}{\var{0}}} \\
    \wrTm{γ}{Γ}{w}{p}{\Quot{A}{B}}}{%
    \wrTm{γ}{Γ}{\qrec{C}{t}{u}{v}{w}}{p}{\substZ{C}{w}}}
  \and
  \inferrule{\eqTy{\consC{Γ}{\Quot{A}{B}}}{C₁}{C₂} \\
    \eqTm{\consC{Γ}{A}}{t₁}{t₂}{\substZUp{C₁}{\class{\var{0}}}} \\
    \eqTm{\consC{\consC{\consC{Γ}{A}}{\wkO{A}}}{B}}{u₁}{u₂}{%
      \respType{A}{B}{C₁}{t₁}} \\
    \eqTm{\consC{\consC{\consC{\consC{\consC{Γ}{\Quot{A}{B}}}{C₁}}{%
            \wkO{C₁}}}{%
          \IdT{(\wkT{C₁})}{\var{1}}{\var{0}}}}{%
        \IdT{(\wkn{3}{C₁})}{\var{2}}{\var{1}}}}{v₁}{v₂}{%
      \IdT{(\IdT{(\wkn{4}{C₁})}{\var{3}}{\var{2}})}{\var{1}}{\var{0}}} \\
    \eqTm{Γ}{w₁}{w₂}{\Quot{A}{B}}}{%
    \eqTm{Γ}{\qrec{C₁}{t₁}{u₁}{v₁}{w₁}}{\qrec{C₂}{t₂}{u₂}{v₂}{w₂}}{%
      \substZ{C₁}{w₁}}}
  \and
  \inferrule{\wfTy{\consC{Γ}{\Quot{A}{B}}}{C} \\\\
    \wfTm{\consC{Γ}{A}}{t}{\substZUp{C}{\class{\var{0}}}} \\
    \wfTm{\consC{\consC{\consC{Γ}{A}}{\wkO{A}}}{B}}{u}{%
      \respType{A}{B}{C}{t}} \\
    \wfTm{\consC{\consC{\consC{\consC{\consC{Γ}{\Quot{A}{B}}}{C}}{%
            \wkO{C}}}{%
          \IdT{(\wkT{C})}{\var{1}}{\var{0}}}}{%
        \IdT{(\wkn{3}{C})}{\var{2}}{\var{1}}}}{v}{%
      \IdT{(\IdT{(\wkn{4}{C})}{\var{3}}{\var{2}})}{\var{1}}{\var{0}}} \\
    \wfTm{Γ}{w}{A}}{%
    \eqTm{Γ}{\qrec{C}{t}{u}{v}{(\class{w})}}{\substZ{t}{w}}{%
      \substZ{C}{\class{w}}}}
\end{mathparpagebreakable}

\subsection{Substitutions}

The definition of parallel substitutions and the function applying a
substitution to a term are not included here, but the following
relation for well-formed substitutions is:
\begin{mathparpagebreakable}
  \inferrule{\wfCtxt{Δ}}{\wrSubst{δ}{Δ}{σ}{p}{\emptyC}}
  \and
  \inferrule{\wrSubst{δ}{Δ}{\tail{σ}}{p}{Γ} \\
    \wrTm{δ}{Δ}{\head{σ}}{p}{\subst{A}{(\tail{σ})}}}{%
    \wrSubst{δ}{Δ}{σ}{p}{\consC{Γ}{A}}}
\end{mathparpagebreakable}
For a substitution $σ$ from a non-empty context $\head{σ}$ refers to
the last term in the substitution, and $\tail{σ}$ refers to the rest
of the substitution.

\subsection{Reduction}

Single-step reduction for terms.
Note that the β-rules for $\JName$ and $\bcName$ have assumptions that
are judgemental equalities:
\begin{mathparpagebreakable}
  \inferrule{\redTm{Γ}{t}{u}{A} \\ \eqTy{Γ}{A}{B}}{\redTm{Γ}{t}{u}{B}}
  \and
  \inferrule{\redTm{Γ}{t₁}{t₂}{\LiftT{l}{A}}}{
    \redTm{Γ}{\lowerTm{t₁}}{\lowerTm{t₂}}{A}}
  \and
  \inferrule{\wfTy{Γ}{A} \\ \wfTm{Γ}{t}{A}}{
    \redTm{Γ}{\lowerTm{(\lift{t})}}{t}{A}}
  \and
  \inferrule{\wfTy{Γ}{A} \\ \redTm{Γ}{t₁}{t₂}{\Empty}}{
    \redTm{Γ}{\emptyrec{p}{A}{t₁}}{\emptyrec{p}{A}{t₂}}{A}}
  \and
  \inferrule{\wfTy{\consC{Γ}{\Unit{\weak}}}{A} \\
    \redTm{Γ}{t₁}{t₂}{\Unit{\weak}} \\\\
    \wfTm{Γ}{u}{\substZ{A}{\starTm{\weak}}}}{
    \redTm{Γ}{\unitrec{p}{A}{t₁}{u}}{\unitrec{p}{A}{t₂}{u}}{
      \substZ{A}{t₁}}}
  \and
  \inferrule{\wfTy{\consC{Γ}{\Unit{\weak}}}{A} \\
    \wfTm{Γ}{t}{\substZ{A}{\starTm{\weak}}}}{
    \redTm{Γ}{\unitrec{p}{A}{\starTm{\weak}}{t}}{t}{
      \substZ{A}{\starTm{\weak}}}}
  \and
  \inferrule{\redTm{Γ}{t₁}{t₂}{\PiT{p}{A}{B}} \\
    \wfTm{Γ}{u}{A}}{%
    \redTm{Γ}{\app{t₁}{p}{u}}{\app{t₂}{p}{u}}{\substZ{B}{u}}}
  \and
  \inferrule{
    \excludeAllowed{\text{$\PiWithLevel{p}$ is allowed} \\}
    \wfTy{\consC{Γ}{A}}{B} \\\\
    \wfTm{\consC{Γ}{A}}{t}{B} \\ \wfTm{Γ}{u}{A}}{%
    \redTm{Γ}{\app{(\lam{p}{t})}{p}{u}}{\substZ{t}{u}}{\substZ{B}{u}}}
  \and
  \inferrule{\wfTy{\consC{Γ}{A}}{B} \\
    \redTm{Γ}{t₁}{t₂}{\SigmaT{\strongS}{p}{A}{B}}}{%
    \redTm{Γ}{\fst{p}{t₁}}{\fst{p}{t₂}}{A}}
  \and
  \inferrule{
    \excludeAllowed{\text{$\SigmaWith{\strongS}{p}$ is allowed} \\}
    \wfTy{\consC{Γ}{A}}{B} \\\\
    \wfTm{Γ}{t}{A} \\ \wfTm{Γ}{u}{\substZ{B}{t}}}{%
    \redTm{Γ}{\fst{p}{\prodTm{\strongS}{p}{t}{u}}}{t}{A}}
  \and
  \inferrule{\wfTy{\consC{Γ}{A}}{B} \\
    \redTm{Γ}{t₁}{t₂}{\SigmaT{\strongS}{p}{A}{B}}}{%
    \redTm{Γ}{\snd{p}{t₁}}{\snd{p}{t₂}}{\substZ{B}{\fst{p}{t₁}}}}
  \and
  \inferrule{
    \excludeAllowed{\text{$\SigmaWith{\strongS}{p}$ is allowed} \\}
    \wfTy{\consC{Γ}{A}}{B} \\\\
    \wfTm{Γ}{t}{A} \\ \wfTm{Γ}{u}{\substZ{B}{t}}}{%
    \redTm{Γ}{\snd{p}{\prodTm{\strongS}{p}{t}{u}}}{u}{\substZ{B}{t}}}
  \and
  \inferrule{\wfTy{\consC{Γ}{\SigmaT{\weak}{p}{A}{B}}}{C} \\
    \redTm{Γ}{t₁}{t₂}{\SigmaT{\weak}{p}{A}{B}} \\
    \wfTm{\consC{\consC{Γ}{A}}{B}}{u}{
      \substZUpT{C}{\prodTm{\weak}{p}{\varO}{\varZ}}}}{%
    \redTm{Γ}{\prodrec{p}{q}{C}{t₁}{u}}{\prodrec{p}{q}{C}{t₂}{u}}{
      \substZ{C}{t₁}}}
  \and
  \inferrule{\wfTy{\consC{Γ}{\SigmaT{\weak}{p}{A}{B}}}{C} \\
    \wfTm{Γ}{t}{A} \\ \wfTm{Γ}{u}{\substZ{B}{t}} \\
    \wfTm{\consC{\consC{Γ}{A}}{B}}{v}{
      \substZUpT{C}{\prodTm{\weak}{p}{\varO}{\varZ}}}}{%
    \redTm{Γ}{\prodrec{p}{q}{C}{\prodTm{\weak}{p}{t}{u}}{v}}{
      \substO{v}{t}{u}}{\substZ{C}{\prodTm{\weak}{p}{t}{u}}}}
  \and
  \inferrule{\wfTm{Γ}{t}{\substZ{A}{\zero}} \\
    \wfTm{\consC{\consC{Γ}{\Nat}}{A}}{u}{\substZUpT{A}{\suc{\varO}}} \\
    \redTm{Γ}{v₁}{v₂}{\Nat}}{%
    \redTm{Γ}{\natrec{p}{q}{A}{t}{u}{v₁}}{
      \natrec{p}{q}{A}{t}{u}{v₂}}{\substZ{A}{v₁}}}
  \and
  \inferrule{\wfTm{Γ}{t}{\substZ{A}{\zero}} \\
    \wfTm{\consC{\consC{Γ}{\Nat}}{A}}{u}{\substZUpT{A}{\suc{\varO}}}}{%
    \redTm{Γ}{\natrec{p}{q}{A}{t}{u}{\zero}}{t}{\substZ{A}{\zero}}}
  \and
  \inferrule{\wfTm{Γ}{t}{\substZ{A}{\zero}} \\
    \wfTm{\consC{\consC{Γ}{\Nat}}{A}}{u}{\substZUpT{A}{\suc{\varO}}} \\
    \wfTm{Γ}{v}{\Nat}}{%
    \redTm{Γ}{\natrec{p}{q}{A}{t}{u}{(\suc{v})}}{
      \substO{u}{v}{\natrec{p}{q}{A}{t}{u}{v}}}{\substZ{A}{\suc{v}}}}
  \and
  \inferrule{\wfTm{Γ}{t}{A} \\
    \wfTy{\consC{\consC{Γ}{A}}{\IdT{(\wkO{A})}{(\wkO{t})}{\varZ}}}{B} \\
    \wfTm{Γ}{u}{\substO{B}{t}{\rfl}} \\
    \wfTm{Γ}{v}{A} \\
    \redTm{Γ}{w₁}{w₂}{\IdT{A}{t}{v}}}{%
    \redTm{Γ}{\J{p}{q}{A}{t}{B}{u}{v}{w₁}}{
      \J{p}{q}{A}{t}{B}{u}{v}{w₂}}{\substO{B}{v}{w₁}}}
  \and
  \inferrule{\wfTm{Γ}{t₁}{A} \\ \wfTm{Γ}{t₂}{A} \\
    \eqTm{Γ}{t₁}{t₂}{A} \\
    \wfTy{\consC{\consC{Γ}{A}}{\IdT{(\wkO{A})}{(\wkO{t₁})}{\varZ}}}{B} \\
    \eqTy{Γ}{\substO{B}{t₁}{\rfl}}{\substO{B}{t₂}{\rfl}} \\
    \wfTm{Γ}{u}{\substO{B}{t₁}{\rfl}}}{
    \redTm{Γ}{\J{p}{q}{A}{t₁}{B}{u}{t₂}{\rfl}}{u}{\substO{B}{t₁}{\rfl}}}
  \and
  \inferrule{\text{K is allowed} \\
    \wfTy{\consC{Γ}{\IdT{A}{t}{t}}}{B} \\\\
    \wfTm{Γ}{u}{\substZ{B}{\rfl}} \\
    \redTm{Γ}{v₁}{v₂}{\IdT{A}{t}{t}}}{%
    \redTm{Γ}{\K{p}{A}{t}{B}{u}{v₁}}{\K{p}{A}{t}{B}{u}{v₂}}{
      \substZ{B}{v₁}}}
  \and
  \inferrule{\text{K is allowed} \\
    \wfTy{\consC{Γ}{\IdT{A}{t}{t}}}{B} \\\\
    \wfTm{Γ}{u}{\substZ{B}{\rfl}}}{
    \redTm{Γ}{\K{p}{A}{t}{B}{u}{\rfl}}{u}{\substZ{B}{\rfl}}}
  \and
  \inferrule{\text{$\bcSWith{s}$ is allowed} \\
    \redTm{Γ}{v₁}{v₂}{\IdT{A}{t}{u}}}{%
    \redTm{Γ}{\bcS{s}{l}{A}{t}{u}{v₁}}{\bcS{s}{l}{A}{t}{u}{v₂}}{%
      \IdT{(\ES{s}{l}{A})}{\bxS{s}{t}}{\bxS{s}{u}}}}
  \and
  \inferrule{\text{$\bcSWith{s}$ is allowed} \\ \eqTm{Γ}{t₁}{t₂}{A}}{%
    \redTm{Γ}{\bcS{s}{l}{A}{t₁}{t₂}{\rfl}}{\rfl}{
      \IdT{(\ES{s}{l}{A})}{\bxS{s}{t₁}}{\bxS{s}{t₂}}}}
  \and
  \inferrule{\text{Equality reflection is allowed} \\
    \wfTy{Γ}{\Quot{A}{B}} \\ \wfTm{Γ}{t}{A} \\
    \wfTm{Γ}{u}{A} \\ \wfTm{Γ}{v}{\substO{B}{t}{u}}}{%
    \redTm{Γ}{\resp{A}{B}{t}{u}{v}}{\rfl}{%
      \IdT{(\Quot{A}{B})}{(\class{t})}{(\class{u})}}}
  \and
  \inferrule{\text{Equality reflection is allowed} \\
    \wfTm{Γ}{t}{\Quot{A}{B}} \\ \wfTm{Γ}{u}{\Quot{A}{B}} \\
    \wfTm{Γ}{v}{\IdT{(\Quot{A}{B})}{t}{u}} \\
    \wfTm{Γ}{w}{\IdT{(\Quot{A}{B})}{t}{u}}}{%
    \redTm{Γ}{\set{A}{B}{t}{u}{v}{w}}{\rfl}{%
      \IdT{(\IdT{(\Quot{A}{B})}{t}{u})}{v}{w}}}
  \and
  \inferrule{\wfTy{\consC{Γ}{\Quot{A}{B}}}{C} \\\\
    \wfTm{\consC{Γ}{A}}{t}{\substZUp{C}{\class{\var{0}}}} \\
    \wfTm{\consC{\consC{\consC{Γ}{A}}{\wkO{A}}}{B}}{u}{%
      \respType{A}{B}{C}{t}} \\
    \wfTm{\consC{\consC{\consC{\consC{\consC{Γ}{\Quot{A}{B}}}{C}}{%
            \wkO{C}}}{%
          \IdT{(\wkT{C})}{\var{1}}{\var{0}}}}{%
        \IdT{(\wkn{3}{C})}{\var{2}}{\var{1}}}}{v}{%
      \IdT{(\IdT{(\wkn{4}{C})}{\var{3}}{\var{2}})}{\var{1}}{\var{0}}} \\
    \redTm{Γ}{w₁}{w₂}{\Quot{A}{B}}}{%
    \redTm{Γ}{\qrec{C}{t}{u}{v}{w₁}}{\qrec{C}{t}{u}{v}{w₂}}{\substZ{C}{w₁}}}
  \and
  \inferrule{\wfTy{\consC{Γ}{\Quot{A}{B}}}{C} \\\\
    \wfTm{\consC{Γ}{A}}{t}{\substZUp{C}{\class{\var{0}}}} \\
    \wfTm{\consC{\consC{\consC{Γ}{A}}{\wkO{A}}}{B}}{u}{%
      \respType{A}{B}{C}{t}} \\
    \wfTm{\consC{\consC{\consC{\consC{\consC{Γ}{\Quot{A}{B}}}{C}}{%
            \wkO{C}}}{%
          \IdT{(\wkT{C})}{\var{1}}{\var{0}}}}{%
        \IdT{(\wkn{3}{C})}{\var{2}}{\var{1}}}}{v}{%
      \IdT{(\IdT{(\wkn{4}{C})}{\var{3}}{\var{2}})}{\var{1}}{\var{0}}} \\
    \wfTm{Γ}{w}{A}}{%
    \redTm{Γ}{\qrec{C}{t}{u}{v}{(\class{w})}}{\substZ{t}{w}}{%
      \substZ{C}{\class{w}}}}
\end{mathparpagebreakable}

Multi-step reduction for terms:
\begin{mathparpagebreakable}
  \inferrule{\wfTm{Γ}{t}{A}}{\redsTm{Γ}{t}{t}{A}}
  \and
  \inferrule{\redTm{Γ}{t}{u}{A} \\ \redsTm{Γ}{u}{v}{A}}{
    \redsTm{Γ}{t}{v}{A}}
\end{mathparpagebreakable}

Single-step reduction for types:
\begin{mathparpagebreakable}
  \inferrule{\redTm{Γ}{A}{B}{\U{l}}}{\redTy{Γ}{A}{B}}
\end{mathparpagebreakable}

Multi-step reduction for types:
\begin{mathparpagebreakable}
  \inferrule{\wfTy{Γ}{A}}{\redsTy{Γ}{A}{A}}
  \and
  \inferrule{\redTy{Γ}{A}{B} \\ \redsTy{Γ}{B}{C}}{\redsTy{Γ}{A}{C}}
\end{mathparpagebreakable}

\ifAppendicesIncluded{}{\newpage}

A single-step reduction relation that allows evaluation under
successor constructors:
\begin{mathparpagebreakable}
  \inferrule{\redTm{Γ}{t}{u}{\Nat}}{\redSuc{Γ}{t}{u}}
  \and
  \inferrule{\redSuc{Γ}{t}{u}}{\redSuc{Γ}{\suc{t}}{\suc{u}}}
\end{mathparpagebreakable}

A multi-step reduction relation that allows evaluation under successor
constructors:
\begin{mathparpagebreakable}
  \inferrule{\wfTm{Γ}{t}{\Nat}}{\redsSuc{Γ}{t}{t}}
  \and
  \inferrule{\redSuc{Γ}{t}{u} \\ \redsSuc{Γ}{u}{v}}{\redsSuc{Γ}{t}{v}}
\end{mathparpagebreakable}

\section{The Target Language}
\label{sec:target-language}

This section contains more details about the target language and the
extraction function presented in the main text.
In addition to what is presented here the formalisation supports
non-strict applications and top-level definitions.

Syntax:
\begin{equation*}
  \begin{pmboxed}
    \>t, u, v ∷=\>~\> \dummy ~|~ \varT{x} ~|~ \starT ~|~ \unitrecT{t}{u} ~|~ \lamT{t} ~|~ \appT{t}{u} ~|~ \prodT{t}{u} ~|~ \fstT{t} ~|~ \sndT{t} ~|~ \prodrecT{t}{u}\\
    \>\hfill|\>\> \zeroName ~|~ \sucT{t} ~|~ \natrecT{t}{u}{v}
  \end{pmboxed}
\end{equation*}
The formalisation uses well-scoped syntax, but such details have been
omitted here: just like in Appendix \ref{sec:full-typing-rules} the
definitions should be read as if everything were well-scoped.

\newcommand{\Value}[1]{\mathsf{Value}\,#1}

Certain terms, including $\dummy$, are classified as \emph{values}:
\begin{mathparpagebreakable}
  \inferrule{ }{\Value{\dummy}}
  \and
  \inferrule{ }{\Value{\starT}}
  \and
  \inferrule{ }{\Value{(\lamT{t})}}
  \and
  \inferrule{ }{\Value{(\prodT{t}{u})}}
  \and
  \inferrule{ }{\Value{\zeroName}}
  \and
  \inferrule{ }{\Value{(\sucT{t})}}
\end{mathparpagebreakable}
Single-step reduction is then defined in the following way, using
substitution:
\begin{mathparpagebreakable}
  \inferrule{\step{t₁}{t₂}}{\step{\unitrecT{t₁}{u}}{\unitrecT{t₂}{u}}}
  \and
  \inferrule{ }{\step{\unitrecT{\starT}{t}}{t}}
  \and
  \inferrule{\step{t₁}{t₂}}{\step{\appT{t₁}{u}}{\appT{t₂}{u}}}
  \and
  \inferrule{\Value{t} \\ \step{u₁}{u₂}}{
    \step{\appT{t}{u₁}}{\appT{t}{u₂}}}
  \and
  \inferrule{\Value{u}}{\step{\appT{(\lamT{t})}{u}}{\substZ{t}{u}}}
  \and
  \inferrule{\step{t₁}{t₂}}{\step{\fstT{t₁}}{\fstT{t₂}}}
  \and
  \inferrule{ }{\step{\fstT{(\prodT{t}{u})}}{t}}
  \and
  \inferrule{\step{t₁}{t₂}}{\step{\sndT{t₁}}{\sndT{t₂}}}
  \and
  \inferrule{ }{\step{\sndT{(\prodT{t}{u})}}{u}}
  \and
  \inferrule{\step{t₁}{t₂}}{\step{\prodrecT{t₁}{u}}{\prodrecT{t₂}{u}}}
  \and
  \inferrule{ }{\step{\prodrecT{(\prodT{t}{u})}{v}}{\substO{v}{t}{u}}}
  \and
  \inferrule{\step{v₁}{v₂}}{
    \step{\natrecT{t}{u}{v₁}}{\natrecT{t}{u}{v₂}}}
  \and
  \inferrule{ }{\step{\natrecT{t}{u}{\zeroName}}{t}}
  \and
  \inferrule{ }{
    \step{\natrecT{t}{u}{(\sucT{v})}}{\substO{u}{v}{\natrecT{t}{u}{v}}}}
\end{mathparpagebreakable}
Multi-step reduction is defined as follows:
\begin{mathparpagebreakable}
  \inferrule{ }{\steps{t}{t}}
  \and
  \inferrule{\step{t}{u} \\ \steps{u}{v}}{\steps{t}{v}}
\end{mathparpagebreakable}

The call-by-value extraction function is defined in the following way
(the call-by-name variant can be found in the accompanying
code):
\begin{center}
  \begin{pboxed}
    \>[][@{}l@{}]$\extract{\var{x}}$                               \>[][@{}l@{}]$= \varT{x}$                                                        \\
    \>           $\extract{\U{l}}$                                 \>           $= \dummy$                                                         \\
    \>           $\extract{(\LiftT{l}{A})}$                        \>           $= \dummy$                                                         \\
    \>           $\extract{(\lift{t})}$                            \>           $= \extract{t}$                                                     \\
    \>           $\extract{(\lowerTm{t})}$                         \>           $= \extract{t}$                                                     \\
    \>           $\extract{\Empty}$                                \>           $= \dummy$                                                         \\
    \>           $\extract{(\emptyrec{p}{A}{t})}$                  \>           $= \loopT$                                                         \\
    \>           $\extract{\Unit{s}}$                              \>           $= \dummy$                                                         \\
    \>           $\extract{\starTm{s}}$                            \>           $= \starT$                                                         \\
    \>           $\extract{(\unitrec{\zeroG}{A}{t}{u})}$           \>           $= \extract{u}$                                                     \\
    \>           $\extract{(\unitrec{\omegaG}{A}{t}{u})}$          \>           $= \unitrecT{(\extract{t})}{(\extract{u})}$                         \\
    \>           $\extract{(\PiT{p}{A}{B})}$                       \>           $= \dummy$                                                         \\
    \>           $\extract{(\lam{p}{t})}$                          \>           $= \lamT{(\extract{t})}$                                            \\
    \>           $\extract{(\app{t}{\omegaG}{u})}$                 \>           $= \appT{(\extract{t})}{(\extract{u})}$                             \\
    \>           $\extract{(\app{t}{\zeroG}{u})}$                  \>           $= \appT{(\extract{t})}{\dummy}$                                    \\
    \>           $\extract{(\SigmaT{s}{p}{A}{B})}$                 \>           $= \dummy$                                                         \\
    \>           $\extract{\prodTm{s}{\zeroG}{t}{u}}$              \>           $= \extract{u}$                                                     \\
    \>           $\extract{\prodTm{s}{\omegaG}{t}{u}}$             \>           $=
                  \appT{(\appT{(\lamT{(\lamT{(\prodT{(\varT{1})}{(\varT{0})})})})}{(\extract{t})})}{(\extract{u})}$                               \\
    \>           $\extract{(\fst{\zeroG}{t})}$                     \>           $= \loopT$                                                         \\
    \>           $\extract{(\fst{\omegaG}{t})}$                    \>           $= \fstT{(\extract{t})}$                                            \\
    \>           $\extract{(\snd{\zeroG}{t})}$                     \>           $= \extract{t}$                                                     \\
    \>           $\extract{(\snd{\omegaG}{t})}$                    \>           $= \sndT{(\extract{t})}$                                            \\
    \>           $\extract{(\prodrec{p}{\zeroG}{A}{t}{u})}$        \>           $= \substO{(\extract{u})}{\loopT}{\loopT}$                          \\
    \>           $\extract{(\prodrec{\zeroG}{\omegaG}{A}{t}{u})}$  \>           $= \appT{(\substZ{(\lamT{(\extract{u})})}{\loopT})}{(\extract{t})}$ \\
    \>           $\extract{(\prodrec{\omegaG}{\omegaG}{A}{t}{u})}$ \>           $= \prodrecT{(\extract{t})}{(\extract{u})}$                         \\
    \>           $\extract{(\IdT{A}{t}{u})}$                       \>           $= \dummy$                                                         \\
    \>           $\extract{\rfl}$                                  \>           $= \dummy$                                                         \\
    \>           $\extract{(\J{p}{q}{A}{t}{B}{u}{v}{w})}$          \>           $= \extract{u}$                                                     \\
    \>           $\extract{(\K{p}{A}{t}{B}{u}{v})}$                \>           $= \extract{u}$                                                     \\
    \>           $\extract{(\bcS{s}{l}{A}{t}{u}{v})}$              \>           $= \dummy$                                                         \\
    \>           $\extract{\Nat}$                                  \>           $= \dummy$                                                         \\
    \>           $\extract{\zeroName}$                             \>           $= \zeroName$                                                      \\
    \>           $\extract{(\suc{t})}$                             \>           $= \appT{(\lamT{(\sucT{(\varT{0})})})}{(\extract{t})}$              \\
    \>           $\extract{(\natrec{p}{q}{A}{t}{u}{v})}$           \>           $= \natrecT{(\extract{t})}{(\extract{u})}{(\extract{v})}$          \\
    \>           $\extract{(\Quot{A}{B})}$                         \>           $= \dummy$                                                         \\
    \>           $\extract{(\class{t})}$                           \>           $= \extract{t}$                                                    \\
    \>           $\extract{(\resp{A}{B}{t}{u}{v})}$                \>           $= \loopT$                                                         \\
    \>           $\extract{(\set{A}{B}{t}{u}{v}{w})}$              \>           $= \loopT$                                                         \\
    \>           $\extract{(\qrec{C}{t}{u}{v}{w})}$                \>           $= \appT{(\lamT{(\extract{t})})}{(\extract{w})}$
  \end{pboxed}
\end{center}
}{}

\end{document}